\documentclass[opre,nonblindrev]{informs3}
\OneAndAHalfSpacedXI % current default line spacing
\usepackage{etex}
\usepackage{natbib}
 \bibpunct[, ]{(}{)}{,}{a}{}{,}%
 \def\bibfont{\small}%

\newtheorem{observation}{Observation}
\newtheorem{approximation}{Approximation}

\usepackage{soul}
\def\pj#1{{\bf{\color{brown} Patrick: #1}}} % some of the questions/changes by patrick for discussions
\def\vm#1{{\color{red}::#1}} % some of the changes by Vahideh for discussions
\def\vmn#1{{\color{black}#1}}

\newcommand{\remove}[1]{}
\renewcommand{\st}{\remove}

\usepackage[framemethod=tikz]{mdframed}
\usepackage{tabularx}
\usepackage{algorithm}
\usepackage{algorithmic}
\usepackage{pst-plot}
\usepackage{pstricks, pst-node, pst-circ}
\usepackage{mathtools}
\usepackage{tikz}
\usepackage{pgfplots}
\usepackage{bm}
\usepackage{enumerate,booktabs}%%
\usetikzlibrary{%
    graphs, graphs.standard, quotes, arrows, shapes,
decorations.fractals,%
decorations.shapes,%
decorations.text,%
decorations.pathmorphing,%
decorations.pathreplacing,%
decorations.footprints,%
decorations.markings}

\TheoremsNumberedThrough     % Preferred (Theorem 1, Lemma 1, Theorem 2)
\ECRepeatTheorems

\EquationsNumberedThrough    % Default: (1), (2), ...
\newcommand{\RG}{{\st{predictable}\vmn{stochastic}}}

\newcommand{\RGS}{{\vmn{\mathcal{S}}}}
\newcommand{\UPG}{\st{unpredictable}\vmn{adversarial}}
\newcommand{\UPGS}{\vmn{\mathcal{A}}}
\newcommand{\E}[1]{\mathbb{E}\left[#1\right]}

\newcommand{\prob}[1]{\mathbb{P}\left( #1 \right)}
\renewcommand{\l}{l}

\begin{document}
%%%%%%%%%%%%%%%%

%\setcounter{page}{1}
%\DOI{10.1287/xxxx.0000.0000}%

\RUNAUTHOR{Hwang et al.}

\RUNTITLE{Online Resource Allocation under Partially\st{Learnable} \vmn{Predictable} Demand}

% \TITLE{Bundling Information Goods of Decreasing Value}
\TITLE{Online Resource Allocation under Partially\st{Learnable} \vmn{Predictable} Demand}
%\thanks{Supported by NSF (grant 1029603) and ONR (grants N00014-12-1-0033 and N00014-15-1-2083).}

\ARTICLEAUTHORS{%
\AUTHOR{Dawsen Hwang}
\AFF{Google, Chicago, IL  60607, \EMAIL{dawsen@gmail.com}}
\AUTHOR{Patrick Jaillet}
\AFF{Department of Electrical Engineering and Computer Science, Laboratory for Information and Decision Systems, Operations Research Center, Massachusetts Institute of Technology, Cambridge, MA 02139, \EMAIL{jaillet@mit.edu}}
\AUTHOR{Vahideh Manshadi}
\AFF{Yale School of Management, New Haven, CT 06511, \EMAIL{vahideh.manshadi@yale.edu}}
}

\ABSTRACT{%
\noindent
For online resource allocation problems, we propose a new demand arrival model where the sequence of arrivals contains both an\st{unpredictable} \vmn{adversarial}
\st{ (non-stationary)}component and a\st{predictable} \vmn{stochastic} one\st{(coming from an  unknown stationary stochastic distribution)}. Our model requires no demand forecasting; however, due to the presence of the\st{predictable} \vmn{stochastic} component, we can partially\st{learn} \vmn{predict} future demand as the sequence of arrivals unfolds. Under the proposed model, we study \st{problem of single-leg revenue management  with $2$ fare classes with great details}\vmn{the problem of the online allocation of a single resource to two types of customers}, and design online algorithms that outperform existing ones.
Our algorithms are adjustable to the relative size of the\st{predictable} \vmn{stochastic} component, and our analysis reveals that as the portion of the\st{predictable} \vmn{stochastic} component grows, the  loss due to making online decisions decreases. This highlights the value of\st{stochastic information and learning} \vmn{(even partial) predictability} in online resource allocation.
\vmn{We impose no conditions on how the resource capacity scales with the maximum number of customers.}
\vmn{However,}\st{Further,} we show that using an adaptive algorithm---which makes online decisions based on observed data---is particularly beneficial when
\st{the initial resource capacity is of the same order of as time horizon maximum number of customers}\vmn{capacity scales linearly with the number of customers}.
Our work serves as a first step in bridging the long-standing gap between  the two well-studied approaches to the\st{fully robust analysis and stochastic modeling}
\vmn{design and analysis of online algorithms based on (1) adversarial models and (2) stochastic ones}.
Using novel algorithm design, we demonstrate that even if the\st{demand process} \vmn{arrival sequence} contains an\st{unpredictable} \vmn{adversarial} component,
we can take advantage of the limited information that the data reveals to improve\st{revenue of the firm} \vmn{allocation decisions}\st{upon a fully robust approach}.\st{At the same time, there is an unavoidable loss due to the presence of an unpredictable component that can be considered as the price of limited robustness.} \vmn{We also study the classical secretary problem under our proposed arrival model, and we show that randomizing over multiple stopping rules may increase the probability of success.}
}%

%\KEYWORDS{deterministic inventory theory; infinite linear programming duality;
%  existence of optimal policies; semi-Markov decision process; cyclic schedule}


\maketitle

\section{Introduction}
\label{sec:intro}

%\vm{I almost re-wrote this section, so would be nice if you could read again. I tried to highlight a few things which were missing before: (1) a bit more motivation for the model (2) this is the first hybrid model which is parameterized between 2 extremes (3) we don't need any forecast(4) the competitive ratio  is also parameterized  and also algorithms have nice piratical properties (5) emphasize the big message which is beyond alg. design.}

E-commerce platforms host markets for perishable resources in various industry sectors ranging from airlines to hotels to internet advertising.
In these markets, demand realizes sequentially, and the firms need to make online (irrevocable) decisions regarding how (and at what price) to allocate resources to arriving demand without precise knowledge of future demand.
The success of any online allocation algorithm  depends crucially on a firm's ability to predict future demand.
\vmn{If demand can be predicted\st{or learned}, then under some conditions on the amount of available resources, making online decisions incurs little loss (as shown in \cite{Agrawal2009b}, among others).} %, and \vm{find a reference from RM}\dawsen{Let's do it after EC.}).
However, in many markets, demand cannot be perfectly\st{forecast} \vmn{predicted} due to unpredictable components such as traffic spikes and strategy changes by competitors.
\st{Demand forecasting  has long been a challenge for applying traditional revenue management, specially in new industries} \st{and for new products.} The emergence of sharing-economy platforms such as Airbnb, which can
%offers a wider range of options and
scale supply at negligible cost and on short notice \citep{AirBnb},  has significantly added to\st{uncertainty in demand and its} \vmn{unpredictable} variability\st{across time} \vmn{in demand} even for products that are not new (e.g., existing hotels).
%More recently, emergence of sharing economy platforms such as Airbnb, that competes with traditional hotels industry, and offers a wider range of options and has the capability to scale supply with negligible cost and on short timescales \cite{AirBnb}, adds to significant uncertainty in demand and its variability across time even for products that are not new.
%
%\footnote{If the demand can be predicted or learned, then under some mild conditions on the amount of available resources, there is little loss incurred for making online decisions (as shown in \cite{Agrawal2009b} among others).} %, and \vm{find a reference from RM}\dawsen{Let's do it after EC.}).
%However, in many markets, the demand cannot be perfectly learned due to unpredictable components such as traffic spikes and competitor's change of strategy.

In such cases firms can take a\st{completely  robust} \vmn{worst-case} approach and assume that demand is\st{not predictable} \vmn{controlled by an imaginary adversary and thus is unpredictable}.\st{ but} \vmn{Such an approach, however,}\st{that} usually results in\st{strategies} \vmn{online policies} that are too conservative (as studied in  \cite{Ball2009} and others). % \dawsen{This sentence is a bit weird because in any case, even when there are available demand prediction, the firm can take robust approach. Perhaps say they ``could choose to take a robust approach'' or something?}
Instead, firms may wish to employ online policies\st{that \emph{try to learn} the demand \footnote{We use the notion of ``learning'' and ``predicting'' interchangeably.
%We will elaborate on this later.
}, but at the same time do not \emph{overfit} to the observed data with the caution that it could embody an unpredictable component.}
\vmn{based on models that assume the future demand can partially be predicted, avoiding being too conservative while not being reliant on fully accurate predictions.}
This paper aims to investigate to what extent the above goal is achievable. We propose a new demand model\vmn{, called {\em partially predictable},} that contains both\st{predictable and unpredictable} \vmn{adversarial (thus unpredictable) and stochastic (predictable) components}\st{,  hence is only partially learnable}.
\vmn{We design novel algorithms to demonstrate that even though demand is assumed to include an unpredictable component,}\st{We show how } firms can\st{still} make use of the limited information that the data reveals and improve upon the completely conservative approach.

We study a basic online\st{resource} allocation problem\st{, known as the single-resource revenue management with $2$ fare classes, under our new demand model called \emph{partially-learnable model}}
\vmn{of a single resource with an arbitrary capacity to a sequence of customers, each of which belongs to one of two types.}
\vmn{Each customer demands one unit of the resource. If the resource is allocated, the firm earns a type-dependent revenue. Type-1 and -2 customers generate revenue of $1$ and $a <1$, respectively.}
Our demand model takes a parameter $0 < p < 1$ and works as follows.
An unknown number of customers of each\st{fare class} \vmn{type} will be revealed to the firm in unknown order. \vmn{Both the number and the order of customers are assumed to be controlled by an imaginary adversary}. However, a fraction $p$ of randomly chosen customers does {\em not follow} this prescribed order and instead arrives at uniformly random times. This group of customers represents the\st{predictable} \vmn{stochastic} component of the demand that is mixed with the\st{unpredictable (unknown)} \vmn{adversarial} element.
Although we cannot identify which customers belong to the\st{predictable} \vmn{stochastic} group, we can still {\em partially}\st{learn} \vmn{predict} future demand, because this group is almost uniformly spread over the time horizon. Therefore parameter $p$ determines the level of\st{learnability} \vmn{predictability} of  demand.

From a practical point of view, our demand model requires {\em no forecast}  for the number of customers \st{from each fare class} \vmn{of each type prior to arrival}; instead, it assumes a rather mild ``regularity'' in the arrival pattern\vmn{: a fraction $p$ of customers of each type is spread throughout the time horizon}. We motivate this through a simple example. Suppose an airline launches a new flight route for which it has no demand forecast. However, using historical data \vmn{on customer booking behavior}, the airline knows that there is heterogeneity in \vmn{booking behavior of customers, namely, the time they request a booking varies across customers of each type.}
\st{the advanced booking behavior of the customers (from each fare class) which results in gradual arrival of at least a portion of requests for each fare-class.}
\vmn{Such heterogeneity results in the gradual arrival of a portion of customers of each type.}
%\footnote{We emphasize that here the historical data is used to learn customer behaviour, and not the demand for this product.}
%it can assume at least a certain fraction (let say $30 \%$) of the demand for each fare-class arrive gradually over time.
For example, \cite{Adv_Booking} illustrates a significant disparity in the advanced booking behavior of business travelers based on their age, gender, and travel frequency. Therefore, the airline can reasonably assume that demand \vmn{from business travelers (who correspond to type-1 in our model)}\st{for business (high) fare class } is, to some degree, spread over the sale horizon.\st{, and this is exactly the regularity condition that our demand model imposes.}
%; on average, female business travelers of age $70$-$75$ book $29$ days in advance, while male business travelers of age $30$-$35$ book only $18$ days in advance. Therefore, if by $23$ days prior to the flight date no business travelers have booked the flight, it is very unlikely that more than $70 \%$ of the rest of the demand belong to business (high) fare class.
\st{In this paper, we show how  the firm can take advantage of this limited information and improve the revenue guarantee upon existing algorithms.}

%Demand forecasting has long been a challenge for applying traditional revenue management, specially in new industries \cite{Lennon2004} and for new products.
%More recently, in the hotel industry, emergence of sharing economy platforms such as Airbnb, that offers a wider range of options and has the capability to scale supply with negligible cost and on short timescales \cite{AirBnb}, adds to significant uncertainty in demand and its variability across time even for products that are not new.

\vmn{From a theoretical point of view, our demand model aims to address the limitations of the main two approaches that have been taken so far in the literature: (1) adversarial models and  (2) stochastic ones.}\footnote{A few papers consider arrival models outside these two categories. We carefully review them and compare them with our model in Section~\ref{sec:review}.} \vmn{Under the adversarial modeling approach,}\st{In the robust approach,}  the sequence of arrivals is assumed to be completely unpredictable. The online algorithms developed for these models aim to perform well in the worst-case scenario, often resulting in very conservative bounds (see \cite{Ball2009} for the single-resource revenue management problem and \cite{Mehta2007a} and \cite{buchbinder2009design} for online allocation problems in internet advertising).
On the other hand, the stochastic modeling approach assumes that demand follows an unknown\st{stationary} distribution \citep{kleinberg2005multiple,devanur2009adwords,Agrawal2009b}.\footnote{In fact these papers assume a more general model, the random order model, that we discuss in Section~\ref{sec:review}.} In this case we can\st{{\em learn} the demand distribution} \vmn{predict future demand after}\st{by} observing a small fraction of it. For example,\st{if} after observing the first $10 \%$ of the demand, \vmn{if} we\st{see} \vmn{observe} that $15 \%$ of customers are\st{business travelers} \vmn{of type-1}, we\st{forecast} \vmn{can predict} that roughly $15 \%$ of the remaining customers are also\st{business travelers} \vmn{of type-1}. The limitation of such an approach is that it cannot model variability across time. In some cases, real data does not confirm the\st{stationary} stochastic structure presumed in these models, as shown in \citet{WangTraffic} and  \citet{Shamsi}\remove{ using real-data}. In fact, as discussed in \citet{Mirrokni2012} and \citet{esfandiari2015online}, large online markets (such as internet advertising systems) often use modified versions of these algorithms to make them less reliant on accurate demand prediction. \vmn{Our model provides a middle ground between the aforementioned approaches}
\st{Our partially-learnable model aims to address the limitation of both of the aforementioned approaches} by assuming that the arrival sequence contains both an
\st{unpredictable (non-stationary)}\vmn{adversarial} component and a\st{predictable} \vmn{stochastic} one.\st{Relative size of the unpredictable component (i.e., $(1-p)$) can be viewed as  ``level of robustness'' or caution that the firm would like to take in regard to what the observed data reveals about the future.
In Section ?, we discuss how to estimate the parameter $p$, and comment on how robust our algorithm is to overestimating $p$.}

For the above problem, we design two online algorithms (a non-adaptive and an adaptive one\footnote{We call an algorithm ``adaptive'' if it makes decisions based on the sequence of arrivals it has observed so far.}) that perform well in the partially\st{learnable} \vmn{predictable} model.
%and study their performance under our partially-learnable model.
We use the metric of competitive ratio, which is commonly used to evaluate the performance of online algorithms. Competitive ratio is the worst-case ratio between the revenue of the online scheme to that of a clairvoyant solution (see Definition~\ref{def:comp}).
The competitive ratio of our algorithms\st{are} \vmn{is} parameterized by $p$, and for both algorithms\st{they are} \vmn{the ratio} increases with $p$: {\em as the relative size of\st{predictable} \vmn{the stochastic} component  grows, the loss due to  making online decisions decreases.} We further show that using an adaptive algorithm is particularly beneficial when the
\st{inventory}\vmn{capacity}\st{is of the same order as the time horizon (maximum number of customers)} \vmn{scales linearly with the maximum number of customers}. Our algorithms are easily adjustable with respect to  parameter $p$. Therefore, if a firm wishes to use different levels of\st{robustness} \vmn{predictability} for different products, then it can\st{still} use the same algorithm with different parameters $p$.

In designing of our algorithms, we keep track of the number of accepted customers of each\st{class} \vmn{type}, and we decide whether to accept an arriving\st{class-$2$} \vmn{type-$2$} customer by comparing the number of already accepted\st{class-$2$} \vmn{type-$2$} customers with optimally designed {\em dynamic} thresholds.\footnote{\vmn{We always accept a type-1 customer if there is remaining inventory.}}
Our non-adaptive algorithm strikes a balance between ``smoothly'' allocating the inventory over time (by not accepting many\st{class-$2$} \vmn{type-$2$} customers toward the beginning) and not protecting too much inventory for potential late-arriving\st{class-$1$} \vmn{type-$1$} customers (see Algorithm~\ref{algorithm:hybrid}  and Theorem~\ref{thm:hybrid}).
Our adaptive algorithm frequently recomputes upper bounds on the number of future customers\st{in each class} \vmn{of each type} based on observed data
%with the knowledge that a constant fraction of them followed a random arrival time. We
and uses these upper bounds to ensure that we protect enough inventory for future\st{class-$1$} \vmn{type-$1$} customers. We show that such an adaptive policy significantly improves the performance guarantee when the initial inventory is large relative to the \vmn{maximum }number of customers (see Algorithm~\ref{algorithm:adaptive-threshold} and Theorem~\ref{thm:adaptive-threshold}). Both algorithms could reject a\st{low-fare} \vmn{type-$2$} customer early on but accept another\st{low-fare} \vmn{type-$2$} customer later. This is consistent with\st{the common}  practice. \vmn{For example, in online airline booking systems, }\st{ where} lower fare classes can open up after being closed out previously \citep{CheapoAir}.

\vmn{From a methodological standpoint, an analysis of  the competitive ratio  of our algorithms presents many new technical challenges arising from the fact that our arrival model contains both an adversarial  and a stochastic component. Our analysis crucially relies on a concentration result
that we establish for our arrival model (see Lemma~\ref{lemma:needed-centrality-result-for-m=2}) as well as fairly intricate case analyses for both algorithms. Further, to prove  the lower bound on the competitive ratio of our adaptive algorithm, we construct a novel {\em factor-revealing} nonlinear mathematical program (see~\ref{MP1} and Section~\ref{subsec:adapt:com}).}

%Analyzing the adaptive algorithm, however, is considerably more challenging. We establish a  lower bound on the competitive ratio by constructing a novel factor revealing mathematical program.\vm{Remainder to myself: add a similar sentence to intro.}

The two extreme cases of our model where all or none of the customers belong to the\st{unpredictable} \vmn{adversarial} group (i.e., $p = 0$ and $p=1$) reduce to the\st{robust} \vmn{adversarial} and stochastic modeling approaches that have been mainly studied in the literature thus far (for instance, \cite{Ball2009} study the former model and \cite{Agrawal2009b} study the latter).
%The online algorithms developed for the former models aim to perform well in the worst-case scenario, and hence the resulting bounds are often very conservative.
%The algorithms designed for the latter can achieve near optimal revenue, however, they  rely on the {\em stationarity} of the arriving sequence.
Our algorithms recover the known performance guarantees for these two extreme cases. For the regime in between (i.e., when $0 < p < 1$), we show that our algorithms achieve competitive ratios better than what can be achieved by any of  the algorithms designed for these extreme cases (or even any combination of them). This  highlights the need to design new algorithms when departing from traditional \st{approaches}\vmn{arrival models}.

%{\color{gray}
%{\bf Further we present two other models closely related to our main model that allow for arbitrary number of fare classes. We study a hypothetical scenario in which we can recognize which customers belong to the\st{predictable} \vmn{stochastic} group, and show that there is an online algorithm with near optimal revenue.\footnote{Under certain condition where inventory is not too small. For a precise statement see Theorem ?. } This further emphasizes the significant impact of learning: by distinguishing predictable group from the unpredictable one, we can almost perfectly learn the demand.} \vmn{I made the the above bold and gray, because we decided to remove this subsection from the paper. If agreed, we can remove}}

We also\st{consider a classic stopping time problem known as} \vmn{study} the classic secretary problem \vmn{under our partially predictable arrival model.}
\vmn{The secretary problem, a stopping time problem, }
corresponds to the setting in which we initially have one unit of inventory; each customer is of a different\st{value (fare-class)} \vmn{type}, and we wish to maximize the probability of\st{accepting the highest value customer} \vmn{allocating the inventory to the type generating the highest revenue}\st{in our partially learnable model}. We show that, unlike the classic setting (which corresponds to $p =1$ in our model), \vmn{the celebrated deterministic stopping rule policy based on}
a \st{fixed}\vmn{deterministic} observation period is no longer\st{best possible} \vmn{optimal}. Due to the presence of the\st{unpredictable} \vmn{adversarial} component, randomizing over the length of the observation period may result in improvement \vmn{(see Algorithm~\ref{algorithm:observation-selection}, Theorem~\ref{thm:b=1}, and Proposition~\ref{thm:randomized-b=1})}.

We conclude this section by highlighting our motivations and contributions. For many applications, demand arrival processes are inherently prone to contain unpredictable components that even advanced information technologies cannot mitigate. \vmn{An allocation policy whose design is based on stochastic modeling cannot incorporate the presence of such unpredictable components.}
At the same time, taking a\st{completely robust} \vmn{worst-case adversarial} approach usually leads to allocation polices that are too conservative. We introduce the {\em first arrival model} that contains {\em both \vmn{adversarial} (thus unpredictable) and \vmn{stochastic}\st{predictable}} components. Through novel algorithm design, we show that (1) we can take advantage of even limited available  information (due to the presence of the\st{predictable} \vmn{stochastic} component) to improve a firm's revenue when compared to algorithms that take  a\st{fully robust } \vmn{worst-case} approach and that (2) there is an unavoidable loss due to the presence of an\st{unpredictable} \vmn{adversarial} component, which emphasizes {\em the value of stochastic information and\st{learning} \vmn{predictability}} in online resource allocation.\st{Concurrently, one can consider such a loss as the price of limited robustness}

%This inherent unpredictability limits the performance of learning based algorithms

%Our results highlight the value of stochastic information and learning in online resource allocation, particularly  if the initial inventory at hand is much smaller than the horizon.
%One can consider the  loss due to the presence of an unpredictable component  as the {\em price of limited robustness} in online resource allocation.
%However, as the ratio of inventory to horizon becomes larger, such loss declines. In this regime, using an adaptive algorithms that reacts to the observed data (by recomputing bounds on the number of customers), we can significantly improve the achievable performance guarantee. Nonetheless, comparing our results to a completely robust approach (a worst-case analysis), we show that even partial learning (due to the presence of the predictable component) can significantly improve the revenue of the firm {\em if} it uses a carefully designed algorithm that takes advantage of the available (partial) information.

%Our work is a first step in bridging the long standing gap between a robust approach and a stochastic modeling one.

The rest of the paper is organized as follows. In Section~\ref{sec:review}, we review the related literature and highlight the differences between the current paper and previous work.
In Section~\ref{sec:prem}, we formally introduce our demand arrival model and our performance metric, and prove a consequential concentration result for the arrival process.
Sections~\ref{sec:alg1} and~\ref{sec:alg2} are dedicated to description and analysis of our two algorithms.
In Section~\ref{sec:model:discussion}, we \st{derive some further fundamental properties of our demand model such as}\vmn{present} upper bounds on \vmn{the} performance of any online algorithm, \vmn{and we compare the performance of our algorithms with that of  existing ones.} \st{model estimation, and robustness issues.}
Section~\ref{sec:secretary} studies \vmn{the secretary problem under our new arrival model.}
\st{two extensions of our model that allow for arbitrary number of fare classes.}
In Section~\ref{sec:conclusion}, we conclude  by outlining several directions for future research.
For the sake of brevity, we include proofs  of only selected results in the main text. Detailed proofs of the remaining  statements are deferred to clearly marked appendices.

\section{Literature Review }
\label{sec:review}

%\vm{Almost re-wrote this. Made it more OR friendly by adding more papers from RM literature. Also, we were missing some important recent papers that I added.   }

Online allocation problems have broad applications in revenue management, internet advertising, scheduling appointments (through web applications) in health care, just to name a few. Thus it has been studied in various forms in operations research and management, as well as computer science.
As discussed in the introduction, the approach taken in modeling the arrival process is the first consequential step in studying these problems. Therefore, in this literature review, we categorize related streams of research by modeling approach rather than by the particular problem formulation and application.

First, we note that the single-leg revenue management (RM) problem and its generalizations have been extensively studied using frameworks other than online resource allocation problems \vmn{and competitive analysis}. Earlier papers assumed {\em low-before-high} models (where all low-fare demand realizes before high-fare demand) with known demand distributions \citep{belobaba1987survey,belobaba1989or,brumelle1993airline,littlewood2005special}
or assumed the arrival process is known, and formulated the problem as a Markov decision problem \citep{lee1993model,lautenbacher1999underlying}. We refer the reader to \cite{talluri2006theory} for a comprehensive review of RM literature.
\vmn{Further, many recent papers in revenue management study dynamic pricing when the underlying price-sensitive demand process is unknown.
See, for example, seminal work by \cite{besbes2009dynamic} and \cite{araman2009dynamic}.
For the sake of brevity, we will not review these streams of work.}

\vspace{2mm}
{\bf Adversarial models:}
\cite{Ball2009} studied the single-leg revenue management problem under an adversarial model and showed that in the two-fare case the optimal competitive ratio is $\frac{1}{2-a}$ where $a <1$ is the ratio of two fares. As discussed in the introduction, our model reduces to that of \cite{Ball2009} for $p=0$. In this special case, our non-adaptive algorithm reduces to the threshold policy of \cite{Ball2009} and recovers the same performance guarantee. However, when $0 < p < 1$, we show that for a certain class
of instances our algorithms perform better than that of \cite{Ball2009} (see Subsection~\ref{sec:bad-instance-for-other-models}), indicating the need for new algorithms for our new arrival model.

Several papers studied the adwords problem under the adversarial model \citep{Mehta2007a,buchbinder2009design}. This problem concerns  allocating ad impressions to budget-constrained advertisers. As mentioned in \cite{Mehta2007a}, even though the optimal competitive ratio under an adversarial model is $1 - 1/e$, one would expect to do better when statistical information is available. Later, \cite{Mirrokni2012} showed that it is impossible to design an algorithm with a near-optimal competitive ratio under both adversarial and random arrival models. Such an impossibility result affirms the need for new modeling approaches to serve as a middle ground between these two models.  Our paper takes a step in this direction.

%This problem is concerned with allocating Ad impressions to budget-constraint advertisers.
%
%As mentioned in the introduction, in adversarial models, the sequence of arrivals is assumed to be determined by an adversary, and cannot be learned (or predicted) as the sequence unfolds.
%Adversarial models have been studied for single-leg revenue management \cite{Ball2009} as well as in more general settings mainly motivated by various forms of Internet advertising \cite{Mehta2007a,buchbinder2009design}.
%The online algorithms developed for these models aim to perform well in the worst-case scenario, and hence the resulting bounds are often very conservative.
%For example, \citet{Ball2009} develop an online algorithm for the single-leg revenue management with $2$ fare classes that guarantees only a $\frac{1}{2-a}$ factor of the maximum revenue that can be achieved in the hindsight. They also show that this bound is the best achievable bound in the adversarial model.
\vspace{2mm}
{\bf Stationary stochastic models:}
A general form of these models is the \emph{random order model}, which assumes that the sequence of arrivals is a random permutation of an arbitrary sequence \citep{kleinberg2005multiple,devanur2009adwords,Agrawal2009b}.
In such a model, after observing a small fraction of the input, one can predict pattern of future demand. This intuition is used to develop  primal- and dual-based online algorithms that achieve near-optimal revenue, under appropriate conditions on the relative amount of available resources to allocate. These algorithms rely heavily on learning from observed data, either once \citep{devanur2009adwords} or repeatedly \citep{kleinberg2005multiple,Agrawal2009b,Kesselheim}. As discussed in the introduction, arrival patterns could experience high variability across time, limiting the performance of these algorithms in practice \citep{Mirrokni2012,esfandiari2015online}.
We note that assuming i.i.d. arrivals with known or unknown distributions also falls into this category of modeling approaches. Several revenue management papers provided asymptotic analysis of linear programming-based (LP-based) approaches for such settings; see \cite{Talluri:1998, Cooper2002} and \cite{Jasin2015}.

%of stationary stochastic settings have also been studied using  asymptotic analysis. See \cite{Cooper2002,Talluri:1998}

%earlier papers in the revenue management literature \cite{Cooper2002,Talluri:1998} studied similar settings and provided asymptotic analysis without using competitive ratio analysis.}

Our model reduces to a special case of the model of \cite{Agrawal2009b} for $p = 1$, and like their algorithm, ours also achieves near-optimal revenue when $p=1$. However, when $0 < p < 1$, we show, in Subsection~\ref{sec:bad-instance-for-other-models}, that for a certain class of instances our algorithms perform better than that of  \cite{Agrawal2009b}.

\vspace{2mm}
{\bf Nonstationary stochastic models:}
Motivated by advanced service reservation and scheduling, \cite{Van-Ahn2} and \cite{Van-Ahn1} studied online allocation problems where demand arrival follows a {\em known} nonhomogeneous Poisson process. For such settings, they developed online algorithms with constant competitive ratios. Further, \cite{Ciocan2012} considered another interesting setting where the (unknown) arrival process belongs to a broad class of stochastic processes. They proved a constant factor guarantee for the case where arrival rates  are uniform.
Our modeling strategy differs from both approaches  by assuming that $(1-p)$ fraction of the input is\st{unpredictable} \vmn{adversarial}. Even for the\st{predictable} \vmn{stochastic} component, we assume no prior knowledge of the distribution. However, we limit the adversary's power by assuming that these two components are {\em mixed}. Also, we note that the aforementioned papers studied more general allocation problems in settings like network revenue management.

%
%
%who studied re-optimization
%policies for a network revenue management problem where the distribution of the valuation of the
%customer (i.e., the distribution of the types) is constant over time, but the size of the market changes
%over time according to a stochastic (e.g., multi-variate Gaussian) process. They showed that a reoptimization
%policy that adjusts prices of the products by solving a linear program obtains about
%one-third of the optimal revenue;
%
%
\vspace{2mm}
{\bf Other models:}
Several earlier papers also acknowledged and addressed the limitation of both the adversarial and random order (or stochastic) models using various approaches.
\citet{Mahdian2007} and \citet{Mirrokni2012} considered allocation problems where the demand can \emph{either} be perfectly estimated or adversarial.
They designed and analyzed algorithms that have good performance guarantees in both worst-case and average-case scenarios.
Unlike these works, our demand model contains \emph{both}\st{predictable} \vmn{stochastic} and\st{unpredictable} \vmn{adversarial} components at the same time, and we design algorithms that take advantage of partial\st{estimation} \vmn{predictability}.

Another approach to address unpredictable patterns in demand is to use robust stochastic optimization \citep{BanTal, Bertsimas_RobustLP}. These papers
aim to optimize allocations when the demand belongs to a class of distributions (or uncertainty set).
This approach limits the adversary's power by restricting the class of demand distributions. Here, we take a different approach.  We do not limit the class of distribution that the adversary can choose from; instead, we assume that a fraction $p$ of the demand will not follow the adversary.

\citet{Lan2008} also took a robust approach, studying the single-leg multi-fare class revenue management problem in a very interesting setting, where the only prior knowledge about demand is the lower and upper bounds on the number of customers from each fare class.
%\dawsen{and they do not update the bounds as the sequence reveals}.
\citet{Lan2008} used fixed upper and lower bounds to develop optimal static policies in the form of nested booking limits, and also showed that dynamically adjusting these policies can improve the competitive ratio.
Unlike their work, we do not assume prior knowledge of lower and upper bounds on the number of customers from each class.
Instead, in our model, we learn the bounds as the sequence unfolds.

\citet{Shamsi} used a real data set from display advertising at AOL/Advertising.com to show that arrival patterns do not satisfy the crucial property implied by assuming a random order model for demand. In particular, they showed that the dual prices of the offline allocation problem at different times can vary significantly.
They used a risk minimization framework to devise allocation rules that outperform existing algorithms when applied to AOL data. Even though the results are practically promising, the paper provides no performance guarantee, nor does it offer insights on how to model traffic in practice.

\st{Finally} \vmn{Further}, \citet{esfandiari2015online} also considered a hybrid arrival model where the input comprises known stochastic i.i.d. demand and an unknown number of arrivals that are chosen by an adversary (which is motivated by traffic spikes).
They do not assume any knowledge of the traffic spikes, but the performance guarantee of their algorithm is parameterized by $\lambda$, roughly the fraction of the revenue in the optimal solution that is obtained from the stochastic (predictable) part of the demand.
Parameter $\lambda$ plays a similar role as parameter $p$ in our model, in that it controls the adversary's power.
However, the underlying arrival processes in these two models differ considerably and cannot be directly compared.
In particular, we do not assume any prior knowledge of the stochastic component; instead we partially\st{learn this} \vmn{predict it}.
However, we do assume that the adversary determines only the initial order of arrivals (i.e., before knowing which customer will eventually follow its order).
%Also, the online allocation problem we study is a special case of the setting studied in~\cite{esfandiari2015online}.

Our work is also closely related to the literature on the secretary problem. In the original formulation of the problem, $n$ secretaries with unique values arrive in uniformly random order; the goal is \vmn{to} maximize the probability of hiring the most valuable secretary.
The optimal solution to this problem is an observation-selection policy: observe the first $n/e$ secretaries, then select the first one whose value exceeds that of the best of the previously observed secretaries
\citep{lindley1961dynamic,dynkin1963optimum,freeman1983secretary,ferguson1989solved}. Recently, \citet{kesselheim2015secretary} relaxed the assumption of uniformly random order, and analyzed the performance of the above policy under certain classes of nonuniform distribution over permutations. Here, we study the secretary problem in our new arrival model (i.e., only a $p$ fraction of secretaries arrive in uniformly random order) and show that a \st{fixed}\vmn{deterministic} observation period is \st{sub-}\vmn{not} optimal.

\section{Model and Preliminaries}
\label{sec:prem}

%\vm{Added  subsection 3.1., Definitions 3.2, 3.3, modified the text after lemma; added def. 3.5. Edited Dawsen's changes for notation and defined some new ones. Patrick: to address your concerns on distinguishing random variables and realization, I clearly defined the convention at the beginning of 3.1. }

%\vm{I changed one notation: $\vec{v'}$ to $\vec{v}_I$ for initial customer sequence.}

A firm is endowed with $b$ (identical) units of a product to sell over $n \geq 3$ periods, where $n \geq b$. In each period, at most one customer arrives demanding one unit of the product; customers belong to two\st{classes} \vmn{types} depending on\st{their willingness-to-pay: class-1 and class-2 customers are willing to pay $\$ 1$ and $ \$ a$ respectively, where $0 < a <1$} \vmn{the revenue they generate. Type-1 and type-2 customers generate revenue of $1$ and $0 < a <1$, respectively}. Upon the arrival of a customer, the firm observes the\st{class} \vmn{type} of the customer and must make an irrevocable decision to accept this customer and allocate one unit, or to reject this customer. If a firm accepts a\st{class} \vmn{type}-1 (\st{class}\vmn{type}-2) customer, it will earn $\$ 1$ ($ \$ a$).
Our goal is to devise online allocation algorithms that maximize the firm's revenue. We evaluate the performance of an algorithm by comparing it to the \st{optimal}\vmn{optimum} offline solution
(i.e., the clairvoyant solution).

Before proceeding with the model, we introduce a few notations and briefly discuss the structure of the problem.
We represent each customer by the value of\st{her requested fare class} \vmn{revenue she generates if accepted}, and the sequence of arrival by $\vec{v}=(v_1,v_2,\dots, v_n)$, where $v_i \in \{0,a,1\}$; $v_i =0$ implies that no customer arrives at period $i$.
%.\dawsen{or a customer acquired by another company.}
We denote the number of\st{class} \vmn{type}-$1$ (\st{class}\vmn{type}-$2$) customers in the entire sequence as $n_1$ ($n_2$).
Note that the \st{optimal}\vmn{optimum} offline solution that we denote by $OPT(\vec{v})$ has the following simple structure: accept all the\st{class} \vmn{type}-$1$ customers, and if $n_1<b$, then accept $\min \{n_2, b-n_1\}$\st{class} \vmn{type}-2 customers. Therefore,
\begin{align}
\label{eq:OPT}
OPT(\vec{v}) = \min\{b, n_1\} + a \min \{n_2, (b-n_1)^{+}\},
\end{align}
where $(x)^+ \triangleq \max \{ x, 0\}$, \vmn{and we use the symbol ``$\triangleq$'' for definitions.}
At each period, a reasonable online algorithm will accept an arriving\st{class} \vmn{type}-$1$ customer if there is inventory left. Thus the main challenge for an online algorithm is to decide whether to accept/reject an arriving\st{class} \vmn{type}-$2$ customer facing the following natural trade-off: accepting a\st{class} \vmn{type}-$2$ customer may result in rejecting a potential future\st{class} \vmn{type}-$1$ customer due to limited inventory; on the other hand, rejecting a\st{class} \vmn{type}-$2$ customer may lead to having unused inventory at the end. Therefore, any good online algorithm needs to strike a balance between accepting \emph{too few} and \emph{too many}\st{class} \vmn{type}-$2$ customers. We denote by $ALG(\vec{v})$ the revenue obtained by an online algorithm.

%In this paper, we take a middle-ground approach in modeling the demand and

\vmn{Next we} introduce \st{a}\vmn{our} \emph{partially \st{learnable}\vmn{predictable}} \vmn{demand arrival model}\st{ which} \vmn{that} works as follows.
The adversary determines an initial  sequence  which we denote by $\vmn{\vec{v}_I}=(v_{I,1},v_{I,2},\dots, v_{I,n})$, where $v_{I,j} \in \{0,a,1\}$, for $1 \leq j \leq n$.
%This can be thought as the sequence that the adversary prescribes to the customers.
However, a subset of customers will not follow this order. We call this subset the {\emph {\RG} group}, which we denote by ${\RGS}$.
Each customer joins the {\RG} group independently and with the same probability $p$.
Other customers are in the {\emph{\UPG} group} denoted by $\UPGS$.
Customers in the {\RG} group are permuted uniformly at random among themselves. Formally, a permutation $\sigma_{{\RGS}}:{\RGS}\rightarrow {\RGS}$ is chosen uniformly at random and determines the order of arrivals among the {\RG} group.
In the resulting overall arriving sequence, the {\UPG} group follows the adversarial sequence according to $\vmn{\vec{v}_I}$, and those in the {\RG} group follow the random order given by $\sigma_{{\RGS}}$.
Given $\vmn{\vec{v}_I}$, we denote  the {\em random} customer arrival sequence by $\vec{V}=(V_1,V_2,\dots , V_n)$,
%, where each of $\{V_1,V_2,\dots , V_n\}$ is a random variable taking values in $\{0,a,1\}$;
and the {\em realization} of it by  $\vec{v}=(v_1,v_2,\dots , v_n)$.

The example presented in Figure~\ref{figure:example} illustrates the arrival process.
The top row (gray nodes) shows the initial sequence ($\vmn{\vec{v}_I}$).
The middle row shows which customers belong to the {\RG} group (the black nodes) and which belong to the {\UPG} group (the white ones).
The bottom row shows both the permutation $\sigma_{{\RGS}}$ and the actual arrival sequence.
In this example, ${\RGS}=\{2, 5, 6, 8\}$, and $\sigma_{{\RGS}}(2)=6, \sigma_{{\RGS}}(5)=2, \sigma_{{\RGS}}(6)=5$, and $\sigma_{{\RGS}}(8)=8$.

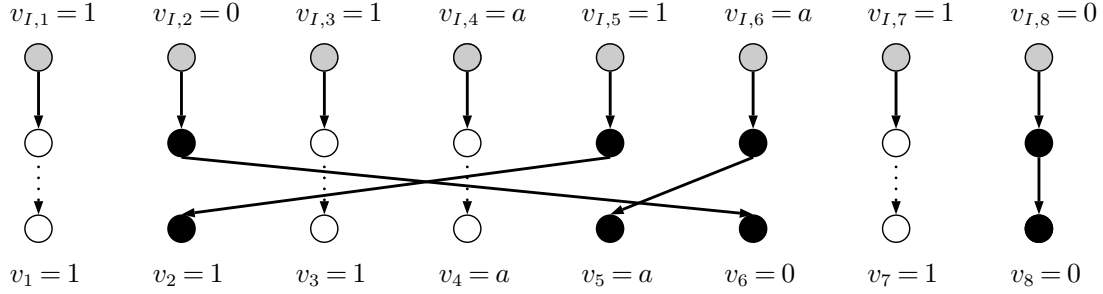
\begin{figure}
\centering
\psscalebox{0.95}{
\psscalebox{1.0 1.0} % Change this value to rescale the drawing.
{
\begin{pspicture}(0,-1.9712988)(15.6,1.9712988)
\definecolor{colour0}{rgb}{0.8,0.8,0.8}
\psdots[linecolor=black, fillstyle=solid,fillcolor=colour0, dotstyle=o, dotsize=0.4, fillcolor=colour0](0.4,1.3046386)
\psdots[linecolor=black, fillstyle=solid,fillcolor=colour0, dotstyle=o, dotsize=0.4, fillcolor=colour0](2.4,1.3046386)
\psdots[linecolor=black, fillstyle=solid,fillcolor=colour0, dotstyle=o, dotsize=0.4, fillcolor=colour0](4.4,1.3046386)
\psdots[linecolor=black, fillstyle=solid,fillcolor=colour0, dotstyle=o, dotsize=0.4, fillcolor=colour0](6.4,1.3046386)
\psdots[linecolor=black, fillstyle=solid,fillcolor=colour0, dotstyle=o, dotsize=0.4, fillcolor=colour0](6.4,1.3046386)
\psdots[linecolor=black, fillstyle=solid,fillcolor=colour0, dotstyle=o, dotsize=0.4, fillcolor=colour0](8.4,1.3046386)
\psdots[linecolor=black, fillstyle=solid,fillcolor=colour0, dotstyle=o, dotsize=0.4, fillcolor=colour0](10.4,1.3046386)
\psdots[linecolor=black, fillstyle=solid,fillcolor=colour0, dotstyle=o, dotsize=0.4, fillcolor=colour0](12.4,1.3046386)
\psdots[linecolor=black, fillstyle=solid,fillcolor=colour0, dotstyle=o, dotsize=0.4, fillcolor=colour0](14.4,1.3046386)
\psdots[linecolor=black, fillstyle=solid, dotstyle=o, dotsize=0.4, fillcolor=white](0.4,0.10463867)
\psdots[linecolor=black, fillstyle=solid, dotstyle=o, dotsize=0.4, fillcolor=white](4.4,0.10463867)
\psdots[linecolor=black, fillstyle=solid, dotstyle=o, dotsize=0.4, fillcolor=white](6.4,0.10463867)
\psdots[linecolor=black, fillstyle=solid, dotstyle=o, dotsize=0.4, fillcolor=white](12.4,0.10463867)
\psdots[linecolor=black, dotsize=0.4](2.4,0.10463867)
\psdots[linecolor=black, dotsize=0.4](8.4,0.10463867)
\psdots[linecolor=black, dotsize=0.4](10.4,0.10463867)
\psdots[linecolor=black, dotsize=0.4](14.4,0.10463867)
\psdots[linecolor=black, dotsize=0.4](2.4,-1.0953614)
\psdots[linecolor=black, dotsize=0.4](8.4,-1.0953614)
\psdots[linecolor=black, dotsize=0.4](10.4,-1.0953614)
\psdots[linecolor=black, dotsize=0.4](14.4,-1.0953614)
\psdots[linecolor=black, dotsize=0.4](14.4,-1.0953614)
\psdots[linecolor=black, fillstyle=solid, dotstyle=o, dotsize=0.4, fillcolor=white](0.4,-1.0953614)
\psdots[linecolor=black, fillstyle=solid, dotstyle=o, dotsize=0.4, fillcolor=white](4.4,-1.0953614)
\psdots[linecolor=black, fillstyle=solid, dotstyle=o, dotsize=0.4, fillcolor=white](6.4,-1.0953614)
\psdots[linecolor=black, fillstyle=solid, dotstyle=o, dotsize=0.4, fillcolor=white](12.4,-1.0953614)
\psline[linecolor=black, linewidth=0.04, arrowsize=0.05291666666666668cm 2.0,arrowlength=1.4,arrowinset=0.0]{->}(0.4,1.1046387)(0.4,0.30463868)
\psline[linecolor=black, linewidth=0.04, arrowsize=0.05291666666666668cm 2.0,arrowlength=1.4,arrowinset=0.0]{->}(2.4,1.1046387)(2.4,0.30463868)
\psline[linecolor=black, linewidth=0.04, arrowsize=0.05291666666666668cm 2.0,arrowlength=1.4,arrowinset=0.0]{->}(4.4,1.1046387)(4.4,0.30463868)
\psline[linecolor=black, linewidth=0.04, arrowsize=0.05291666666666668cm 2.0,arrowlength=1.4,arrowinset=0.0]{->}(6.4,1.1046387)(6.4,0.30463868)
\psline[linecolor=black, linewidth=0.04, arrowsize=0.05291666666666668cm 2.0,arrowlength=1.4,arrowinset=0.0]{->}(8.4,1.1046387)(8.4,0.30463868)
\psline[linecolor=black, linewidth=0.04, arrowsize=0.05291666666666668cm 2.0,arrowlength=1.4,arrowinset=0.0]{->}(10.4,1.1046387)(10.4,0.30463868)
\psline[linecolor=black, linewidth=0.04, arrowsize=0.05291666666666668cm 2.0,arrowlength=1.4,arrowinset=0.0]{->}(12.4,1.1046387)(12.4,0.30463868)
\psline[linecolor=black, linewidth=0.04, arrowsize=0.05291666666666668cm 2.0,arrowlength=1.4,arrowinset=0.0]{->}(14.4,1.1046387)(14.4,0.30463868)
\psline[linecolor=black, linewidth=0.04, linestyle=dotted, dotsep=0.10583334cm, arrowsize=0.05291666666666668cm 2.0,arrowlength=1.4,arrowinset=0.0]{->}(0.4,-0.09536133)(0.4,-0.8953613)
\psline[linecolor=black, linewidth=0.04, linestyle=dotted, dotsep=0.10583334cm, arrowsize=0.05291666666666668cm 2.0,arrowlength=1.4,arrowinset=0.0]{->}(4.4,-0.09536133)(4.4,-0.8953613)
\psline[linecolor=black, linewidth=0.04, linestyle=dotted, dotsep=0.10583334cm, arrowsize=0.05291666666666668cm 2.0,arrowlength=1.4,arrowinset=0.0]{->}(6.4,-0.09536133)(6.4,-0.8953613)
\psline[linecolor=black, linewidth=0.04, linestyle=dotted, dotsep=0.10583334cm, arrowsize=0.05291666666666668cm 2.0,arrowlength=1.4,arrowinset=0.0]{->}(12.4,-0.09536133)(12.4,-0.8953613)
\psline[linecolor=black, linewidth=0.04, arrowsize=0.05291666666666668cm 2.0,arrowlength=1.4,arrowinset=0.0]{->}(14.4,-0.09536133)(14.4,-0.8953613)
\psline[linecolor=black, linewidth=0.04, arrowsize=0.05291666666666668cm 2.0,arrowlength=1.4,arrowinset=0.0]{->}(2.4,-0.09536133)(10.4,-0.8953613)
\psline[linecolor=black, linewidth=0.04, arrowsize=0.05291666666666668cm 2.0,arrowlength=1.4,arrowinset=0.0]{->}(10.4,-0.09536133)(8.4,-0.8953613)
\psline[linecolor=black, linewidth=0.04, arrowsize=0.05291666666666668cm 2.0,arrowlength=1.4,arrowinset=0.0]{->}(8.4,-0.09536133)(2.4,-0.8953613)
\rput[bl](0.0,1.7046387){$v_{I,1}=1$}
\rput[bl](2.0,1.7046387){$v_{I,2}=0$}
\rput[bl](4.0,1.7046387){$v_{I,3}=1$}
\rput[bl](6.0,1.7046387){$v_{I,4}=a$}
\rput[bl](8.0,1.7046387){$v_{I,5}=1$}
\rput[bl](10.0,1.7046387){$v_{I,6}=a$}
\rput[bl](12.0,1.7046387){$v_{I,7}=1$}
\rput[bl](14.0,1.7046387){$v_{I,8}=0$}
\rput[bl](0.0,-1.8953613){$v_1=1$}
\rput[bl](2.0,-1.8953613){$v_2=1$}
\rput[bl](4.0,-1.8953613){$v_3=1$}
\rput[bl](6.0,-1.8953613){$v_4=a$}
\rput[bl](8.0,-1.8953613){$v_5=a$}
\rput[bl](10.0,-1.8953613){$v_6=0$}
\rput[bl](12.0,-1.8953613){$v_7=1$}
\rput[bl](14.0,-1.8953613){$v_8=0$}
\end{pspicture}
}
}
\caption{Illustration of the customer arrival model} \label{figure:example}
\end{figure}

Note that the extreme cases $p = 0$ and $p=1$ correspond to the adversarial and random order models that have been studied before (e.g., \cite{Ball2009} and \cite{Agrawal2009b}, respectively). Hereafter, we assume that $0 < p < 1$. For a given $p \in (0,1)$, at any time over the horizon, we can use the number of past observed\st{class} \vmn{type}-$1$ (\st{class}\vmn{type}-$2$) customers to obtain \st{``rough'' estimate}\vmn{bounds} on the number of customers of each\st{class} \vmn{type} to be expected over the rest of the horizon.
This idea is formalized later in Subsection~\ref{subsec:concentration} along with further analysis of our model.

Having described the arrival process, we now define the competitive ratio of an online algorithm under the proposed partially\st{learnable} \vmn{predictable} model as follows:

\begin{definition}An online algorithm is $c$-competitive in the proposed partially\st{learnable} \vmn{predictable} model if for any adversarial instance $\vmn{\vec{v}_I}$, $$ \E{ALG(\vec{V})} \geq c OPT(\vmn{\vec{v}_I}),$$
where the expectation is taken over which customers belong to the {\RG} group (i.e., subset $\RGS$), the choice of the random permutation $\sigma_{{\RGS}}$, and any possibly randomized decisions of the online algorithm.
\label{def:comp}
\end{definition}
Note that $OPT(\vec{V})=OPT(\vmn{\vec{v}_I})$
%for any random permutation of the arrivals of any subset of customers,
and thus, in the above definition, $ \E{ALG(\vec{V})} \geq c OPT(\vmn{\vec{v}_I})$ is equivalent to  $ \E{ALG(\vec{V})} \geq c \E{OPT(\vec{V})}$.

In Sections~\ref{sec:alg1} and~\ref{sec:alg2}, we present two online algorithms that perform well in the proposed partially\st{learnable} \vmn{predictable} model for various ranges of $b$ and $n$. Before introducing our online algorithms, in the following subsections we introduce a series of notations used throughout the paper and state a consequential concentration result that will allow us to partially\st{learn} \vmn{predict} future demand using past observed data.

\subsection{Notational Conventions}
\label{subsec:notation}

Throughout the paper, we use uppercase letters for random variables and lowercase ones for realizations. We have already used this convention in defining $\vec{V}$ vs. $\vec{v}$.
%We normalize the time horizon to $1$, and represent time steps by $\lambda = 0, 1/n, 2/n, \dots, 1$.
We normalize the time horizon to $1$, and represent time steps by $\lambda = 1/n, 2/n, \dots, 1$.
%\pj{Note that I have removed the time step $\lambda = 0$, as I dont think that this was needed.}
First, we introduce notations related to the random customer arrival sequence $\vec{V}$.
At any time step $\lambda$, for $j=1,2$, the number of\st{class} \vmn{type}-$j$ customers \emph{to be observed} by the online algorithms up to time $\lambda$ is denoted by $O_j(\lambda)$.
%Note that $O_j(\lambda)$ is a random variable where the randomness comes from the randomness of $\vec{V}$.
Further, we denote by $O_j^\RGS(\lambda)$ the number of\st{class} \vmn{type}-$j$ customers in the {\RG} group that arrive up to time $\lambda$ in $\vec{V}$.
Note that the online algorithm cannot distinguish between customers in the {\RG} group and customers in the {\UPG} group. Therefore, the online algorithm does {\em not} observe the realizations of $O_j^\RGS(\lambda)$. %$O_j^\RGS(\lambda)$ is related to $\vec{V}$ but online algorithms cannot observe it.
We denote realizations of $O_j(\lambda)$ and $O_j^\RGS(\lambda)$ by $o_j(\lambda)$ and $o_j^\RGS(\lambda)$, respectively.

Next, we introduce notations related to the initial adversarial  sequence $\vmn{\vec{v}_I}$.
As discussed earlier, we denote the total number of\st{class} \vmn{type}-$j$ customers in $\vmn{\vec{v}_I}$ by $n_j$.
In addition, given the sequence $\vmn{\vec{v}_I}$, we denote the total number of\st{class} \vmn{type}-$j$ customers among the first $\lambda n$ customers by  $\eta_j(\lambda)$.
Note that both $n_j$ and $\eta_j(\lambda)$ are deterministic.
Also, we define $\tilde o_j(\lambda) \triangleq (1-p) \eta_j(\lambda) + p\lambda n_j$ and $\tilde o^\RGS_j(\lambda) \triangleq  p\lambda n_j$ which will serve as deterministic approximations for $O_j(\lambda)$ and $O^\RGS_j(\lambda)$, respectively (see Lemma~\ref{lemma:needed-centrality-result-for-m=2} and the subsequent discussion for motivation of this definition).

Here we return to the example
in Figure~\ref{figure:example} and review the notations.
Suppose $\lambda = 5/8$ and $p = 0.5$; in this example, looking at the bottom row that shows \vmn{the} sequence $\vec{v}$, we have: $o_1(5/8) = 3$, $o^\RGS_1(5/8)=1$, which are realizations of
random variables $O_1(5/8)$ and $O^\RGS_1(5/8)$, respectively. Looking at the top row that shows sequence $\vmn{\vec{v}_I}$, we have: $n_1=4$, $\eta_1(5/8) = 3$, $\tilde o_1(5/8) = 0.5\times 3 + 0.5\times 4 \times  (5/8)$ = 2.75, and $\tilde o^\RGS_1(5/8) = 0.5\times 4 \times  (5/8)=1.25$ that are all deterministic quantities.
Similarly, for\st{class} \vmn{type}-$2$ customers, $o_2(5/8) = 2$, $o^\RGS_2(5/8) = 1$, $n_2=2$, $\eta_2(5/8) = 1$, $\tilde o_2(5/8) = 0.5\times 1 + 0.5\times 2 \times (5/8)=1.125$, and $\tilde o^\RGS_2(5/8) = 0.5\times 2 \times (5/8)=0.625$.

For convenience of reference,  in Table~\ref{table:notations-partial} we present a summary of the defined notations.

\begin{table}[htbp]
\centering
\caption{Notations}
\begin{tabular}{| c | l |}
\hline
$\vmn{\vec{v}_I}$& $\vmn{\vec{v}_I}=(v_{I,1}, v_{I,2}, \dots, v_{I,n})$,\st{adversarial } initial customer sequence \\\hline
$\RGS$& subset of customers in the  {\RG} group   \\\hline
$\UPGS$& subset of customers in the  {\UPG} group   \\\hline
$\vec{V}$& $\vec{V}=(V_1,V_2, \dots, V_n)$, random customer arrival sequence \\\hline
$\vec{v}$& $\vec{v}=(v_1,v_2, \dots, v_n)$, a realization of $\vec{V}$ (what online algorithm actually observes)\\\hline
$n_j$ & number of\st{class} \vmn{type}-$j$, $j=1,2$, customers in $\vmn{\vec{v}_I}$ (which is the same as in $\vec{V}$)\\ \hline
%$\lambda$ & normalized time: $\lambda =0,1/n, \dots, 1$ \\ \hline
$\lambda$ & normalized time: $\lambda =1/n, \dots, 1$ \\ \hline
$O_j(\lambda)$ & random number of\st{class} \vmn{type}-$j$ customers arriving up to time $\lambda$ \\ \hline
$o_j(\lambda)$ & a realization of  $O_j(\lambda)$\\ \hline
$O^\RGS_j(\lambda)$ & random number of\st{class} \vmn{type}-$j$ customers in ${\RGS}$  arriving up to time $\lambda$ \\ \hline
$o^\RGS_j(\lambda)$ & a realization of $O^\RGS_j(\lambda)$ \\ \hline
$\eta_j(\lambda)$ & number of\st{class} \vmn{type}-$j$, $j=1,2$, customers among the first $\lambda n$ ones in $\vmn{\vec{v}_I}$  \\ \hline
$\tilde o_j(\lambda)$ & $(1-p)\eta_j(\lambda)+p \lambda n_j$ (a deterministic approximation of $O_j(\lambda)$) \\ \hline
$
\tilde o^\RGS_j(\lambda)$ &  $p \lambda n_j$ (a deterministic approximation of $O^\RGS_j(\lambda)$)  \\
\hline
\end{tabular}\label{table:notations-partial}
\end{table}

Finally, to avoid carrying cumbersome expressions in the statement of our results for second-order quantities (e.g., bounds on approximation errors), we use the following approximation notations.

\begin{definition}
\label{def:bigO}
Suppose $f,g:{\mathcal{X}}\rightarrow {\mathbb{R}}$ are two functions defined on set $\mathcal{X}$. We use the notation $f = O(g)$ if there exists a constant $k$ such that $f(x) < k g(x)$ for all $x \in \mathcal{X}$.
\end{definition}

\begin{definition}
\label{def:bigO}
Suppose $f,g: {\mathbb{N}}\rightarrow {\mathbb{R}}$ are two functions defined on natural numbers. We use the notation $f = o(g)$ if $\lim_{n \rightarrow \infty} \frac{f(n)}{g(n)} = 0$, and the notation $f = \omega(g)$ if  $\lim_{n \rightarrow \infty} \lvert{\frac{f(n)}{g(n)}}\rvert = \infty$.
\end{definition}
%for each $j=1,2$, the adversary specifies the total number of class-$j$ customers, denoted by $n_j$.
%In addition to $n_j$, at each time step $\lambda = 0,1/n, \dots, 1$, the adversary also specifies the number of class-$j$ customers among the first $\lambda n$ customers in $\vmn{\vec{v}_I}$, denoted by $\eta_j(\lambda)$.
%Note that both $n_j$ and $\eta_j(\lambda)$ are deterministic.

\subsection{Estimating Future Demand}
\label{subsec:concentration}

At time $\lambda < 1$, upon observing $o_j(\lambda)$, $j=1,2$ (but not $n_j$ and $\eta_j(\lambda)$), we wish to estimate future demand, or equivalently the total demand $n_j$.
To make such an estimation, we establish the following concentration result:

\begin{lemma}\label{lemma:needed-centrality-result-for-m=2}
\vmn{Define constants  $\alpha \triangleq 10 + 2 \sqrt{6}$, $\bar\epsilon \triangleq 1/24$, and $k  \triangleq 16$.
For any $\epsilon \in [\frac{1}{n}, \bar\epsilon]$, with probability at least $1 - \epsilon$, all the following statements hold:
}
%There exist positive constants $\alpha$, $\bar\epsilon$ and $k$ such that when $\frac{1}{n} \leq \epsilon \leq \bar\epsilon $, with probability at least $1 - \epsilon$, all the following hold:
\begin{itemize}
\item If $n_1 \geq \frac{k}{p^2} \log n$, then for all $\lambda \in \{0,1/n, 2/n, \dots, 1\}$,
\begin{subequations}
\begin{align} & \left| O_1(\lambda)-\tilde o_1 (\lambda) \right| < \alpha \sqrt{n_1 \log n} \text{, and } \label{inequality:good-approximation-o_1}
\\ & \left| O_1(\lambda)+O_2(\lambda) - (\tilde o_1 (\lambda) +\tilde o_2(\lambda))\right| < \alpha \sqrt{ (n_1+n_2) \log n} \label{inequality:good-approximation-o_1+o_2}
\end{align}
\end{subequations}
\item If $n_2 \geq \frac{k}{p^2} \log n$, then for all $\lambda \in \{ 0,1/n, 2/n, \dots, 1\}$,
\begin{subequations}
\begin{align} & \left| O_2(\lambda)-\tilde o_2 (\lambda) \right| < \alpha \sqrt{n_2 \log n} \text{, and } \label{inequality:good-approximation-o_2}
\\ & \left| O^\RGS_2(\lambda)-\tilde o^\RGS_2 (\lambda) \right| < \alpha \sqrt{n_2 \log n}. \label{inequality:good-approximation-app--o^R_2}
\end{align}
\end{subequations}
\end{itemize}
\end{lemma}
The lemma is proved in Appendix~\ref{sec:proof-of-lemma-m=2}. Given that there are two layers of randomization (selection of subset $\RGS$ and the random permutation), proving the above concentration results requires a fairly delicate analysis that builds upon several existing concentration bounds. Because proving concentration results is not the main focus of our work, we will not outline the proof in the main text, and refer the interested reader to Appendix~\ref{sec:proof-of-lemma-m=2}.\footnote{\vmn{We present the values of the constants, defined in the statement of the lemma, only to clarify that they exist and do not depend on $n$; however, they are not optimized.}} %\vm{Question: should we say that the constants are defined in the proof?}
Here we focus on the following two questions: (i) what is our motivation for using deterministic approximations $\tilde o_j(\lambda)$ and $\tilde o^\RGS_j(\lambda)$? and (ii) how do such approximations help us to estimate $n_j$?

To answer the first question, let us count the number of\st{class} \vmn{type}-$j$ customers in $O_j(\lambda)$ that belong to the {\RG} and  {\UPG} groups separately.
We start with the {\RG} group.
Roughly, a total of $p n_j$\st{class} \vmn{type}-$j$ customers belong to the {\RG} group, and a $\lambda$ fraction of them arrive by time $\lambda$, because these customers are spread almost uniformly over the entire time horizon.
As a result, there are approximately $p n_j \lambda$\st{class} \vmn{type}-$j$ customers from ${\RGS}$ arriving up to time $\lambda$.
Now we move on to the {\UPG} group: there are a total of $\eta_j(\lambda)$ of\st{class} \vmn{type}-$j$ customers in the first $\lambda n$ customers in $\vmn{\vec{v}_I}$.
Since with probability $1-p$ each of them will be in the {\UPG} group, the total number of\st{class} \vmn{type}-$j$ customers from the {\UPG} group arriving up to time $\lambda$ is approximately  $(1-p)\eta_j(\lambda)$. Combining these two approximate counting arguments  gives us:
%Therefore, for $j=1,2$, with high probability,
\begin{align}
\label{eq:estimate}
O_j(\lambda) \approx (1-p) \eta_j(\lambda) + p\lambda n_j  = \tilde o_j(\lambda) .
\end{align}
A similar argument shows that $O_j^\RGS(\lambda)\approx p\lambda n_j  = \tilde o^\RGS_j(\lambda) $. Lemma~\ref{lemma:needed-centrality-result-for-m=2} confirms  that these  approximations hold with high probability. Lemma~\ref{lemma:needed-centrality-result-for-m=2} also it provides upper bounds on the corresponding approximation errors. Further, we note that $\tilde o_j(\lambda) \neq \E{O_j(\lambda)}$, as shown in \vmn{Appendix}
\ref{subsec:remark}.
%~\ref{remark:deterministic-approx-vs-expected-value}.
However, the difference between the two is very small \vmn{and vanishing in $n$}. Given that $\tilde o_j(\lambda)$ provides a very intuitive deterministic
approximation for random variable $O_j(\lambda)$ and admits a simple closed-form expression, we use it instead of the $\E{O_j(\lambda)}$.

Now, let us answer the second posed question. There are simple relations between $n_j$ and $\eta_j(\lambda)$ such as $n_j \geq \eta_j(\lambda)$ and $\eta_j(\lambda) + (1-\lambda)n \geq n_j$.\footnote{The first inequality follows from definition. The second one also follows from definition and from the observation that the number of\st{class} \vmn{type}-$j$ customers arriving between $\lambda$ and $1$ cannot be more than the number of remaining time steps, i.e., $(1-\lambda)n$.} Combining these with our deterministic approximations leads us to compute upper bounds on the total number of customers as established in Lemma~\ref{prop:full-Ubounds}.

Finally, based on Lemma~\ref{lemma:needed-centrality-result-for-m=2}, \st{fixing $\epsilon \in [\frac{1}{n}, \bar\epsilon]$,} we partition the sample space of arriving sequences into two subsets, $\mathcal{E}$ and its complement $\bar{\mathcal{E}}$, and  define event $\mathcal{E}$ as follows:

\begin{definition}
\label{def:event}
Event $\mathcal{E}$ occurs if the realized arrival sequence $\vec{v}$ satisfies all the conditions of Lemma~\ref{lemma:needed-centrality-result-for-m=2}, i.e.,
\begin{itemize}
\item If $n_1 \geq \frac{k}{p^2} \log n$, then for all $\lambda \in \{0,1/n, 2/n, \dots, 1\}$, $$ \left| o_1(\lambda)-\tilde o_1 (\lambda) \right| < \alpha \sqrt{n_1 \log n}  \quad \quad     \text{and}  \quad \quad \left| o_1(\lambda)+o_2(\lambda) - (\tilde o_1 (\lambda) +\tilde o_2(\lambda))\right| < \alpha \sqrt{ (n_1+n_2) \log n}, $$
\item  If $n_2 \geq \frac{k}{p^2} \log n$, then for all $\lambda \in \{ 0,1/n, 2/n, \dots, 1\}$, $$ \left| o_2(\lambda)-\tilde o_2 (\lambda) \right| < \alpha \sqrt{n_2 \log n} \quad \quad \text{and}  \quad \quad \left| o^\RGS_2(\lambda)-\tilde o^\RGS_2 (\lambda) \right| < \alpha \sqrt{n_2 \log n}.$$
\end{itemize}
\end{definition}

Lemma~\ref{lemma:needed-centrality-result-for-m=2} confirms that event $\mathcal{E}$ occurs {\em with high probability}. In all our analyses, we use the above definition to focus on the\st{(very likely case)} {\vmn{event}} that the deterministic approximations (i.e., $\tilde o_j (\lambda)$) are in fact ``very close'' to the observed sequence. This greatly helps us simplify the analysis and its presentation.

\section{A Non-Adaptive Algorithm}
\label{sec:alg1}
In this section, we \st{design}\vmn{present} and analyze \st{a}\vmn{our first online}\st{non-adaptive} algorithm \vmn{for the resource allocation problem and the demand model described in Section~\ref{sec:prem}}.
First, in Section~\ref{sec:alg1-alg}, we describe the algorithm.
\vmn{Then, in Section~\ref{sec:HBY-analysis}, we present the analysis of its competitive ratio}.

\subsection{The Algorithm}\label{sec:alg1-alg}
Our first algorithm is a non-adaptive online algorithm that uses predetermined \vmn{dynamic} thresholds to accept or reject customers.
This algorithm combines some ideas from the primal algorithm of~\cite{Kesselheim} and the threshold algorithm of~\cite{Ball2009} to \st{capture}\vmn{generate} maximal revenue from both the \st{predictable}\vmn{stochastic} and \st{unpredictable}\vmn{adversarial} components of the demand.

In particular, our non-adaptive algorithm makes use of the fact that customers from the {\RG}  group are \vmn{uniformly} spread over the entire horizon.
Therefore, at least a fraction $p$ of the inventory should be allocated at a roughly constant rate.
To this end, we define an \emph{evolving threshold} that works as follows: at any time $\lambda$, accept a \st{class}\vmn{type}-$2$ customer if the total number of accepted customers by this rule does not exceed \st{$\lambda pb$}\vmn{${\lfloor\lambda pb \rfloor}$}.

However, the arrival pattern of the other $1-p$ fraction can take any arbitrary form. In particular, if the adversary puts many \st{class}\vmn{type}-$2$ customers at the very beginning of the time horizon but none toward the end, then we may reject too many \st{class}\vmn{type}-$2$ customers\st{ from the \UPG group} early on.
To prevent this loss, we keep another quota for a \st{class}\vmn{type}-$2$ customer rejected by the evolving threshold. We only reject that customer if the number of such \st{class}\vmn{type}-$2$ customers accepted so far exceeds the \emph{fixed} threshold of $\theta \triangleq \frac{1-p}{2-a}$. When $p=0$, this is the same threshold as in~\cite{Ball2009}.

The formal definition of our algorithm is presented in Algorithm~\ref{algorithm:hybrid}.
\vmn{Note that $q_1$,  $q_{2,e}$, and $q_{2,f}$ respectively represent counters for the number of accepted \st{class}\vmn{type}-$1$ customers, the number of \st{class}\vmn{type}-$2$ customers accepted by the evolving threshold, and the number of \st{class}\vmn{type}-$2$ customers accepted by the fixed threshold.}
\st{
Note that $q_1$ represents  the number of accepted \st{class}\vmn{type}-$1$ customers; $q_{2,e}$ represents the number of \st{class}\vmn{type}-$2$ customers accepted by the evolving threshold, and $q_{2,f}$ represents the number of \st{class}\vmn{type}-$2$ customers accepted by the fixed threshold.}

%do not depend on the observed data. This algorithm makes use of the fact that the \RG s are spread in the entire horizon.
%Therefore, at least a $p$ fraction of the inventory should be allocated at a constant rate.
%However, the adversarial part of the demand can mislead us by having many class-$1$ customers in the beginning, which would result in rejecting many class-$2$ customers early on, but none towards the end.
%\dawsen{this sentence is not correct as, like you mentioned in a previous sentence, whether we accept a class-$2$ customer does not depend on what we observed. ``However, if the adversary put many class-$2$ customers in the very beginning of the time horizon but none towards the end, then we will reject too many class-$2$ unfollowers since we have a low booking-limit threshold early on.''}
%To avoid this, we keep another threshold that guarantees we are not accepting too few low fare customers.
%We show that this algorithm achieves competitive ratio of ? for $b/n ? $ (Theorem ?) which is the convex combination of the known results for $p=0$ \cite{?} and $p=1$ \vm{find a better citation than agrawal here... her results do not apply directly...}
%\dawsen{Kleinberg's paper will do (Kleinberg, R. (2005). A multiple-choice secretary algorithm with applications to online auctions.
%In Proceedings of the sixteenth annual ACM-SIAM symposium on Discrete algorithms,
%pages 630–631. Society for Industrial and Applied Mathematics.)}. In Proposition ?, we show that this bound is the best achievable for any online algorithm for $b/n?$.

\begin{algorithm}[H]
\begin{enumerate}
\item Initialize $q_1, q_{2,e}, q_{2,f} \leftarrow 0$, and define $\theta \triangleq \frac{1-p}{2-a}$.
\item Repeat for time $\lambda = 1/n, 2/n, \dots, 1$, accept  customer $i=\lambda n$ arriving at time $\lambda$ if there is remaining inventory and one of the following conditions holds:
\begin{enumerate}
\item $v_i = 1$; update $q_1\leftarrow q_1+1$.
\item {\bf Evolving threshold rule:} $v_i = a$ and $q_1+q_{2,e}< {\lfloor\lambda pb \rfloor}$; update  $q_{2,e} \leftarrow q_{2,e} +1$.
\item {\bf Fixed threshold rule:} $v_i = a$ and $q_{2,f} < {\lfloor \theta b\rfloor}$; update $q_{2,f} \leftarrow q_{2,f} +1$.
\end{enumerate}
We prioritize the evolving threshold rule if both of the last two conditions are satisfied.
\end{enumerate}
\caption{ Online Non-adaptive Algorithm ($ALG_1$)}\label{algorithm:hybrid}
\end{algorithm}
%\end{document}

%\input{main_head}
\subsection{Competitive Analysis}\label{sec:HBY-analysis}
\st{\vm{The theorem was stated twice! We can't just copy the whole appendix here...  }
\vm{I restructured the proof; defined $v_A$ in Lemma 4.7 instead of $\bar{v}$; fixed numerous typos. One question: we don't use $\theta b$ integer anymore, right? I changed all of them to $\lfloor \theta b \rfloor$. }}

In this subsection, we analyze  the competitive ratio of Algorithm~\ref{algorithm:hybrid}.
Our main result is the following theorem:
\begin{theorem}\label{thm:hybrid}
For $p\in (0,1)$, the competitive ratio of Algorithm~\ref{algorithm:hybrid} is at least $p+\frac{1-p}{2-a}- O\left(\frac{1}{a(1-p)p} \sqrt{\frac{\log n}{b}}\right)$ in the partially \st{learnalbe}\vmn{predictable} model.
\end{theorem}

\noindent \vmn{Before proceeding to the proof of the above theorem, we make the following remarks:}

\begin{remark}
\label{rem:alg1:3}
Our competitive analysis of Algorithm~\ref{algorithm:hybrid} is tight (up to an $O\left(\sqrt{\frac{\log n}{b}}\right)$ term). In particular, for the following instance,  Algorithm~\ref{algorithm:hybrid} can attain only a $p + \frac{1-p}{2-a}$ fraction of the \vmn{optimum} offline solution: Suppose $b = n$ and all customers \st{belong to}\vmn{are of} \st{class}\vmn{type}-$2$. The revenue of the \st{optimal}\vmn{optimum} offline algorithm is $ab$.
On the other hand, if we employ Algorithm~\ref{algorithm:hybrid}, at the end we will have $q_1=0$, $q_{2,e}\leq pb$ and $q_{2,f}\leq \vmn{\theta b} $. This results in a \vmn{competitive ratio} of at most $p+\theta \vmn{ = p+\frac{1-p}{2-a}}.$
\end{remark}

\begin{remark}
\label{rem:alg1:2}
In Subsection~\ref{sec:upper-bounds}, we prove that no online algorithm can have a competitive ratio larger than $p+\frac{1-p}{2-a}+ o(1)$ when $b  = o\left(\sqrt{n}\right)$. On the other hand, Theorem~\ref{thm:hybrid} indicates that Algorithm~\ref{algorithm:hybrid} achieves a competitive ratio of $p+\frac{1-p}{2-a}-o(1)$ when $b = \omega (\log n) $. \st{Comparing}\vmn{Combining} the two results \vmn{implies that} for fixed $a$ and  $p$, Algorithm~\ref{algorithm:hybrid} achieves the best possible competitive ratio (up to an $o(1)$ term) in the regime where conditions  $b=\omega (\log n)$ and $b= o(\sqrt{n})$ hold simultaneously.
\end{remark}

\begin{remark}
\label{rem:alg1:1}
Note that even though $p+\frac{1-p}{2-a}$ is the convex combination of the \st{worst-case bound}\vmn{competitive ratios} of~\cite{Ball2009} and \st{the average-case one}of~\cite{Agrawal2009b}, it cannot be achieved by simply randomizing between these two algorithms. Suppose we flip a biased coin;  with probability $p$, we  follow the algorithm of~\cite{Agrawal2009b}  (or any other algorithms designed for a random order model such as~\cite{Kesselheim}) and with probability $(1-p)$ we  follow the fixed threshold algorithm of~\cite{Ball2009}.
In Subsection~\ref{sec:bad-instance-for-other-models} we show that for a certain class of instances, such a randomized algorithm does not generate $p+\frac{1-p}{2-a}$ fraction of the \vmn{optimum} offline solution.
%
%This will not result in a $p+\frac{1-p}{2-a}$ competitive algorithm. We show , because as shown in Section~\ref{sec:bad-instance-for-other-models}, under partially-learnable model the learning based algorithms do not generate almost optimal revenue. % but only a  $p+\frac{b}{n}(1-p)$ fraction of the optimal revenue.
\end{remark}
%\vm{in the proof, I tried to take care of the the constraint that $\lambda$ belongs to the discrete set of $1/n, \ldots$. Please keep this in mind when proofreading.}

\begin{proof}{\textbf{Proof of Theorem~\ref{thm:hybrid}:}}
\vmn{We start the proof by making the following observation:
Theorem~\ref{thm:hybrid} is nontrivial only if $\sqrt{\frac{\log n}{b}}$ is small enough, such that the approximation term $O(\cdot)$ is negligible. Therefore, without loss of generality,  we can restrict attention to the case where $\sqrt{\frac{\log n}{b}}$ is small. In particular,
\vmn{recalling that we defined constant $\bar\epsilon = 1/24$ in Lemma~\ref{lemma:needed-centrality-result-for-m=2}, if $\frac{1}{a(1-p)p}\sqrt{\frac {\log n}{b}} \geq \bar\epsilon$,}
\if false
%If $b$ is too small then the approximation term (i.e., $O(.)$ term) in Theorem~\ref{thm:hybrid} becomes
let us define constant $\vm{k'}$ as $\vm{k'}\triangleq \min\left\{ \frac{1}{\sqrt{k}}, \bar\epsilon, 1 \right\}$, where constants $k$ and $\bar\epsilon$ are defined in Lemma~\ref{lemma:needed-centrality-result-for-m=2}. \pj{Since $k=16$ and $\bar\epsilon=1/24$  in Lemma~\ref{lemma:needed-centrality-result-for-m=2}, then $\vm{k'}$ simply becomes the same as $\bar\epsilon$. So why introduce $\vm{k'}$?}
If $\frac{1}{a(1-p)p}\sqrt{\frac {\log n}{b}} \geq \vm{k'}$,
\fi
then $O\left(\frac{1}{a(1-p)p}\sqrt{\frac {\log n}{b}}\right)$ becomes $O(1)$ and Theorem~\ref{thm:hybrid} becomes trivial.
Therefore,  without loss of generality, we assume $\frac{1}{a(1-p)p}\sqrt{\frac {\log n}{b}} < \vmn{\bar\epsilon}$, or equivalently,
\begin{align}
b > \frac{1}{\vmn{\bar\epsilon}^2}\frac{\log n}{a^2(1-p)^2p^2}.\label{ineq:Alg1-non-trivial-case}
\end{align}
}

\st{For any initial customer sequence $\vmn{\vec{v}_I}$, }We denote the random revenue generated by Algorithm~\ref{algorithm:hybrid} by $ALG_1(\vec{V})$. To analyze $\E{ALG_1(\vec{V})}$ we condition it on the event $\mathcal{E}$\st{for an appropriately defined $\epsilon$}. \vmn{Thus we have:}
\vmn{
\begin{align*}
\frac{\E{ALG_1(\vec{V})}}{OPT(\vmn{\vec{v}_I})} \geq  \frac{\E{ALG_1(\vec{V})|\mathcal{E}}\prob{\mathcal{E}}}{OPT(\vmn{\vec{v}_I})}.
\end{align*}
}
\vmn{Define $\epsilon \triangleq \frac{1}{a(1-p)p}\sqrt{\frac {\log n}{b}}$. For $b$ that satisfies condition \eqref{ineq:Alg1-non-trivial-case}, \vmn{and assuming
that $n \geq 3$\st{is large enough}, }we have $\frac{1}{n} \leq \epsilon \leq \bar\epsilon$.
%\pj{condition \eqref{ineq:Alg1-non-trivial-case} only implies that $\epsilon \leq \bar\epsilon$, the other inequality requires $n$ to be large enough.}
Therefore, we can apply Lemma~\ref{lemma:needed-centrality-result-for-m=2} to get:}

\vmn{
\begin{align*}
\frac{\E{ALG_1(\vec{V})}}{OPT(\vmn{\vec{v}_I})} \geq  \frac{\E{ALG_1(\vec{V})|\mathcal{E}}\prob{\mathcal{E}}}{OPT(\vmn{\vec{v}_I})} \geq \frac{\E{ALG_1(\vec{V})|\mathcal{E}}}{OPT(\vmn{\vec{v}_I})} \left(1 - \epsilon \right).
\end{align*}}

\vmn{
This will allow us to focus on the realizations that belong to event $\mathcal{E}$.
%Claim~\ref{claim:epsilon-con} (stated at the end of the proof) proves that in fact condition $\frac{1}{n} \leq \epsilon \leq \bar\epsilon$ holds.
In the main part of the proof we show that, for any realization $\vec{v}$ belonging to event $\mathcal{E}$,}
\vmn{
\begin{align*}
\frac{{ALG_1(\vec{v})}}{OPT(\vmn{\vec{v}_I})} \geq  p+\frac{1-p}{2-a} - O(\epsilon).
\end{align*}
}

\if false
In particular
%let $\epsilon \triangleq \frac{1}{a(1-p)p}\sqrt{\frac {\log n}{b}}$. If we can show that
let $\alpha$, $k$, and $\bar\epsilon$ denote the constants  specified in Lemma~\ref{lemma:needed-centrality-result-for-m=2}, and $\epsilon \triangleq \frac{1}{a(1-p)p}\sqrt{\frac {\log n}{b}}$. If we show that $\frac{1}{n} \leq \epsilon \leq \bar\epsilon$, then we can apply Lemma~\ref{lemma:needed-centrality-result-for-m=2} which implies:

$$\frac{\E{ALG_1(\vec{V})}}{OPT(\vmn{\vec{v}_I})} \geq  \frac{\E{ALG_1(\vec{V})|\mathcal{E}}\prob{\mathcal{E}}}{OPT(\vmn{\vec{v}_I})} \geq \frac{\E{ALG_1(\vec{V})|\mathcal{E}}}{OPT(\vmn{\vec{v}_I})} \left(1 - \epsilon \right)$$

This will allow us to focus on the realizations that belong to event $\mathcal{E}$. Claim~\ref{claim:epsilon-con} (stated at the end of the proof) proves that in fact condition $\frac{1}{n} \leq \epsilon \leq \bar\epsilon$ holds. In the main part of the proof we show that for any realization $\vec{v}$ belonging to event $\mathcal{E}$,

$$\frac{{ALG_1(\vec{v})}}{OPT(\vmn{\vec{v}_I})} \geq  p+\frac{1-p}{2-a} - O(\epsilon).$$
%where $\Delta \triangleq \alpha \sqrt{b \log n}$.
This in turn implies the result (combined with a few small steps). Before proceeding with the proof, we make the following observation:
Theorem~\ref{thm:hybrid} is only non-trivial if $\sqrt{\frac{\log n}{b}}$ is small enough, so that the approximation term $O(\cdot)$ is negligible. Therefore, without loss of generality we can restrict attention to the case where $\sqrt{\frac{\log n}{b}}$ is small. In particular,
%If $b$ is too small then the approximation term (i.e., $O(.)$ term) in Theorem~\ref{thm:hybrid} becomes
let $\vm{k'}\triangleq \min\left\{ \frac{1}{\sqrt{k}}, \bar\epsilon, 1 \right\}$.
If $\frac{1}{a(1-p)p}\sqrt{\frac {\log n}{b}} \geq \vm{k'}$, then $O\left(\frac{1}{a(1-p)p}\sqrt{\frac {\log n}{b}}\right)$ becomes $O(1)$ and Theorem~\ref{thm:hybrid} becomes trivial.
Therefore,  without loss of generality, we assume $\frac{1}{a(1-p)p}\sqrt{\frac {\log n}{b}} < \vm{k'}$, or equivalently,
\begin{align}
b > \frac{1}{\vm{k'}^2}\frac{\log n}{a^2(1-p)^2p^2}.\label{ineq:Alg1-non-trivial-case}
\end{align}

\fi
Fixing a realization $\vec{v}$ that belongs to  event $\mathcal{E}$, we define $q_1(\lambda)$, $q_{2,e}(\lambda)$, and $q_{2,f}(\lambda)$ \vmn{to be }the values of \vmn{counters} $q_1$, $q_{2,e}$, and $q_{2,f}$ right after the algorithm determines whether to accept the customer arriving at time $\lambda$. Further, we define $\Delta \triangleq \alpha \sqrt{b \log n}$ (constant $\alpha$ is defined in Lemma~\ref{lemma:needed-centrality-result-for-m=2}). To analyze the competitive ratio  we analyze three cases separately.

\medskip
\noindent { \bf Case (i): $n_1 \geq \frac{k}{p^2} \log n$, and Algorithm~\ref{algorithm:hybrid} exhausts the inventory.}
\medskip

Note that when $n_1 \geq \frac{k}{p^2} \log n$, we can apply \vmn{the concentration result} \eqref{inequality:good-approximation-o_1} from Lemma~\ref{lemma:needed-centrality-result-for-m=2}.
When Algorithm~\ref{algorithm:hybrid} exhausts the inventory it is possible that the algorithm accepts \emph{too many} \st{class}\vmn{type}-$2$ customers, which results in rejecting \st{class}\vmn{type}-$1$ customers and losing revenue.  We control for this loss by establishing the following upper bound on the number of \st{class}\vmn{type}-$2$ customers accepted by the evolving threshold.\footnote{Note that we already have an upper bound on the number of \st{class}\vmn{type}-$2$ customers accepted by the fixed threshold: $q_{2,f}(1) \leq \theta b$.}
 In particular, we have the following lemma:
\begin{lemma}
\label{lem:HYB:full-case1}
Under event $\mathcal{E}$, if $n_1 \geq \frac{k}{p^2} \log n$, then
\begin{align*}
q_{2,e}(1) \leq p(b-n_1)^+ + \Delta.
%\label{ub:q2e}
\end{align*}
\end{lemma}
%\noindent{where $\Delta \triangleq \alpha \sqrt{b \log n}$.}
\begin{proof}{\textbf{Proof:}}
We assume, without loss of generality, that $n_1 \leq b$.
Otherwise, we \vmn{construct a modified adversarial instance, denoted by $\vec{v}_{I,M}$, as follows:}
\vmn{keep}
an arbitrary subset of \st{class}\vmn{type}-$1$ customers with size $b$ in $\vmn{\vec{v}_I}$ (before the random permutation), \vmn{and remove the remaining type-$1$ customers (e.g., set their revenue to be $0$)}.
\vmn{For the same realization of the {\RG} group and random permutation, we claim that at any time $\lambda \in \{1/n, \ldots, 1\}$, the number of \st{class}\vmn{type}-$2$ customers accepted through the evolving threshold rule
in the original instance is not larger than that in the modified one.
This holds because $o_1(\lambda, \vec{v}) \geq o_{1}(\lambda, \vec{v}_{M})$, where the second argument is added to $o_1(\cdot, \cdot)$ to  indicate the corresponding instance.
Note that because the algorithm accepts all type-$1$ customers, this implies $q_1(\lambda, \vec{v}) \geq q_{1}(\lambda, \vec{v}_{M})$, which proves our claim (i.e., $q_{2,e}(\lambda, \vec{v}) \leq q_{2,e}(\lambda, \vec{v}_{M})$). Thus, without loss of generality, we assume $n_1 \leq b$.
Further, note that because of condition \eqref{ineq:Alg1-non-trivial-case}, we have $n_{1}(\vec{v}_{M}) = b \geq \frac{k}{p^2} \log n$.\footnote{This follows from Condition \eqref{ineq:Alg1-non-trivial-case} and  the fact that $\frac{1}{a^2(1-p)^2} > 1$ and by definition (given in Lemma~\ref{lemma:needed-centrality-result-for-m=2}) $\frac{1}{\vmn{\bar\epsilon}^2} \geq k$ which imply $\frac{1}{\vmn{\bar\epsilon}^2}\frac{1}{a^2(1-p)^2} \geq k$.} Thus we are still in Case (i) for the
modified instance.
}

\if false
\vmn{We argue that the number of \st{class}\vmn{type}-$2$ customers accepted through the evolving threshold rule in the original case is not larger than that in the case with a subset of $b$ customers.}
For any $\lambda$, with the same realization of the {\RG} group and random permutation, the value of $o_1(\lambda)$ corresponding to these $b$ \st{class}\vmn{type}-$1$ customers cannot be larger than the original ones.
As a result, the number of \st{class}\vmn{type}-$2$ customers accepted through the evolving threshold rule in the original case is not larger than that in the case with a subset of $b$ customers.
Hence the upper bound proven in the case with a subset of $b$ customers is still valid for the original case.
Note that we can still apply Inequalities~\eqref{inequality:good-approximation-o_1} to these $b$ customers.
\fi

If no \st{class}\vmn{type}-$2$ customer is accepted by the evolving threshold, then $q_{2,e}(1)=0$ and \st{we are done}\vmn{the proof is complete}.
Otherwise, let $\bar{\lambda} \leq 1$ be the last time that a \st{class}\vmn{type}-$2$ customer is accepted by the evolving rule.
Then we have
\begin{align*}
q_{2,e}(1) = & q_{2,e}(\bar{\lambda}) \leq \bar{\lambda} pb - o_1(\bar{\lambda}) &(\text{Evolving threshold rule}) \\
\leq & \bar{\lambda} pb - (\bar \lambda p n_1 + (1-p)\eta_1(\bar\lambda) -\Delta) ) &(\eqref{inequality:good-approximation-o_1}) \\
\leq & p(b-n_1) + \Delta . &(\eta_1(\bar\lambda) \geq 0 \text{, } n_1\leq b   \text{, and } \bar\lambda \leq 1)
\end{align*}

\vmn{The reason for each inequality appears in the same line. We remark that in the second inequality,  we crucially use the concentration result of  Lemma~\ref{lemma:needed-centrality-result-for-m=2}.}
\end{proof}

\noindent Using Lemma~\ref{lem:HYB:full-case1}, we prove, in Appendix~\ref{sec:proof-of-alg1}, the following lemma \vmn{that gives a lower bound on the competitive ratio for Case (i)}:

\begin{lemma}
\label{claim:full-case-comp1}
Under event $\mathcal{E}$, if $n_1 \geq \frac{k}{p^2} \log n$ and $q_1(1) + q_{2,e}(1) +q_{2,f}(1) =b$, then
$\frac{ALG_1(\vec{v})}{OPT(\vec{v})} \geq p + \frac{1-p}{2-a} - \frac{(1-a)\Delta}{ab}$.
\end{lemma}
%\begin{proof}
%The revenue of Algorithm~\ref{algorithm:hybrid} in this case is $ ALG_1(\vec{v}) = b - (1-a) \left[ q_{2,e}(1) + q_{2,f}(1) \right],$ which is decreasing in $q_{2,e}(1) + q_{2,f}(1)$.
%Note that due to the fixed threshold rule, we already have an upper bound on $q_{2,f}(1)$, i.e., $q_{2,f}(1) \leq \theta b$.
%As a result, using Lemma~\ref{lem:HYB:full-case1}, $ALG_1(\vec{v}) \geq b - (1-a)(\theta b + p(b-n_1)^+ +\Delta)$.
%Note that $OPT(\vec{v}) \leq b- (1- a) \left(b-n_1\right)^+$, and thus
%\begin{align*} \frac{ALG_1(\vec{v})}{OPT(\vec{v})} \geq & \frac{b - (1-a)(\theta b + p(b-n_1)^+ + \Delta)}{b- (1- a) \left(b-n_1\right)^+}
%\\ \geq & \frac{b - (1-a)(\theta b + (p + \theta) (b-n_1)^+ + \Delta)}{b- (1- a) \left(b-n_1\right)^+} &(\theta \geq 0)
%\\ = &p + \frac{1-p}{2-a}- \frac{(1-a)\Delta}{OPT(\vec{v})}.
%\end{align*}
%The rest follows from the simple inequality $OPT(\vec{v}) \geq ab$ due to $q_1(1) + q_{2,e}(1) +q_{2,f}(1) =b$.
%\end{proof}

\medskip
\noindent { \bf Case (ii): $n_1 \geq \frac{k}{p^2} \log n$, and Algorithm~\ref{algorithm:hybrid} does not exhaust the inventory.}
\medskip

In this case, all \st{class}\vmn{type}-$1$ customers are accepted.
Therefore, the ratio between $ALG_1(\vec{v})$ and $OPT(\vec{v})$ can be expressed as:
$$
\frac{ALG_1(\vec{v})}{OPT(\vec{v})} = \frac{n_1 + a \left[q_{2,e}(1) +q_{2,f}(1)\right]}{n_1 + a \min \{n_2,(b-n_1)\}}.
$$

The only ``mistake'' that the algorithm may make is to reject too many \st{class}\vmn{type}-$2$ customers. The following lemma establishes a lower bound on the number of accepted \st{class}\vmn{type}-$2$ customers:
\begin{lemma}
\label{lem:HYB:full-case2}
Under event $\mathcal{E}$, if $n_1 \geq \frac{k}{p^2} \log n$ and $q_1(1) + q_{2,e}(1) +q_{2,f}(1) <b$, then one of the following conditions holds:
\begin{enumerate}[(a)]
\item $ q_{2,e}(1) +q_{2,f}(1) = n_2$,
\item $q_{2,f}(1) = \lfloor \theta b\rfloor$ and $n_1 > bp - 3 \Delta $, or
\item $q_{2,f}(1) =\lfloor \theta b\rfloor$, $n_1 \leq bp - 3 \Delta $, and $q_{2,e}(1) \geq \left(p(n_1+n_2) -n_1 - 5 \Delta \right)^+$.
\end{enumerate}
%\noindent{where $\Delta \triangleq \alpha \sqrt{b \log n}$.}
\end{lemma}
\begin{proof}{\textbf{Proof:}}
First note that $q_{2,f}(1) < \lfloor \theta b \rfloor$ means that Algorithm~\ref{algorithm:hybrid} never rejects a \st{class}\vmn{type}-$2$ customer.
This implies that $ q_{2,e}(1) +q_{2,f}(1) = n_2$, \vmn{i.e., condition (a) holds}.
Now suppose $q_{2,f}(1) = \lfloor \theta b \rfloor $.
If $n_1 > bp - 3\Delta $, then condition (b) holds.
The most interesting case is when $q_{2,f}(1) = \lfloor \theta b \rfloor $, and $n_1 \leq bp - 3\Delta$.
\vmn{In the following, we show that in this case, condition (c) will hold.}

\st{Now we consider the case where $q_{2,f}(1) = \lfloor \theta b  \rfloor $, and $n_1 \leq bp - 3\Delta$.}
In this case, without loss of generality, we can assume $n_1+n_2 \leq b$.
\vmn{Otherwise, we construct an alternative adversarial instance, denoted by $\vec{v}_{I,A}$, as follows:}
\vmn{keep an arbitrary subset of \st{class}\vmn{type}-$2$ customers with size $b-n_1$ in $\vmn{\vec{v}_I}$ (before the random permutation), \vmn{and remove the remaining type-$2$ customers (e.g., set their revenue to be $0$)}.
With the same realization of the {\RG} group and random permutation, we claim that:}

\vmn{
\begin{align}\label{eq:induc:1}
q_{2,e}(\lambda, \vec{v}) \geq q_{2,e}(\lambda, \vec{v}_A), ~~ \lambda \in \{0, 1/n, \ldots, 1\}.
\end{align}
}

%$q_{2,e}(\lambda, \vec{v}) \geq q_{2,e}(\lambda, \vec{v}_A)$, for $\lambda \in \{0, 1/n, \ldots, 1\}$.
\st{Otherwise, we consider an alternative adversarial instance $\vec{v}_{I,A}$ that keeps the values of arbitrary $b-n_1$ \st{class}\vmn{type}-$2$ customers in $\vmn{\vec{v}_I}$ but modifies the values of all other \st{class}\vmn{type}-$2$ customers in $\vmn{\vec{v}_I}$ to $0$'s.
In the following we show that, with the same realization of the {\RG} group and random permutation, for any $\lambda$, the value of $q_{2,e}(\lambda)$ corresponding to $\vmn{\vec{v}_I}$ (denoted by $q_{2,e}(\lambda, \vec{v})$) is not smaller than that corresponding to ${\vec{v_A'}}$ (denoted by $q_{2,e}(\lambda, {\vec{v_A}})$), that is, for all $\lambda$,$$q_{2,e}(\lambda, \vec{v}) \geq q_{2,e}(\lambda,  {\vec{v_A}}).$$}

To \st{see} \vmn{show \eqref{eq:induc:1},} \st{this,} we use induction.
The \vmn{base case, corresponding to taking }$\lambda = 0$,  is trivial.
\vmn{Suppose \eqref{eq:induc:1} holds for $\lambda- 1/n$. We show it will hold for $\lambda$ as well.}
At time $\lambda$, if $q_{2,e}(\lambda,  \vec{v}_A) = q_{2,e}(\lambda- 1/n, \vec{v}_A)$, then  \eqref{eq:induc:1} holds \vmn{because $q_{2,e}(\lambda,  {\vec{v}}) \geq q_{2,e}(\lambda- 1/n, {\vec{v}})$}.
Otherwise, $q_{2,e}(\lambda,  \vec{v}_A) = q_{2,e}(\lambda- 1/n,  \vec{v}_A)+1$.
This implies that a \st{class}\vmn{type}-$2$ customer arrives at time $\lambda$ in $\vec{v}_A$, and thus also in $\vec{v}$.
If $q_{2,e}(\lambda, \vec{v}) = q_{2,e}(\lambda- 1/n, \vec{v}) +1$, then \eqref{eq:induc:1} again holds.
Otherwise, under customer arrival sequence $\vec{v}$, we do not accept the \st{class}\vmn{type}-$2$ customer at time $\lambda$ \vmn{by the evolving threshold rule}, which means that $o_1(\lambda, \vec{v}) + q_{2,e}(\lambda, \vec{v}) = {\lfloor\lambda pb \rfloor}.$
Because $o_1(\lambda,{\vec{v}}_A) + q_{2,e}(\lambda,  {\vec{v}}_A) \leq {\lfloor\lambda pb \rfloor}$, \vmn{and  $o_1(\lambda,{\vec{v}}) =  o_1(\lambda,{\vec{v}}_A)$, }we can conclude that \st{ $q_{2,e}(\lambda,  {\vec{v}}_A) \leq q_{2,e}(\lambda, \vec{v})$} \vmn{\eqref{eq:induc:1} holds in the last case as well}. This concludes the induction.
%Because $q_{2,e}(1, I) \geq q_{2,e}(1, \bar {\vec{v}}) $, $q_{2,e}(1, \bar {\vec{v}})$ is a valid lower bound for what we want to establish.
Thus, without loss of generality, we assume $n_1+n_2 \leq b$.

\vmn{To prove that condition (c) holds when $q_{2,f}(1) =\lfloor \theta b\rfloor$ and $n_1 \leq bp - 3 \Delta$, we make two important observations: (i)
In this case, the number of type-$2$ customers is large enough to apply the concentration results of~\eqref{inequality:good-approximation-app--o^R_2}. In particular, we have:}
\vmn{
\begin{align}
\label{epsilon-condition-n_2-big}
n_2\geq \theta b \geq \frac{ k \log n}{p^2}
\end{align}
}\vmn{where the last inequality holds because of \eqref{ineq:Alg1-non-trivial-case}, and definitions of $\theta =\frac{1-p}{2-a}$ and $k$ (defined in Lemma~\ref{lemma:needed-centrality-result-for-m=2}). %and $\vm{k'}\leq 1/\sqrt{k}$.
(ii) The number of type-$1$ customers
 is so small that after a certain time the evolving threshold accepts a sufficient number of type-$2$ customers that ensures  condition (c) holds.
 In particular, define $$ \bar\lambda \triangleq \frac{1}{n} \lceil \frac{n(n_1(1-p)+ 3\Delta)}{p(b-n_1)}\rceil.$$ Note that $\bar \lambda \leq 1$ when $n_1 \leq bp-3\Delta$.
For any $\lambda \geq \bar\lambda$, we have:
}
\vmn{
\begin{align*}
o_1(\lambda)+o^\RGS_2(\lambda)-o^\RGS_2(\bar \lambda) & \leq \lambda p n_1+(1-p) \eta_1 (\lambda) + \Delta +\lambda pn_2 +\Delta - (\bar\lambda pn_2 - \Delta ) &(\eqref{inequality:good-approximation-o_1},\eqref{inequality:good-approximation-app--o^R_2})\\
&\leq \lambda p n_1+(1-p) n_1 +(\lambda - \bar\lambda)pn_2 + 3 \Delta &(\eta_1(\lambda)\leq n_1)\\
& = \bar \lambda p n_1+(1-p) n_1 +(\lambda - \bar\lambda)p (n_1+n_2) + 3 \Delta \\
& \leq \bar \lambda p n_1+(1-p) n_1 +(\lambda - \bar\lambda)p b + 3 \Delta &(n_1+n_2 \leq b)  \\
& \leq \lambda pb. &(\vmn{\text{definition of $\bar\lambda$}})
\end{align*}
}

\noindent Note that because $o_1(\lambda)+o^\RGS_2(\lambda)-o^\RGS_2(\bar \lambda)$ is an integer, the above inequality also implies
\begin{align}
\label{ineq:case2:a3}
o_1(\lambda)+o^\RGS_2(\lambda)-o^\RGS_2(\bar \lambda) \leq \lfloor\lambda pb \rfloor \quad \quad \text{ for all }\lambda \geq \bar \lambda.
\end{align}
%$o_1(\lambda)+o^\RGS_2(\lambda)-o^\RGS_2(\bar \lambda) \leq \lfloor\lambda pb \rfloor$.}
\vmn{Further, the above inequality implies that for $\lambda \geq \bar\lambda$, there is a gap between $o_1(\lambda)$ and the evolving threshold $ \lfloor\lambda pb \rfloor$, which in turn implies that the evolving threshold will accept type-$2$ customers.
Next, for  $\lambda \geq \bar\lambda$, we establish a lower bound on the number of type-$2$ customers that the evolving threshold accepts.
In particular, we show that
}

\begin{align}
\label{ineq:case2:a2}
q_{2,e}(\lambda) \geq o_2^\RGS (\lambda) - o_2^\RGS (\bar \lambda) \quad \quad \text{ for all }\lambda \geq \bar \lambda.
\end{align}

\vmn{We show \eqref{ineq:case2:a2} by induction. }
The base case $\lambda = \bar{ \lambda}$ is trivial.
\vmn{Suppose \eqref{ineq:case2:a2} holds for $\lambda-1/n \geq \bar{ \lambda}$.
We show it will also hold for $\lambda$: }
If the arriving customer is not a \st{class}\vmn{type}-$2$ customer belonging to the {\RG} group, then $o_2^\RGS ({\lambda}) = o_2^\RGS (\lambda -1/n)$; but $q_{2,e}(\lambda) \geq q_{2,e}(\lambda-1/n)$, and thus \eqref{ineq:case2:a2} holds.
\vmn{Otherwise,} we have $o_2^\RGS ({\lambda}) = o_2^\RGS ({\lambda-1/n}) + 1$. Now if this customer is accepted \vmn{by the evolving threshold rule}, then both sides of \eqref{ineq:case2:a2} are increased by one, and thus inequality \eqref{ineq:case2:a2}  still holds.
Otherwise, if the customer is not accepted, it implies we have reached the threshold. Therefore

\begin{align}
\label{eq:intermed1}
q_{2,e}(\lambda)= \lfloor \lambda p b\rfloor - o_1(\lambda).
\end{align}

\vmn{Now we utilize the gap between $\lfloor \lambda p b\rfloor$ and $o_1(\lambda)$ that we established above in~\eqref{ineq:case2:a3}.}
%\vm{In particular, we have shown that:
%\begin{align}
%\label{eq:intermed2}
%o_1(\lambda)+o^\RGS_2(\lambda)-o^\RGS_2(\bar \lambda) \leq \lfloor \lambda pb \rfloor \quad \quad \text{ for all }\lambda \geq \bar \lambda.
%\end{align}}
%
\vmn{Combining \eqref{eq:intermed1} and \eqref{ineq:case2:a3} proves that \eqref{ineq:case2:a2} holds in this case as well. This completes the induction, and thus
the proof of \eqref{ineq:case2:a2}.}

\vmn{We complete the proof of the lemma by} using \eqref{ineq:case2:a2} with $\lambda =1$, to have the following lower bound:
\begin{align*}
q_{2,e}(1) \geq & o_2^\RGS (1) - o_2^\RGS (\bar \lambda) &(\eqref{ineq:case2:a2}) \\
\geq & pn_2 - \Delta - (\bar\lambda pn_2 + \Delta) &(\eqref{inequality:good-approximation-app--o^R_2}) \\
\geq & pn_2 - (n_1(1-p)+ 3\Delta) - 2\Delta &(b-n_1 \geq n_2) \\
= & p(n_1+n_2) -n_1- 5\Delta.
\end{align*}

\if false

In order to apply Inequalities~\eqref{inequality:good-approximation-o_2} and~\eqref{inequality:good-approximation-app--o^R_2}, we need $n_2 \geq \frac{k \log n}{p^2}$.
Because $n_2\geq \theta b$, $n_2 \geq \frac{k \log n}{p^2}$ holds if
\begin{align}
\theta b \geq \frac{ k \log n}{p^2},\label{epsilon-condition-n_2-big}
\end{align}
which is given by \eqref{ineq:Alg1-non-trivial-case} and $\vm{k'}\leq 1/\sqrt{k}$.
%\begin{align}
%\vm{k'}\leq 1/\sqrt{k} \label{condition:ALG_11:\vm{k'}>1/sqrt-k}.
%\end{align}
We first show that, for any $\bar\lambda$ that satisfies the following condition:
\begin{align}
o_1(\lambda)+o^\RGS_2(\lambda)-o^\RGS_2(\bar \lambda) \leq \lambda pb \quad \quad \text{ for all }\lambda \geq \bar \lambda,\label{condition-for-bar-lambda}
\end{align}
the following inequalities hold:

\begin{align}
\label{ineq:case2:a2}
q_{2,e}(\lambda) \geq o_2^\RGS (\lambda) - o_2^\RGS (\bar \lambda) \quad \quad \text{ for all }\lambda \geq \bar \lambda.
\end{align}
We prove that \eqref{condition-for-bar-lambda} implies \eqref{ineq:case2:a2} by induction on $\lambda$.
The base case $\lambda = \bar{ \lambda}$ is trivial. At time $\lambda > \bar{ \lambda}$, if the arriving customer is not a \st{class}\vmn{type}-$2$ belonging to the {\RG} group, then $o_2^\RGS ({\lambda}) = o_2^\RGS (\lambda -1/n)$ but $q_{2,e}(\lambda) \geq q_{2,e}(\lambda-1/n)$, so we are done. Else, the customer is a \st{class}\vmn{type}-$2$ {\RG}, then $o_2^\RGS ({\lambda}) = o_2^\RGS ({\lambda-1/n}) + 1$. Now if this customer is accepted, then both sides of \eqref{ineq:case2:a2} are increased by one, so we are done.
Otherwise, the customers is not accepted, which implies we have reached the threshold, and therefore $o_1(\lambda)+q_{2,e}(\lambda)=\lambda p b,$ which, combined with~\eqref{condition-for-bar-lambda}, implies~\eqref{ineq:case2:a2}.

Here we find a sufficient condition on $\bar\lambda$ for Condition~\eqref{condition-for-bar-lambda} to hold.
\begin{align*}
o_1(\lambda)+o^\RGS_2(\lambda)-o^\RGS_2(\bar \lambda) & \leq \lambda p n_1+(1-p) \eta_1 (\lambda) + \Delta +\lambda pn_2 +\Delta - (\bar\lambda pn_2 - \Delta ) &(\eqref{inequality:good-approximation-o_1},\eqref{inequality:good-approximation-app--o^R_2})\\
&\leq \lambda p n_1+(1-p) n_1 +(\lambda - \bar\lambda)pn_2 + 3 \Delta &(\eta_1(\lambda)\leq n_1)\\
& = \bar \lambda p n_1+(1-p) n_1 +(\lambda - \bar\lambda)p (n_1+n_2) + 3 \Delta \\
& \leq \bar \lambda p n_1+(1-p) n_1 +(\lambda - \bar\lambda)p b + 3 \Delta. &(n_1+n_2 \leq b)
\end{align*}
As a result, Condition~\eqref{condition-for-bar-lambda} holds if
$ \bar \lambda p n_1 + (1-p) n_1 + 3\Delta \leq \bar \lambda p b.$
Thus, defining
$$ \bar\lambda \triangleq \frac{n_1(1-p)+ 3\Delta}{p(b-n_1)}$$
(note that $\bar \lambda \leq 1$ when $n_1 \leq bp-3\Delta$) and using \eqref{ineq:case2:a2} with $\lambda =1$, we have the bound:
\begin{align*}
q_{2,e}(1) \geq & o_2^\RGS (1) - o_2^\RGS (\bar \lambda) &(\eqref{ineq:case2:a2}) \\
\geq & pn_2 - \Delta - (\bar\lambda pn_2 + \Delta) &(\eqref{inequality:good-approximation-app--o^R_2}) \\
\geq & pn_2 - (n_1(1-p)+ 3\Delta) - 2\Delta &(b-n_1 \geq n_2) \\
= & p(n_1+n_2) -n_1- 5\Delta.
\end{align*}
\fi
\end{proof}
Using Lemma~\ref{lem:HYB:full-case2}, we prove, in Appendix~\ref{sec:proof-of-alg1}, the following lemma \vmn{that gives a lower bound on the competitive ratio for Case (ii)}:
\begin{lemma}
\label{claim:full-comp2}
Under event $\mathcal{E}$, if $n_1 \geq \frac{k}{p^2} \log n$ and $q_1(1) + q_{2,e}(1) +q_{2,f}(1) <b$, then
$$\frac{ALG_1(\vec{v})}{OPT(\vec{v})} \geq p + \frac{1-p}{2-a} - \frac{5 \Delta }{\theta b }.$$
%\noindent{where $\Delta \triangleq \alpha \sqrt{b \log n}$.}
\end{lemma}
%\begin{proof}
%We consider three cases of Lemma~\ref{lem:HYB:full-case2} separately.
%\begin{enumerate}
%\item $ q_{2,e}(1) +q_{2,f}(1) = n_2$:
%\\ Algorithm~\ref{algorithm:hybrid} accepts all customers and hence achieves the optimal revenue.
%\item $q_{2,f}(1) = {\lfloor \theta b\rfloor} $ and $n_1  > bp - 3 \Delta$:
%\\ $ ALG_1(\vec{v}) \geq n_1 + a \theta b $ and $OPT(\vec{v}) \leq ab+n_1(1-a)$. Therefore,
%$$ \frac{ALG_1(\vec{v})}{OPT(\vec{v})} \geq \frac{n_1+ a \theta b}{ab+n_1(1-a)}, $$
%which is increasing in $n_1$, so the ratio is minimized at $n_1=bp - 3 \Delta$, which is a special case of the third case.
%\item$q_{2,f}(1) =\lfloor \theta b\rfloor$, $n_1 \leq bp - 3 \Delta$, and $q_{2,e}(1) \geq \left(p(n_1+n_2) -n_1 - 5 \Delta \right)^+$:
%\begin{align*} ALG_1(\vec{v}) \geq & n_1+ a(p(n_1+n_2)-n_1 -5\Delta+ \theta b) \\
%\geq & n_1+ a(p(n_1+n_2)-n_1 -5\Delta+ \theta (n_1+n_2)) &(b \geq n_1+n_2)
%\\
%= & n_1(1-a+pa+\theta a) + n_2(p+\theta)a -5\Delta a \\
%\geq & n_1((1-a)(p+\theta)+a(p+\theta )) + n_2(p+\theta)a -5\Delta a &(p+\theta \leq 1) \\
%= &(p+\theta)(n_1+ n_2 a) -5\Delta a \geq (p+\theta)OPT(\vec{v}) -5\Delta a \\
%= &\left( p + \frac{1-p}{2-a}\right)OPT(\vec{v}) - 5\Delta a. &(p+\theta = p + \frac{1-p}{2-a})
%\end{align*}
%Since $OPT(\vec{v}) \geq q_{2,f}(1)a = \theta b a$,
%$$ \frac{ALG_1(\vec{v})}{OPT(\vec{v})} \geq p + \frac{1-p}{2-a} - \frac{5 \Delta a}{\theta b a} = p + \frac{1-p}{2-a} - \frac{5 \Delta }{\theta b }.$$
%\end{enumerate}
%\end{proof}

\medskip
\noindent{ \bf Case (iii): $n_1 < \frac{k}{p^2} \log n$.}
\medskip

The competitive ratio analysis for Case (iii) is fairly similar to that for Case (ii). It follows from the next two lemmas. The proofs are deferred to Appendix~\ref{sec:proof-of-alg1}.

\begin{lemma}
\label{lem:HYB:full-case3}
Under event $\mathcal{E}$, if $n_1 < \frac{k}{p^2} \log n$, then one of the following conditions holds:
\begin{enumerate}[(a)]
\item $ q_1(1) + q_{2,e}(1) +q_{2,f}(1) = b$,
\item $q_1(1)=n_1$ and $ q_{2,e}(1) +q_{2,f}(1) = n_2$, or
\item $q_1(1)=n_1$, $q_{2,f}(1) =\lfloor \theta b\rfloor$ and $q_{2,e}(1) \geq pn_2 - \frac{k}{p^2} \log n - 4 \Delta$.
\end{enumerate}
\end{lemma}
Using Lemma~\ref{lem:HYB:full-case3}, the following lemma (proven in Appendix~\ref{sec:proof-of-alg1}) gives a lower bound on the competitive ratio for Case(iii):

\begin{lemma}
\label{claim:full-comp3}
Under event $\mathcal{E}$, if $n_1 < \frac{k}{p^2} \log n$, then $$\frac{ALG_1(\vec{v})}{OPT(\vec{v})} = \frac{n_1 + a \left[q_{2,e}(1) +q_{2,f}(1)\right]}{n_1 + a \min \{n_2,(b-n_1)\}} \geq \min \left\{ p+\frac{1-p}{2-a}- \frac{\frac{k}{p^2} \log n}{ab} , p+\frac{1-p}{2-a} - \frac{\frac{k}{p^2} \log n + 4 \Delta}{\theta b}\right\}.$$
\end{lemma}
%\begin{proof}
%Following the discussion in the proof of Lemma~\ref{lem:HYB:full-case2}, we assume, without loss of generality, $n_1+n_2\leq b$.
%We consider three cases in Lemma~\ref{lem:HYB:full-case3} separately.
%
%If case $1$ in Lemma~\ref{lem:HYB:full-case3} happens, then $ALG_1(\vec{v})+n_1 \geq OPT(\vec{v})$ and $OPT(\vec{v}) \geq ab$.
%As a result,
%$$ \frac{ALG_1(\vec{v})}{OPT(\vec{v})} \geq 1-\frac{n_1}{OPT(\vec{v})} \geq 1- \frac{\frac{k}{p^2} \log n}{ab} \geq  p+\frac{1-p}{2-a}- \frac{\frac{k}{p^2} \log n}{ab} .$$
%
%If case $2$ in Lemma~\ref{lem:HYB:full-case3} happens, then $ALG_1(\vec{v})=OPT(\vec{v})$, and we are done.
%
%If case $3$ happens, then
%\begin{align*}
%ALG_1(\vec{v}) \geq n_1 + \left(pn_2 - \frac{k}{p^2} \log n - 4 \Delta + \theta b \right)a.
%\end{align*}
%Because $n_1 \geq \left( p+\frac{1-p}{2-a} \right) n_1 $, $p n_2 +\theta b \geq (p+\theta)n_2 = \left( p+\frac{1-p}{2-a}\right)n_2$ and $OPT(\vec{v}) = n_1 + a n_2 \geq a \theta b$, we have
%\begin{align}
%\frac{ALG_1(\vec{v})}{OPT(\vec{v})} \geq p+\frac{1-p}{2-a}
%-\frac{a\left(\frac{k}{p^2} \log n + 4 \Delta \right)}{{OPT(\vec{v})}} \geq p+\frac{1-p}{2-a} - \frac{\frac{k}{p^2} \log n + 4 \Delta}{\theta b} .\label{condition:HYA-epsilon-3}
%\end{align}
%\end{proof}

 Using Lemmas~\ref{claim:full-case-comp1},~\ref{claim:full-comp2}, and~\ref{claim:full-comp3}, we have lower bounds on the competitive ratio of Algorithm~\ref{algorithm:hybrid} for all possible cases. \vmn{We complete the proof of the theorem by the following lemma (proven in Appendix~\ref{sec:proof-of-alg1}) that ensures that the error terms in Lemmas~\ref{claim:full-case-comp1},~\ref{claim:full-comp2}, and~\ref{claim:full-comp3} are $O(\epsilon)$.}
\if false
To complete the proof, we just need to show the following two claims that are proven in Appendix~\ref{sec:proof-of-alg1}.

\begin{claim}
\label{claim:epsilon-con}
$\frac{1}{n}\leq \epsilon = \frac{1}{a(1-p)p}\sqrt{\frac {\log n}{b}} \leq \bar\epsilon$.
\end{claim}
\fi

\begin{lemma}
\label{claim:error-terms}
%The error terms in Lemmas ~\ref{claim:full-case-comp1},~\ref{claim:full-comp2}, and~\ref{claim:full-comp3} are $O(\epsilon)$, i.e.,
The error terms in Lemmas~\ref{claim:full-case-comp1},~\ref{claim:full-comp2}, and~\ref{claim:full-comp3} are $O(\epsilon)$, i.e., we have: (a)~ $ \frac{(1-a)\Delta}{ab} =O(\epsilon)$, (b)~ $\frac{5 \Delta }{\theta b }=O(\epsilon)$, (c)~ $\frac{\frac{k}{p^2} \log n}{ab}=O(\epsilon)$, and
(d)~ $\frac{\frac{k}{p^2} \log n + 4 \Delta}{\theta b} =O(\epsilon)$.
%
%
%(a) $ \frac{(1-a)\Delta}{ab} =O\left(\frac{1}{a(1-p)p} \sqrt{\frac{\log n}{b}}\right)$,
%(b) $  \frac{5 \Delta }{\theta b }=O\left(\frac{1}{a(1-p)p} \sqrt{\frac{\log n}{b}}\right)$,
%(c) $\frac{\frac{k}{p^2} \log n}{ab}=O\left(\frac{1}{a(1-p)p} \sqrt{\frac{\log n}{b}}\right)$, and
%(d) $\frac{\frac{k}{p^2} \log n + 4 \Delta}{\theta b} =O\left(\frac{1}{a(1-p)p} \sqrt{\frac{\log n}{b}}\right)$.
\end{lemma}
\end{proof}

\section{The Adaptive Algorithm}
\label{sec:alg2}
%\vm{The theorem was stated twice! We can't just copy the whole appendix here...  }
%\vm{I restructured the proof; fixed numerous typos. I replaced Lemma 5.1 with 5.7...Dawsen: I left a couple of questions for you.}

In the design of Algorithm~\ref{algorithm:hybrid}, we used the observation that in the partially \st{learnable}\vmn{predictable} model, the demand has a \st{predictable}\vmn{stochastic} component \vmn{that is} uniformly spread over the entire horizon. \vmn{This observation motivated us to define the evolving threshold rule. We remark that in Algorithm~\ref{algorithm:hybrid} neither the evolving threshold rule nor the fixed threshold rule adapts to the observed data, which makes Algorithm~\ref{algorithm:hybrid} a non-adaptive algorithm.}
As noted in Remark~\ref{rem:alg1:1}, when \vmn{the initial }inventory $b$ is small compared to the horizon $n$, the competitive ratio of Algorithm~\ref{algorithm:hybrid}, $p + \frac{1-p}{2-a}$, is in fact the best possible, and it can be achieved with our non-adaptive algorithm. Therefore, in this regime, \st{learning}\vmn{adapting to the data, i.e., setting thresholds based on the observed data,}\st{from data} would not improve the performance.
More precisely, when $b = o(\sqrt{n})$ \st{there is not enough time for learning}\vmn{the inventory is so small compared to the time horizon that there may not be enough time to effectively adapt to the
observed data.} The adversary can mislead us to allocate all the inventory before \st{seeing}\vmn{we can observe} \st{enough of the}\vmn{a sufficient portion of the} data.
However, as $b$ becomes larger, we will have more chance to observe \st{and learn from}\vmn{and adapt to} the data before allocating a significant part of the inventory.
In this section, \vmn{in fact, }we design an adaptive algorithm that achieves a better competitive ratio for large enough $b$ \vmn{(relative to $n$)}.
In Section~\ref{sec:alg2-alg}, we \vmn{first} \st{describe}\vmn{present} \vmn{the ideas behind our adaptive algorithm along with its formal description}\st{the algorithm}.
\vmn{Then, in Section~\ref{subsec:adapt:com}, we} analyze the competitive ratio of \st{the proposed}\vmn{our} algorithm.
\subsection{The Algorithm}\label{sec:alg2-alg}
In this section, we describe \st{the Adaptive}\vmn{our adaptive} algorithm, denoted by $ALG_{2,c}$, which takes $c\in [0,1]$ as a parameter.
For a certain range of $c$, we show that $ALG_{2,c}$  attains a competitive ratio of $c$ (up to an error term); however, if $c$ becomes too large (for example if $c=1$),
then  $ALG_{2,c}$ no longer  guarantees a $c$ fraction of the \vmn{optimum} offline solution.
%
%, where $c$ is the ``targeted'' competitive ratio of the algorithm.
%We discuss the range of $c$ for which $ALG_{2,c}$ in fact attains a competitive ratio close to $c$ in the analysis.
We call this algorithm  adaptive because it makes decisions based on the sequence of arrivals it has observed so far. In particular, this algorithm repeatedly computes upper bounds on the total number of \st{class}\vmn{type}-$1$/-$2$ customers based on the observed data and uses these upper bounds to decide whether to accept an arriving \st{class}\vmn{type}-$2$ customer or not.
Before proceeding with the algorithm, we first introduce two functions $u_1(\lambda)$ and $u_{1,2}(\lambda)$ that will prove useful in constructing \vmn{the aforementioned }upper bounds.
In particular we define:

\begin{align*}
u_1 (\lambda) & \triangleq \begin{cases}
b &\text{ if }\lambda<\delta \text{ (not enough data \st{to learn}\vmn{observed})}.\\
\min \left\{ \frac{o_1(\lambda)}{\lambda p}, \frac{o_1(\lambda) + (1-\lambda) (1-p) n}{1-p+\lambda p} \right\}& \text{ if }\lambda \geq \delta.
\end{cases} \\
u_{1,2} (\lambda) & \triangleq \begin{cases}
b &\text{ if }\lambda<\delta \text{ (not enough data \st{to learn}\vmn{observed})}.\\
\min \left\{ \frac{o_1(\lambda)+o_2(\lambda)}{\lambda p}, \frac{o_1(\lambda) +o_2(\lambda) + (1-\lambda) (1-p) n}{1-p+\lambda p} \right\}& \text{ if }\lambda \geq \delta,
\end{cases}
\end{align*}
where $\delta \triangleq \frac{(1-c)b}{(1-a)n}$. Note that $u_1(\lambda)$ and $u_{1,2}(\lambda)$ are functions of the observed data $o_1(\lambda)$ and $o_2(\lambda)$.
In the following lemma, we show how $u_1(\lambda)$ and $u_{1,2}(\lambda)$  provide upper bounds on $n_1$ and $n_1+n_2$
\vmn{when the realized sequence, $\vec{v}$, belongs to event $\mathcal{E}$ and}
\st{when}\vmn{the} number of \st{class}\vmn{type}-$1$ customers \vmn{as well as the initial inventory $b$} is large enough \vmn{(as specified in the lemma's statement)}.
\if false
In particular, similar to the analysis for Algorithm ~\ref{algorithm:hybrid}, we let $\alpha$, $k$, and $\bar\epsilon$ denote the constants  specified in Lemma~\ref{lemma:needed-centrality-result-for-m=2}, and define $\epsilon \triangleq \frac{1}{(1-c)^2ap^{3/2}}\sqrt{\frac{n^2 \log n}{b^3}} $.
In Claim~\ref{claim:check-epsilon} we show that $\frac{1}{n} \leq \epsilon \leq \bar\epsilon$ which implies that we can apply Lemma ~\ref{lemma:needed-centrality-result-for-m=2} to the corresponding event $\mathcal{E}$.
\fi
\vmn{Recall that we defined $\Delta$ to be $\alpha \sqrt{b \log n}$, where constant $\alpha$ itself is defined in Lemma~\ref{lemma:needed-centrality-result-for-m=2}. }
\if false
Further we define $\Delta \triangleq \alpha \sqrt{b \log n}$. We have:
\fi

\begin{lemma}
\label{prop:full-Ubounds}
Under event $\mathcal{E}$, \st{if}\vmn{suppose} $n_1\geq \frac{k}{p^2} \log n$ \vmn{and $b > \left({\frac{1}{\vmn{\bar\epsilon}} \frac{n\sqrt{\log n}}{(1-c)^2ap^{3/2}}}\right)^\frac{2}{3}$, where
%$\vm{\tilde{k}} = \min\left\{ \frac{3}{2\alpha}, \frac{1}{\sqrt{k}}, \frac{1}{\sqrt[4]{k}}, \bar\epsilon \right\}$, and constants $\alpha$, $k$, and $\bar\epsilon$ are defined in
constants $k$ and $\bar\epsilon$ are defined in  Lemma~\ref{lemma:needed-centrality-result-for-m=2}}.
%\pj{Here again, no need to introduce a new symbol since $\min\left\{ \frac{3}{2\alpha}, \frac{1}{\sqrt{k}}, \frac{1}{\sqrt[4]{k}}, \bar\epsilon \right\}$ is nothing else than $\bar\epsilon$, given the values of the constants defined in Lemma~\ref{lemma:needed-centrality-result-for-m=2}.}
Then for all $\lambda \in \{1/n, 2/n, \dots, 1\}$,
\begin{subequations}
\begin{align} & u_1(\lambda) \geq \min \left\{ b, n_1-\frac{2\Delta}{\delta p}\right\} \text{, and }\label{eq:full-u_1}
\\ & u_{1,2}(\lambda) \geq \min \left\{ b, n_1 + n_2 -\frac{2\Delta}{\delta p}\right\}. \label{eq:full-u_2}
\end{align}
\end{subequations}
\end{lemma}

\noindent{Lemma~\ref{prop:full-Ubounds} is proven in Appendix~\ref{sec:proof-value-MP1}.}

Having defined $u_1(\lambda)$ and $u_{1,2}(\lambda)$, now we describe how the adaptive algorithm determines whether to accept an arriving \st{class}\vmn{type}-$2$ customer when there is remaining inventory.
In the following, $q_j(\lambda)$, $j=1,2$, represents the number of \st{class}\vmn{type}-$j$ customers accepted by the algorithm up to time $\lambda$ (for a particular realization $\vec{v}$).
\vmn{Suppose the arriving customer at time $\lambda$ is of type-$2$.} If $u_{1,2}(\lambda) < b$, then we accept the customer, because  \eqref{eq:full-u_2} implies  that the total number of \st{class}\vmn{type}-$1$ and \st{class}\vmn{type}-$2$ customers will not exceed $b$ (neglecting the error term), and thus  we will have extra inventory at the end.
%, that $n_1 + n_2 \leq u_{1,2}(\lambda)$
%if $n_1 + n_2 \leq u_{1,2}(\lambda)$, then we will have extra inventory at the end.
On the other hand, if $u_{1,2}(\lambda) \geq b$, we may want to reject this customer to reserve inventory for a future \st{class}\vmn{type}-$1$ customer.
The decision of whether to accept the customer is based on the following two observations:
\begin{observation} If $u_1(\lambda) \geq n_1$, then \label{obs:upper-bound-OPT}
$$OPT(\vec{v}) \leq \min\{n_1, b\} + a (b-n_1)^+ = (1-a) \min\{n_1, b\}+ab \leq \min \{ u_1(\lambda) , b\}(1-a)+ab.$$
\end{observation}

%\vm{Dawsen: can we replace ``If $u_1(\lambda) \geq n_1$'' with ``If $u_1(\lambda) \geq \min\{ n_1, b \}$''? this way we cover the entire range, and I think nothing changes..}
%\vm{Actually I think the above is correct.. otherwise if becomes meaningless, because it will always holds according to 11a}
\begin{observation} \label{observation:maximum-revenue}
If we accept the current \st{class}\vmn{type}-$2$ customer, then the maximum revenue we can get is
$\left( b-\left( {q}_2\left(\lambda -1/n \right)+1\right)\right) + a \left( {q}_2\left(\lambda -1/n \right) +1 \right).$
\end{observation}
To have a competitive ratio of at least $c$, Observations~\ref{obs:upper-bound-OPT} and~\ref{observation:maximum-revenue} motivate us to accept the \st{class}\vmn{type}-$2$ customer only if \begin{align} \frac{\left( b-\left( {q}_2\left(\lambda -1/n\right)+1\right)\right) + a \left( {q}_2\left(\lambda - 1/n \right) +1 \right) } {\min \{ u_1(\lambda ) , b\}(1-a)+ab}\geq c. \label{inequality:cond-accept-class-2}
\end{align}
After rearranging terms, we get the following threshold for accepting the \st{class}\vmn{type}-$2$ customer:
\begin{align} \label{eq:Ccomp}
q_2(\lambda - 1/n ) +1 \leq \frac{1-c}{1-a} b + c \left(b - u_1(\lambda) \right)^+.
\end{align}
\st{In summary,}\vmn{Thus} when $u_{1,2}(\lambda) \geq b$, we use Condition~\eqref{eq:Ccomp} to accept/reject a \st{class}\vmn{type}-$2$ customer.
For notational convenience, we define $\phi \triangleq \frac{1-c}{1-a}$.
We point out the right-hand side of~\eqref{eq:Ccomp} may not be an integer; thus, in our algorithm, we use a slightly modified version of it, defined as follows:
\begin{align} \label{eq:Ccomp2}
q_2(\lambda - 1/n ) \leq \lfloor \frac{1-c}{1-a} b + c \left(b - u_1(\lambda) \right)^+ \rfloor.
\end{align}
\st{Note that $\phi b$ is the lowest threshold for accepting a \st{class}\vmn{type}-$2$ customer.}
\vmn{Note that by the definition of the threshold given in \eqref{eq:Ccomp2}, we always accept the first $\lfloor \phi b \rfloor$ \vmn{type}-$2$ customers.}
The formal definition of our algorithm is presented in Algorithm~\ref{algorithm:adaptive-threshold}.
In Algorithm~\ref{algorithm:adaptive-threshold}, $q_j$  represents the counter \vmn{for the} number of accepted customers \vmn{of} \st{class}\vmn{type}-$j$ so far.

\begin{algorithm}[H]\begin{enumerate}
\item Initialize $q_1, q_2 \leftarrow 0$, and define $\phi \triangleq \frac{1-c}{1-a}$, and $\delta \triangleq \frac{\phi b}{n}$.
\item Repeat for time $\lambda = 1/n, 2/n, \dots, 1$:
\begin{enumerate}
\item Calculate functions $u_1 (\lambda)$ and $u_{1,2} (\lambda)$ (to construct upper bounds for $n_1$ and $n_1+n_2$):
\begin{align*}
u_1 (\lambda) & \triangleq \begin{cases}
b &\text{ if }\lambda<\delta \text{ (not enough data \st{to learn}\vmn{observed})}.\\
\min \left\{ \frac{o_1(\lambda)}{\lambda p}, \frac{o_1(\lambda) + (1-\lambda) (1-p) n}{1-p+\lambda p} \right\}& \text{ if }\lambda \geq \delta.
\end{cases} \\
u_{1,2} (\lambda) & \triangleq \begin{cases}
b &\text{ if }\lambda<\delta \text{ (not enough data \st{to learn}\vmn{observed})}.\\
\min \left\{ \frac{o_1(\lambda)+o_2(\lambda)}{\lambda p}, \frac{o_1(\lambda) +o_2(\lambda) + (1-\lambda) (1-p) n}{1-p+\lambda p} \right\}& \text{ if }\lambda \geq \delta .
\end{cases}
\end{align*}
\item Accept  customer $i=\lambda n$ arriving at time $\lambda$ if there is remaining inventory and one of the following conditions holds:
\begin{enumerate}
\item $v_i = 1$; update $q_1\leftarrow q_1+1$.
\item $v_i = a$ and $u_{1,2}(\lambda) < b$; update $q_2 \leftarrow q_2 +1$.
\item $v_i = a$ and $q_2 \leq \lfloor \phi b + c \left(b - u_1(\lambda) \right)^+ \rfloor$; update $q_2 \leftarrow q_2 +1$.
\end{enumerate}
We prioritize the second condition if both the second and the third ones hold.
\end{enumerate}
\end{enumerate}
\caption{Online Adaptive Algorithm ($ALG_{2,c}$)}\label{algorithm:adaptive-threshold}
\end{algorithm}

Before we analyze the algorithm, we highlight two key properties of threshold $\lfloor \phi b + c \left(b - u_1(\lambda) \right)^+ \rfloor$: (i) The threshold is decreasing in $u_1(\lambda)$; the smaller $u_1(\lambda)$ is, the less inventory we reserve for future \st{class}\vmn{type}-$1$ customers. (ii) The threshold  is decreasing in $c$ as well (the right-hand side   of \eqref{eq:Ccomp2} can be expressed as $\lfloor \frac{1}{1-a} b -c (\frac{b}{1-a} - \left(b - u_1(\lambda) \right)^+)\rfloor$).  When $c$ is too large, we may reject too many \st{class}\vmn{type}-$2$ customers, which in turn hurts the revenue in a certain class of instances.
Said another way, note that Inequality~\eqref{eq:Ccomp2} only gives a ``necessary'' condition for achieving $c$-competitiveness.
%The threshold given by \eqref{eq:Ccomp} is decreasing in $c$: This can be seen by re-writing the right hand side of \eqref{eq:Ccomp} as $\frac{1}{1-a} b -c \left[\frac{b}{1-a} - \left(b - u_1(\lambda) \right)^+\right]$. For $c$ too large, we may reject too many \st{class}\vmn{type}-$2$ customers, which in turn, hurt the revenue in certain class of instances.
We identify the sufficient condition \vmn{for $c$-competitiveness} by solving \st{the following mathematical program whose feasibility region contains all such instances. T}the \vmn{{\em factor-revealing}} mathematical program \st{is} presented in (\ref{MP1}). \vmn{We will explain the construction of this program in the analysis of the competitive ratio (in Section~\ref{subsec:adapt:com}). On a high level, we construct the feasible region such that it contains any valid instance that can violate the $c$-competitiveness; by minimizing over $c$, we find the smallest value of $c$ for which the feasible region is not empty.}
%Denote $c^*$ the optimal objective value of (\ref{MP1}).
%To obtain a good competitive ratio, we should choose $c\leq c^*$.
%
%
%The threshold  is decreasing in $c$ as well (the right hand side   of \eqref{eq:Ccomp} can be expressed as $\frac{1}{1-a} b -c \left[\frac{b}{1-a} - \left(b - u_1(\lambda) \right)^+\right]$).
%
%
%Recall that our algorithm is parameterized by the targeted competitive ratio $c$.
%Note that the right hand side   of \eqref{eq:Ccomp} can be expressed as $\frac{1}{1-a} b -c \left[\frac{b}{1-a} - \left(b - u_1(\lambda) \right)^+\right]$, which is decreasing in $c$.
%Therefore, when $c$ is too large, we may reject too many \st{class}\vmn{type}-$2$ customers and thus the algorithm may not achieve the targed competitive ratio.
%
%The valid range of $c$ is related to the mathematical program presented in (\ref{MP1}).
%Denote $c^*$ the optimal objective value of (\ref{MP1}).
%To obtain a good competitive ratio, we should choose $c\leq c^*$.

%\vm{I changed $\lambda$ to $\l$ in the MP1 and also $\bar \lambda $ to $\l$ in the analysis of case (ii)  which talks about constructing MP1. }

\bigskip
\begin{mdframed}
\begin{subequations}
\begin{equation}\underset{(\vmn{\l} , n_1, n_2, \eta_1, \eta_2, c)\st{\in \mathbb{R}^6_{\geq 0}}}{\text{Minimize}}  c \tag{MP1} \label{MP1}
\end{equation}
subject to
  \begin{equation}
 c  \geq \frac{a(n_2-\tilde o_2+\frac{b}{1-a})+n_1}{a\min\{n_1+n_2, b\}+(1-a)n_1+\frac{a^2b}{1-a}+a \min\{\tilde u_1, b\}}\label{constraint:not-c-competitive}
  \end{equation}
%  \vspace{-10pt}
  \begin{equation}
%   \vspace{-10pt}
  \tilde  u_{1,2} \geq b \label{constraint:u_2>=b}
%  \vspace{-30pt}
  \end{equation}
% \vspace{-50pt}
\vspace{-20pt}
\\
\vspace{-35pt}
%\noindent\begin{tabularx}{\textwidth}{@{}>{\hsize=0.8\hsize}X>{\hsize=1.2\hsize}X>{\hsize=1\hsize}X>{\hsize=1\hsize}X@{}}
\noindent\begin{tabularx}{\textwidth}{@{}>{\hsize=1\hsize}X>{\hsize=1.2\hsize}X>{\hsize=1\hsize}X>{\hsize=1\hsize}X@{}}
  \begin{equation}
 \vmn{\l}  \leq 1 \label{constraint:x<=1}  \end{equation} &
  \begin{equation}
\eta_1+\eta_2 \leq \vmn{\l} n \label{constraint:eta_1+eta_2<=lambdan} \end{equation} &
\begin{equation}
\eta_1 \leq n_1 \label{constraint:n1'<=n1}  \end{equation} &
\begin{equation}
\eta_2 \leq n_2 \label{constraint:n2'<=n2} \end{equation}
\end{tabularx}
% \vspace{-20pt}
\vspace{-20pt}
\noindent\begin{tabularx}{\textwidth}{>{\hsize=0.7\hsize}X>{\hsize=1\hsize}X>{\hsize=1.5\hsize}X}%{@{}XXX@{}}
\begin{equation}
n_1 \leq b \label{constraint:n1<=b}  \end{equation} &
\begin{equation}
n_1+n_2 \leq n \label{constraint:n1+n2<=n}
 \end{equation} &
\begin{equation}
n_1+n_2  \leq \eta_1+\eta_2+(1-\vmn{\l})n\label{constraint:n1'+n2'big}  \end{equation}
\end{tabularx}
\end{subequations}
where $\tilde o_1 \triangleq (1-p)\eta_1+p n_1 \vmn{\l} $, $\tilde o_2 \triangleq (1-p)\eta_2+p n_2 \vmn{\l}$, $\tilde u_1 \triangleq  \min \left\{ \frac{\tilde o_1}{\vmn{\l} p }, \frac{\tilde o_1 + (1-\vmn{\l} ) (1-p)n}{(1-p+\vmn{\l} p)}  \right\} $, and $\tilde u_{1,2} \triangleq \min \left\{ \frac{\tilde o_1+\tilde o_2}{\vmn{\l}  p }, \frac{\tilde o_1+\tilde o_2 + (1-\vmn{\l} ) (1-p)n}{(1-p+\vmn{\l}  p)} \right\} $.
\end{mdframed}
\bigskip
Before we analyze Algorithm~\ref{algorithm:adaptive-threshold}, we \vmn{also} evaluate the solution of (\ref{MP1}).
\vmn{Denote  the optimal objective value of (\ref{MP1}) \vmn{by $c^*$}. As will be stated in Theorem~\ref{thm:adaptive-threshold}, $ALG_{2,{c^*}}$ achieves a competitive ratio of $c^*$ (minus an error term).}
First, we solve (\ref{MP1}) numerically for the regime where $b= \kappa n$ (where $0< \kappa \leq 1$ is a constant), and show that if $b/n > 0.5$, then Algorithm~\ref{algorithm:adaptive-threshold} achieves a better competitive ratio than Algorithm~\ref{algorithm:hybrid}.
\begin{figure}[h]
\centering
\noindent\begin{tabularx}{\textwidth}{@{}>{\hsize=1\hsize}X>{\hsize=1\hsize}X@{}}
\psscalebox{0.5}{ \centering
\psset{xunit=12cm,yunit=27.0415cm}
\begin{pspicture*}(-0.16774,0.60276)(1.1226,1.0322)
\psaxes[Ox=0,Oy=0.65,Dx=0.1,Dy=0.05,axesstyle=frame]{-}(0,0.65)(0,0.65)(1,1)

\rput[t](0.5,0.62359){
\begin{tabular}{c}
\Large$p$\\
\end{tabular}
}

\rput[b]{90}(-0.083333,0.825){
\begin{tabular}{c}
\Large Solution of~\ref{MP1}\\
\end{tabular}
}
\newrgbcolor{color178.004}{0        0.75        0.75}
\savedata{\mydata}[{
{0.05,0.68333}{0.1,0.7}{0.15,0.71667}{0.2,0.73333}
{0.25,0.75}{0.3,0.76667}{0.35,0.78333}{0.4,0.8}{0.45,0.81667}{0.5,0.83333}
{0.55,0.85}{0.6,0.86667}{0.65,0.88333}{0.7,0.9}{0.75,0.91667}{0.8,0.93333}
{0.85,0.95}{0.9,0.96667}{0.95,0.98333}
}]
\dataplot[plotstyle=line,linestyle=dashed,linecolor=color178.004]{\mydata}
\newrgbcolor{color177.004}{1  0  0}
\savedata{\mydata}[{
{0.05,0.68787}{0.1,0.70817}{0.15,0.72685}{0.2,0.74484}
{0.25,0.76389}{0.3,0.78082}{0.35,0.79945}{0.4,0.81538}{0.45,0.83303}{0.5,0.85166}
{0.55,0.86701}{0.6,0.88367}{0.65,0.89993}{0.7,0.91554}{0.75,0.93056}{0.8,0.94605}
{0.85,0.96003}{0.9,0.97408}{0.95,0.9877}
}]
\dataplot[plotstyle=line,linestyle=solid,linecolor=color177.004,linewidth=2pt]{\mydata}
\newrgbcolor{color176.004}{0         0.5           0}
\savedata{\mydata}[{
{0.05,0.83841}{0.1,0.85228}{0.15,0.86561}{0.2,0.87866}
{0.25,0.89155}{0.3,0.90407}{0.35,0.91241}{0.4,0.91987}{0.45,0.92733}{0.5,0.93492}
{0.55,0.94201}{0.6,0.9491}{0.65,0.95601}{0.7,0.96293}{0.75,0.97008}{0.8,0.97691}
{0.85,0.98312}{0.9,0.98896}{0.95,0.99457}
}]
\dataplot[plotstyle=line,linestyle=dotted,linecolor=color176.004,linewidth=2.5pt]{\mydata}
\newrgbcolor{color175.0045}{0  0  1}
\savedata{\mydata}[{
{0.05,0.95293}{0.1,0.95787}{0.15,0.96243}{0.2,0.96726}
{0.25,0.97185}{0.3,0.97495}{0.35,0.97688}{0.4,0.97871}{0.45,0.98084}{0.5,0.98226}
{0.55,0.98426}{0.6,0.98661}{0.65,0.98875}{0.7,0.99068}{0.75,0.99257}{0.8,0.99413}
{0.85,0.99575}{0.9,0.99718}{0.95,0.99862}
}]
\dataplot[plotstyle=line,linestyle=dashed,linecolor=color175.0045,linewidth=2.5pt]{\mydata}
\rput[bl](0.55,0.65492)
%\rput[bl](0.71918,0.65492)
%\rput[bl](0.60034,0.65492)
{
\psframebox[framesep=0]{\psframebox*{\begin{tabular}{l}
\Rnode{a1}{\hspace*{0.0ex}} \hspace*{0.7cm} \Rnode{a2}{\Large~~b/n=0.9} \\
\Rnode{a3}{\hspace*{0.0ex}} \hspace*{0.7cm} \Rnode{a4}{\Large~~b/n=0.7} \\
\Rnode{a5}{\hspace*{0.0ex}} \hspace*{0.7cm} \Rnode{a6}{\Large~~b/n=0.5} \\
\Rnode{a7}{\hspace*{0.0ex}} \hspace*{0.7cm} \Rnode{a8}{\Large Algorithm~\ref{algorithm:hybrid}} \\
\end{tabular}}
\ncline[linestyle=dashed,linecolor=color175.0045,linewidth=2.5pt]{a1}{a2}
\ncline[linestyle=dotted,linecolor=color176.004,linewidth=2.5pt]{a3}{a4}
\ncline[linestyle=solid,linecolor=color177.004,linewidth=2pt]{a5}{a6}
\ncline[,linestyle=dashed,linecolor=color178.004]{a7}{a8}
}
}
\end{pspicture*}
}
&\psscalebox{0.5}{ \centering
\psset{xunit=12cm,yunit=37.8581cm}
\begin{pspicture*}(-0.16774,0.71626)(1.1226,1.023)
\psaxes[Ox=0,Oy=0.75,Dx=0.1,Dy=0.05,axesstyle=frame]{-}(0,0.75)(0,0.75)(1,1)

%\rput[t](0.5,0.74113){
\rput[t](0.5,0.73113){
\begin{tabular}{c}
\Large $p$\\
\end{tabular}
}

\rput[b]{90}(-0.083333,0.875){
\begin{tabular}{c}
\Large Solution of~\ref{MP1}\\
\end{tabular}
}
\newrgbcolor{color178.004}{0        0.75        0.75}
\savedata{\mydata}[{
{0.05,0.78077}{0.1,0.79231}{0.15,0.80385}{0.2,0.81538}
{0.25,0.82692}{0.3,0.83846}{0.35,0.85}{0.4,0.86154}{0.45,0.87308}{0.5,0.88462}
{0.55,0.89615}{0.6,0.90769}{0.65,0.91923}{0.7,0.93077}{0.75,0.94231}{0.8,0.95385}
{0.85,0.96538}{0.9,0.97692}{0.95,0.98846}
}]
\dataplot[plotstyle=line,linestyle=dashed,linecolor=color178.004]{\mydata}
\newrgbcolor{color177.004}{1  0  0}
\savedata{\mydata}[{
{0.05,0.7824}{0.1,0.79514}{0.15,0.80746}{0.2,0.81901}
{0.25,0.83087}{0.3,0.84197}{0.35,0.85459}{0.4,0.86488}{0.45,0.87678}{0.5,0.88949}
{0.55,0.90003}{0.6,0.91168}{0.65,0.9233}{0.7,0.9345}{0.75,0.94554}{0.8,0.95715}
{0.85,0.96785}{0.9,0.97887}{0.95,0.98982}
}]
\dataplot[plotstyle=line,linestyle=solid,linecolor=color177.004,linewidth=2pt]{\mydata}
\newrgbcolor{color176.004}{0         0.5           0}
\savedata{\mydata}[{
{0.05,0.89502}{0.1,0.90345}{0.15,0.91163}{0.2,0.91975}
{0.25,0.92715}{0.3,0.93212}{0.35,0.93729}{0.4,0.94221}{0.45,0.9472}{0.5,0.95206}
{0.55,0.95704}{0.6,0.96196}{0.65,0.96693}{0.7,0.97174}{0.75,0.97672}{0.8,0.98156}
{0.85,0.98628}{0.9,0.99091}{0.95,0.99548}
}]
\dataplot[plotstyle=line,linestyle=dotted,linecolor=color176.004,linewidth=2.5pt]{\mydata}
\newrgbcolor{color175.0045}{0  0  1}
\savedata{\mydata}[{
{0.05,0.97102}{0.1,0.97395}{0.15,0.97679}{0.2,0.97969}
{0.25,0.98105}{0.3,0.98238}{0.35,0.9837}{0.4,0.98492}{0.45,0.98628}{0.5,0.98743}
{0.55,0.98883}{0.6,0.99016}{0.65,0.99149}{0.7,0.99279}{0.75,0.99406}{0.8,0.99527}
{0.85,0.99649}{0.9,0.99766}{0.95,0.99884}
}]
\dataplot[plotstyle=line,linestyle=dashed,linecolor=color175.0045,linewidth=2.5pt]{\mydata}
%\rput[bl](0.60034,0.75242)
%\rput[bl](0.60034,0.75242)
\rput[bl](0.55,0.755)
{
\psframebox[framesep=0]{\psframebox*{\begin{tabular}{l}
\Rnode{a1}{\hspace*{0.0ex}} \hspace*{0.7cm} \Rnode{a2}{\Large ~~b/n=0.9} \\
\Rnode{a3}{\hspace*{0.0ex}} \hspace*{0.7cm} \Rnode{a4}{\Large ~~b/n=0.7} \\
\Rnode{a5}{\hspace*{0.0ex}} \hspace*{0.7cm} \Rnode{a6}{\Large ~~b/n=0.5} \\
\Rnode{a7}{\hspace*{0.0ex}} \hspace*{0.7cm} \Rnode{a8}{\Large Algorithm~\ref{algorithm:hybrid}} \\
\end{tabular}}
\ncline[linestyle=dashed,linecolor=color175.0045,linewidth=2.5pt]{a1}{a2}
\ncline[linestyle=dotted,linecolor=color176.004,linewidth=2.5pt]{a3}{a4}
\ncline[linestyle=solid,linecolor=color177.004,linewidth=2pt]{a5}{a6}
\ncline[,linestyle=dashed,linecolor=color178.004]{a7}{a8}
}
}
\end{pspicture*}
}
\\ \centering
$a=0.50$
& \centering
$a=0.70$
\end{tabularx}
\caption{Solution of (\ref{MP1}), \vmn{$c^*$}, vs. $p$ for $a=0.50$ and $ 0.70$}
\label{fig:adaptive_threshold_vs_hybrid}
\end{figure}
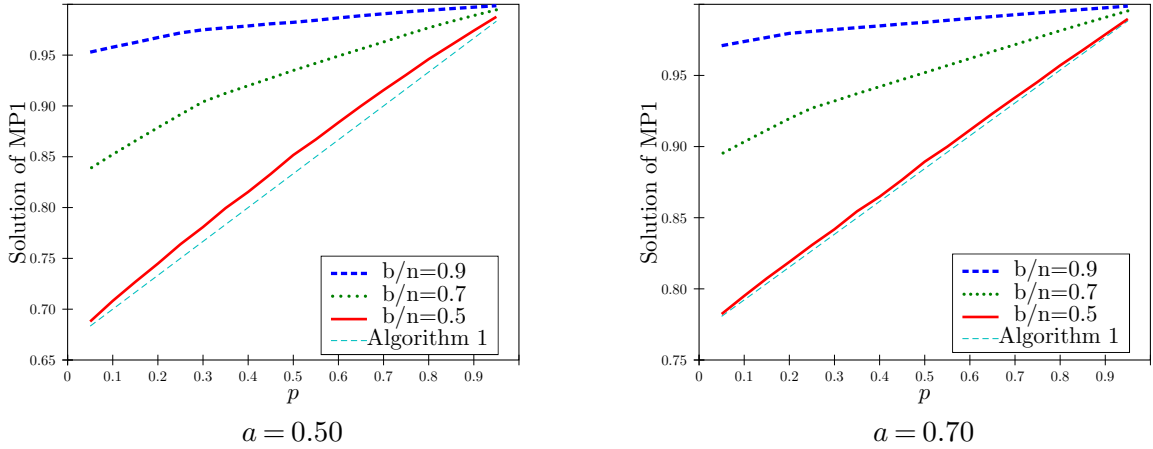
%\vm{Dawsen: I fixed some inconsistency in the figures. Can you please make the fonts of the both axis and legends a bit larger?}
%\vm{Dawsen: the figure was from $0$ all the way to $1$, I removed the first and last points from the data in ratio50.tex and ratio70.tex. is that fine?}

In Figure~\ref{fig:adaptive_threshold_vs_hybrid}, we fix $a=0.5,0.7$, and plot $c^*$ for $p = 0.05, 0.1,\dots, 0.95$ for three cases of $b/n=0.9$, $0.7$, and $0.5$.
\vmn{Figure~\ref{fig:adaptive_threshold_vs_hybrid} leads us to make the following observation:}
\st{As the figure illustrates that} \vmn{The competitive ratio of $ALG_{2,{c^*}}$ is at least that of $ALG_{1}$, and it is significantly larger when (i) $p$ is small and (ii) $b/n$ is large.} \st{(ii) For every $p$, the solution $c^*$ increases as $b/n$ ratio becomes larger. This increase is more significant for small probabilities.}
\st{This}\vmn{This observation} highlights the power of \st{adaptive (even partially) learning}\vmn{adapting to the data, even though it contains an adversarial component}: Consider $a=0.7$, $b = 0.7 n$ and $p = 0.2$; this means that $80 \%$ of the demand belongs to the {\UPG} group\st{, i.e., we can only learn from $20 \%$ of the data}. Our adaptive algorithm guarantees $10 \%$ more revenue than the non-adaptive algorithm does.

In addition, we note that as the \vmn{initial} inventory $b$ becomes larger (for a fixed time horizon $n$), the adversary's power naturally \st{reduces}\vmn{declines}.
\vmn{Thus one would expect that a ``smart'' algorithm achieves a higher competitive ratio.}
Our adaptive algorithm \st{takes advantage of this limiting phenomenon, and provides better}\vmn{ indeed attains a higher} competitive ratio \st{for larger initial inventory values}\vmn{as the initial inventory increases}. In contrast, \st{the}\vmn{the competitive ratio of our} non-adaptive algorithm \st{is ignorant of this phenomenon}\vmn{remains the same}.

\vmn{We conclude our study of (\ref{MP1}) by establishing a lower bound on its optimum solution. The following proposition states that $c^*$ is at least $p+\frac{1-p}{2-a}$, which is the competitive ratio of Algorithm~\ref{algorithm:hybrid} (ignoring the error term).}
\if false
Finally, in the next proposition (proved in Appendix~\ref{sec:proof-value-MP1}), we find lower bounds on $c^*$:
\fi
%, the optimal objective value of (\ref{MP1}):
\begin{proposition}
\label{prop:values-math-program}
For any $b\leq n$, we have: $c^* \geq p+\frac{1-p}{2-a}$. Further, if $b=n$, then $c^* =1$.
\end{proposition}
%\end{document}
%

%\input{main_head}
%
\subsection{Competitive Analysis}\label{subsec:adapt:com}
In this section we analyze the competitive ratio of Algorithm~\ref{algorithm:adaptive-threshold} and prove the following theorem:

\begin{theorem}\label{thm:adaptive-threshold}
For $p\in(0,1)$, let $c^*$ be the optimal objective value of (\ref{MP1}). For any $c \leq c^*$ \vmn{such that $c<1$}, $ALG_{2,c}$ is $c - O\left(\frac{1}{(1-c)^2ap^{3/2}}\sqrt{\frac{n^2 \log n}{b^3}}\right)$ competitive in the partially \st{learnable}\vmn{predictable} model.
\end{theorem}

\vmn{The above theorem implies that if $c^* <1$, then $ALG_{2,c^*}$ is  $c^* - O\left(\frac{1}{(1-c^*)^2ap^{3/2}}\sqrt{\frac{n^2 \log n}{b^3}}\right)$ competitive. However, the same does not hold
when $c^* = 1$. For this special case,}
%For the special case of $c^*=1$ and $c=1- \sqrt[3]{\frac{1}{ap^{3/2}}\sqrt{\frac{n^2 \log n}{b^3}}}$,
we have the following corollary of  Theorem~\ref{thm:adaptive-threshold}:
%In addition, we have the following corollary:
\begin{corollary}
\label{cor:cstar1}
When $c^*=1$, for $c = 1- \sqrt[3]{\frac{1}{ap^{3/2}}\sqrt{\frac{n^2 \log n}{b^3}}}$ the competitive ratio of $ALG_{2,c}$ is $1-O\left( \sqrt[3]{\frac{1}{ap^{3/2}}\sqrt{\frac{n^2 \log n}{b^3}}}\right)$.
%for $c = 1- \sqrt[3]{\frac{1}{ap^{3/2}}\sqrt{\frac{n^2 \log n}{b^3}}}$, where $c^*$ is the optimal objective value of (\ref{MP1}).
\end{corollary}
%\vm{I brought back the corollary and its proofs. Also added a line above to clarify.}

%\vm{Dawsen: the statement of this corollary is a bit weird... do we mean to say that for $c \leq 1- \sqrt[3]{\frac{1}{ap^{3/2}}\sqrt{\frac{n^2 \log n}{b^3}}}$ we get
%$c-O\left( \sqrt[3]{\frac{1}{ap^{3/2}}\sqrt{\frac{n^2 \log n}{b^3}}}\right)$? }
%\vm{Also the corollary didn't have a proof, so I suggest removing?}
%\begin{proof}
%Theorem~\ref{thm:adaptive-threshold} with $c=1- \sqrt[3]{\frac{1}{ap^{3/2}}\sqrt{\frac{n^2 \log n}{b^3}}}$ proves the corollary.
%\end{proof}

\begin{remark}
Theorem~\ref{thm:adaptive-threshold} combined with Proposition~\ref{prop:values-math-program} shows that in the asymptotic regime (where $n$ and $b$ both grow), if the scaling factor  $\sqrt{\frac{n^2 \log n}{b^3}}$ (which appears in the error term of the competitive ratio) is vanishing (i.e., order of $o(1)$), then our adaptive algorithm outperforms our non-adaptive one. For instance, the aforementioned condition holds if $b = \kappa n$ \vmn{where $0< \kappa \leq 1$ is a constant}\st{with $\kappa$ being a positive constant}.
\end{remark}

\if false
\begin{remark}
\label{remark:alg2-2}
In Subsection~\ref{sec:asymptotic-ub}, we present a class of instances that provides evidence that no online algorithm can achieve a competitive ratio better than $c^*$ in the asymptotic regime and for the case $b = \kappa n$. Our partial analysis is based on two approximations; making this a rigorous statement is beyond the scope of our current paper.
\end{remark}
\fi
%However, making this a rigorous statement is beyond the scope of our paper.
%
%we provide some evidence that esent a class of instances for which, under some technical conditions, no online algorithm can achieve a competitive ratio better than $c^*$ (the solution of (\ref{MP1})) in the asymptotic case where $b$ and $n$ go to infinity with a fixed $b/n$ ratio.%\end{remark}

\begin{proof}{\textbf{Proof of Theorem~\ref{thm:adaptive-threshold}:}}
\vmn{Similar to the proof of Theorem~\ref{thm:hybrid}, we start by making the observation that}
Theorem~\ref{thm:adaptive-threshold} is nontrivial only if $\sqrt{\frac{n^2 \log n}{b^3}}$ is small enough \st{so}\vmn{such} that the approximation term $O(\cdot)$ is negligible. Therefore, without loss of generality, we can restrict attention to the case where $\sqrt{\frac{n^2 \log n}{b^3}}$ is small. In particular,
\if false
\st{we define}\vmn{let us define constant} $\vmn{\vm{\tilde{k}}}$ \vmn{as} $\min\left\{ \frac{3}{2\alpha}, \frac{1}{\sqrt{k}}, \frac{1}{\sqrt[4]{k}}, \bar\epsilon \right\}$\vmn{, where constants $\alpha$, $k$, and $\bar\epsilon$ are defined in Lemma~\ref{lemma:needed-centrality-result-for-m=2}. We note that we have already defined $\vm{\tilde{k}}$ in Lemma~\ref{prop:full-Ubounds}.} \pj{See my previous comment on the fact that we dont need to introduce $\vm{\tilde{k}}$, as this is nothing else than $\bar\epsilon$, the smallest in $\min\left\{ \frac{3}{2\alpha}, \frac{1}{\sqrt{k}}, \frac{1}{\sqrt[4]{k}}, \bar\epsilon \right\}$.}
%\footnote{The definition of $\vmn{\vm{\tilde{k}}}$ is motivated by \eqref{ineq:ADP-\vmn{\vm{\tilde{k}}}}, \eqref{ineq:\vmn{\vm{\tilde{k}}}-condition-ADP}, \eqref{ineq:\vmn{\vm{\tilde{k}}}-condition-ADP<=k4}, and \eqref{condition:ALG2:\vmn{\vm{\tilde{k}}}<=bar-epsilon}.}
\fi
\st{When}\vmn{if} $\frac{1}{(1-c)^2ap^{3/2}}\sqrt{\frac{n^2 \log n}{b^3}} \geq \vmn{\bar\epsilon}$, \vmn{then} $O\left(\frac{1}{(1-c)^2ap^{3/2}}\sqrt{\frac{n^2 \log n}{b^3}}\right)$ becomes $O(1)$ and Theorem~\ref{thm:adaptive-threshold} becomes trivial \vmn{(recall that constant $\bar\epsilon = 1/24$ is defined in Lemma~\ref{lemma:needed-centrality-result-for-m=2})}.
Therefore, without loss of generality, we assume $\frac{1}{(1-c)^2ap^{3/2}}\sqrt{\frac{n^2 \log n}{b^3}} < \vmn{\bar\epsilon}$, or equivalently,
\begin{align}
b ^{\frac{3}{2}}> \frac{1}{\vmn{\bar\epsilon}} \frac{n\sqrt{\log n}}{(1-c)^2ap^{3/2}}.\label{ineq:ADP-non-trivial-case}
\end{align}
{We remark that we impose the same condition on $b$ in Lemma~\ref{prop:full-Ubounds}.}
\vmn{We denote the random revenue generated by Algorithm~\ref{algorithm:adaptive-threshold} by $ALG_{2,c}(\vec{V})$. Similar to the proof of Theorem~\ref{thm:hybrid}, we define
an appropriate $\epsilon$ that allows us to focus on the realizations that belong to event $\mathcal{E}$. In particular, let
$\epsilon = \frac{1}{(1-c)^2ap^{3/2}}\sqrt{\frac{n^2 \log n}{b^3}}$. For $b$ that satisfies condition \eqref{ineq:ADP-non-trivial-case}, \vmn{and assuming
that $n \geq 3$\st{is large enough}, }we have $\frac{1}{n} \leq \epsilon \leq \bar\epsilon$. Therefore, we can apply Lemma~\ref{lemma:needed-centrality-result-for-m=2} to get:}

\vmn{
\begin{align*}
\frac{\E{ALG_{2,c}(\vec{V})}}{OPT(\vmn{\vec{v}_I})} \geq  \frac{\E{ALG_{2,c}(\vec{V})|\mathcal{E}}\prob{\mathcal{E}}}{OPT(\vmn{\vec{v}_I})} \geq \frac{\E{ALG_{2,c}(\vec{V})|\mathcal{E}}}{OPT(\vmn{\vec{v}_I})} \left(1 - \epsilon \right).
\end{align*}}

\noindent \vmn{In the main part of the proof, we show that for any realization $\vec{v}$ belonging to event $\mathcal{E}$,}
\vmn{
\begin{align*}
\frac{{ALG_{2,c}(\vec{v})}}{OPT(\vmn{\vec{v}_I})} \geq  c - O(\epsilon).
\end{align*}
}

%\vmn{Finally, to ease the exposition, we treat the threshold of accepting \st{class}\vmn{type}-$2$ customer, $\phi b + c \left(b - u_1(\lambda) \right)^+$, to be integers for all $\lambda$.}
\noindent \vmn{To analyze the competitive ratio we analyze three cases separately.}

\if false
For any initial customer sequence $\vec{v'}\vmn{ZZ\vec{v}_I}$, we denote the random revenue generated by Algorithm~\ref{algorithm:adaptive-threshold} by $ALG_{2,c}(\vec{V})$. In order to analyze $\E{ALG_{2,c}(\vec{V})}$ we condition it on the event $\mathcal{E}$ for an appropriately defined $\epsilon$. In fact we have already defined $\epsilon$ in Subsection~\ref{sec:alg2-alg}, right before Lemma~\ref{prop:full-Ubounds}; $\epsilon \triangleq \frac{1}{(1-c)^2ap^{3/2}}\sqrt{\frac{n^2 \log n}{b^3}} $. As mentioned before, in Claim~\ref{claim:check-epsilon} we show that $\frac{1}{n} \leq \epsilon \leq \bar\epsilon$; this allows us to
apply Lemma ~\ref{lemma:needed-centrality-result-for-m=2} which in turn  implies:

%In particular
%%let $\epsilon \triangleq \frac{1}{a(1-p)p}\sqrt{\frac {\log n}{b}}$. If we can show that
%let $\alpha$, $k$, and $\bar\epsilon$ denote the constants  specified in Lemma~\ref{lemma:needed-centrality-result-for-m=2}, and $\epsilon \triangleq \frac{1}{a(1-p)p}\sqrt{\frac {\log n}{b}}$. If we show that $\frac{1}{n} \leq \epsilon \leq \bar\epsilon$, then we can apply Lemma ~\ref{lemma:needed-centrality-result-for-m=2} which implies:

$$\frac{\E{ALG_{2,c}(\vec{V})}}{OPT(\vec{v'}\vmn{ZZ\vec{v}_I})} \geq  \frac{\E{ALG_{2,c}(\vec{V})|\mathcal{E}}\prob{\mathcal{E}}}{OPT(\vec{v'}\vmn{ZZ\vec{v}_I})} \geq \frac{\E{ALG_{2,c}(\vec{V})|\mathcal{E}}}{OPT(\vec{v'}\vmn{ZZ\vec{v}_I})} \left(1 - \epsilon \right)$$

This will allow us to focus on the realizations that belong to event $\mathcal{E}$.
%Claim~\ref{claim:epsilon-con} (stated at the end of the proof) proves that in fact condition $\frac{1}{n} \leq \epsilon \leq \bar\epsilon$ holds.
In the main part of the proof we show that for any realization $\vec{v}$ belonging to event $\mathcal{E}$,

$$\frac{{ALG_{2,c}(\vec{v})}}{OPT(\vec{v'}\vmn{ZZ\vec{v}_I})} \geq  c - O(\epsilon).$$
%where $\Delta \triangleq \alpha \sqrt{b \log n}$.
This in turn implies the result (combined with a few small steps). Before proceeding with the proof, we make the following observation:

Theorem~\ref{thm:adaptive-threshold} is only non-trivial if $\sqrt{\frac{n^2 \log n}{b^3}}$ is small enough, so that the approximation term $O(\cdot)$ is negligible. Therefore, without loss of generality we can restrict attention to the case where $\sqrt{\frac{n^2 \log n}{b^3}}$ is small. In particular,
we define $\vmn{\vm{\tilde{k}}} \triangleq \min\left\{ \frac{3}{2\alpha}, \frac{1}{\sqrt{k}}, \frac{1}{\sqrt[4]{k}}, \bar\epsilon \right\}$.
%\footnote{The definition of $\vmn{\vm{\tilde{k}}}$ is motivated by \eqref{ineq:ADP-k'}, \eqref{ineq:k'-condition-ADP}, \eqref{ineq:k'-condition-ADP<=k4}, and \eqref{condition:ALG2:k'<=bar-epsilon}.}
When $\frac{1}{(1-c)^2ap^{3/2}}\sqrt{\frac{n^2 \log n}{b^3}} \geq \vmn{\vm{\tilde{k}}}$, $O\left(\frac{1}{(1-c)^2ap^{3/2}}\sqrt{\frac{n^2 \log n}{b^3}}\right)$ becomes $O(1)$ and Theorem~\ref{thm:adaptive-threshold} becomes trivial.
Therefore, we assume, without loss of generality, $\frac{1}{(1-c)^2ap^{3/2}}\sqrt{\frac{n^2 \log n}{b^3}} < \vmn{\vm{\tilde{k}}}$, or equivalently,
\begin{align}
b ^{\frac{3}{2}}> \frac{1}{\vmn{\vm{\tilde{k}}}} \frac{n\sqrt{\log n}}{(1-c)^2ap^{3/2}}. \label{ineq:ADP-non-trivial-case}
\end{align}

Finally, for simplicity, we treat the threshold of accepting \st{class}\vmn{type}-$2$ customer, $\phi b + c \left(b - u_1(\lambda) \right)^+$, to be integers for all $\lambda$. \vm{do we need this?}
In order to analyze the competitive ratio  we analyze $3$ cases separately.
\fi

\medskip
\noindent { \bf Case (i): $n_1 \geq \frac{k}{p^2} \log n$, and Algorithm~\ref{algorithm:adaptive-threshold} exhausts the inventory.}
\medskip

When $n_1 \geq \frac{k}{p^2} \log n$, we can apply \eqref{inequality:good-approximation-o_1} from Lemma~\ref{lemma:needed-centrality-result-for-m=2} and Lemma~\ref{prop:full-Ubounds}. \vmn{Because Algorithm~\ref{algorithm:adaptive-threshold} exhausts the inventory, we know that $n_1 + n_2 \geq b$. Now we have either (a) $n_1 + n_2 -\frac{2\Delta}{\delta p} \leq b$ or (b) $n_1 + n_2 -\frac{2\Delta}{\delta p} > b$. If (a) happens, then (according to Lemma~\ref{prop:full-Ubounds}) we may have $u_{1,2} (\lambda) < b$, which may result in accepting a \vmn{type}-$2$ customer through the second condition that we should have rejected. However, in this case, we also have a tight upper bound on the optimum offline solution.  As shown in the proof of Lemma~\ref{claim:q1+q2=b-ADP-ratio}\vmn{---which analyzes the competitive ratio of the two cases (a) and (b) separately---}such a bound allows us to establish the desired lower bound on the competitive ratio.} \vmn{Case (b) is the more interesting case, which accepts type-$2$ customers through the third condition of Algorithm~\ref{algorithm:adaptive-threshold}. It
is possible that the algorithm accepts \emph{too many} \st{class}\vmn{type}-$2$ customers through this condition, resulting in  rejecting \st{class}\vmn{type}-$1$ customers, and thus in revenue loss. In the following lemma, we control for this loss by establishing an upper bound on the number of accepted \st{class}\vmn{type}-$2$ customers. The proof of the lemma, which uses similar ideas to those in Lemma~\ref{lem:HYB:full-case1}, is deferred to Appendix~\ref{sec:proof-value-MP1}.}
\if false
Note that it is possible that $u_{1,2}(\lambda)<b$ but $n_1+n_2\geq b$, which misleads us to accept a \st{class}\vmn{type}-$2$ customer which we should have rejected.
In the following lemma (proven in Appendix~\ref{sec:proof-value-MP1}), using Lemma~\ref{prop:full-Ubounds}, we show that this does not happen when $n_1+n_2 > b+\frac{2\Delta}{\delta p}$:
\fi
\begin{lemma}\label{lemma:upper-bound-q2-ADP-case-1}
Under event $\mathcal{E}$, if $n_1\geq \frac{k}{p^2} \log n$, then one of the following conditions holds:
\begin{enumerate}[(a)]
\item $n_1+n_2 - \frac{2\Delta}{\delta p} \leq b $, or
\item $n_1+n_2 - \frac{2\Delta}{\delta p} > b  $ and $q_2(1) \leq \frac{1-c}{1-a} b + c \left(b - n_1 \right)^+ + c \frac{2\Delta}{\delta p} +1.$
\end{enumerate}
\end{lemma}
%\begin{proof}
%The only interesting case is when $n_1+n_2 > b+ \frac{2\Delta}{\delta}$.
%If $q_2(1)=0$, then we are done.
%Otherwise, let $\bar\lambda$ be the last time we accept a \st{class}\vmn{type}-$2$ customer.
%By Lemma~\ref{prop:full-Ubounds}, $u_{1,2}(\bar\lambda) \geq \min \left\{ b, n_1 + n_2 -\frac{2\Delta}{\delta p}\right\} = b$.
%Therefore, according to the definition of $\bar\lambda$, Condition~\eqref{eq:Ccomp} must be satisfied.
%Thus,
%\begin{align*}
%q_2(1) = & q_2 ( \bar \lambda) \\ \leq & \frac{1-c}{1-a} b + c \left(b - u_1( \bar \lambda) \right)^+ &(\text{Condition~\eqref{eq:Ccomp}}) \\
%\leq & \frac{1-c}{1-a} b + c \left(b - \min \left\{ b, n_1-\frac{2\Delta}{\delta p}\right\} \right)^+ &(\text{Lemma~\ref{prop:full-Ubounds}}) \\
%\leq & \frac{1-c}{1-a} b + c \left(b - n_1 \right)^+ + c \frac{2\Delta}{\delta p}.
%\end{align*}
%\end{proof}
\if false
Using the Lemma~\ref{lemma:upper-bound-q2-ADP-case-1} \vmn{and the discussion before the lemma}, we find a lower bound on the competitive ratio:
\fi
Using Lemma~\ref{lemma:upper-bound-q2-ADP-case-1} \vmn{and the discussion before the lemma}, in Appendix~\ref{sec:proof-value-MP1}, we prove the following lemma, \vmn{which gives a lower bound on the competitive ratio for Case (i)}:
\begin{lemma} \label{claim:q1+q2=b-ADP-ratio}
Under event $\mathcal{E}$, if $n_1\geq \frac{k}{p^2} \log n$ and $q_1(1)+ q_2(1)=b$, then
$$\frac{ALG_{2,c}({\vec{v}})}{OPT(\vec{v})} \geq c - \frac{{3}\Delta}{ab\delta p}.$$
\end{lemma}
\if false
{\noindent The proof is deferred to Appendix~\ref{sec:proof-value-MP1}.}
\fi

\medskip
\noindent { \bf Case (ii): $n_1 \geq \frac{k}{p^2} \log n$, and Algorithm~\ref{algorithm:adaptive-threshold} does not exhaust the inventory.}
\medskip

First note that in this case $OPT(\vec{v}) = n_1 + a \min \{b-n_1,n_2\}$.
Also, in this case, we accept all \st{class}\vmn{type}-$1$ customers. Therefore, $q_1(1)=n_1$. To \vmn{lower-}bound the competitive ratio, we need to show only that we do not reject too many \st{class}\vmn{type}-$2$ customers, i.e., $q_2(1)$ is large enough.
Note that if for all $\lambda\vmn{\in \{1/n, 2/n, \ldots, 1\}}$, condition \eqref{eq:Ccomp2} holds, then all \st{class}\vmn{type}-$2$ customers are accepted, and we have $q_2(1) = n_2$. This implies that $ALG_{2,c}({\vec{v}}) = OPT(\vec{v})$\st{, and we are done}.
The more interesting case is when there exists at least one time step for which condition \eqref{eq:Ccomp2} is violated. Let $\vmn{\l}$ be the last time that we reject a \st{class}\vmn{type}-$2$ customer. This means that at time $\vmn{\l}$, we have:

\noindent\begin{tabularx}{\textwidth}{>{\hsize=0.7\hsize}X>{\hsize=1.3\hsize}X}%{@{}XXX@{}}
\begin{equation}
u_{1,2}(\vmn{\l}) \geq b \label{ineq:u2>=b},\end{equation} &
\begin{equation}
q_2(\vmn{\l}) \geq \frac{1-c}{1-a} b + c \left(b - u_1(\vmn{\l}) \right)^+.\label{eq:Ccomp5} \end{equation}
\end{tabularx}

%\begin{align}
%
%u_{1,2}(\bar{\lambda}) &\geq b, and \label{ineq:u2>=b}\\
\noindent This also provides the following lower bound on the number of accepted \st{class}\vmn{type}-$2$ customers:
\begin{align}
q_2(1) = q_2(\vmn{\l}) + \left[n_2 - o_2(\vmn{\l})\right] \geq \frac{1-c}{1-a} b + c \left(b - u_1(\vmn{\l}) \right)^+ + \left[n_2 - o_2(\bar{\lambda})\right].
\end{align}
Therefore, when $q_1(1)+q_2(1)<b$,
\begin{align}
\label{ineq:lower-bound-ratio-ADP}
\frac{ALG_{2,c}({\vec{v}})}{OPT(\vec{v})}
\geq \frac{n_1 + a\left(\frac{1-c}{1-a} b + c \left(b - u_1(\vmn{\l}) \right)^+ + \left[n_2 - o_2(\vmn{\l})\right]\right) }{n_1 + a \min \{b-n_1,n_2\}}.
\end{align}
%Re-arranging terms, the right hand side of \eqref{ineq:lower-bound-ratio-ADP} being smaller or equal to $c$ is equivalent to:
%\begin{align}
%c \geq \frac{a(n_2- o_2(\bar\lambda)+\frac{b}{1-a})+n_1}{a\min\{n_1+n_2, b\}+(1-a)n_1+\frac{a^2b}{1-a}+a \min\{ u_1(\bar\lambda), b\}}. \label{eq:Ccomp6}
%\end{align}

For a fixed $c$, if for all possible instances the right-hand side of \eqref{ineq:lower-bound-ratio-ADP} is greater than $c$, then $ALG_{2,c}$ would be $c$-competitive.
However, if $c$ is too large, then there will be instances for which \vmn{the right-hand side of} \eqref{ineq:lower-bound-ratio-ADP} \st{is violated}\vmn{will be less than $c$}.
We identify a superset of these instances by all possible combinations of $(\vmn{\l}, n_1, n_2, \eta_1(\vmn{\l}), \eta_2(\vmn{\l}))$ that satisfy certain \st{``regularity'' conditions}\vmn{constraints to ensure they correspond to valid instances}.
As a reminder, $\eta_j(\vmn{\l})$ \st{and $\eta_2(\vmn{\l})$}represents the number of \st{class}\vmn{type}-$j$ customers\st{and\st{class}\vmn{type}-$2$} by time $\vmn{\l}$ in the initial sequence (determined by the adversary, i.e., $\vec{v}_I$).
As we describe these \st{conditions}\vmn{constraints} below, it becomes clear that \vmn{(1)} any instance of the problem would satisfy all these \st{conditions}\vmn{constraints}, and \vmn{(2) these constraints correspond to the feasible region of the mathematical program  in (\ref{MP1})}.

We start with the straightforward \st{conditions}\vmn{constraints}: for every instance, $n_1 + n_2 \leq n$. Also, $\eta_1(\vmn{\l}) \leq n_1$, and $\eta_2(\vmn{\l}) \leq n_2$. Further, in the initial customer sequence $\vec{v}_I$, at time $\vmn{\l}$ we cannot have more than $\vmn{\l} n$ customers, thus $\eta_1(\vmn{\l}) + \eta_2(\vmn{\l}) \leq \vmn{\l} n$. Similarly after time $\vmn{\l}$, we cannot have more than $(1 - \vmn{\l}) n$ customers, and therefore $n_1 + n_2 - \left[\eta_1(\vmn{\l}) + \eta_2(\vmn{\l})\right] \leq (1-\vmn{\l}) n$. By definition of $\vmn{\l}$, we have $\vmn{\l} \leq 1$.  We also add the condition $n_1 \leq b$, which is always true under the case when $q_1(1)+q_2(1)<b$. \vmn{Note that these are Constraints~\eqref{constraint:x<=1}-\eqref{constraint:n1'+n2'big} in (\ref{MP1}), where in (\ref{MP1}), with a slight abuse of notation, we simplify by substituting $\eta_j$ for $\eta_j(\vmn{\l})$.}

For a moment, suppose $o_j(\vmn{\l}) =\tilde o_j(\vmn{\l})$. First, we remind the reader that $\tilde o_j(\vmn{\l}) =  (1-p)\eta_j(\vmn{\l})+p \vmn{\l} n_j$ \vmn{is the deterministic approximation of $o_j(\vmn{\l})$ that we introduced in Section~\ref{sec:prem}, and also is redefined in (\ref{MP1})} (at the bottom). \st{Also,}\vmn{Further,} note that this is just to explain the idea behind constructing (\ref{MP1}). Later in the proof, we address the difference between
$\tilde o_j(\vmn{\l})$ and $ o_j(\vmn{\l})$.
In this case, \vmn{we have:}

\vmn{
\begin{align}
\label{eq:const:b}
\tilde u_{1,2}(\vmn{\l}) \triangleq \min \left\{ \frac{\tilde o_1(\vmn{\l})+\tilde o_2(\vmn{\l})}{\vmn{\l}  p }, \frac{\tilde o_1(\vmn{\l})+\tilde o_2(\vmn{\l}) + (1- \vmn{\l} ) (1-p)n}{(1-p+\vmn{\l}  p)} \right\} = u_{1,2}(\vmn{\l})\geq b
\end{align}
}
\vmn{where the last inequality is the same as inequality \eqref{ineq:u2>=b}.
Further note that rejecting a customer at time $\vmn{\l}$ implies that
%Further, recall that we reject a customer at time $\vmn{\l}$, which implies that $\vmn{\l} n  \geq o_2(\vmn{\l}) \geq \phi b$. This in turn implies that
$\vmn{\l} \geq \frac{\phi b}{n}= \delta$, and thus by definition $u_{1,2}(\vmn{\l})  = \min \left\{ \frac{o_1(\vmn{\l})+o_2(\vmn{\l})}{\vmn{\l} p}, \frac{o_1(\vmn{\l}) +o_2(\vmn{\l}) + (1-\vmn{\l}) (1-p) n}{1-p+\vmn{\l} p} \right\}$.\footnote{\vmn{Note that when $\vmn{\l} < \delta$ Algorithm~\ref{algorithm:adaptive-threshold} never rejects a customer, because $q_2(\l) \leq \l n < \delta n = \phi b$.}}
Note that Inequality \eqref{eq:const:b} is Constraint \eqref{constraint:u_2>=b},  where in (\ref{MP1}), again with a slight abuse of notation, we simplify by substituting $\tilde u_{1,2}$ for $\tilde u_{1,2}(\vmn{\l})$ and $\tilde  o_j$ for $\tilde o_j(\vmn{\l})$.}

%\vm{Dawsen: the equality above only holds if $\l > \delta$. is this always true for this regime? I added the explanation above, please double-check!}

\if false
$\tilde u_{1,2}(\bar{\lambda}) = u_{1,2}({\lambda})\geq b$ (see \vmn{\ref{MP1} for the definition of $\tilde u_{1,2}(\bar{\lambda})$, and }inequality \eqref{ineq:u2>=b}).
\fi
%and $u_{1,2}(\bar{\lambda}) \geq b$ (see inequality \eqref{ineq:u2>=b}). We also add the condition $n_1 \leq b$, which is always true under the case that $q_1(1)+q_2(1)<b$.
Further,  the most interesting \vmn{constraint, Constraint \eqref{constraint:not-c-competitive},}
\st{condition}comes from condition \eqref{ineq:lower-bound-ratio-ADP}.
By rearranging terms, we can show that the right-hand side of \eqref{ineq:lower-bound-ratio-ADP} being smaller or equal to $c$ is equivalent to:
\begin{align}
c \geq \frac{a(n_2- o_2(\vmn{\l})+\frac{b}{1-a})+n_1}{a\min\{n_1+n_2, b\}+(1-a)n_1+\frac{a^2b}{1-a}+a \min\{ u_1(\vmn{\l}), b\}} \label{eq:Ccomp6}
\end{align}
\vmn{which is Constraint \eqref{constraint:not-c-competitive} after substituting
$o_2(\vmn{\l})$ with $\tilde o_2$ and $u_1(\vmn{\l})$ with $\tilde u_1$.}

\vmn{Overall,} the above conditions define the feasible region of the math program (\ref{MP1}).
%Note that in~\ref{MP1} we use $\eta_1$ and $\eta_2$ instead of $\eta_1(\lambda)$ and $\eta_2(\lambda)$ to make their dependency on $\lambda$ implicit.
%Also, in~\ref{MP1}, we use the deterministic approximations notion
%$\tilde o_1$, $\tilde o_2$, $\tilde u_1$, $\tilde u_2$, as defined in Equation~\eqref{eq:estimate}.
By minimizing $c$, we find the threshold for making (\ref{MP1}) infeasible: Let $c^*$ be the solution of (\ref{MP1}); for any $c < c^*$, (\ref{MP1}) is infeasible, and the only constraint that $(\vmn{\l} , n_1, n_2, \eta_1(\vmn{\l}), \eta_2(\vmn{\l}), c)$ can violate is ~\eqref{constraint:not-c-competitive} (same as \eqref{eq:Ccomp6}). \vmn{This implies}\st{implying} that $ALG_{2,c}$ is $c$-competitive.

%
%
%
%\vm{edited here}
%
%The rest of the analysis is related to (\ref{MP1}).
%Note that the tuple $(\bar \lambda , n_1, n_2, \eta_1(\bar \lambda), \eta_2(\bar \lambda), c)$ satisfies Constraints~\eqref{constraint:x<=1}-\eqref{constraint:n1'+n2'big}.
%Now we explain how we construct (\ref{MP1}) by considering the special case where $(o_1(\bar\lambda), o_2(\bar\lambda)) =(\tilde o_1(\bar\lambda), \tilde o_2(\bar\lambda))$.
%In this case, Constraint~\eqref{constraint:u_2>=b} is satisfied due to \eqref{ineq:u2>=b}.
%Therefore, if $c<c^*$ (the optimal objective value of (\ref{MP1})), then the tuple $(\bar \lambda , n_1, n_2, \eta_1(\bar \lambda), \eta_2(\bar \lambda), c)$ is not in the feasible set of (\ref{MP1}), and hence Constraint~\eqref{constraint:not-c-competitive} is violated, which means \eqref{eq:Ccomp6} is violated and hence due to \eqref{ineq:lower-bound-ratio-ADP}, $\frac{ALG_{2,c}}{OPT(\vec{v})} > c$.

We now go back and address the issue that $ \tilde o_j(\vmn{\l})$ and $o_j(\vmn{\l})$ are not equal.
Due to the difference between $o_j(\vmn{\l})$ and $ \tilde o_j(\vmn{\l})$, \vmn{(1)} Constraint~\eqref{constraint:u_2>=b} might be violated (even though \eqref{ineq:u2>=b} is satisfied) and \vmn{(2)} violating Constraint~\eqref{constraint:not-c-competitive} does not imply violating \eqref{eq:Ccomp6}.
\st{To partially \vm{Dawsen: why partially?!}} \vmn{To} address these issues, \st{in the following lemma}\vmn{first in Lemma~\ref{lemma:feasible-sol}}, we give a slightly modified tuple that satisfies Constraints~\eqref{constraint:u_2>=b}-\eqref{constraint:n1'+n2'big}; \vmn{then, in Lemma~\ref{claim:ADP-non-exahust}, we prove that for any $c\leq c^*$, if Constraint~\eqref{constraint:not-c-competitive} is violated, then $\frac{ALG_{2,c}({\vec{v}})}{OPT(\vec{v})} \geq c- \frac{4\Delta n}{\phi^2 b^2 p}$.}
\vmn{The proofs of both lemmas are deferred to Appendix~\ref{sec:proof-value-MP1}, and they amount to applying the concentration results of Lemma~\ref{lemma:needed-centrality-result-for-m=2} and carefully analyzing the error terms.
These two lemmas complete the analysis of competitive ratio in Case (ii).}

\begin{lemma} \label{lemma:feasible-sol}
Under event $\mathcal{E}$, if $n_1\geq \frac{k}{p^2} \log n$ \vmn{and $q_1(1)+q_2(1) < b$}, then the tuple $(\vmn{\l}', n_1', n_2', \eta_1', \eta_2', c')\triangleq (\vmn{\l}, n_1, n_2 + \xi, \eta_1(\vmn{\l}), \eta_2(\vmn{\l}) + \bar \xi, c)$ satisfies Constraints~\eqref{constraint:u_2>=b}-\eqref{constraint:n1'+n2'big}, where
\begin{align*}
\xi &\triangleq \begin{cases} 0 &\text{ if }n_1+n_2 \geq b ,\\
\min \left\{ n - (n_1 + n_2), \frac{\Delta n}{\phi b p} \right\} &\text{ if }n_1+n_2< b;
\end{cases}\\
\bar \xi &\triangleq \begin{cases} 0 &\text{ if }n_1+n_2 \geq b, \\
\min \left\{ \xi,  \vmn{\l} n - (\eta_1 (\vmn{\l})+ \eta_2(\vmn{\l})) \right\} &\text{ if }n_1+n_2< b,
\end{cases}
\end{align*}
\vmn{and where $\Delta = \alpha \sqrt{b \log n}$, $\phi = \frac{1-c}{1-a}$, and $\vmn{\l}$ is the last time that we reject a \vmn{type}-$2$ customer.}
\end{lemma}

\if flase
\vm{Dawse, there is ? reference in the proof. please check}
\noindent{The proof of the lemma is presented in Appendix~\ref{sec:proof-value-MP1}}.
Using the tuple $(\lambda', n_1', n_2', \eta_1', \eta_2', c')$
%of~(\ref{MP1})
given in Lemma~\ref{lemma:feasible-sol}, we show the following lemma (proven in Appendix~\ref{sec:proof-value-MP1}):
\vm{Dawsen: let's add a sentence here: analyzing the objective of (\ref{MP1}) for tuple... }
\fi
\begin{lemma} \label{claim:ADP-non-exahust}
Under event $\mathcal{E}$, if $n_1\geq \frac{k}{p^2} \log n$ and $q_1(1)+q_2(1) < b$, then $$\frac{ALG_{2,c}({\vec{v}})}{OPT(\vec{v})} \geq c- \frac{4\Delta n}{\phi^2 b^2 p}.$$
\end{lemma}

\medskip
\noindent{ \bf Case (iii): $n_1 < \frac{k}{p^2} \log n$.}
\medskip

The competitive ratio analysis for this case uses ideas similar to those in the previous two cases, and it follows from the next two lemmas. The proofs are deferred to Appendix~\ref{sec:proof-value-MP1}.

\begin{lemma}\label{lemma:n1-small-ADP-full}
Under event $\mathcal{E}$, if $n_1 < \frac{k}{p^2} \log n$, then one of the following three conditions holds:
\begin{enumerate}[(a)]
\item $q_1(1)+q_2(1)=b$;
\item $q_1(1)=n_1$ and $q_2(1)=n_2$; or
\item $q_1(1)=n_1$ and $q_2(1) \geq cb$.
\end{enumerate}
\end{lemma}
%\begin{proof}
%Note that $q_1(1)=n_1$ when $q_1(1)+q_2(1)<b$.
%Therefore, what is remaining is to show that if $q_1(1)+q_2(1)<b$ and $q_2(1)<n_2$, then $q_2(1) \geq cb$.
%
%
%Let $\bar \lambda$ be the last time when a customer is rejected.
%Then, Similar to earlier discussion, Inequalities~\eqref{ineq:big-bar-lambda} is satisfied.
%Therefore,
%\begin{align}
%u_1(\bar\lambda) = &\min \left \{ \frac{o_1(\bar \lambda)}{\bar \lambda p}, \frac{ o_1(\bar \lambda) + (1-\bar \lambda) (1-p) n}{1-p+\bar \lambda p} \right\} &(\eqref{ineq:big-bar-lambda})\nonumber \\
%\leq & \frac{o_1(\bar \lambda)}{\bar \lambda p} \nonumber \\
%\leq & \frac{n_1 n}{\phi b p} & (\eqref{ineq:big-bar-lambda}\text{ and }o_1(\bar\lambda)\leq n_1) \nonumber \\
%< & \frac{n \frac{k}{p^2} \log n }{\phi b p}=\frac{kn \log n}{\phi b p^3}. &(n_1< \frac{k}{p^2} \log n) \nonumber % \label{ineq:ub-u1-ADP-small-n1}
%\end{align}
%As a result,
%$ q_2(1) = \phi b + c(b-u_1(\bar\lambda))^+ > (\phi + c) b- c\frac{kn \log n}{\phi b p^3},$
%which is greater than $cb$ when
%$ b^2 \geq c \frac{kn \log n}{\phi ^2 p^3}.$
%Using $\phi = \frac{1-c}{1-a} \geq 1-c$ and Inequality \eqref{ineq:ADP-non-trivial-case} and
%\begin{align}k'\leq \frac{1}{\sqrt[4]{k}} \label{ineq:k'-condition-ADP<=k4},
%\end{align}
%we have $$b^2 = \frac{\left( b^{\frac{3}{2}}\right)^2 }{b} \geq \frac{\left( b^{\frac{3}{2}}\right)^2 }{n} > \frac{1}{k'^4}\frac{n\log n}{(1-c)^4a^2p^3}\geq c \frac{kn \log n}{\phi ^2 p^3}.$$
%Therefore, $ b^2 \geq c \frac{kn \log n}{\phi ^2 p^3}$ and thus $q_2(1) \geq cb$.
%\end{proof}
Using Lemma~\ref{lemma:n1-small-ADP-full}, \st{we prove the} \vmn{in the } following lemma, \vmn{we establish a lower bound on the competitive ratio for
Case (iii):}
\begin{lemma}
Under event $\mathcal{E}$, if $n_1 < \frac{k}{p^2} \log n$, then
$\frac{ALG_{2,c}({\vec{v}})}{OPT(\vec{v})} \geq c$. \label{claim:small-n1}
\end{lemma}
%\begin{proof}
%We consider each case of Lemma~\ref{lemma:n1-small-ADP-full} separately.
%For the first case, $q_1(1)+q_2(1)=b$, since $n_1 < \frac{k}{p^2} \log n$, it is easy to see that
%\begin{align*}
%\frac{ALG_{2,c}({\vec{v}})}{OPT(\vec{v})} \geq \frac{ab}{ab+\frac{k}{p^2} \log n} \geq \frac{ab - \frac{k}{p^2} \log n}{ab} = 1- \frac{k \log n}{ab p^2},
%\end{align*}
%which is at least $c$ if $b \geq \frac{k \log n}{a(1-c) p^2}$.
%Inequality \eqref{ineq:ADP-non-trivial-case} and \eqref{ineq:k'-condition-ADP} imply
%$$\sqrt{b} = \frac{b^{\frac{3}{2}}}{b} \geq \frac{b^{\frac{3}{2}}}{n} > \frac{1}{k'}\frac{\sqrt{\log n}}{(1-c)^2a^2p^{3/2}}\geq \sqrt{k}\frac{\sqrt{\log n}}{p\sqrt{a(1-c)}},$$
%and thus $b \geq \frac{k \log n}{a(1-c) p^2}$ and $\frac{ALG_{2,c}({\vec{v}})}{OPT(\vec{v})} \geq c$.
%
%
%For the second and third cases, $q_2(1) \geq \min \{n_2, b\}$, we have
%$$ ALG_{2,c}({\vec{v}}) \geq n_1+c(\min \{n_2, b\})a \geq c (n_1+\min \{n_2, b\}a) \geq c OPT(\vec{v}).$$
%\end{proof}

%Since the value $OPT(\vec{v})$ determined by $\vec{v'}\vmn{ZZ\vec{v}_I}$, we drop $\vec{v}$ to simplify the notation.

Having Lemmas~\ref{claim:q1+q2=b-ADP-ratio},~\ref{claim:ADP-non-exahust}, and~\ref{claim:small-n1}, we have lower bounds on the competitive ratio of Algorithm~\ref{algorithm:adaptive-threshold} for all possible cases.
\vmn{We complete the proof of the theorem by the following lemma (proven in
Appendix~\ref{sec:proof-value-MP1}) that ensures that the error terms in  Lemmas~\ref{claim:q1+q2=b-ADP-ratio} and~\ref{claim:ADP-non-exahust} are $O(\epsilon)$.}
\if false
To complete the proof, we just need to show the following two claims that are proven in Appendix~\ref{sec:proof-value-MP1}.
\fi

\if false
\begin{claim}
\label{claim:check-epsilon}
$\frac{1}{n} \leq \epsilon \leq \bar\epsilon$.
\end{claim}
\fi

\begin{lemma}
\label{claim:check-error-terms}
The error terms in Lemmas~\ref{claim:q1+q2=b-ADP-ratio} and~\ref{claim:ADP-non-exahust}  are $O(\epsilon)$, i.e.,
\vmn{(a)} $\frac{{3}\Delta}{ab\delta p} =O\left( \epsilon \right)$ and \vmn{(b)} $ \frac{4\Delta n}{\phi^2 b^2 p} = O\left( \epsilon \right)$.
%$\frac{2\Delta}{ab\delta p} =O\left( \frac{1}{(1-c)^2ap^{3/2}}\sqrt{\frac{n^2 \log n}{b^3}} \right)$ and $ \frac{4\Delta n}{\phi^2 b^2 p} = O\left( \frac{1}{(1-c)^2ap^{3/2}}\sqrt{\frac{n^2 \log n}{b^3}} \right)$.
\end{lemma}

\end{proof}

\section{Discussion of the Model}\label{sec:model:discussion}

%\vm{I moved most of the statements/ proof to Appendix. Essentially everything we can't proof rigorously... I really don't know if we wanna keep the robustness analysis. We just have 2 lines for it in the main text, and then 15 pages of analysis in the appendix.. I wonder if it gets us in to more trouble than helping us... can't we just say that in a longer version we show some indications of robustness? One other idea is to bring back the statement of propositions calling them conjectures, and then saying proof of conjectures are removed.}

In this section, we\st{derive some fundamental properties of} \vmn{further study the performance of online algorithms in} our demand model.
%First in Section~\ref{sec:bad-instance-for-other-models} we present a problem instance in the partially \st{learnable}\vmn{predictable} model for which our algorithms outperform the existing ones.
\vmn{First, }in Section~\ref{sec:upper-bounds}, we present \vmn{an} upper bound on the competitive ratio achievable by any online algorithm \vmn{under our demand model when the initial inventory $b$ is small---more precisely,  $b = o\left(\sqrt{{n}}\right)$.}
%In Section~\ref{sec:unknown-p}, we show how to estimate $p$ from offline data for future runs of the process. Also, we discuss how the competitive ratio changes if we slightly overestimate or underestimate $p$.
\vmn{Next,} in Section~\ref{sec:bad-instance-for-other-models}, we \vmn{highlight the need for our new online algorithms by }presenting a problem instance for which our algorithms outperform existing ones in \st{the}\vmn{our} partially \st{learnable}\vmn{predictable} model.\st{ for which our algorithms outperform the existing ones.}

\subsection{Upper Bounds}\label{sec:upper-bounds}
In this section, we present \vmn{an} upper bound on the competitive ratio of any online algorithm when $b = o(\sqrt{n})$.
%We also comment on our findings for the asymptotic case that $b = \kappa n$ (where $\kappa$ is a constant) and $n$ grows.
%under two different regimes of $b$ as compared to $n$.
We start with a warm-up example that illustrates a fundamental limit of any online algorithm in \vmn{the} partially \st{learnable}\vmn{predictable} model.
Figure~\ref{figure:hard:example} shows two instances with $n = 8$ \st{arriving customers}.
The bottom row shows the sequence that the online algorithm will see; \vmn{as a reminder,}\st{note that} we represent the nodes of the {\RG} group as filled (even though the online algorithm cannot distinguish between the two groups of customers).
Suppose $b=4$;
in the instance presented on the left, the \vmn{optimum} offline solution rejects all \st{class}\vmn{type}-$2$ customers, and in the instance on the right, it accepts all of them.
Now, by time $\lambda = \vmn{b/n = }4/8$, online algorithms cannot distinguish between these two instances, and hence cannot perform as well as the optimal offline algorithm on {\em both} of these instances.
Similar to this example, \vmn{in the following proposition}, we \st{prove}\vmn{establish} the upper bound by \st{presenting}\vmn{constructing} two problem instances that are ``difficult'' for online algorithms to distinguish between up to\st{a certain}  time \vmn{$\frac{b}{n}$}, and show that the trade-off between accepting too many or too few \st{class}\vmn{type}-$2$ customers limits the competitive ratio of any online algorithm.
\st{First we focus on the case $b = o\left(\sqrt{{n}}\right)$, and show:}
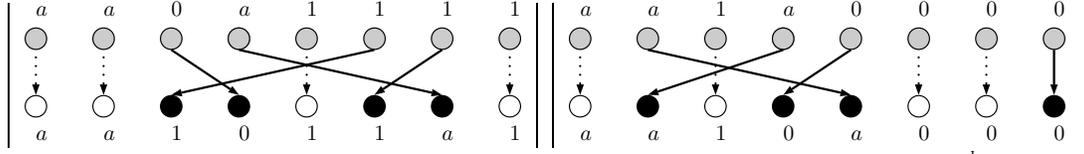
\begin{figure}
\centering
\psscalebox{0.75}{
\begin{tabular}{| l | l | l | }
\psscalebox{1.0 1.0} % Change this value to rescale the drawing.
{
\begin{pspicture}(0,-1.21)(8.802778,1.21)
\definecolor{colour0}{rgb}{0.8,0.8,0.8}
\psdots[linecolor=black, fillstyle=solid,fillcolor=colour0, dotstyle=o, dotsize=0.4, fillcolor=colour0](0.20138885,0.59)
\psdots[linecolor=black, fillstyle=solid,fillcolor=colour0, dotstyle=o, dotsize=0.4, fillcolor=colour0](1.4013889,0.59)
\psdots[linecolor=black, fillstyle=solid,fillcolor=colour0, dotstyle=o, dotsize=0.4, fillcolor=colour0](2.601389,0.59)
\psdots[linecolor=black, fillstyle=solid,fillcolor=colour0, dotstyle=o, dotsize=0.4, fillcolor=colour0](3.8013887,0.59)
\psdots[linecolor=black, fillstyle=solid,fillcolor=colour0, dotstyle=o, dotsize=0.4, fillcolor=colour0](3.8013887,0.59)
\psdots[linecolor=black, fillstyle=solid,fillcolor=colour0, dotstyle=o, dotsize=0.4, fillcolor=colour0](5.001389,0.59)
\psdots[linecolor=black, fillstyle=solid,fillcolor=colour0, dotstyle=o, dotsize=0.4, fillcolor=colour0](6.201389,0.59)
\psdots[linecolor=black, fillstyle=solid,fillcolor=colour0, dotstyle=o, dotsize=0.4, fillcolor=colour0](7.4013886,0.59)
\psdots[linecolor=black, fillstyle=solid,fillcolor=colour0, dotstyle=o, dotsize=0.4, fillcolor=colour0](8.601389,0.59)
\psdots[linecolor=black, fillstyle=solid, dotstyle=o, dotsize=0.4, fillcolor=white](0.20138885,-0.61)
\psdots[linecolor=black, dotsize=0.4](2.601389,-0.61)
\psdots[linecolor=black, dotsize=0.4](3.8013887,-0.61)
\psdots[linecolor=black, dotsize=0.4](7.4013886,-0.61)
\psdots[linecolor=black, fillstyle=solid, dotstyle=o, dotsize=0.4, fillcolor=white](1.4013889,-0.61)
\psdots[linecolor=black, fillstyle=solid, dotstyle=o, dotsize=0.4, fillcolor=white](5.001389,-0.61)
\psdots[linecolor=black, dotsize=0.4](6.201389,-0.61)
\psdots[linecolor=black, fillstyle=solid, dotstyle=o, dotsize=0.4, fillcolor=white](8.601389,-0.61)
\rput[bl](0.20138885,0.99){$a$}
\rput[bl](1.4013889,0.99){$a$}
\rput[bl](2.601389,0.99){$0$}
\rput[bl](3.8013887,0.99){$a$}
\rput[bl](5.001389,0.99){$1$}
\rput[bl](6.201389,0.99){$1$}
\rput[bl](7.4013886,0.99){$1$}
\rput[bl](8.601389,0.99){$1$}
\psline[linecolor=black, linewidth=0.04, arrowsize=0.05291666666666668cm 2.0,arrowlength=1.4,arrowinset=0.0]{->}(3.8013887,0.39)(7.4013886,-0.41)
\psline[linecolor=black, linewidth=0.04, arrowsize=0.05291666666666668cm 2.0,arrowlength=1.4,arrowinset=0.0]{->}(6.201389,0.39)(2.601389,-0.41)
\psline[linecolor=black, linewidth=0.04, arrowsize=0.05291666666666668cm 2.0,arrowlength=1.4,arrowinset=0.0]{->}(7.4013886,0.39)(6.201389,-0.41)
\psline[linecolor=black, linewidth=0.04, linestyle=dotted, dotsep=0.10583334cm, arrowsize=0.05291666666666668cm 2.0,arrowlength=1.4,arrowinset=0.0]{->}(0.20138885,0.39)(0.20138885,-0.41)
\psline[linecolor=black, linewidth=0.04, linestyle=dotted, dotsep=0.10583334cm, arrowsize=0.05291666666666668cm 2.0,arrowlength=1.4,arrowinset=0.0]{->}(5.001389,0.39)(5.001389,-0.41)
\psline[linecolor=black, linewidth=0.04, linestyle=dotted, dotsep=0.10583334cm, arrowsize=0.05291666666666668cm 2.0,arrowlength=1.4,arrowinset=0.0]{->}(8.601389,0.39)(8.601389,-0.41)
\rput[bl](0.20138885,-1.21){$a$}
\rput[bl](1.4013889,-1.21){$a$}
\rput[bl](2.601389,-1.21){$1$}
\rput[bl](3.8013887,-1.21){$0$}
\rput[bl](5.001389,-1.21){$1$}
\rput[bl](6.201389,-1.21){$1$}
\rput[bl](7.4013886,-1.21){$a$}
\rput[bl](8.601389,-1.21){$1$}
\psline[linecolor=black, linewidth=0.04, linestyle=dotted, dotsep=0.10583334cm, arrowsize=0.05291666666666668cm 2.0,arrowlength=1.4,arrowinset=0.0]{->}(1.4013889,0.39)(1.4013889,-0.41)
\psline[linecolor=black, linewidth=0.04, arrowsize=0.05291666666666668cm 2.0,arrowlength=1.4,arrowinset=0.0]{->}(2.601389,0.39)(3.8013887,-0.41)
\end{pspicture}
}& \quad &\psscalebox{1.0 1.0} % Change this value to rescale the drawing.
{
\begin{pspicture}(0,-1.21)(8.802778,1.21)
\definecolor{colour0}{rgb}{0.8,0.8,0.8}
\psdots[linecolor=black, fillstyle=solid,fillcolor=colour0, dotstyle=o, dotsize=0.4, fillcolor=colour0](0.20138885,0.59)
\psdots[linecolor=black, fillstyle=solid,fillcolor=colour0, dotstyle=o, dotsize=0.4, fillcolor=colour0](1.4013889,0.59)
\psdots[linecolor=black, fillstyle=solid,fillcolor=colour0, dotstyle=o, dotsize=0.4, fillcolor=colour0](2.601389,0.59)
\psdots[linecolor=black, fillstyle=solid,fillcolor=colour0, dotstyle=o, dotsize=0.4, fillcolor=colour0](3.8013887,0.59)
\psdots[linecolor=black, fillstyle=solid,fillcolor=colour0, dotstyle=o, dotsize=0.4, fillcolor=colour0](3.8013887,0.59)
\psdots[linecolor=black, fillstyle=solid,fillcolor=colour0, dotstyle=o, dotsize=0.4, fillcolor=colour0](5.001389,0.59)
\psdots[linecolor=black, fillstyle=solid,fillcolor=colour0, dotstyle=o, dotsize=0.4, fillcolor=colour0](6.201389,0.59)
\psdots[linecolor=black, fillstyle=solid,fillcolor=colour0, dotstyle=o, dotsize=0.4, fillcolor=colour0](7.4013886,0.59)
\psdots[linecolor=black, fillstyle=solid,fillcolor=colour0, dotstyle=o, dotsize=0.4, fillcolor=colour0](8.601389,0.59)
\psdots[linecolor=black, fillstyle=solid, dotstyle=o, dotsize=0.4, fillcolor=white](0.20138885,-0.61)
\psdots[linecolor=black, fillstyle=solid, dotstyle=o, dotsize=0.4, fillcolor=white](2.601389,-0.61)
\psdots[linecolor=black, fillstyle=solid,fillcolor=black, dotsize=0.4](3.8013887,-0.61)
\psdots[linecolor=black, fillstyle=solid, dotstyle=o, dotsize=0.4, fillcolor=white](7.4013886,-0.61)
\psdots[linecolor=black, dotsize=0.4](1.4013889,-0.61)
\psdots[linecolor=black, dotsize=0.4](5.001389,-0.61)
\psdots[linecolor=black, dotstyle=o, dotsize=0.4, fillcolor=white](6.201389,-0.61)
\psdots[linecolor=black, dotsize=0.4](8.601389,-0.61)
\rput[bl](0.20138885,0.99){$a$}
\rput[bl](1.4013889,0.99){$a$}
\rput[bl](2.601389,0.99){$1$}
\rput[bl](3.8013887,0.99){$a$}
\rput[bl](5.001389,0.99){$0$}
\rput[bl](6.201389,0.99){$0$}
\rput[bl](7.4013886,0.99){$0$}
\rput[bl](8.601389,0.99){$0$}
\psline[linecolor=black, linewidth=0.04, arrowsize=0.05291666666666667cm 2.0,arrowlength=1.4,arrowinset=0.0]{->}(1.4013889,0.39)(5.001389,-0.41)
\psline[linecolor=black, linewidth=0.04, arrowsize=0.05291666666666667cm 2.0,arrowlength=1.4,arrowinset=0.0]{->}(3.8013887,0.39)(1.4013889,-0.41)
\psline[linecolor=black, linewidth=0.04, arrowsize=0.05291666666666667cm 2.0,arrowlength=1.4,arrowinset=0.0]{->}(5.001389,0.39)(3.8013887,-0.41)
\psline[linecolor=black, linewidth=0.04, arrowsize=0.05291666666666667cm 2.0,arrowlength=1.4,arrowinset=0.0]{->}(8.601389,0.39)(8.601389,-0.41)
\psline[linecolor=black, linewidth=0.04, linestyle=dotted, dotsep=0.10583334cm, arrowsize=0.05291666666666667cm 2.0,arrowlength=1.4,arrowinset=0.0]{->}(0.20138885,0.39)(0.20138885,-0.41)
\psline[linecolor=black, linewidth=0.04, linestyle=dotted, dotsep=0.10583334cm, arrowsize=0.05291666666666667cm 2.0,arrowlength=1.4,arrowinset=0.0]{->}(2.601389,0.39)(2.601389,-0.41)
\psline[linecolor=black, linewidth=0.04, linestyle=dotted, dotsep=0.10583334cm, arrowsize=0.05291666666666667cm 2.0,arrowlength=1.4,arrowinset=0.0]{->}(6.201389,0.39)(6.201389,-0.41)
\psline[linecolor=black, linewidth=0.04, linestyle=dotted, dotsep=0.10583334cm, arrowsize=0.05291666666666667cm 2.0,arrowlength=1.4,arrowinset=0.0]{->}(7.4013886,0.39)(7.4013886,-0.41)
\rput[bl](0.20138885,-1.21){$a$}
\rput[bl](1.4013889,-1.21){$a$}
\rput[bl](2.601389,-1.21){$1$}
\rput[bl](3.8013887,-1.21){$0$}
\rput[bl](5.001389,-1.21){$a$}
\rput[bl](6.201389,-1.21){$0$}
\rput[bl](7.4013886,-1.21){$0$}
\rput[bl](8.601389,-1.21){$0$}
\end{pspicture}
}
\end{tabular}}
\caption{Two problem instances \vmn{between which} online algorithms cannot distinguish at time\st{$\frac{4}{8}$} \vmn{$\frac{b}{n}$, where $b = 4$ and $n = 8$}.} \label{figure:hard:example}
\end{figure}
%\vm{remove subsection}
%\subsubsection{General Case}\label{subsec:upperB}

%\vm{remove this:}For general $b$ and $n$, we have the following proposition:
\begin{proposition}\label{thm:impossibility}
Under the partially \st{learnable}\vmn{predictable arrival} model\vmn{, and for any $p \in (0,1)$}, no online algorithm, deterministic and randomized, can achieve a competitive ratio better than $\frac{1-p}{2-a} + p + O\left( \frac{pb^2}{n} \right)$.
Therefore, when $b = o\left(\sqrt{{n}}\right)$, no online algorithm can achieve a competitive ratio better than $\frac{1-p}{2-a} + p + o\left(1\right)$.
\end{proposition}

\noindent{The details of the proof are deferred to Appendix~\ref{sec:app:model:discussion}. As explained above, the main idea of the proof is to \st{come up with}\vmn{construct} two instances that are almost indistinguishable \vmn{up to time $\frac{b}{n}$} to any online algorithm. In the proof we show that the following two instances ${\vec{v}_{I}}$ and ${\vec{w}_I}$\st{will} serve our purpose: }

\if false
\begin{equation*}
v_{I,j} = \left\{
\begin{array}{lll}
a, \qquad & i = 1, 2, \ldots, b, \\
1, \qquad & i = b+1, b+2, \ldots, 2b \text{ for } {\vec{v_2'}} ,\\
0, \qquad & i = b+1, b+2, \ldots, 2b \text{ for } {\vec{v_1'}} , \\
0, \qquad & i > 2b.
\end{array}
\right
w_{I,j} = \left\{
\begin{array}{lll}
a, \qquad & i = 1, 2, \ldots, b, \\
1, \qquad & i = b+1, b+2, \ldots, 2b \text{ for } {\vec{v_2'}} ,\\
0, \qquad & i = b+1, b+2, \ldots, 2b \text{ for } {\vec{v_1'}} , \\
0, \qquad & i > 2b.
\end{array}
\right.
\end{equation*}
\fi

\begin{equation*}
v_{I,j} = \begin{cases}
a, \qquad & 1 \leq j \leq b, \\
0, \qquad & b < j \leq 2b, \\
0, \qquad & j > 2b.
\end{cases} ~~~~~
w_{I,j} = \begin{cases}
a, \qquad & 1 \leq j \leq b, \\
1, \qquad & b < j \leq 2b, \\
0, \qquad & j > 2b.
\end{cases}
\end{equation*}

\subsection{Comparison with Existing Algorithms}\label{sec:bad-instance-for-other-models}
In this section, we show that,\st{if the customers follow} \vmn{under} our \vmn{demand arrival }model, \st{then for a certain}\vmn{there exists a} class of instances \vmn{for which }our algorithms \st{have better performance}\vmn{achieve higher revenue} than algorithms designed for either the worst-case \citep{Ball2009} or the random-order model \citep{devanur2009adwords,Agrawal2009b}\vmn{, which respectively correspond to $p = 0$, and $p=1$ in our model}.
To \st{see this}\vmn{this end}, we consider\st{the following adversarial} instance $\vmn{{\vec{v}_{I}}}$ where \vmn{$$v_{I,j}=\begin{cases}a&\text{ for }1 \leq j \leq b, \\ 0 &\text{ for } j > b .\end{cases}$$}

\begin{table}[h]
{
\footnotesize
%\tbl{}
\centering
\begin{tabular}[h]{|c|c|c|c|c|}

\hline
Algorithm & Worst-Case & Random-Order & Algorithm~\ref{algorithm:hybrid}  & Algorithm~\ref{algorithm:adaptive-threshold} \\
 & {\tiny{(\cite{Ball2009})}} & {\tiny{(\vmn{Idea of}~\cite{Agrawal2009b})}} & \vmn{(Non-Adaptive Algorithm)} & \vmn{(Adaptive Algorithm)}\\
\hline
Ratio & $\frac{1}{2-a}$ & at most $p+\frac{b}{n}(1-p)$ & $p+\frac{1-p}{2-a}-O\left(\frac{1}{a(1-p)p}\sqrt{\frac{\log n}{b}}\right)$ & $1$ \\
\hline
\end{tabular}
}
\caption{\st{Expected} Ratio \vmn{between the expected revenue} of different algorithms \vmn{and the optimum offline solution}.}
%\Note{}{expected ratios of different algorithms (worst-case algorithm\cite{Ball2009}, random-order algorithms, \cite{Agrawal2009b,devanur2009adwords}, and our algorithms)}
\label{table:compare-algorithms}
\end{table}

Table~\ref{table:compare-algorithms} \st{states}\vmn{presents} the ratio between the expected revenue of different online algorithms and that of the \st{optimal}\vmn{optimum} offline \st{one}\vmn{solution}.
In the following, we will explain how we compute these bounds.
%As we can see, when $b = o(n)$, for instance ${\vec{v'}}$, our algorithms outperform algorithms designed for the worst-case and the random-order models.
Before\st{doing} that, we discuss the implications of this example.
This instance class shows that, for any $p \in (0,1)$, when $b = \omega (\sqrt{\log n})$ and $b = o(n)$ the\st{competitive} ratio for both of our algorithms is better than existing ones.
\vmn{Further, note that the ratio for Algorithm~\ref{algorithm:hybrid} is in fact its competitive ratio; thus the same ratio holds for any other instance as well.
%However, the algorithms of \cite{Ball2009} and~\cite{Agrawal2009b} can result in even worse ratios for other instances.
This implies that the competitive ratio of our non-adaptive algorithm is higher than those of
\cite{Ball2009} and~\cite{Agrawal2009b} under the partially predictable model.}
%Hence, if our model describes the true demand model, then our algorithms have a better performance guarantees than the existing ones.
Also note that for the same instance, randomizing between the algorithm of \cite{Ball2009} (with probability $\vmn{1 - }p$) and that of \cite{Agrawal2009b} (with probability $\st{1 - }p$)
\vmn{leads to a ratio of $\frac{1-p}{2-a} + p^2 + o(1)$, which is not}
\st{does not lead to a  that is}the convex combination of the \st{guarantee}\vmn{competitive ratios} of these two algorithms (as \vmn{also} pointed out in Remark~\ref{rem:alg1:1}).

\vmn{Next, we}\st{Here we} calculate the ratios listed in Table~\ref{table:compare-algorithms}. The offline solution is $OPT(\vmn{{\vec{v}_{I}}})=ab$.
The algorithm of \cite{Ball2009}, proposed for the \st{worst-case}\vmn{adversarial} model, has a fixed threshold of $\frac{1}{2-a}b$ \vmn{for accepting type-$2$ customers}, and hence accepts \st{a total of} $\frac{1}{2-a}b$ \st{class}\vmn{type}-$2$ customers.

\vmn{Next we compute the ratio for algorithms designed for the random-order model (e.g., \cite{devanur2009adwords,Agrawal2009b}, and \cite{Kesselheim}).}
\vmn{We note that, for the sake of brevity, we present an analysis based on the idea of these papers, which is allocating inventory at a roughly unform rate over the entire horizon.}
\vmn{In particular, these}  algorithms\st{proposed for the random-order model (e.g., \cite{Agrawal2009b,devanur2009adwords})} accept roughly $\lambda b$ customers at any time $\lambda\in[0,1]$.
%Further online algorithms of \cite{kleinberg2005multiple,devanur2009adwords} are based on the same idea. }}
As a result, \vmn{for this instance,} \st{these algorithms }\vmn{they} accept at most $b^2/n$ \vmn{type-$2$} customers up to time $\lambda=b/n$.
According to our model, in the arriving instance ${\vec{v}}$, there are approximately  $(1-b/n)bp$ \vmn{type-$2$ }customers arriving after time $b/n$.
Therefore, these algorithms can accept at most \st{roughly}$b^2/n + (1-b/n)bp $ \st{class}\vmn{type}-$2$ customers, which corresponds to a ratio of at most $ p + \frac{b}{n}(1-p)$.
Note that $p + \frac{b}{n}(1-p) < p+\frac{1-p}{2-a}$ for any $b < \frac{n}{2-a}$.

Our Algorithm~\ref{algorithm:hybrid} achieves a ratio of at least \st{the performance guarantee}\vmn{its competitive ratio} as given in Theorem~\ref{thm:hybrid}
and the ratio is tight for this instance (up to an additive error term of $O\left(\frac{1}{a(1-p)p}\sqrt{\frac{\log n}{b}}\right)$).
%Remark~\ref{rem:alg1:3} shows that the ratio for this instance is tight (up to an additive error term of $O\left(\frac{1}{a(1-p)p}\sqrt{\frac{\log n}{b}}\right)$).
\vmn{For  Algorithm~\ref{algorithm:adaptive-threshold}, let $c \in (0,1)$ be an arbitrary constant. We show that $ALG_{2,c}$ achieves the ratio of $1$ because the
third condition in Algorithm~\ref{algorithm:adaptive-threshold}, i.e., the dynamic threshold, is never violated. To see this we compute the threshold as follows:}
\vmn{
\begin{align*}
\lfloor \phi b + c \left(b - u_1(\lambda) \right)^+ \rfloor =
\begin{cases}
\lfloor \phi b  \rfloor  & \lambda < \delta = \frac{\phi b}{n}\\
\lfloor \phi b + c b  \rfloor  & \lambda \geq \delta
\end{cases}
\end{align*}}\vmn{where we use the fact that $u_1(\lambda) = b$ for $\lambda < \delta$, and $u_1(\lambda) = 0$ for $\lambda \geq \delta$. }
\vmn{In both cases we have $\lfloor \phi b + c \left(b - u_1(\lambda) \right)^+ \rfloor > \lambda$, which implies that the algorithm never rejects a type-$2$ customer because
 $o_2(\lambda) \leq \lambda < \lfloor \phi b + c \left(b - u_1(\lambda) \right)^+ \rfloor$.}
\if false
The second algorithm accepts all \st{class}\vmn{type}-$2$ customers before time $\delta$ (since we have $\delta \leq \phi$).
The algorithm again accepts all customers for time $\lambda>\delta$ because $u_1(\lambda) \leq \frac{o_1(\lambda)}{\lambda p}=0$.
Hence, the adaptive threshold algorithm will accept all customers and give a ratio of $1$.
\fi

\if false

For the regime where $b = \kappa n$, and $n \rightarrow \infty$, we are able to identify problem instances that provide us evidence that the following conjecture holds.

\begin{conjecture}\label{conj:ata-tightness}
Under the partially \st{learnable}\vmn{predictable} model, when  $ b= \kappa n$ where $\kappa$ is a positive constant and $n \rightarrow \infty$, no online algorithm, deterministic and randomized, can achieve a competitive ratio better than $c^*$, the optimal objective value  of (\ref{MP1}).
\end{conjecture}

In Appendix~\ref{sec:app:approx} we show how to prove the above conjecture under two approximations. However, the rigorous proof requires carefully analyzing all the approximation error terms, and it is beyond the scope of the current manuscript.

\subsection{Offline Estimation of $p$ and Robustness to Errors}\label{sec:unknown-p}
In this section, we discuss what Algorithm~\ref{algorithm:adaptive-threshold} can do if it does not know the level of learnability $p$.
We first show how we may estimate $p$ if there are multiple runs of the same process with the same $p$ (but different instances, i.e, different ${\vec{v'}\vmn{ZZ\vec{v}_I}}$'s). Further we discuss how robust Algorithm~\ref{algorithm:adaptive-threshold} is with respect to over-/under-estimating $p$.

%In Section~\ref{sec:offline-estimation}, we discuss how we may estimate $p$ if there are multiple runs of the same process with the same $p$ (but different ${\vec{v'}\vmn{ZZ\vec{v}_I}}$).
%In Section~\ref{sec:robust}, we discuss the competitive ratio of Algorithm~\ref{algorithm:adaptive-threshold} when value $p$ that Algorithm~\ref{algorithm:adaptive-threshold} uses is slightly off from the true parameter.
%\subsubsection{Offline Estimation of $p$}
%
%\label{sec:offline-estimation}
%
%%\subsection{On Knowing $p$}\label{sec:unknown-p}
%
%In this section, we discuss how we may estimate the degree of randomness $p$ using the previous runs of the process (when available).
%Note that here we do not provide any results regarding the quality of our estimation of $p$.

We start with offline estimattion of $p$.
Our estimation  is based on the fact that the inequalities in
Lemma~\ref{lemma:needed-centrality-result-for-m=2} hold with high probability.
After observing the sequence, we know $n_j$ and $o_j(\lambda)$ for $0 \leq \lambda \leq 1$ and $j = 1,2$.
However, we do not observe the arrival order of customers set by adversary, thus we do not have functions $\eta_1(\lambda)$, and $\eta_2(\lambda)$.
We find values of $p\in(0,1)$ for which Inequalities \eqref{inequality:good-approximation-o_1}, \eqref{inequality:good-approximation-o_1+o_2}, and \eqref{inequality:good-approximation-o_2} hold for at least one set of functions $\eta_j(\lambda)$, $j = 1,2$.
We estimate level of learnability by maximizing such $p$.
This optimization problem is described in (\ref{MP2}).

\begin{mdframed}
\footnotesize
\begin{subequations}
\begin{equation}\underset{(\eta_1(0), \eta_1(1/n)\dots \eta_1(1),\eta_2(0), \eta_2(1/n)\dots \eta_2(1), p) \in [0,1]^{2n+3}} {\text{Maximize}}  p \tag{MP2} \label{MP2}
\end{equation}
subject to
\begin{align*}
|\tilde o_1(\lambda) - o_1(\lambda)| &\leq \alpha \sqrt{n_1 \log n} & \text{ for all }\lambda = \frac{1}{n},\dots, 1 \text{ (if }n_1\geq \frac{k}{p^2}\log n), \\
|\tilde o_1(\lambda)
+\tilde o_2(\lambda)  - (o_1(\lambda)+o_2(\lambda))| &\leq \alpha \sqrt{(n_1+n_2) \log n} & \text{ for all }\lambda = \frac{1}{n},\dots, 1 \text{ (if }n_1\geq \frac{k}{p^2}\log n), \\
|\tilde o_2(\lambda)  - o_2(\lambda)| &\leq \alpha \sqrt{n_2 \log n} & \text{ for all }\lambda = \frac{1}{n},\dots, 1 \text{ (if }n_2\geq \frac{k}{p^2}\log n),\\
\eta_1\left(\lambda-\frac{1}{n}\right) \leq \eta_1(\lambda);&\eta_2\left(\lambda-\frac{1}{n}\right)\leq \eta_2(\lambda)& \text{ for all }\lambda = \frac{1}{n},\dots, 1, \\
\eta_1(\lambda)+\eta_2(\lambda)   \leq&\eta_1\left(\lambda-\frac{1}{n}\right)+\eta_2\left(\lambda-\frac{1}{n}\right)+1& \text{ for all }\lambda = \frac{1}{n},\dots, 1,
\end{align*}
where $ \eta_1(1) = n_1 $, $ \eta_2(1) = n_2$, $\eta_1(0) = 0$, $\eta_2(0) = 0$, $\tilde o_1(\lambda) \triangleq (1-p)\eta_1(\lambda)+p n_1 \lambda $, and $\tilde o_2 (\lambda)\triangleq (1-p)\eta_2(\lambda)+p n_2 \lambda$.
\end{subequations}
\end{mdframed}

In (\ref{MP2}), the first three sets of constraints correspond to Inequalities~\eqref{inequality:good-approximation-o_1},~\eqref{inequality:good-approximation-o_1+o_2}, and~\eqref{inequality:good-approximation-o_2} and the last two sets ensure that $\eta_j(\cdot )$  corresponds to a valid problem instance: the functions are non-decreasing, and in one time step at most one customer arrives.
(under the Approximation~\ref{assumption:continuous instance}) \vm{Dawsen: we don't need Approximation~\ref{assumption:continuous instance} to be true.. Also don't we need to comment on the case that $n_1$ or $n_2$ are too small?}

Note that the above estimation is only useful if we have access to many past runs of the process. In fact, if we only have access to one offline instance then the adversary can mislead us to estimate $p$ to be $1$ by
giving us a randomly order sequence. However, if we use the strategy of repeating the estimation, and always taking the minimum estimated $p$, then in the long run the adversary may not benefit by misleading us to overestimate $p$.

For the regime where $b = \kappa n$, and $n \rightarrow \infty$, under two approximations, we are able to show that Algorithm~\ref{algorithm:adaptive-threshold} is robust against a slightly off estimation of the  parameter $p$. In Subsection~\ref{sec:robust} we present these approximate analysis and results. However, the rigorous statement and proof require carefully analyzing all the approximation error terms, and it is beyond the scope of the current manuscript.
\vm{Dawsen in the prof are there are some ?'s...}

\fi 
\if false

\section{Problem Extensions}\label{sec:extension}

\vm{Minor edits here + putting the proofs back to the appendix. }
In this section, we study two extensions of the problem.
In both of the extensions, we relax the constraint that $v_i' \in \{0,a,1\}$ and let $v_i'$ take any positive values.
In Section~\ref{sec:discerning-online-algorithm}, we study the hypothetical situation where the algorithms can observe whether a customer is from the {\RG} group or not.
To distinguish with the original setting, we call such algorithms \emph{discerning} online algorithms.
In Section~\ref{sec:secretary}, we study the classical online secretary problem (where $b=1$ and $n\to \infty$) under our arrival model.
Since $b=1$ and the values can take any positive numbers, the objective is to maximize the probability of successfully selecting the highest-valued customer.

\subsection{Discerning Online Algorithm}\label{sec:discerning-online-algorithm}
In this section we assume that different customers take different values.
To ensure that this happens, we can add a small and independent random disturbance to the customer values.

When online algorithms can distinguish whether a customer is from the {\RG} group or not, we can learn the pattern from the first few samples of customers in the {\RG} group, and derive an online algorithm with a competitive ratio of $1-O(\epsilon)$ with some $\epsilon >0$.
%In particular, when $\epsilon$ is roughly of the order $ \Omega \left( \sqrt[3]{\frac{n \log n}{pb^2}}\right)$ (see Theorem~\ref{theorem:one-time-learning} for details), we design a near-optimal one-time learning algorithm (Algorithm~\ref{algorithm:one-time-learing}) with a competitive ratio of $1-O(\epsilon)$.

Our one-time learning algorithm builds upon the ideas of \cite{Agrawal2009b}. We solve a scaled version of the allocation linear program after observing ``enough'' customers from the {\RG} group, and we use the dual prices from the linear program to make online decisions after this initial period.
However, we modify both the initial phase and also the linear program of \cite{Agrawal2009b} to account for the adversarial component of the demand.
In particular, unlike \cite{Agrawal2009b}, we accept all of the customers in the learning phase. In the next paragraph, we show that this is indeed necessary when $p<1$.

Let us consider the example where $v_1'=b^2$ and $v_i'<1$ for all $i\geq 2$.
To achieve at least a $1/b$ fraction of the optimal revenue, an online algorithm must accept the customer with value $v_1'$.
However, if this customer  is in the {\UPG} group, which happens with probability $1-p$, then the one-time-learning algorithm of \cite{Agrawal2009b} rejects the customer with value $v_1=v_1'$.
Thus the competitive ratio of the classical one-time-learning is at most $p + \frac{1-p}{b}$, which is clearly not near-optimal when $b\geq 2$.

Note that if we use a learning period of $\epsilon $ (same as \cite{Agrawal2009b}), then we might exhaust all the inventory in this initial phase (if $b \leq \epsilon n$).
Therefore we use a different criterion in defining our learning period.
In particular, let $\zeta \triangleq \frac{p b \epsilon}{2n}$, our learning period continues until we observe $\zeta n$ customers from the {\RG} group.
Note that our learning period is not deterministic.
The following lemma (proven in Appendix~\ref{sec:proof-b=1}) gives an upper bound on the number of customers arriving during the learning period, denoted by $T$.
(By convention, we define $T=\infty$ if the learning period never ends.)
\begin{lemma}\label{lemma:stopping_time}
If $n\geq \frac{4}{\zeta}\log \frac{1}{\epsilon}$, then, $$ \prob{T\geq \frac{2\zeta n}{p}} \leq \epsilon.$$
\end{lemma}
%\begin{proof}
%Let $X$ denote the number of customers from the {\RG} group among the first $\frac{2\zeta n}{p}$ of all customers.
%Then, $X$ is drawn from the binomial distribution $Bin\left( \frac{2\zeta n}{p}, p\right)$.
%Note that $T\geq \frac{2\zeta n}{p}$ implies $X\leq \zeta n$.
%Thus, $\prob{T\geq \frac{2\zeta n}{p}} \leq \prob{X\leq \zeta n}$.
%Chernoff bound gives, for any $\delta \in (0,1)$ and $\mu = \mathbb{E}(X)$, $ \mathbb{P}(X < (1-\delta)\mu) < e^{-\delta^2\mu/2}.$
%Therefore, setting $\delta = 1/2$ and $\mu=\zeta n$,
%\begin{align*}
%\prob{T\geq \frac{2\zeta n}{p}} \leq \prob{X\leq \zeta n} \leq e^{-\frac{1}{4}\zeta n} \leq \epsilon,
%\end{align*}
%where the last inequality follows from our assumption $n\geq \frac{4}{\zeta}\log \frac{1}{\epsilon}$.
%\end{proof}

At time $T$, we solve a scaled offline linear program (\ref{Primal}) and its dual (\ref{Dual}) to estimate the dual price.
In the following, we use $S_\RGS$ to denote the set of customers from the {\RG} group we have observed in the learning period.
In~\ref{Primal}, we first reduce the inventory to $b - \frac{2\zeta n}{p}$ to compensate for the customers we have accepted during the learning period (note that by Lemma~\ref{lemma:stopping_time}, $\prob{T\geq \frac{2\zeta n}{p}} \leq \epsilon$).
Then, we scale the inventory by $\zeta$ (because $|S_{\RGS}|=\zeta n$).
Finally we scale the inventory by a factor of $(1-\epsilon)$ to ensure that, with high probability, what we learn will select \vm{Dawsen: what do you mean here?!} at most $b - \frac{2\zeta n}{p}$ future customers.
% Program~\ref{Primal} describes the LP and its dual.
After this initial phase we use $s^*$ (optimal solution of variable $s$ of~\ref{Dual}) to decide whether to accept/reject a \st{class}\vmn{type}-$2$ customer.
The formal definition of our algorithm is given in Algorithm~\ref{algorithm:one-time-learing}.

\begin{mdframed}
\begin{align}
\underset{x_i \in [0,1] \text{ for all } i\in S_\RGS} {\text{Maximize}}&&& \sum_{i\in S_\RGS} v_ix_i \tag{Primal} \label{Primal} \\
\text{subject to}&&& \sum_{i\in S_\RGS} x_i \leq (1-\epsilon) \zeta \left(b - \frac{2\zeta n}{p}\right) .\nonumber \\
\hline \nonumber \\
\underset{s\geq 0; y_i \geq 0 \text{ for all } i\in S_\RGS} {\text{Minimize}}&&& (1-\epsilon) \zeta \left( b - \frac{2\zeta n}{p} \right) s+\sum_{i\in S_\RGS} y_i \tag{Dual} \label{Dual} \\
\text{subject to}&&&\nonumber \\
\text{for all }i \in S_\RGS &&& s+y_i \geq v_i. \nonumber
\end{align}
\end{mdframed}
\begin{algorithm}[H]
\begin{enumerate}
\item Initialize $S_\RGS \leftarrow \emptyset$ and define $\zeta \triangleq \frac{p b \epsilon}{2n}$.
\item Repeat for time $\lambda = 1/n,2/n,\dots, 1$, accept the customer $i(=\lambda n)$ arriving at time $\lambda$ if there is remaining inventory and one of the following holds:
\begin{enumerate}
\item \textbf{Before learning:} $|S_\RGS|< \zeta n$. \\If customer $i$ is from the {\RG} group, then update $S_\RGS \leftarrow S_\RGS \cup \{ i \}$; If after updating, $|S_\RGS| \geq \zeta n$, then find the optimal solution $s^*$ of variable $s$ of~\ref{Dual}.
\item \textbf{After learning:} $|S_\RGS| \geq \zeta n$ and $v_i > s^*$.
\end{enumerate}
\end{enumerate}
\caption{The One-Time Learning Algorithm. (OTL)} \label{algorithm:one-time-learing}
\end{algorithm}
We show that the one-time-learning algorithm achieves the following near-optimal results:
\begin{theorem} \label{theorem:one-time-learning}
If $\epsilon < \frac{p}{2}$ and $b^2 \geq \frac{12n}{p \epsilon ^3} \log \left( \frac{n+1}{\epsilon}\right)$, then the one-time learning algorithm is $(1-O(\epsilon))-$competitive in the partially \st{learnable}\vmn{predictable} model.
\end{theorem}
Roughly speaking, (if we ignore the $\log \frac{1}{\epsilon}$ terms and make some minor assumption \vm{Dawsen make assumption on what?!}), the one-time-learning algorithm is near optimal in the regime where $b=\omega\left( \sqrt{{n\log n}}\right)$.
On the other hand, Proposition~\ref{thm:impossibility} also applies to discerning online algorithms.
Therefore, it is not possible to achieve near-optimality when $b = o\left( \sqrt{{n}}\right)$.
Thus we have achieved near-optimality in a fairly wide range of $b$ (with respect to $n$) where it is still possible.
The proof of this theorem uses similar ideas as in the proofs of ~\cite{Agrawal2009b}. We defer the proof to Appendix~\ref{sec:proof-b=1}.

\fi
\section{The Secretary Problem \vmn{under Partially Predictable Demand}}
\label{sec:secretary}

In \vmn{this} section, we study the online secretary problem under our \vmn{new} arrival model.
\vmn{In our setting, the secretary problem corresponds to having one unit of inventory, i.e., $b = 1$, and $n$ customers, where $v_{I,j} \in \mathbb{R}^{+}$ for $1\leq j \leq n$, i.e., we relax the assumption that there are only two types. The objective is to maximize the probability of\st{successfully} selecting the highest-revenue customer in the asymptotic regime, where $n\to \infty$.}

\vmn{In the classical setting, the arrival sequence  is assumed to be a uniformly random permutation of $n$ customers, which corresponds to the extreme case of $p = 1$ under our
partially predictable model.}
In this \st{case}\vmn{setting}, it is well known that the best-possible online algorithm \st{works as follows}\vmn{is the following {\em deterministic} algorithm} \citep{lindley1961dynamic,dynkin1963optimum,ferguson1989solved,freeman1983secretary}:
Observe the first $\left\lfloor \gamma n \right\rfloor$ customers, where $\gamma = \frac{1}{e}$;  then accept the next one that has the highest \st{value}\vmn{revenue} so far (if any). \vmn{The success probability of this algorithm approaches $\frac{1}{e} \approx 0.37$ as $n\to \infty$.}
\vmn{We generalize the classical setting by studying the problem under our demand model}. \vmn{First, we analyze the success probability of a similar class of algorithms for any $p \in (0,1]$. Next, we show that under our demand model where $p <1$---i.e., in the presence of an \vmn{adversarial} component---this class of algorithms is not necessarily the best possible.}

\if false
(where $b=1$ and $n\to \infty$) under our arrival model.
Since $b=1$ and the values can take any positive numbers, the objective is to maximize the probability of successfully selecting the highest-valued customer.
\fi
\if false
When $p=1$, our setting  reduces to the classical online secretary problem.
In this case, it is well-known that the best-possible online algorithm works as follows (\cite{dynkin1963optimum,lindley1961dynamic,ferguson1989solved,freeman1983secretary}):
Observe the fist $\left\lfloor \gamma n \right\rfloor$ customers, where $\gamma = \frac{1}{e}$;  then accept the next one that has the highest value so far (if any).
\fi

\vmn{For any $\gamma \in (0,1)$, we define the Observation-Selection Algorithm ($\text{OSA}_\gamma$), which works similarly to the classical algorithm described above.}
%Under our partially \st{learnable}\vmn{predictable} model (i.e, when $p \in (0,1)$), we analyze a similar class of algorithms that we call Observation-Selection Algorithm ($\text{OSA}_\gamma$), for $\gamma \in [0,1]$.
The formal definition of the algorithm is presented in Algorithm~\ref{algorithm:observation-selection}.

\begin{algorithm}[H]
\begin{enumerate}
\item Initialize $v_{\max} \leftarrow 0$.
\item \textbf{Observation period:} Repeat for customer $i = 1,2,\dots, \lfloor \gamma n \rfloor$: reject customer $i$ and update $v_{\max} \leftarrow \max\{v_{\max} ,v_i\}$.
\item \textbf{Selection period:} Repeat for customer $i = \lfloor \gamma n \rfloor+ 1, \lfloor \gamma n \rfloor+ 2\dots, n$:
\begin{itemize}
\item If $v_i \geq v_{\max}$, then select customer $i$ and stop the algorithm.
\item Otherwise, reject customer $i$.
\end{itemize}
\end{enumerate}
\caption{: Observation-Selection Algorithm ($\text{OSA}_\gamma$, $\gamma \in (0,1)$)}
\label{algorithm:observation-selection}
\end{algorithm}

\noindent{In Appendix~\ref{sec:proof-b=1}, we analyze the success probability of Algorithm~\ref{algorithm:observation-selection} \st{(i.e., the probability that it selects the customer with the highest value)} and prove the following theorem:}

%in our partially \st{learnable}\vmn{predictable} model (i.e, when $p \in (0,1)$).
%
%, described in Algorithm~\ref{algorithm:observation-selection}, for general $\gamma$ in our partially \st{learnable}\vmn{predictable} model.
%We first list our results and then prove them.
%
%Our first result is the following theorem (proved in Appendix~\ref{sec:proof-b=1}):
\begin{theorem}\label{thm:b=1}
Under the partially \st{learnable}\vmn{predictable} model, \st{when}\vmn{in the limit} $n \rightarrow \infty$, the success probability of $\text{OSA}_\gamma$ approches $\gamma p \log \frac{1}{\gamma p + 1-p}.$
%Under the partially \st{learnable}\vmn{predictable} model where the probability a customer being at the {\RG} group is $p$, as $n\to \infty$, the success probability of $\text{OSA}_\gamma$ is $\gamma p \log \frac{1}{\gamma p + 1-p}.$
\end{theorem}
By optimizing over $\gamma$, we obtain the following corollary:
\begin{corollary}\label{cor:best-obs-length}
Let $\gamma^* \in(0,1)$ be the unique solution to $$ \log (\gamma^* p + 1-p) + \frac{\gamma^* p}{\gamma^* p + 1-p}=0;$$
then, $\text{OSA}_{\gamma^*}$ achieves the highest success probability among $\text{OSA}_\gamma$ for all $\gamma \in (0,1)$.
\end{corollary}
%\begin{proof}
%Using Theorem~\ref{thm:b=1} and taking the derivative, we obtain the corollary.
%\end{proof}
\begin{table}[h]
{
%\footnotesize
%\tbl{}
%\centering
\begin{center}
\begin{tabular}[h]{|c|c|c|c|c|c|c|c|c|c|c|}
\hline
$p$ & $0.1$& $0.2$ & $0.3$&$0.4$&$0.5$&$0.6$&$0.7$&$0.8$&$0.9$&$1$\\
\hline
$\gamma^*$ & $0.4935$ & $0.4863$ & $0.4784$ & $0.4696$ & $0.4597$ & $ 0.4482$ & $ 0.4348$ & $ 0.4184$ & $ 0.3975$ & $0.3679$ \\
\hline
$OSA_{\gamma^*}$ & $0.0026$ & $ 0.0105$ & $0.0244$ & $0.0448$ & $ 0.0724$ & $ 0.1081 $ & $ 0.1533 $ & $ 0.2095 $ & $ 0.2796 $ & $ 0.3679$ \\
\hline
\end{tabular}
\end{center}
}
\caption{The optimal length of the observation period, $\gamma^*$, and the success probability of $OSA_{\gamma^*}$ \st{for different values of}\vmn{vs.} $p$.}
%\Note{}{the optimal length of the observation period $\gamma^*$ and the success probability of $OSA_{\gamma^*}$ for different values of $p$}
\label{table:b=1}
\end{table}
Table~\ref{table:b=1} \st{shows}\vmn{presents} the optimal length of the observation period, $\gamma^*$, and the success probability of $OSA_{\gamma^*}$ for different values of $p$.
\vmn{We observe that as the size of the stochastic component increases, i.e., as $p$ increases, the length of the observation period decreases, whereas the success probability increases.}

Next, \vmn{in the following proposition,} we \vmn{establish a lower bound on the success probability when we randomize over the length of the observation period ($\gamma$); further, we present an example that shows such randomization increases the success probability for $p<1$. This illustrates the benefit of employing randomized algorithms in the presence of an adversarial component in the arrival sequence.}
\st{show that when $p<1$, randomizing over the length of the observation period ($\gamma$) may increase the success probability.}
\st{Intuitively, when $p<1$, the customer arrival sequence has an ``adversarial component'', and thus randomization may help.
In particular, we have the following proposition (proven in Appendix~\ref{sec:proof-b=1}):}
\begin{proposition}\label{thm:randomized-b=1}
Under the partially \st{learnable}\vmn{predictable} model, for any $0 < \gamma_1 < \gamma_2 < 1$ and $0<q<1$, the randomized algorithm that runs $\text{OSA}_{\gamma_1}$ with probability $q$ and $\text{OSA}_{\gamma_2}$ with probability $1-q$ has \vmn{an asymptotic} success probability of at least $$ q s_1 +(1-q) s_2 +\min \left\{ (1-q)p(1-p)(1-\gamma_2), q(1-p)\frac{\gamma_2-\gamma_1}{1-\gamma_1}s_1 \right\} $$
where for $i=1,2$, $s_i$ denotes the success probability of $\text{OSA}_{\gamma_i}$.
\end{proposition}

\vmn{The proposition is proven in Appendix~\ref{sec:proof-b=1}.} \vmn{Suppose $p=0.5$; randomizing over $\gamma_1=0.427$ and $\gamma_2=0.69$ with $q=0.824$}
\vmn{results in a success probability of at least $0.083$ (utilizing the result of  Proposition~\ref{thm:randomized-b=1}).}
\st{When $p<1$, this lower bound on the success probability shows that randomizing the length of the observation period can lead to higher success probability.
For example, when $p=0.5$, if we set $\gamma_1=0.427$, $\gamma_2=0.69$, and $q=0.824$, then the success probability is $0.083$.}
On the other hand, the success probability of the best possible \vmn{deterministic observation period} $\text{OSA}_\gamma$, given in Theorem~\ref{thm:b=1} and Corollary~\ref{cor:best-obs-length}, is $0.072$.

%\input{extensions-app}
%
%
%\end{document}
%
%
%
%
%

%\input{head-thesis}
\section{Conclusion}\label{sec:conclusion}

%\vm{Added a summary, highlight main points. Added some more future research.}

%In such cases the firm can take a worst-case approach and assume that the demand is controlled
%by an imaginary adversary, and thus it is unpredictable. Such an approach, however, usually results
%in online policies that are too conservative (as studied in Ball and Queyranne (2009) and others).
%Instead, the firm may wish to employ online policies that are based on models that assume
%the future demand can partially be predicted, avoiding being too conservative while not being
%reliant on fully accurate prediction. This paper aims to investigate to what extent the above goal is
%achievable: we propose a new demand model, called partially predictable, that contains both adversarial
%(thus unpredictable) and stochastic (predictable) components. We design novel algorithms
%to demonstrate that even though the demand is assumed to include an unpredictable component,
%the firm can make use of the limited information that the data reveals and improve upon the
%completely conservative approach.
%We study a basic online allocation problem of a single resource with an arbitrary capacity to a
%sequence of customers that belongs to two types. Each customer demands one unit of the resource.
%If the resource is allocated, the firm earns a type-dependent revenue.

Online resource allocation is a central problem in the operations of numerous \st{industry sectors}\vmn{online platforms} ranging from airline \vmn{booking systems} to hotel \vmn{booking systems} to internet \st{advertisement}\vmn{advertising}.
Despite advances in information technology, demand arrival processes are rarely perfectly predictable. The presence of unpredictable patterns limits the
performance of most allocation algorithms that rely \vmn{on fully accurate prediction of}\st{predicting} future demand based on observed data.
At the same time, ignoring available information and taking a completely \st{robust approach}\vmn{worst-case approach} usually leads to \vmn{online} allocation policies that are too conservative.
In this paper we take a middle ground approach and introduce the first arrival model that contains both \vmn{adversarial (thus unpredictable)} and \vmn{stochastic} (predictable) components.
Our demand model requires no forecast of demand; however, the stochastic component allows us to partially \st{learn}\vmn{predict} future demand as the sequence of arrivals unfolds.
\vmn{In our model, }the relative size of the \vmn{stochastic} component, $p$, represents the level of \st{learnablity}\vmn{predictability} of the demand.

Under our proposed demand model, \vmn{we study the basic yet fundamental  problem of allocating a single resource with an arbitrary initial inventory to a
sequence of customers that belong to two types, with type-$1$ generating  higher revenue.}\st{we study a basic yet fundamental online resource allocation problem known as the single-resource revenue management with $2$ fare classes.}
For this problem, we design\st{two online algorithms} \vmn{a non-adaptive algorithm as well as an adaptive one.} We analyze the competitive ratios of our algorithms
    and show that they outperform existing ones  under our proposed demand model.
The first implication of our analysis is that\vmn{, by employing our algorithms, }we can take advantage of limited available  information (due to the presence of the \st{predictable}\vmn{stochastic} component) to improve the revenue of the firm compared to a fully \st{robust}\vmn{conservative} approach.
Indeed, the competitive ratios of our algorithms are parameterized by $p$, and for both algorithms the ratio increases with $p$ (the relative size of the \st{predictable}\vmn{stochastic} component), which highlights the value of \st{learnablity}\vmn{even partial predictability}.

\vmn{Further, we show that our adaptive algorithm---which repeatedly computes upper bounds on the total number of customers of each type based on observed data, and makes online decisions based on those bounds---achieves a higher competitive ratio when the initial inventory $b$ is sufficiently large. This underlines the significance of adapting to the data, even though it contains an adversarial component. Analyzing the adaptive algorithm, however, is considerably more challenging. We establish a  lower bound on the competitive ratio by constructing a novel factor-revealing mathematical program.}

\vmn{On the other hand, when $b$ is small (more precisely, when $b = o(\sqrt{n})$), we prove an upper bound on the competitive ratio of any deterministic or randomized online algorithm that matches the competitive ratio of
our non-adaptive algorithm (up to an error term). This implies (1) our non-adaptive algorithm is the best possible in this regime, and (2) when the initial inventory is small relative to the
time horizon, we may not be able to effectively adapt to observed data before allocating most of the inventory.}
\vmn{We also have heuristic arguments---in which we do not characterize the error terms---that indicate that (1) our adaptive algorithm achieves the best possible competitive ratio in the regime where $b = \kappa {n}$ (where $\kappa \in (0,1]$ is a constant) and (2)  underestimating parameter $p$ does not affect the competitive ratio of our adaptive algorithm, whereas (3) if we overestimate $p$ by (a small amount), its competitive ratio decreases only slightly.
Because making the above results rigorous will make the paper prohibitively long, these results are not included in the paper.}

%For the regime where $b = \kappa {n}$  we have identified a class of instances
%We also have strong evidence that when  $b = \kappa {n}$  (with $\kappa$ being a positive constant), in the asymptotic regime no online algorithm can achieve a competitive ratio better than that of Algorithm~\ref{algorithm:adaptive-threshold}.

\vmn{To illustrate the application of our model to other online allocation problems, we study the secretary problem under our demand model. We analyze the celebrated policy of selecting the highest revenue customer after an observation period with a deterministic length of $\gamma$ under our new model, and find the optimum value of $\gamma$ (which is parameterized by $p$). We further show that, in the presence of an adversarial component and unlike the classical setting, randomizing over the length of the observation period may increase the probability of selecting the highest revenue customer.}

\vmn{In this paper, we} use a discrete time model and also assume that the arrival times of customers from the {\RG} group are randomly permuted among their predetermined  positions.
We believe similar results can be obtained for a model where a total of $n$ customers from the two groups  (i.e., the {\RG} and {\UPG} group) arrive according to independent Poisson processes with rates $p$ and $1-p$. We leave the rigorous treatment of this alternative model for future research.

\vmn{Studying other online allocation problems under our new demand model is a promising direction for future research.
Our consequential concentration result from Lemma~\ref{lemma:needed-centrality-result-for-m=2} can be extended to any finite number of types.
Further, we believe that by combining our ideas for adaptively \st{learning}\vmn{computing} bounds on the demand of each \vmn{type} with those of \cite{Lan2008}, and utilizing the concentration results, one can generalize our algorithms to a setting with any finite number of types.
Such extensions are, however, beyond the scope of this paper.
}

\if false

Focusing on this simple online allocation problem allows us to characterize the fundamental limits of our proposed demand model.
In particular, Remark~\ref{rem:alg1:2} shows that for $b = o(\sqrt{n})$ no online algorithm can  achieve a competitive ratio better than that of Algorithm~\ref{algorithm:hybrid} (ignoring the error term). Further,  as discussed in Remark~\ref{remark:alg2-2}, we have strong evidence that when  $b = \kappa {n}$  (with $\kappa$ being a positive constant), in the asymptotic regime no online algorithm can achieve a competitive ratio better than that of Algorithm~\ref{algorithm:adaptive-threshold}.
Therefore, the parameterized competitive ratios of our algorithms not only  reveal the fundamental limit that there is an unavoidable loss due to the presence of an unpredictable component, but also characterizes this loss - which can be thought as as the price of partial robustness - for a broad ranges of $b/n$.

%further   reveal the following fundamental limit:  there is an unavoidable loss due to the presence of an unpredictable component.
%One can consider such a loss as the price of partial robustness.

%The premise of our partially-\st{learnable}\vmn{predictable} demand model is to find a compromise between giving adversary the full power to determine the arrival order and none at all.
%In our model, the adversary can only determine the order without knowing which customer is going to join the {\RG} group. Such a model cannot capture sudden unpredictable change in the demand.

We conclude the paper by outlining several directions for future research.
Here we use a discrete time model, and also assume that the arrival times of customers from the {\RG} group are randomly permuted among their pre-determined  positions.
We believe similar results can be obtained for a model where a total of $n$ customers of the two groups (i.e., the {\RG} and {\UPG} group) arrive according to independent Poisson processes with rate $p$ and $1-p$. We leave the rigorous treatment of this alternative model for future research.

In this paper, we focus on the special case where customers belong to only $2$ \st{class}\vmn{type}es.
We believe that by combining our ideas for adaptively learning bounds on the demand of each \st{class}\vmn{type} with those of \cite{Lan2008}, one can generalize our algorithms to the multiple-fare \st{class}\vmn{type} problem. In particular, achieving the convex combination of the worst-case and the average-case bounds in the generalized model is an interesting open question.
Finally, studying more general resource allocation and matching problems under our proposed demand model is a promising direction for future research.

\fi

\ACKNOWLEDGMENT{Jaillet acknowledges the research support of the Office of Naval Research grants N00014-12-1-0999 and N00014-16-1-2786. We would like to thank L$\rm \hat{e}$ Nguy$\rm\hat{e}$n Hoang for fruitful discussions on the concentration result.} 

\bibliographystyle{chicago}
\bibliography{bibdata}

\begin{thebibliography}{}

\bibitem[\protect\citeauthoryear{Agrawal, Wang, and Ye}{Agrawal
  et~al.}{2014}]{Agrawal2009b}
Agrawal, S., Z.~Wang, and Y.~Ye (2014).
\newblock {A Dynamic Near-Optimal Algorithm for Online Linear Programming}.
\newblock {\em Operations Research\/}~{\em 62\/}(4), 876--890.

\bibitem[\protect\citeauthoryear{Araman and Caldentey}{Araman and
  Caldentey}{2009}]{araman2009dynamic}
Araman, V.~F. and R.~Caldentey (2009).
\newblock Dynamic pricing for nonperishable products with demand learning.
\newblock {\em Operations research\/}~{\em 57\/}(5), 1169--1188.

\bibitem[\protect\citeauthoryear{Ball and Queyranne}{Ball and
  Queyranne}{2009}]{Ball2009}
Ball, M.~O. and M.~Queyranne (2009).
\newblock Toward robust revenue management: Competitive analysis of online
  booking.
\newblock {\em Operations Research\/}~{\em 57\/}(4), 950--963.

\bibitem[\protect\citeauthoryear{Belobaba}{Belobaba}{1987}]{belobaba1987survey}
Belobaba, P.~P. (1987).
\newblock Survey paper-airline yield management an overview of seat inventory
  control.
\newblock {\em Transportation Science\/}~{\em 21\/}(2), 63--73.

\bibitem[\protect\citeauthoryear{Belobaba}{Belobaba}{1989}]{belobaba1989or}
Belobaba, P.~P. (1989).
\newblock {OR} practice-application of a probabilistic decision model to
  airline seat inventory control.
\newblock {\em Operations Research\/}~{\em 37\/}(2), 183--197.

\bibitem[\protect\citeauthoryear{Ben-Tal and Nemirovski}{Ben-Tal and
  Nemirovski}{2002}]{BanTal}
Ben-Tal, A. and A.~Nemirovski (2002).
\newblock Robust optimization -- methodology and applications.
\newblock {\em Mathematical Programming\/}~{\em 92\/}(3), 453--480.

\bibitem[\protect\citeauthoryear{Bertsimas, Pachamanova, and Sim}{Bertsimas
  et~al.}{2004}]{Bertsimas_RobustLP}
Bertsimas, D., D.~Pachamanova, and M.~Sim (2004, November).
\newblock Robust linear optimization under general norms.
\newblock {\em Operations Research Letters\/}~{\em 32\/}(6), 510--516.

\bibitem[\protect\citeauthoryear{Besbes and Zeevi}{Besbes and
  Zeevi}{2009}]{besbes2009dynamic}
Besbes, O. and A.~Zeevi (2009).
\newblock Dynamic pricing without knowing the demand function: Risk bounds and
  near-optimal algorithms.
\newblock {\em Operations Research\/}~{\em 57\/}(6), 1407--1420.

\bibitem[\protect\citeauthoryear{Brumelle and McGill}{Brumelle and
  McGill}{1993}]{brumelle1993airline}
Brumelle, S.~L. and J.~I. McGill (1993).
\newblock {Airline seat allocation with multiple nested fare classes}.
\newblock {\em Operations Research\/}~{\em 41\/}(1), 127--137.

\bibitem[\protect\citeauthoryear{Buchbinder and Naor}{Buchbinder and
  Naor}{2009}]{buchbinder2009design}
Buchbinder, N. and J.~Naor (2009).
\newblock {The design of competitive online algorithms via a primal: dual
  approach}.
\newblock {\em Foundations and Trends$\textregistered$ in Theoretical Computer
  Science\/}~{\em 3\/}(2--3), 93--263.

\bibitem[\protect\citeauthoryear{Cheapair}{Cheapair}{2016}]{CheapoAir}
Cheapair (2016).
\newblock What the airlines never tell you about airfares.
\newblock {\em
  \url{https://www.cheapair.com/blog/travel-tips/what-the-airlines-never-tell-you-about-airfares/}\/}.

\bibitem[\protect\citeauthoryear{Chernoff}{Chernoff}{1952}]{chernoff1952measure}
Chernoff, H. (1952).
\newblock {A measure of asymptotic efficiency for tests of a hypothesis based
  on the sum of observations}.
\newblock {\em The Annals of Mathematical Statistics\/}, 493--507.

\bibitem[\protect\citeauthoryear{Ciocan and Farias}{Ciocan and
  Farias}{2012}]{Ciocan2012}
Ciocan, D.~F. and V.~Farias (2012).
\newblock {Model Predictive Control for Dynamic Resource Allocation}.
\newblock {\em Mathematics of Operations Research\/}~{\em 37\/}(3), 501--525.

\bibitem[\protect\citeauthoryear{Cooper}{Cooper}{2002}]{Cooper2002}
Cooper, W.~L. (2002).
\newblock {Asymptotic Behavior of an Allocation Policy for Revenue Management}.
\newblock {\em Operations Research\/}~{\em 50\/}(4), 720--727.

\bibitem[\protect\citeauthoryear{CWT}{CWT}{2016}]{Adv_Booking}
CWT (2016).
\newblock Gender differences in booking business travel.
\newblock {\em Carson Wagonlit Travel Report\/}.

\bibitem[\protect\citeauthoryear{Devanur and Hayes}{Devanur and
  Hayes}{2009}]{devanur2009adwords}
Devanur, N.~R. and T.~P. Hayes (2009).
\newblock {The adwords problem: online keyword matching with budgeted bidders
  under random permutations}.
\newblock In {\em Proceedings of the 10th ACM conference on Electronic
  Commerce}, pp.\  71--78.

\bibitem[\protect\citeauthoryear{Dynkin}{Dynkin}{1963}]{dynkin1963optimum}
Dynkin, E.~B. (1963).
\newblock The optimum choice of the instant for stopping a markov process.
\newblock In {\em Soviet Math. Dokl}, Volume~4.

\bibitem[\protect\citeauthoryear{Esfandiari, Korula, and Mirrokni}{Esfandiari
  et~al.}{2015}]{esfandiari2015online}
Esfandiari, H., N.~Korula, and V.~Mirrokni (2015).
\newblock {Online Allocation with Traffic Spikes: Mixing Adversarial and
  Stochastic Models}.
\newblock In {\em Proceedings of the Sixteenth ACM Conference on Economics and
  Computation}, pp.\  169--186.

\bibitem[\protect\citeauthoryear{Ferguson}{Ferguson}{1989}]{ferguson1989solved}
Ferguson, T.~S. (1989).
\newblock Who solved the secretary problem?
\newblock {\em Statistical science\/}, 282--289.

\bibitem[\protect\citeauthoryear{Freeman}{Freeman}{1983}]{freeman1983secretary}
Freeman, P.~R. (1983).
\newblock {The secretary problem and its extensions: A review}.
\newblock {\em International Statistical Review/Revue Internationale de
  Statistique\/}, 189--206.

\bibitem[\protect\citeauthoryear{Hush and Scovel}{Hush and
  Scovel}{2005}]{hush2005concentration}
Hush, D. and C.~Scovel (2005).
\newblock {Concentration of the hypergeometric distribution}.
\newblock {\em Statistics \& Probability Letters\/}~{\em 75\/}(2), 127--132.

\bibitem[\protect\citeauthoryear{Jasin}{Jasin}{2015}]{Jasin2015}
Jasin, S. (2015).
\newblock {Performance of an LP-Based Control for Revenue Management with
  Unknown Demand Parameters}.
\newblock {\em Operations Research\/}~{\em 63\/}(4), 909–--915.

\bibitem[\protect\citeauthoryear{Kesselheim, Kleinberg, and
  Niazadeh}{Kesselheim et~al.}{2015}]{kesselheim2015secretary}
Kesselheim, T., R.~Kleinberg, and R.~Niazadeh (2015).
\newblock Secretary problems with non-uniform arrival order.
\newblock In {\em Proceedings of the Forty-Seventh annual ACM on symposium on
  Theory of Computing}, pp.\  879--888.

\bibitem[\protect\citeauthoryear{Kesselheim, T\"{o}nnis, Radke, and
  V\"{o}cking}{Kesselheim et~al.}{2014}]{Kesselheim}
Kesselheim, T., A.~T\"{o}nnis, K.~Radke, and B.~V\"{o}cking (2014).
\newblock {Primal Beats Dual on Online Packing LPs in the Random-Order Model}.
\newblock In {\em Proceedings of the 46th annual ACM symposium on Theory of
  Computing}, Number~1, pp.\  303--312.

\bibitem[\protect\citeauthoryear{Kleinberg}{Kleinberg}{2005}]{kleinberg2005multiple}
Kleinberg, R. (2005).
\newblock {A multiple-choice secretary algorithm with applications to online
  auctions}.
\newblock In {\em Proceedings of the sixteenth annual ACM-SIAM symposium on
  Discrete Algorithms}, pp.\  630--631.

\bibitem[\protect\citeauthoryear{Lan, Gao, Ball, and Karaesmen}{Lan
  et~al.}{2008}]{Lan2008}
Lan, Y., H.~Gao, M.~O. Ball, and I.~Karaesmen (2008).
\newblock {Revenue Management with Limited Demand Information}.
\newblock {\em Management Science\/}~{\em 54\/}(9), 1594--1609.

\bibitem[\protect\citeauthoryear{Lautenbacher and {Stidham Jr.}}{Lautenbacher
  and {Stidham Jr.}}{1999}]{lautenbacher1999underlying}
Lautenbacher, C.~J. and S.~{Stidham Jr.} (1999).
\newblock {The underlying Markov decision process in the single-leg airline
  yield-management problem}.
\newblock {\em Transportation Science\/}~{\em 33\/}(2), 136--146.

\bibitem[\protect\citeauthoryear{Lee and Hersh}{Lee and
  Hersh}{1993}]{lee1993model}
Lee, T.~C. and M.~Hersh (1993).
\newblock {A model for dynamic airline seat inventory control with multiple
  seat bookings}.
\newblock {\em Transportation Science\/}~{\em 27\/}(3), 252--265.

\bibitem[\protect\citeauthoryear{Lindley}{Lindley}{1961}]{lindley1961dynamic}
Lindley, D.~V. (1961).
\newblock {Dynamic programming and decision theory}.
\newblock {\em Applied Statistics\/}, 39--51.

\bibitem[\protect\citeauthoryear{Littlewood}{Littlewood}{2005}]{littlewood2005special}
Littlewood, K. (2005).
\newblock {Special issue papers: Forecasting and control of passenger
  bookings}.
\newblock {\em Journal of Revenue and Pricing Management\/}~{\em 4\/}(2),
  111--123.

\bibitem[\protect\citeauthoryear{Mahdian, Nazerzadeh, and Saberi}{Mahdian
  et~al.}{2007}]{Mahdian2007}
Mahdian, M., H.~Nazerzadeh, and A.~Saberi (2007).
\newblock {Allocating online advertisement space with unreliable estimates}.
\newblock In {\em Proceedings of the 8th ACM conference on Electronic
  Commerce}, pp.\  288--294.

\bibitem[\protect\citeauthoryear{McDiarmid}{McDiarmid}{1998}]{mcdiarmid1998concentration}
McDiarmid, C. (1998).
\newblock {Concentration}.
\newblock In {\em Probabilistic methods for algorithmic discrete mathematics},
  pp.\  195--248. Springer.

\bibitem[\protect\citeauthoryear{Mehta, Saberi, Vazirani, and Vazirani}{Mehta
  et~al.}{2007}]{Mehta2007a}
Mehta, A., A.~Saberi, U.~Vazirani, and V.~Vazirani (2007).
\newblock {Adwords and generalized online matching}.
\newblock {\em Journal of the ACM (JACM)\/}~{\em 54\/}(5), 22.

\bibitem[\protect\citeauthoryear{Mirrokni, Oveis~Gharan, and
  Zadimoghaddam}{Mirrokni et~al.}{2012}]{Mirrokni2012}
Mirrokni, V., S.~Oveis~Gharan, and M.~Zadimoghaddam (2012).
\newblock {Simultaneous approximations for adversarial and stochastic online
  budgeted allocation}.
\newblock In {\em Proceedings of the Twenty-Third annual ACM-SIAM symposium on
  Discrete Algorithms}, pp.\  1690--1701. SIAM.

\bibitem[\protect\citeauthoryear{Shamsi, Holtan, Luenberger, and Ye}{Shamsi
  et~al.}{2014}]{Shamsi}
Shamsi, D., M.~Holtan, R.~Luenberger, and Y.~Ye (2014).
\newblock Online allocation rules in display advertising.
\newblock {\em arXiv preprint arXiv:1407.5710\/}.

\bibitem[\protect\citeauthoryear{Stein, Truong, and Wang}{Stein
  et~al.}{2016}]{Van-Ahn1}
Stein, C., V.-A. Truong, and X.~Wang (2016).
\newblock {Advance Service Reservations with Heterogeneous Customers}.
\newblock {\em Working paper\/}.

\bibitem[\protect\citeauthoryear{Talluri and Ryzin}{Talluri and
  Ryzin}{1998}]{Talluri:1998}
Talluri, K. and G.~V. Ryzin (1998, November).
\newblock An analysis of bid-price controls for network revenue management.
\newblock {\em Manage. Sci.\/}~{\em 44\/}(11), 1577--1593.

\bibitem[\protect\citeauthoryear{Talluri and {Van Ryzin}}{Talluri and {Van
  Ryzin}}{2006}]{talluri2006theory}
Talluri, K.~T. and G.~J. {Van Ryzin} (2006).
\newblock {\em {The theory and practice of revenue management}}, Volume~68.
\newblock Springer Science \& Business Media.

\bibitem[\protect\citeauthoryear{Wang, Xie, Qiu, Yang, Zhang, and
  Greenberg}{Wang et~al.}{2006}]{WangTraffic}
Wang, H., H.~Xie, L.~Qiu, Y.~R. Yang, Y.~Zhang, and A.~Greenberg (2006,
  August).
\newblock Cope: Traffic engineering in dynamic networks.
\newblock {\em ACM SIGCOMM Computer Communication Review\/}~{\em 36\/}(4),
  99--110.

\bibitem[\protect\citeauthoryear{Wang and Truong}{Wang and
  Truong}{2015}]{Van-Ahn2}
Wang, X. and V.-A. Truong (2015).
\newblock {Online Advance Admission Scheduling for Services, with Customer
  Preferences}.
\newblock {\em Working paper\/}.

\bibitem[\protect\citeauthoryear{Zervas, Proserpio, and Byers}{Zervas
  et~al.}{2016}]{AirBnb}
Zervas, G., D.~Proserpio, and J.~W. Byers (2016).
\newblock The rise of the sharing economy: Estimating the impact of airbnb on
  the hotel industry.
\newblock {\em Forthcoming, Journal of Marketing Research\/}.

\end{thebibliography}

\ECSwitch

%\ECDisclaimer
%%%%%%%%%%%%%%%%%%%%%%%%%%%%%%%%%%%%%%%%%%%%%%%%%%%%%%%%%%

%%% Main head for the e-companion
\ECHead{Appendix}

%A general heading for the whole e-companion should be provided here as in the example above this paragraph.

% \input{main_head}
%
%\input{prem}
\section{Proof of Lemma~\ref{lemma:needed-centrality-result-for-m=2}}\label{sec:proof-of-lemma-m=2}
%\vm{Patrick answering your question about the length of the proof. We worked on it a lot, and this is the best we could come up with. Because there are 2 steps are randomization, it requires some work... }
%\vm{Dawsen: looks like you changed the proof structure a bit. Can you explain what you did, and also how you applied Patrick's comments?}
%\vm{Updates: I define all the constant within the statements. Also, I made the proof more modular by postponing
%the proofs (that are mainly algebra) to Section~\ref{subsec:aux}. Also, I highlighted the proof sketch, and removed some repetition.
%Now the main part of the proof is about 5 pages.}

\noindent The proof of Lemma~\ref{lemma:needed-centrality-result-for-m=2} is based on the following lemma:
\begin{lemma}\label{lemma:low-tail-o1}
\vmn{Define constants $\alpha_{\ref{lemma:low-tail-o1}} \triangleq  5 + \sqrt{6}$, $\bar\epsilon_{\ref{lemma:low-tail-o1}} \triangleq 1/24$, and $k_{\ref{lemma:low-tail-o1}} \triangleq 4$.
If $\epsilon' \in (0, \bar\epsilon_{\ref{lemma:low-tail-o1}}]$ and  $n_1 > \frac{k_{\ref{lemma:low-tail-o1}}}{p^2} \log \left( \frac{1}{\epsilon' } \right)$,  for any
%(but not all)
$\lambda\in \{ 1/n, 2/n, \dots , n/n\}$, we have:
$$
\prob{\left| O_1(\lambda)-\tilde o_1 (\lambda) \right| \geq \alpha_{\ref{lemma:low-tail-o1}} \sqrt{n_1 \log \left( \frac{1}{\epsilon'}\right)}} \leq \epsilon'.
$$}
%There exists positive real numbers $\alpha_{\ref{lemma:low-tail-o1}}$, $\bar\epsilon_{\ref{lemma:low-tail-o1}}$ and $k_{\ref{lemma:low-tail-o1}}$ such that when $0<\epsilon' \leq \bar\epsilon_{\ref{lemma:low-tail-o1}} $ and $n_1 > \frac{k_{\ref{lemma:low-tail-o1}}}{p^2} \log \left( \frac{1}{\epsilon' } \right)$, for any (but not all) $\lambda\in \{ 1/n, 2/n, \dots , n/n\}$, with probability at most $\epsilon'$,
%$$ \left| O_1(\lambda)-\tilde o_1 (\lambda) \right| \geq \alpha_{\ref{lemma:low-tail-o1}} \sqrt{n_1 \log \left( \frac{1}{\epsilon'}\right)}.$$
\end{lemma}

To prove Lemma~\ref{lemma:low-tail-o1}, we use two existing concentration bounds for random variables obtained from sampling with \vmn{and} without replacement.
Before proceeding to the proof, we state these concentration bounds. \vmn{Using the existing concentration bounds, Lemma~\ref{lemma:low-tail-o1} is proven through
a series of auxiliary corollaries (of the concentration results) and lemmas whose proofs are deferred to Section~\ref{subsec:aux}}.

First, we use a well-know variant of the classical Chernoff bound (\citet{chernoff1952measure}) regarding the concentration of binomial random variables as given in \cite{mcdiarmid1998concentration}:
\begin{theorem}[\cite{mcdiarmid1998concentration}]\label{thm:binomial-bound}
Let $0<p<1$, let $X_1, X_2, \dots, X_n$ be independent binary random variables, with $\prob{X_k=1}=p$ and $\prob{X_k=0}=1-p$ for each $k$, and let $S_n=\sum_{k=1}^n X_k$.
Then for any $t\geq 0$,
$$ \prob{|S_n-np| \geq nt} \leq 2e^{-2nt^2}.$$
\end{theorem}

When applying this theorem in our proof, we find it more insightful and convenient to use the following form of the above concentration result:
\begin{corollary} \label{coro:binomial}
\vmn{For any $k \geq 0$, define constants $\alpha_{\ref{coro:binomial},k} \triangleq 1$ and $\bar\epsilon_{\ref{coro:binomial},k} \triangleq  k/2$.
For $\epsilon \in (0, \bar\epsilon_{\ref{coro:binomial},k})$, under the same setting as in Theorem~\ref{thm:binomial-bound}, we have:}
%For any positive real number $k$, there exist positive real numbers $\alpha_{\ref{coro:binomial},k}, \bar\epsilon_{\ref{coro:binomial},k},$ such that, when $0<\epsilon < \bar\epsilon_{\ref{coro:binomial},k}$, under the same setting as in Theorem~\ref{thm:binomial-bound},
$$ \prob{ |S_n-np| \geq \alpha_{\ref{coro:binomial},k} \sqrt{n \log \left( \frac{1}{\epsilon} \right) } } \leq k \epsilon. $$
\end{corollary}

Second, we use a concentration result for random variables drawn from the hypergeometric distribution given by \citet{hush2005concentration}. Recall that hypergeometric distribution is similar to binomial distribution when sampling without replacement is performed. It is defined precisely within the following theorem.
%\citet{chvatal1979tail} gives an analogue of Theorem~\ref{thm:binomial-bound}., which performs better when the initial success probability is small.

\begin{theorem}[\cite{hush2005concentration}]\label{thm:hypergeometric-bound}
Let $K \sim \text{Hyper}(n_1, n, m)$ denote the hypergeometric random variable describing the process of counting how many defectives are selected when $n_1$ items are randomly selected without replacement from a population of $n$ items of which $m$ are defective. Let $\gamma \geq 2$.
Then,
$$\prob{K-\E{K}>\gamma } < e^{-2\alpha_{n_1,n,m}(\gamma^2-1)}$$
and
$$\prob{K-\E{K}<-\gamma } < e^{-2\alpha_{n_1,n,m}(\gamma^2-1)},$$
where
$$\alpha_{n_1,n,m}=\max \left\{\frac{1}{n_1+1}+\frac{1}{n-n_1+1}, \frac{1}{m+1}+\frac{1}{n-m+1} \right\}. $$
\end{theorem}
%\dawsen{According to some comments in~\cite{mcdiarmid1998concentration}, it seems that if we consider $\max\{ |S_k-pk|\}$, then we can get the same error probability, so we might be able to remove the $\log n$ term in our concentration lemmas. However, \cite{hush2005concentration} does not support a similar extension, so if we want to use an analogue to that result, then we must read \cite{hush2005concentration} carefully and see if we can obtain something similar, which takes a lot of time, so I skip for now. This is a point to re-visit though.}
%
Similar to the concentration result for binomial distribution, we find it easier to use the following form of the above concentration result:
\begin{corollary} \label{coro:hyper}
\vmn{For any $k \geq 0$, define constants $\alpha_{\ref{coro:hyper},k} \triangleq 2$ and $\bar\epsilon_{\ref{coro:hyper},k} \triangleq k/2$ , $\underline m_{\ref{coro:hyper},k} \triangleq \max \left\{ \left( \log \frac{1}{\bar \epsilon_{\ref{coro:hyper},k}}\right)^{-1} , 1\right\}$.
For $\epsilon \in (0, \bar\epsilon_{\ref{coro:hyper},k})$ and $m \geq \underline m_{\ref{coro:hyper},k}$, under the same setting as in Theorem~\ref{thm:hypergeometric-bound}, we have:}
%For any positive real number $k$, there exist positive real numbers $\alpha_{\ref{coro:hyper},k}, \bar\epsilon_{\ref{coro:hyper},k}, \underline m_{\ref{coro:hyper},k}$ such that, when $0<\epsilon < \bar\epsilon_{\ref{coro:hyper},k}$ and $m \geq \underline m_{\ref{coro:hyper},k}$, under the same setting as in Theorem~\ref{thm:hypergeometric-bound},
$$ \prob{ |K - \E{K}| \geq \alpha_{\ref{coro:hyper},k} \sqrt{m \log \left( \frac{1}{\epsilon} \right) } } \leq k \epsilon. $$
\end{corollary}

\vmn{\textbf{Proof Sketch of Lemma~\ref{lemma:low-tail-o1}:}}
Before proceeding to the proof of Lemma~\ref{lemma:low-tail-o1}, we explain the idea of the proof by going back to the example of Figure~\ref{figure:example} from Section~\ref{sec:prem}.
Let us consider $\lambda = 5/8$.
In the following, we count the number of customers in the {\RG} group and the {\UPG} group in $O_1(\lambda)$ separately.

We begin by counting the number of \st{class}\vmn{type}-$1$ customers in the {\RG} group that arrive no later than time $5/8$ in $\vec{V}$ in Figure~\ref{figure:example}.
Among the five customers arriving by time $5/8$, two of them are in the {\RG} group: customers at positions $2$ and $5$.
We aim to count the number of \st{class}\vmn{type}-$1$ customers in these two positions.
There are a total of four customers in the {\RG} group (there are four black nodes in the middle row).
Note that only one of them is \st{class}\vmn{type}-$1$ ($\vmn{v_{I,5}} = 1$).
Now we take two samples without replacement from the four customers to fill the two positions ($2$ and $5$).
Thus, given the realization of the {\RG} group (the middle row), the number of \st{class}\vmn{type}-$1$ customers in these two positions follows a hypergeometric distribution with parameters $(2,4,1)$ (which, as defined in Theorem~\ref{thm:hypergeometric-bound}, corresponds to taking two samples without replacement from four customers among which one is \st{class}\vmn{type}-$1$).
In the particular realization of Figure~\ref{figure:example}, the \st{class}\vmn{type}-$1$ customer in the {\RG} group is placed in position $2$.

Now we count the number of \st{class}\vmn{type}-$1$ customers in the {\UPG} group that arrive no later than time $5/8$ in $\vec{V}$ in Figure~\ref{figure:example}.
In the adversarial sequence $\vmn{\vec{v}_I}$, there are three \st{class}\vmn{type}-$1$ customers (at positions $1$, $3$, and $5$).
Any of these three customers will be in the {\UPG} group independent of each other and with probability $(1-p)$, and hence the number of \st{class}\vmn{type}-$1$ in the {\UPG} group that arrive no later than time $5/8$ in $\vec{V}$ follows the binomial distribution $\text{Bin}(3, 1-p)$.
In the particular realization of Figure~\ref{figure:example}, among the three \st{class}\vmn{type}-$1$ customers arriving no later than time $5/8$ in $\vmn{\vec{v}_I}$, two of them are in the {\UPG} group: customers at position $1$ and $3$.
Therefore, the number of \st{class}\vmn{type}-$1$ customers in the {\UPG} group that arrive no later than time $5/8$ in the particular realization $\vec{v}$ is two.

In the proof of Lemma~\ref{lemma:low-tail-o1}, we use the method described in the above example to count the number of customers in $O_1(\lambda)$.
For counting the number of customers in the {\RG} group in $O_1(\lambda)$:
(i) First we count the number of positions before time $\lambda $ that belong to the {\RG} group.
Call this number $Z$.
(ii) Next we count the number of \st{class}\vmn{type}-$1$ customers in the {\RG} group. Call the total number of customers in the {\RG} group $R$ and the number of \st{class}\vmn{type}-$1$ customers in the {\RG} group $R_1$.
(iii) We compute the number of \st{class}\vmn{type}-$1$ customers in the {\RG} group that fill one of these $Z$ positions.
Call this number $Z_1$.
As mentioned above, this is equivalent to taking $Z$ samples without replacement from $R$ customers among which $R_1$ are \st{class}\vmn{type}-$1$.
The number $Z_1$ is the number of customers in the {\RG} group in $O_1(\lambda)$.
Counting the customers in the {\UPG} group in $O_1(\lambda)$ is relatively simple. Call this number $\zeta_1$.
Finally, we obtain $O_1(\lambda)$ with the equation $O_1(\lambda) = Z_1+\zeta_1$.
In summary, the random variables have the following distributions:
\begin{itemize}
\item $R \sim \text{Bin}(n, p)$,
\item $R_1 \sim \text{Bin}(n_1, p)$,
\item $Z \sim \text{Bin}(\lambda n, p)$,
\item $Z_1 \sim  \text{Hyper}(Z, R, R_1)$\footnote{Note that $Z$, $R$, and $R_1$ are not necessarily independent.},
\item $\zeta_1 \sim \text{Bin}(\eta_1(\lambda), 1-p)$.
\end{itemize}

The proof of Lemma~\ref{lemma:low-tail-o1} includes \vmn{establishing concentration results for $Z_1$ and $\zeta_1$ through a series of auxiliary lemmas; Lemmas~\ref{claim:BinomialBounds},~\ref{claim:expectationofZ1}, and~\ref{claim:concentraionZ1} are concerned with the former random variable, and Lemma~\ref{claim:adversary} is concerned with the latter one. The proof of these lemmas is deferred to Section~\ref{subsec:aux}.}

The first lemma focuses on analyzing $R$, $R_1$, and $Z$.
In particular, we use Corollary~\ref{coro:binomial} along with the union bound to show the following:
%\pj{I have not made specific recommendations about the introductions of constant in all the remaining lemmas, but I definitely think that we should only introduce new constants only when needed ... along the lines I have suggested above. Making the proofs as clean as possible along these lines would help any readers interested in checking their validity.}
\begin{lemma}
\label{claim:BinomialBounds}
\vmn{Define constants $\bar\epsilon_{\ref{claim:BinomialBounds}} \triangleq 1/24$ and $\alpha_{\ref{claim:BinomialBounds}} \triangleq 1$. For $ \epsilon \in (0, \bar\epsilon_{\ref{claim:BinomialBounds}}]$,}
%There exist real numbers $\bar\epsilon_{\ref{claim:BinomialBounds}}$ and $\alpha_{\ref{claim:BinomialBounds}}$ such that when $0 < \epsilon \leq \bar\epsilon_{\ref{claim:BinomialBounds}}$,
with probability at least $1 - \epsilon/4$, all the following three events happen:
\begin{subequations}
\begin{align}
R & \in \left( np - \alpha_{\ref{claim:BinomialBounds}} \sqrt{ n \log \left( \frac{1}{\epsilon} \right)}, np + \alpha_{\ref{claim:BinomialBounds}} \sqrt{ n \log \left( \frac{1}{\epsilon} \right)} \right)\label{event:R},\\
R_1& \in \left( n_1p - \alpha_{\ref{claim:BinomialBounds}} \sqrt{ n_1 \log \left( \frac{1}{\epsilon} \right)}, n_1p + \alpha_{\ref{claim:BinomialBounds}} \sqrt{ n_1 \log \left( \frac{1}{\epsilon} \right)} \right)\text{, and }\label{event:R1+bar-R1}\\
Z & \in \left(\lambda np - \alpha_{\ref{claim:BinomialBounds}} \sqrt{ \lambda n \log \left( \frac{1}{\epsilon} \right)}, \lambda np + \alpha_{\ref{claim:BinomialBounds}} \sqrt{ \lambda n \log \left( \frac{1}{\epsilon} \right)} \right). \label{event:R1+R-1}
\end{align}
\end{subequations}
\end{lemma}
%\begin{proof}[\textbf{Proof of Lemma~\ref{claim:BinomialBounds}}]
%Let $k = 1/12$, Corollary~\ref{coro:binomial} implies that there exist $\bar\epsilon_{\ref{coro:binomial},k}$ and $\alpha_{\ref{coro:binomial},k}$ such that when $0<\epsilon \leq \bar\epsilon_{\ref{coro:binomial},k}$,
%\begin{align}
%& 1 - \prob{R \in \left( np - \alpha_{\ref{coro:binomial},k} \sqrt{ n \log \left( \frac{1}{\epsilon} \right)}, np + \alpha_{\ref{coro:binomial},k} \sqrt{ n \log \left( \frac{1}{\epsilon} \right)} \right)} \nonumber \\ & =
%\prob{|R - np| \geq \alpha_{\ref{coro:binomial},k} \sqrt{ n \log \left( \frac{1}{\epsilon} \right)}} \leq \epsilon/12.
%\end{align}
%Defining $\bar\epsilon_{\ref{claim:BinomialBounds}} \triangleq \bar\epsilon_{\ref{coro:binomial},k}$ and $\alpha_{\ref{claim:BinomialBounds}} \triangleq \alpha_{\ref{coro:binomial},k}$, repeating the same for $R_1$ and $Z$, and applying the union bound imply the statement.
%\end{proof}
\noindent We note that conditioned on $R$, $R_1$, and $Z$, the expected value of $Z_1$ is $\frac{ZR_1}{R}$. Thus we have:
\begin{align}
\E{Z_1} = \E{\E{Z_1|R,R_1,Z}} = \E{\frac{ZR_1}{R}}. \label{def:EZ1}
\end{align}
The last expectation is a non-linear function of the three random variables $R$, $R_1$, and $Z$. Instead of computing the expectation directly, we
use the concentration bounds of \eqref{event:R} - \eqref{event:R1+R-1} to show the following lemma:
\begin{lemma}
\label{claim:expectationofZ1}
\vmn{Define constants $\alpha_{\ref{claim:expectationofZ1}} \triangleq 4$, $k_{\ref{claim:expectationofZ1}} \triangleq 4$. }
Conditioned on the events \eqref{event:R}- \eqref{event:R1+R-1},
%there exists positive real numbers $\alpha_{\ref{claim:expectationofZ1}}$, $k_{\ref{claim:expectationofZ1}}$ such that
and when $n_1 > \frac{k_{\ref{claim:expectationofZ1}}}{p^2}\log \left(\frac{1}{\epsilon}\right)$, we have:
\begin{align}
\frac{Z R_1 }{R} \in \left( \lambda pn_1 - \alpha_{\ref{claim:expectationofZ1}} \sqrt{n_1\log\left( \frac{1}{\epsilon} \right)}, \lambda pn_1 + \alpha_{\ref{claim:expectationofZ1}} \sqrt{n_1\log\left( \frac{1}{\epsilon} \right) }\right).\label{inequality:fractional-thing-to-prove}
\end{align}
\end{lemma}
\noindent Lemmas~\ref{claim:BinomialBounds} and~\ref{claim:expectationofZ1} together imply that, \vmn{for $ \epsilon \in (0, \bar\epsilon_{\ref{claim:BinomialBounds}}]$}:
\begin{align}
\prob{ \left| \frac{R_1 Z}{R} - pn_1 \lambda\right| \geq \alpha_{\ref{claim:expectationofZ1}} \sqrt{n_1\log\left( \frac{1}{\epsilon} \right)} } \leq \frac{\epsilon}{4}
\label{ineq:expectationError}
\end{align}

Having \eqref{def:EZ1} and Lemma~\ref{claim:expectationofZ1}, we are ready to establish a concentration result for $Z_1$.
We partition the sample space of ($R, R_1, Z$) into two events as follows: the event where \eqref{event:R}-\eqref{event:R1+R-1} hold, denoted by $\mathcal{E}$; the complement event, denoted by $\mathcal{E}^c$.\footnote{We note that this event is only locally defined within this appendix, and it is not the same as the one defined in Definition~\ref{def:event}.}
Note that Lemma~\ref{claim:BinomialBounds} implies that $\prob{\mathcal{E}^c} \leq \frac{\epsilon}{4}$. Using the law of total probability, we have: for any $\tilde{ \alpha} > 0$,

\begin{align}
& \prob{ \left| Z_1 - \E{Z_1 | R,R_1,Z} \right| \geq \tilde{\alpha} \sqrt{n_1\log\left( \frac{1}{\epsilon} \right)} }
\nonumber \\ & = \prob{\mathcal{E}^c}\prob{ \left| Z_1 - \E{Z_1 | R,R_1,Z} \right| \geq \tilde{\alpha} \sqrt{n_1\log\left( \frac{1}{\epsilon} \right)} \Big|~\mathcal{E}^c}
\nonumber \\ & + \prob{\mathcal{E}} \prob{ \left| Z_1 - \E{Z_1 | R,R_1,Z} \right| \geq \tilde{\alpha} \sqrt{n_1\log\left( \frac{1}{\epsilon} \right)} \Big|~\mathcal{E}}
\nonumber \\ & \leq \frac{\epsilon}{4} \cdot 1 +1\cdot \prob{ \left| Z_1 - \E{Z_1 | R,R_1,Z} \right| \geq \tilde{\alpha} \sqrt{n_1\log\left( \frac{1}{\epsilon} \right)} \Big|~\mathcal{E}}.
\label{ineq:totalProb}
\end{align}

\noindent Using Corollary~\ref{coro:hyper} and the definition of events \eqref{event:R}-\eqref{event:R1+R-1}, we show the following lemma:
\begin{lemma}
\label{claim:concentraionZ1}
\vmn{Define constants $\bar\epsilon_{\ref{claim:concentraionZ1}} \triangleq 1/24$, $\alpha_{\ref{claim:concentraionZ1}} \triangleq \sqrt{6}$, and $k_{\ref{claim:concentraionZ1}} \triangleq 4$. For $\epsilon \in (0, \bar\epsilon_{\ref{claim:concentraionZ1}}]$, }
%There exists positive real numbers $\bar\epsilon_{\ref{claim:concentraionZ1}}$, $\alpha_{\ref{claim:concentraionZ1}}$ and $k_{\ref{claim:concentraionZ1}}$ such that, for all $0< \epsilon \leq \bar\epsilon_{\ref{claim:concentraionZ1}}$
if $n_1 > \frac{k_{\ref{claim:concentraionZ1}}}{p^2}\log \left(\frac{1}{\epsilon}\right)$, and  $(R,R_1,Z) \in \mathcal{E}$, we have:
\begin{align}
\prob{ \left| Z_1 - \E{Z_1 | R,R_1,Z} \right| \geq \alpha_{\ref{claim:concentraionZ1}} \sqrt{n_1\log\left( \frac{1}{\epsilon} \right)} \Big|~R,R_1,Z} \leq \frac{\epsilon}{4}.
\end{align}
\end{lemma}

\noindent Putting Lemma~\ref{claim:concentraionZ1} back to \eqref{ineq:totalProb} and setting $\tilde{\alpha} = \alpha_{\ref{claim:concentraionZ1}} $, we get:
\begin{align}
\prob{ \left| Z_1 - \E{Z_1 | R,R_1,Z} \right| \geq \alpha_{\ref{claim:concentraionZ1}} \sqrt{n_1\log\left( \frac{1}{\epsilon} \right)} }\leq \frac{\epsilon}{2}
\label{ineq:totalProb2}
\end{align}
Finally, we have the following lemma regarding $\zeta_1$:
\begin{lemma}
\label{claim:adversary}
\vmn{Define constants $\bar\epsilon_{\ref{claim:adversary}} \triangleq 1/8$ and $\alpha_{\ref{claim:adversary}} \triangleq 1$. For
$\epsilon \in (0, \bar\epsilon_{\ref{claim:adversary}}]$, we have:}
%There exist real numbers $\bar\epsilon_{\ref{claim:adversary}}$ and $\alpha_{\ref{claim:adversary}}$ such that when $\epsilon \leq \bar\epsilon_{\ref{claim:adversary}}$,
\begin{align}
\prob{ \left| \zeta_1 - (1-p) \eta_1(\lambda)\right| \geq \alpha_{\ref{claim:adversary}} \sqrt{n_1\log\left( \frac{1}{\epsilon} \right)} } \leq \frac{\epsilon}{4}
\label{ineq:adversary}
\end{align}
\end{lemma}
%\begin{proof}[\textbf{Proof of Lemma~\ref{claim:adversary}}]
%Recall that $\zeta_1$ follows the binomial distribution $\text{Bin}(\eta_1(\lambda), 1-p)$.
%Hence, the lemma follows straightforwardly from Corollary~\ref{coro:binomial}.
%\end{proof}
\noindent{With the lemmas above, we are ready to prove Lemma~\ref{lemma:low-tail-o1}: }

%, which we reiterate as follows:
%\begin{replemma}{lemma:low-tail-o1}
%There exists positive real numbers $\alpha_{\ref{lemma:low-tail-o1}}$, $\bar\epsilon_{\ref{lemma:low-tail-o1}}$ and $k_{\ref{lemma:low-tail-o1}}$ such that when $0<\epsilon' \leq \bar\epsilon_{\ref{lemma:low-tail-o1}} $ and $n_1 > \frac{k_{\ref{lemma:low-tail-o1}}}{p^2} \log \left( \frac{1}{\epsilon' } \right)$, for any (but not all) $\lambda\in \{ 1/n, 2/n, \dots , n/n\}$, with probability at most $\epsilon'$,
%$$ \left| O_1(\lambda)-\tilde o_1 (\lambda) \right| \geq \alpha_{\ref{lemma:low-tail-o1}} \sqrt{n_1 \log \left( \frac{1}{\epsilon'}\right)}.$$
%\end{replemma}
\begin{proof}{\textbf{Proof of Lemma~\ref{lemma:low-tail-o1}:}}
%For simplicity, we denote $\epsilon = \epsilon'$.
First, \vmn{we  note that we set the constants in Lemma~\ref{lemma:low-tail-o1} and Lemmas~\ref{claim:BinomialBounds}-\ref{claim:adversary} such that we get}  $\bar\epsilon_{\ref{lemma:low-tail-o1}} = \min\{\bar\epsilon_{\ref{claim:BinomialBounds}}, \bar\epsilon_{\ref{claim:concentraionZ1}}, \bar\epsilon_{\ref{claim:adversary}}\}$, $ k_{\ref{lemma:low-tail-o1}} = \max \{ k_{\ref{claim:concentraionZ1}},k_{\ref{claim:expectationofZ1}}\}$ and $\alpha_{\ref{lemma:low-tail-o1}} = \alpha_{\ref{claim:expectationofZ1}}+ \alpha_{\ref{claim:concentraionZ1}}+\alpha_{\ref{claim:adversary}}$.
Now we can apply the union bound on \eqref{ineq:totalProb2}, \eqref{ineq:expectationError}, and \eqref{ineq:adversary} and obtain: when $0<\epsilon' \leq \bar\epsilon_{\ref{lemma:low-tail-o1}} $ and $n_1 > \frac{k_{\ref{lemma:low-tail-o1}}}{p^2} \log \left( \frac{1}{\epsilon'} \right)$, with probability at least $1-\epsilon'$,
\begin{align*}
& \left| Z_1 - \E{Z_1 | R,R_1,Z} \right| < \alpha_{\ref{claim:concentraionZ1}} \sqrt{n_1\log\left( \frac{1}{\epsilon'} \right)}, \\
& \left| \frac{R_1 Z}{R} - pn_1 \lambda\right| < \alpha_{\ref{claim:expectationofZ1}} \sqrt{n_1\log\left( \frac{1}{\epsilon'} \right)}\text{, and} \\
& \left| \zeta_1 - (1-p) \eta_1(\lambda)\right| < \alpha_{\ref{claim:adversary}} \sqrt{n_1\log\left( \frac{1}{\epsilon'} \right)}. \\
\end{align*}
% It is easy to show that when the three inequalities hold, $ \left| o_1(\lambda) - \tilde o_1 (\lambda) \right| < \alpha_{\ref{lemma:low-tail-o1}} \sqrt{n_1\log\left( \frac{1}{\epsilon} \right)}$, which completes the proof:
We have $O_1(\lambda) = Z_1 +\zeta_1$, and thus, by using the triangular inequality (note that $\E{Z_1 | R,R_1,Z} = \frac{R_1 Z}{R}$),
$$ \left| O_1(\lambda) - \tilde o_1 (\lambda) \right| \leq \left| Z_1 - \E{Z_1 | R,R_1,Z} \right|+ \left| \frac{R_1 Z}{R} - pn_1 \lambda\right| + \left| \zeta_1 - (1-p) \eta_1(\lambda) \right|, $$
which, according to the three inequalities above and the definition $\alpha_{\ref{lemma:low-tail-o1}} = \alpha_{\ref{claim:expectationofZ1}}+ \alpha_{\ref{claim:concentraionZ1}}+\alpha_{\ref{claim:adversary}}$, is smaller than $\alpha_{\ref{lemma:low-tail-o1}} \sqrt{n_1\log\left( \frac{1}{\epsilon'} \right)}$.
\end{proof}

\noindent With Lemma~\ref{lemma:low-tail-o1}, we are ready to prove Lemma~\ref{lemma:needed-centrality-result-for-m=2}:

%, which we reiterate as follows:
%\begin{replemma}{lemma:needed-centrality-result-for-m=2}
%There exist positive real numbers $\alpha$, $\bar\epsilon$ and $k$ such that when $\frac{1}{n} \leq \epsilon \leq \bar\epsilon $, with probability at least $1 - \epsilon$, all the following hold:
%\begin{itemize}
%\item If $n_1 \geq \frac{k}{p^2} \log n$, then for all $\lambda \in \{0,1/n, 2/n, \dots, n/n\}$,
%\begin{subequations}
%\begin{align} & \left| O_1(\lambda)-\tilde o_1 (\lambda) \right| < \alpha \sqrt{n_1 \log n} \text{, and } \label{inequality:good-approximation-o_1-app}
%\\ & \left| O_1(\lambda)+O_2(\lambda) - (\tilde o_1 (\lambda) +\tilde o_2(\lambda))\right| < \alpha \sqrt{ (n_1+n_2) \log n} \label{inequality:good-approximation-o_1+o_2-app}
%\end{align}
%\end{subequations}
%\item If $n_2 \geq \frac{k}{p^2} \log n$, then for all $\lambda \in \{ 0,1/n, 2/n, \dots, n/n\}$,
%\begin{subequations}
%\begin{align} & \left| O_2(\lambda)-\tilde o_2 (\lambda) \right| < \alpha \sqrt{n_2 \log n} \text{, and } \label{inequality:good-approximation-o_2-app}
%\\ & \left| O^\RGS_2(\lambda)-\tilde o^\RGS_2 (\lambda) \right| < \alpha \sqrt{n_2 \log n}, \label{inequality:good-approximation-app--o^R_2-app}
%\end{align}
%\end{subequations}
%\end{itemize}
%
%\end{replemma}
\begin{proof}{\textbf{Proof of Lemma~\ref{lemma:needed-centrality-result-for-m=2}:}}
The proof consists of two steps:
First, we note that the \vmn{concentration} results similar to Lemma~\ref{lemma:low-tail-o1} can be \st{applied to}\vmn{obtained for} the other three random variables $O_1(\lambda)+ O_2(\lambda)$, $O_2(\lambda)$, and $O_2^\RGS(\lambda)$.
When applied to $O_2^\RGS(\lambda)$, the only modification is that we do not need to consider Lemma~\ref{claim:adversary} and~\eqref{ineq:adversary}.

Second, we apply the union bound on the probability that at least one of the $4n$ events in \eqref{inequality:good-approximation-o_1}-\eqref{inequality:good-approximation-app--o^R_2} is violated (note that $\lambda$ takes $n$ different non-zero values: when $\lambda=0$, $|O_1(\lambda) - \tilde o_1(\lambda)|=|0-0| =0$ and the same holds for the other three random variables), and choose the appropriate constants.

To prove this lemma, \vmn{we first note that we set the constants in Lemmas~\ref{lemma:needed-centrality-result-for-m=2} and~\ref{lemma:low-tail-o1} such that we get} $\alpha = 2\alpha_{\ref{lemma:low-tail-o1}}$, $\bar\epsilon = \bar\epsilon_{\ref{lemma:low-tail-o1}}$, and $k = 4 k_{\ref{lemma:low-tail-o1}}$.

Using Lemma~\ref{lemma:low-tail-o1} with $\epsilon' = \frac{\epsilon}{4n}$, the lemma holds because all of the following three statements are true:
\begin{subequations}
\begin{align}
\text{If }\epsilon \leq \bar\epsilon & \text{, then } \epsilon ' \leq \bar\epsilon_{\ref{lemma:low-tail-o1}}. \label{cond-epsilon-bar}\\
\text{If }n_1 > \frac{k}{p^2}\log n & \text{, then } n_1 > \frac{k_{\ref{lemma:low-tail-o1}}}{p^2}\log \left( \frac{1}{\epsilon'}\right). \label{cond-k}\\
\text{If } \left| O_1(\lambda)-\tilde o_1 (\lambda) \right| \geq \alpha \sqrt{n_1 \log n} &\text{, then } \left| O_1(\lambda)-\tilde o_1 (\lambda) \right| \geq \alpha_{\ref{lemma:low-tail-o1}} \sqrt{n_1 \log \left( \frac{1}{\epsilon'}\right)}.\label{cond-alpha}
\end{align}
\end{subequations}
Because $\epsilon'<\epsilon$, and $\bar\epsilon = \bar\epsilon_{\ref{lemma:low-tail-o1}}$, \eqref{cond-epsilon-bar} holds.
Before proceeding to the other two conditions, we first note that
$\log \left( \frac{1}{\epsilon'}\right) = \log \left( \frac{4n}{\epsilon}\right) \leq \log \left( \frac{n^3}{\epsilon}\right) \leq \log \left( n^4\right) = 4 \log n$, where we use $n\geq 2$ in the first inequality and $\epsilon \geq \frac{1}{n}$ in the second inequality.
Therefore, \eqref{cond-k} and~\eqref{cond-alpha} hold because $k = 4 k_{\ref{lemma:low-tail-o1}}$, and $\alpha = 2\alpha_{\ref{lemma:low-tail-o1}}$.
\end{proof}

\subsection{\vmn{Further Remak on Deterministic Approximations}}
\label{subsec:remark}

\begin{remark}\label{remark:deterministic-approx-vs-expected-value}

In Lemma~\ref{lemma:needed-centrality-result-for-m=2}, we use the deterministic value $\tilde o_j(\lambda)$ rather than $\E{O_j(\lambda)}$ to estimate $O_j(\lambda)$ because $\tilde o_j(\lambda)$ is a very simple function of $n_j$ and $\eta_j(\lambda)$.
Here we provide an example to show that $\tilde o_j(\lambda)$ and $\E{O_j(\lambda)}$ are not necessarily the same, which explains why we do not write the term $\tilde o_j(\lambda)$ as $\E{O_j(\lambda)}$.

Let us consider an example where $n=2$ and $\vmn{\vec{v}_I}=(v_{I,1}, v_{I,2})=(1,0)$ at $\lambda =1/2$.
First we compute $\E{O_1(1/2)}$.
Because $O_1(1/2)$ consists of only one customer,
$$\E{O_1(1/2)}=\prob{V_1=1}.$$
We can then use the law of total probability to express the probability as
\begin{align*}
\prob{V_1=1} = & \prob{V_1=1|1\notin \RGS}\prob{1\notin \RGS} \\ + & \prob{V_1=1|1\in \RGS,2\notin \RGS}\prob{1\in \RGS,2\notin \RGS} \\ + & \prob{V_1=1| 1, 2\in \RGS}\prob{1,2 \in \RGS}.
\end{align*}
Following the definitions,
\begin{align*}
\E{O_1(1/2)} = \prob{V_1=1} = 1 \cdot (1-p) + 1 \cdot p(1-p) + & \frac{1}{2} \cdot p^2 = 1-\frac{p^2}{2}.
\end{align*}
On the other hand,
$$\tilde o_1(1/2) = (1-p)\eta_1(\frac{1}{2})+p\frac{1}{2}n_1= 1-\frac{p}{2}.$$
Therefore, for all $p\in(0,1)$, $\E{O_1(1/2)} \neq \tilde o_1(1/2)$.
\end{remark}
%\end{document}
%
%

\subsection{\vmn{Proof of Auxiliary Corollaries and Lemmas}}
\label{subsec:aux}

\begin{proof}{\textbf{Proof of Corollary~\ref{coro:binomial}:}}
In order to apply Theorem~\ref{thm:binomial-bound}, we define $t$ such that
$2e^{-2nt^2} \leq k\epsilon$, which corresponds to
\begin{align*}
t \geq \sqrt{\frac{1}{2n} \log \frac{2}{ k\epsilon}}.
\end{align*}
By setting $t$ to be $ \sqrt{\frac{1}{2n} \log \frac{2}{ k\epsilon}}$, what is remaining to prove is that we can find $\alpha_{\ref{coro:binomial},k}, \bar \epsilon_{\ref{coro:binomial},k} $ such that when $0<\epsilon<\bar\epsilon_{\ref{coro:binomial},k}$,
\begin{align*}
n \sqrt{\frac{1}{2n} \log \frac{2}{ k\epsilon}} \leq \alpha_{\ref{coro:binomial},k} \sqrt{n \log \left( \frac{1}{\epsilon} \right) }.
\end{align*}
This can be achieved by setting $\bar \epsilon_{\ref{coro:binomial},k} =k/2 $ and $\alpha_{\ref{coro:binomial},k} =1$:

\begin{align*}
\epsilon \leq k/2 \quad \Rightarrow \quad \frac{2}{k \epsilon} \leq \frac{1}{\epsilon^2} \quad \Rightarrow \quad \log \frac{2}{ k\epsilon} \leq 2 \log \frac{1}{\epsilon} \quad \Rightarrow \quad n \sqrt{\frac{1}{2n} \log \frac{2}{ k\epsilon}} \leq \alpha_{\ref{coro:binomial},k} \sqrt{n \log \left( \frac{1}{\epsilon} \right) }.
\end{align*}
\end{proof}

\begin{proof}{\textbf{Proof of Corollary~\ref{coro:hyper}:}}
According to Theorem~\ref{thm:hypergeometric-bound}, when $\gamma \geq 2$,
\begin{align*}
\prob{\left| K-\E{K} \right| > \gamma } < 2 e^{-2\alpha_{n_1,n,m}(\gamma^2-1)}.
\end{align*}
We first find an upper bound of the above right-hand-side probability when $\gamma$ and $m$ are large enough.
When $m\geq 1$,
\begin{align*}
\alpha_{n_1,n,m}=\max \left\{\frac{1}{n_1+1}+\frac{1}{n-n_1+1}, \frac{1}{m+1}+\frac{1}{n-m+1} \right\} \geq \frac{1}{m+1} \geq \frac{1}{2m}.
\end{align*}
Further, when $\gamma \geq 2$, we have: $\gamma^2 -1 \geq \gamma ^2 /2$. Putting these two together,

\begin{align*}
2 e^{-2\alpha_{n_1,n,m}(\gamma^2-1)} \leq 2 e^{-\frac{1}{m} \frac{\gamma^2}{2}}.
\end{align*}
Therefore, if $\alpha_{\ref{coro:hyper},k} \sqrt{m \log \left( \frac{1}{\epsilon} \right) } \geq 2$ and $m\geq 1$,
\begin{align*}
\prob{ |K - \E{K}| \geq \alpha_{\ref{coro:hyper},k} \sqrt{m \log \left( \frac{1}{\epsilon} \right) } } < 2 \exp \left( - \frac{1}{m} \frac{\alpha_{\ref{coro:hyper},k}^2 m \log \left( \frac{1}{\epsilon} \right)}{2} \right)
= 2 \epsilon^{\alpha_{\ref{coro:hyper},k}^2/2}.
\end{align*}
Thus, it is sufficient to have $\alpha_{\ref{coro:hyper},k} \sqrt{m \log \left( \frac{1}{\epsilon} \right) } \geq 2$, $m\geq 1$, and
\begin{align*}
2\epsilon^{\alpha_{\ref{coro:hyper},k}^2/2} \leq k\epsilon.
\end{align*}
The last condition holds by setting $\alpha_{\ref{coro:hyper},k} = 2$ and $\bar \epsilon_{\ref{coro:hyper},k} = k/2$ ($ \epsilon \leq \bar \epsilon_{\ref{coro:hyper},k} = k/2 \Rightarrow \epsilon^2 \leq k \epsilon/2$).
The first two conditions hold by defining $\underline m_{\ref{coro:hyper},k} \triangleq \max \left\{ \left( \log \frac{1}{\bar \epsilon_{\ref{coro:hyper},k}}\right)^{-1} , 1\right\}$.

\end{proof}

\begin{proof}{\textbf{Proof of Lemma~\ref{claim:BinomialBounds}:}}
Let $k = 1/12$, Corollary~\ref{coro:binomial} implies that there exist $\bar\epsilon_{\ref{coro:binomial},k}$ and $\alpha_{\ref{coro:binomial},k}$ such that when $0<\epsilon \leq \bar\epsilon_{\ref{coro:binomial},k}$,
\begin{align}
& 1 - \prob{R \in \left( np - \alpha_{\ref{coro:binomial},k} \sqrt{ n \log \left( \frac{1}{\epsilon} \right)}, np + \alpha_{\ref{coro:binomial},k} \sqrt{ n \log \left( \frac{1}{\epsilon} \right)} \right)} \nonumber \\ & =
\prob{|R - np| \geq \alpha_{\ref{coro:binomial},k} \sqrt{ n \log \left( \frac{1}{\epsilon} \right)}} \leq \epsilon/12.
\end{align}
Defining $\bar\epsilon_{\ref{claim:BinomialBounds}} \triangleq \bar\epsilon_{\ref{coro:binomial},k}$ and $\alpha_{\ref{claim:BinomialBounds}} \triangleq \alpha_{\ref{coro:binomial},k}$, repeating the same for $R_1$ and $Z$, and applying the union bound imply the statement.
\end{proof}

\begin{proof}{\textbf{Proof of Lemma~\ref{claim:expectationofZ1}:}}
First we define the constant $\alpha_{\ref{claim:expectationofZ1}}\triangleq 3 \alpha_{\ref{claim:BinomialBounds}} + 2\alpha_{\ref{claim:BinomialBounds}}^2 \sqrt{1/k_{\ref{claim:expectationofZ1}}}$ where $k_{\ref{claim:expectationofZ1}} \triangleq 4 \alpha_{\ref{claim:BinomialBounds}}^2$.
The reason for definition of the constants becomes clear in the process of the proof.
We prove the lower bound first.
Because $0 < \frac{Z}{R} \leq 1$, the ratio does not increase by subtracting the same positive number from both the denominator and the numerator if the denominator remains positive after the subtraction. In particular, we subtract $\alpha_{\ref{claim:BinomialBounds}} \sqrt{n \log \left( \frac{1}{\epsilon} \right)}$ from both the denominator and the numerator,
Therefore,
\begin{align}
\frac{Z}{R} > \frac{Z - \alpha_{\ref{claim:BinomialBounds}} \sqrt{n \log \left( \frac{1}{\epsilon} \right)} }{R - \alpha_{\ref{claim:BinomialBounds}} \sqrt{n \log \left( \frac{1}{\epsilon} \right)}}. \label{ineq:ratio}
\end{align}
Note that $R - \alpha_{\ref{claim:BinomialBounds}} \sqrt{n \log \left( \frac{1}{\epsilon} \right)} > 0$, because under event \eqref{event:R}, we have:

\begin{align*}
& R - \alpha_{\ref{claim:BinomialBounds}} \sqrt{n \log \left( \frac{1}{\epsilon} \right)} \geq np - 2 \alpha_{\ref{claim:BinomialBounds}} \sqrt{n \log \left( \frac{1}{\epsilon} \right)}. \\
\end{align*}
Therefore,
\begin{align}
n > \frac{4 \alpha_{\ref{claim:BinomialBounds}}^2}{p^2} \log \left( \frac{1}{\epsilon} \right) \Rightarrow R - \alpha_{\ref{claim:BinomialBounds}} \sqrt{n \log \left( \frac{1}{\epsilon} \right)} > 0 \label{ineq:conditionN1}.
\end{align}
The first inequality in \eqref{ineq:conditionN1} holds because $n \geq n_1$, and, by assumption in the lemma, $n_1 > \frac{k_{\ref{claim:expectationofZ1}}}{p^2} \log \left( \frac{1}{\epsilon} \right) = \frac{4\alpha_{\ref{claim:BinomialBounds}}^2}{p^2} \log \left( \frac{1}{\epsilon} \right) $ where we use the fact that we defined $k_{\ref{claim:expectationofZ1}} = 4 \alpha_{\ref{claim:BinomialBounds}}^2$.

Going back to \eqref{ineq:ratio}, under the events \eqref{event:R} and \eqref{event:R1+R-1}, we have:

\begin{align}
\frac{Z}{R} & > \frac{Z - \alpha_{\ref{claim:BinomialBounds}} \sqrt{n \log \left( \frac{1}{\epsilon} \right)} }{R - \alpha_{\ref{claim:BinomialBounds}} \sqrt{n \log \left( \frac{1}{\epsilon} \right)}} \nonumber \\
& > \frac{\lambda n p - \alpha_{\ref{claim:BinomialBounds}} \sqrt{\lambda n \log \left( \frac{1}{\epsilon} \right) } - \alpha_{\ref{claim:BinomialBounds}} \sqrt{n \log \left( \frac{1}{\epsilon} \right)} }{np + \alpha_{\ref{claim:BinomialBounds}} \sqrt{n \log \left( \frac{1}{\epsilon} \right)} - \alpha_{\ref{claim:BinomialBounds}} \sqrt{n \log \left( \frac{1}{\epsilon} \right)}}
\geq \lambda - \frac{2\alpha_{\ref{claim:BinomialBounds}} \sqrt{n \log \left( \frac{1}{\epsilon} \right)}}{np}. \label{ineq:ZtoR}
\end{align}

Combining \eqref{ineq:ZtoR} and with the lower-bound on $R_1$ under event \eqref{event:R1+bar-R1}, we get:

\begin{align*}
& \frac{Z R_1}{R} > \left( \lambda - \frac{2\alpha_{\ref{claim:BinomialBounds}}}{p} \sqrt{\frac{ \log \left( \frac{1}{\epsilon} \right)}{n}} \right)
\left( n_1p - \alpha_{\ref{claim:BinomialBounds}} \sqrt{ n_1 \log \left( \frac{1}{\epsilon} \right)} \right)
\\
= &\lambda n_1 p -2\alpha_{\ref{claim:BinomialBounds}} n_1 \sqrt{\frac{ \log \left( \frac{1}{\epsilon} \right)}{n}} - \lambda \alpha_{\ref{claim:BinomialBounds}} \sqrt{ n_1 \log \left( \frac{1}{\epsilon} \right)}
+ 2\frac{\alpha_{\ref{claim:BinomialBounds}}^2}{p} \sqrt{\frac{n_1}{n}} \log \left( \frac{1}{\epsilon} \right).
\end{align*}

Because $n_1\leq n$, we have $n_1 \sqrt{\frac{1}{n}} \leq \sqrt{n_1}$.
By definition, $\alpha_{\ref{claim:expectationofZ1}} \geq 3 \alpha_{\ref{claim:BinomialBounds}} \geq (2+\lambda ) \alpha_{\ref{claim:BinomialBounds}}$, and therefore, the right hand side of the above inequality is at least
$$\lambda n_1 p - \alpha_{\ref{claim:expectationofZ1}} \sqrt{ n_1 \log \left( \frac{1}{\epsilon} \right)}. $$

Hence we complete the proof for the lower bound part of Inequality~(\ref{inequality:fractional-thing-to-prove}).
When it comes to the upper bound, we can use the same argument as the lower bound and obtain,
\begin{align}
& \frac{Z R_1 }{R} < \left( \lambda + \frac{2\alpha_{\ref{claim:BinomialBounds}}}{p} \sqrt{\frac{ \log \left( \frac{1}{\epsilon} \right)}{n}} \right)
\left( n_1p + \alpha_{\ref{claim:BinomialBounds}} \sqrt{ n_1 \log \left( \frac{1}{\epsilon} \right)} \right)
\nonumber \\
= &\lambda n_1 p +2\alpha_{\ref{claim:BinomialBounds}} n_1 \sqrt{\frac{ \log \left( \frac{1}{\epsilon} \right)}{n}} + \lambda \alpha_{\ref{claim:BinomialBounds}} \sqrt{ n_1 \log \left( \frac{1}{\epsilon} \right)}
+ 2\frac{\alpha_{\ref{claim:BinomialBounds}}^2}{p} \sqrt{\frac{n_1}{n}} \log \left( \frac{1}{\epsilon} \right) \nonumber\\
\leq & \lambda n_1 p +2\alpha_{\ref{claim:BinomialBounds}} \sqrt{n_1 \log \left( \frac{1}{\epsilon} \right)} + \alpha_{\ref{claim:BinomialBounds}} \sqrt{ n_1 \log \left( \frac{1}{\epsilon} \right)}
+ 2\frac{\alpha_{\ref{claim:BinomialBounds}}^2}{p} \sqrt{n_1 \log \left( \frac{1}{\epsilon} \right)} \sqrt{\frac{ \log \left( \frac{1}{\epsilon} \right)}{n}} \nonumber \\
= & \lambda n_1 p + \left(3 \alpha_{\ref{claim:BinomialBounds}} + 2\alpha_{\ref{claim:BinomialBounds}}^2 \frac{1}{p}\sqrt{\frac{\log\left(\frac{1}{\epsilon}\right)}{n}} \right) \sqrt{n_1 \log\left(\frac{1}{\epsilon}\right)}, \label{ineq:lower_bound}
\end{align}
where the inequality in the third line comes from the fact that $\sqrt{\frac{n_1}{n}} \leq 1$ and $\lambda \leq 1$.
Combining the definition $\alpha_{\ref{claim:expectationofZ1}}= 3 \alpha_{\ref{claim:BinomialBounds}} + 2\alpha_{\ref{claim:BinomialBounds}}^2 \sqrt{1/k_{\ref{claim:expectationofZ1}}}$ and~\eqref{ineq:lower_bound}, it is sufficient to have
$$ \frac{1}{p}\sqrt{\frac{\log\left(\frac{1}{\epsilon}\right)}{n}}$$
upper bounded by the constant $\sqrt{1/k_{\ref{claim:expectationofZ1}}}$.
By the assumption of this lemma, $n_1 \geq \frac{k_{\ref{claim:expectationofZ1}}}{p^2} \log \left( \frac{1}{\epsilon} \right)$, and thus
$ \frac{1}{p}\sqrt{\frac{\log\left(\frac{1}{\epsilon}\right)}{n_1}} \leq \sqrt{1/k_{\ref{claim:expectationofZ1}}}$.
Further, because $n\geq n_1$, we have: $ \frac{1}{p}\sqrt{\frac{\log\left(\frac{1}{\epsilon}\right)}{n}} \leq \frac{1}{p}\sqrt{\frac{\log\left(\frac{1}{\epsilon}\right)}{n_1}} \leq \sqrt{1/k_{\ref{claim:expectationofZ1}}}$.
Thus we have:
$$ 3 \alpha_{\ref{claim:BinomialBounds}} + 2\alpha_{\ref{claim:BinomialBounds}}^2 \frac{1}{p}\sqrt{\frac{\log\left(\frac{1}{\epsilon}\right)}{n}} \leq 3\alpha_{\ref{claim:BinomialBounds}}+ 2\alpha_{\ref{claim:BinomialBounds}}^2 \sqrt{1/k_{\ref{claim:expectationofZ1}}} \leq \alpha_{\ref{claim:expectationofZ1}}.$$
Putting this back in \eqref{ineq:lower_bound} completes the proof of the lemma.
\end{proof}

\begin{proof}{\textbf{Proof of Lemma~\ref{claim:concentraionZ1}:}}
Applying Corollary~\ref{coro:hyper} with $k=1/4$, there exist positive real numbers $\alpha_{\ref{claim:concentraionZ1}}'\triangleq \alpha_{\ref{coro:hyper},k}$, $\bar\epsilon_{\ref{coro:hyper},k}$, and $\underline m \triangleq \underline m_{\ref{coro:hyper},k}$ such that when $0< \epsilon \leq \bar\epsilon_{\ref{coro:hyper},k}$ and $R_1 \geq \underline m$,
\begin{align*}
\prob{ \left| Z_1 - \E{Z_1 | R,R_1,Z} \right| \geq \alpha_{\ref{claim:concentraionZ1}}'\sqrt{ R_1 \log \left( \frac{1}{\epsilon} \right) }} \leq \frac{ \epsilon}{4}.
\end{align*}
Now, we define $\bar\epsilon_{\ref{claim:concentraionZ1}} \triangleq \min \{ \bar\epsilon_{\ref{coro:hyper},k} , \bar\epsilon_{\ref{claim:BinomialBounds}} \}$, $k_{\ref{claim:concentraionZ1}}\triangleq \max{\left\{ 4 \alpha_{\ref{claim:BinomialBounds}}^2, \frac{2 \underline{m}}{\log \left( \frac{1}{\bar\epsilon_{\ref{claim:concentraionZ1}}}\right)}\right\}}$ and $\alpha_{\ref{claim:concentraionZ1}} \triangleq \alpha_{\ref{claim:concentraionZ1}}' \sqrt{\frac{3}{2}}$.

First we check the condition $R_1 \geq \underline m$:
Because $n_1 > \frac{k_{\ref{claim:concentraionZ1}}}{p^2}\log\left( \frac{1}{\epsilon} \right) \geq \frac{4 \alpha_{\ref{claim:BinomialBounds}}^2}{p^2}\log\left( \frac{1}{\epsilon} \right) $,
\begin{align}
\alpha_{\ref{claim:BinomialBounds}} \sqrt{n_1 \log \left( \frac{1}{\epsilon} \right) } \leq \frac{n_1p}{2}. \label{ineq:n_1p-big-enough}
\end{align}
Therefore, under event \eqref{event:R1+bar-R1}, we have:
$$ R_1 > n_1p - \alpha_{\ref{claim:concentraionZ1}}' \sqrt{n_1 \log \left( \frac{1}{\epsilon} \right) } \geq \frac{1}{2} n_1p \geq \underline m,$$
where the first inequality follows the definition of event \eqref{event:R1+bar-R1}, second one uses Inequality~\eqref{ineq:n_1p-big-enough}, and the last one uses $k_{\ref{claim:concentraionZ1}} \geq \frac{2 \underline{m}}{\log \left( \frac{1}{\bar\epsilon}\right)}$ (which gives $n_1 \geq \frac{k_{\ref{claim:concentraionZ1}}}{p^2}\log \left( \frac{1}{\epsilon}\right) \geq \frac{2 \underline {m}}{p^2}\geq\frac{2 \underline {m}}{p}$).

Next we show that because we defined $\alpha_{\ref{claim:concentraionZ1}} = \alpha_{\ref{claim:concentraionZ1}}' \sqrt{\frac{3}{2}}$, we have:

\begin{align}
\alpha_{\ref{claim:concentraionZ1}}'\sqrt{ R_1 \log \left( \frac{1}{\epsilon} \right) } \leq \alpha_{\ref{claim:concentraionZ1}} \sqrt{n_1 \log \left( \frac{1}{\epsilon} \right) }. \label{ineq:alpha}
\end{align}
This again follows by the definition of $k_{\ref{claim:concentraionZ1}}$ and event \eqref{event:R1+bar-R1} and Inequality~\eqref{ineq:n_1p-big-enough}:
$$ R_1 < n_1p + \alpha_{\ref{claim:concentraionZ1}}' \sqrt{n_1 \log \left( \frac{1}{\epsilon} \right) } \leq \frac{3}{2} n_1p \leq \frac{3}{2} n_1.$$
Inequality \eqref{ineq:alpha} implies that:
\begin{align*}
&\prob{ \left| Z_1 - \E{Z_1 | R,R_1,Z} \right| \geq \alpha_{\ref{claim:concentraionZ1}} \sqrt{n_1 \log \left( \frac{1}{\epsilon} \right) }} \\ &\leq \prob{ \left| Z_1 - \E{Z_1 | R,R_1,Z} \right| \geq \alpha_{\ref{claim:concentraionZ1}}'\sqrt{ R_1 \log \left( \frac{1}{\epsilon} \right) }} \leq \frac{ \epsilon}{4}, %\label{inequality:Y-centralized-we-have}
\end{align*}
which completes the proof of the lemma.
\end{proof}

\begin{proof}{\textbf{Proof of Lemma~\ref{claim:adversary}:}}
Recall that $\zeta_1$ follows the binomial distribution $\text{Bin}(\eta_1(\lambda), 1-p)$.
Hence, the lemma follows straightforwardly from Corollary~\ref{coro:binomial}.
\end{proof}

\section{Missing proofs of Subsection~\ref{sec:HBY-analysis}}
\label{sec:proof-of-alg1}
\begin{proof}{\textbf{Proof of Lemma~\ref{claim:full-case-comp1}:}}
The revenue of Algorithm~\ref{algorithm:hybrid} in this case is $ ALG_1(\vec{v}) = b - (1-a) \left[ q_{2,e}(1) + q_{2,f}(1) \right],$ which is decreasing in $q_{2,e}(1) + q_{2,f}(1)$.
Note that due to the fixed threshold rule, we already have an upper bound on $q_{2,f}(1)$, i.e., $q_{2,f}(1) \leq \theta b$.
As a result, using Lemma~\ref{lem:HYB:full-case1}, $ALG_1(\vec{v}) \geq b - (1-a)(\theta b + p(b-n_1)^+ +\Delta)$.
Note that $OPT(\vec{v}) \leq b- (1- a) \left(b-n_1\right)^+$, and thus
\begin{align*} \frac{ALG_1(\vec{v})}{OPT(\vec{v})} \geq & \frac{b - (1-a)(\theta b + p(b-n_1)^+ + \Delta)}{b- (1- a) \left(b-n_1\right)^+}
\\ \geq & \frac{b - (1-a)(\theta b + (p + \theta) (b-n_1)^+ + \Delta)}{b- (1- a) \left(b-n_1\right)^+} &(\theta \geq 0)
\\ = &p + \frac{1-p}{2-a}- \frac{(1-a)\Delta}{OPT(\vec{v})}. &(1-(1-a)\theta = p+(1-p)/(2-a))
\end{align*}
The rest follows from the simple inequality $OPT(\vec{v}) \geq ab$ due to $q_1(1) + q_{2,e}(1) +q_{2,f}(1) =b$.
\end{proof}

\begin{proof}{\textbf{Proof of Lemma~\ref{claim:full-comp2}:}}
We consider three cases of Lemma~\ref{lem:HYB:full-case2} separately.
\begin{enumerate}[(a)]
\item $ q_{2,e}(1) +q_{2,f}(1) = n_2$:
\\ Algorithm~\ref{algorithm:hybrid} accepts all customers and hence achieves the optimal revenue, \vmn{i.e., $\frac{ALG_1(\vec{v})}{OPT(\vec{v})} = 1$}.
\item $q_{2,f}(1) = {\lfloor \theta b\rfloor} $ and $n_1  > bp - 3 \Delta$:
\\ $ ALG_1(\vec{v}) \geq n_1 + a \theta b $ and $OPT(\vec{v}) \leq ab+n_1(1-a)$. Therefore,
$$ \frac{ALG_1(\vec{v})}{OPT(\vec{v})} \geq \frac{n_1+ a \theta b}{ab+n_1(1-a)}, $$
which is increasing in $n_1$, so the ratio is minimized at $n_1=bp - 3 \Delta$, which is a special case of the last case, \vmn{which we analyze next.}
\item$q_{2,f}(1) =\lfloor \theta b\rfloor$, $n_1 \leq bp - 3 \Delta $, and $q_{2,e}(1) \geq \left(p(n_1+n_2) -n_1 - 5 \Delta \right)^+$:
{First, we remark that following the discussion in the proof of Lemma~\ref{lem:HYB:full-case2}, we assume, without loss of generality, $n_1+n_2\leq b$. Note that by construction, the alternative adversarial instance has the same optimum offline solution, i.e.,  $OPT(\vec{v}) = OPT(\vec{v}_A)$.}
\begin{align*} ALG_1(\vec{v}) \geq & n_1+ a(p(n_1+n_2)-n_1 -5\Delta+ \theta b) \\
\geq & n_1+ a(p(n_1+n_2)-n_1 -5\Delta+ \theta (n_1+n_2)) &(b \geq n_1+n_2)
\\
= & n_1(1-a+pa+\theta a) + n_2(p+\theta)a -5\Delta a \\
\geq & n_1((1-a)(p+\theta)+a(p+\theta )) + n_2(p+\theta)a -5\Delta a &(p+\theta \leq 1) \\
= &(p+\theta)(n_1+ n_2 a) -5\Delta a \geq (p+\theta)OPT(\vec{v}) -5\Delta a \\
= &\left( p + \frac{1-p}{2-a}\right)OPT(\vec{v}) - 5\Delta a. &(p+\theta = p + \frac{1-p}{2-a})
\end{align*}
Since $OPT(\vec{v}) \geq q_{2,f}(1)a = \theta b a$,
$$ \frac{ALG_1(\vec{v})}{OPT(\vec{v})} \geq p + \frac{1-p}{2-a} - \frac{5 \Delta a}{\theta b a} = p + \frac{1-p}{2-a} - \frac{5 \Delta }{\theta b }.$$
\end{enumerate}
\end{proof}

\begin{proof}{\textbf{Proof of Lemma~\ref{lem:HYB:full-case3}:}}
Clearly, either \vmn{condition} \vmn{(a)} \vmn{holds} or $q_1(1)=n_1$.
Below we consider the cases where $q_1(1)=n_1$.
If $q_{2,f}(1)<\vmn{\lfloor \theta b \rfloor}$, then we do not reject any \st{class}\vmn{type}-$2$ customer, and thus \vmn{condition} \vmn{(b) holds}.
The interesting case is when $q_{2,f}(1)= \lfloor \theta b \rfloor$. \vmn{We prove that condition (c) will hold.}
\vmn{First note that} in this case, we know $n_2 \geq \theta b$.
As a result, \eqref{epsilon-condition-n_2-big} implies we can use \vmn{the concentration result of }\eqref{inequality:good-approximation-app--o^R_2}.

Following the discussion in the proof of Lemma~\ref{lem:HYB:full-case2}, we assume, without loss of generality, $n_1+n_2\leq b$.
\vmn{Further, if we find a time $\hat \lambda$ for which we have:}

\vmn{
\begin{align}
\label{eq:app:intermed2}
o_1(\lambda)+o^\RGS_2(\lambda)-o^\RGS_2(\hat \lambda) \leq \lfloor \lambda pb \rfloor \quad \quad \text{ for all }\lambda \geq \hat \lambda,
\end{align}}
\vmn{then, using a similar induction to the one in Lemma~\ref{lem:HYB:full-case2}, we can show:}

\vmn{
\begin{align}
\label{ineq:app:case2:a2}
q_{2,e}(\lambda) \geq o_2^\RGS (\lambda) - o_2^\RGS (\hat \lambda) \quad \quad \text{ for all }\lambda \geq \hat \lambda.
\end{align}
}
\if false
For a $\bar \lambda$ satisfying Condition~(\ref{condition-for-bar-lambda}), Inequality~\eqref{ineq:case2:a2} holds.\fi
%\begin{align}
%\bar \lambda pn_2 \geq \zeta b -p \lambda (b-n_2) + \Delta.
%\end{align}
Here we find a sufficient condition on $\hat\lambda$ for Condition~\eqref{eq:app:intermed2} to hold.

\begin{align*}
o_1(\lambda)+o^\RGS_2(\lambda)-o^\RGS_2(\hat \lambda) & < \frac{k}{p^2} \log n +\lambda pn_2 +\Delta - (\hat\lambda pn_2 - \Delta ) &(o_1(\lambda) \leq n_1 < \frac{k}{p^2} \log n,\eqref{inequality:good-approximation-app--o^R_2})\\
& = \frac{k}{p^2} \log n +(\lambda - \hat\lambda)pn_2 + 2 \Delta \\
& \leq \frac{k}{p^2} \log n +(\lambda - \hat\lambda)pb + 2 \Delta &(n_1+n_2 \leq b).
\end{align*}
As a result, Condition~\eqref{eq:app:intermed2} holds if
$$ \frac{k}{p^2} \log n +(\lambda - \hat\lambda)pb + 2 \Delta \leq \lambda p b,$$
which can be achieved when
\begin{align*}
\hat\lambda \triangleq \frac{\frac{k}{p^2} \log n+2\Delta}{pb}.% \label{epsilon:lambda-bar-smaller-than-1}
\end{align*}
If $\hat\lambda \leq 1$, then
\begin{align*}
q_{2,e}(1) \geq & o_2^\RGS(1) - o_2^\RGS(\hat\lambda) &(\text{Inequality}~\eqref{ineq:app:case2:a2})\\
\geq & pn_2 - \Delta - (\hat\lambda pn_2 + \Delta ) &(\eqref{epsilon-condition-n_2-big}\text{ and }\eqref{inequality:good-approximation-app--o^R_2}) \\
\geq & pn_2 - (\frac{k}{p^2} \log n+2\Delta) - 2 \Delta &(\hat\lambda \leq \frac{\frac{k}{p^2} \log n+2\Delta}{pn_2} ) \\
= & pn_2 - \frac{k}{p^2} \log n - 4 \Delta .\end{align*}
If $\hat\lambda >1$, then
\begin{align*}
pn_2 - \frac{k}{p^2} \log n - 4 \Delta
= & pn_2 - (\frac{k}{p^2} \log n + 2\Delta) - 2 \Delta \\
< & pn_2 - pb - 2 \Delta < p(n_2-b) < 0 &(\hat\lambda >1, n_2\leq b) \\
\leq & q_{2,e}(1).
\end{align*}
\end{proof}

\begin{proof}{\textbf{Proof of Lemma~\ref{claim:full-comp3}:}}
Following the discussion in the proof of Lemma~\ref{lem:HYB:full-case2}, we assume, without loss of generality, $n_1+n_2\leq b$.
We consider three cases in Lemma~\ref{lem:HYB:full-case3} separately.

\noindent If case \vmn{(a)} in Lemma~\ref{lem:HYB:full-case3} happens, then $ALG_1(\vec{v})+n_1 \geq OPT(\vec{v})$ and $OPT(\vec{v}) \geq ab$.
As a result,
$$ \frac{ALG_1(\vec{v})}{OPT(\vec{v})} \geq 1-\frac{n_1}{OPT(\vec{v})} \geq 1- \frac{\frac{k}{p^2} \log n}{ab} \geq  p+\frac{1-p}{2-a}- \frac{\frac{k}{p^2} \log n}{ab} .$$

\noindent If case \vmn{(b)} in Lemma~\ref{lem:HYB:full-case3} happens, then $ALG_1(\vec{v})=OPT(\vec{v})$, and we are done.

\noindent If case \vmn{(c)} happens, then
\begin{align*}
ALG_1(\vec{v}) \geq n_1 + \left(pn_2 - \frac{k}{p^2} \log n - 4 \Delta + \theta b \right)a.
\end{align*}
Because $n_1 \geq \left( p+\frac{1-p}{2-a} \right) n_1 $, $p n_2 +\theta b \geq (p+\theta)n_2 = \left( p+\frac{1-p}{2-a}\right)n_2$ and $OPT(\vec{v}) = n_1 + a n_2 \geq a \theta b$, we have
\begin{align}
\frac{ALG_1(\vec{v})}{OPT(\vec{v})} \geq p+\frac{1-p}{2-a}
-\frac{a\left(\frac{k}{p^2} \log n + 4 \Delta \right)}{{OPT(\vec{v})}} \geq p+\frac{1-p}{2-a} - \frac{\frac{k}{p^2} \log n + 4 \Delta}{\theta b} .\label{condition:HYA-epsilon-3}
\end{align}
\end{proof}

\if false

\begin{proof}[Proof of Claim~\ref{claim:epsilon-con}]
We have $\frac{1}{a(1-p)p}\sqrt{\frac {\log n}{b}} \geq \sqrt{\frac{1}{b}} \geq \frac{1}{n}$.
In addition, $\frac{1}{a(1-p)p}\sqrt{\frac {\log n}{b}} \leq \bar\epsilon $ is implied by \eqref{ineq:Alg1-non-trivial-case} and \begin{align} \vm{k'} \leq \bar\epsilon. \label{condition:ALG_11:k'<=bar-epsilon}
\end{align} %O\left( \frac{1}{a(1-p)p}\sqrt{\frac {\log n}{b}} \right)$ becomes $O(1)$ (recall that $\bar\epsilon$ is a constant), so our theorem (although not informative) holds.
Therefore, Condition of the claim  holds.
\end{proof}
\fi

\begin{proof}{\textbf{Proof of Lemma~\ref{claim:error-terms}:}}
It is easy to check \vmn{(a)} $ \frac{(1-a)\Delta}{ab} =O\left(\frac{1}{a(1-p)p} \sqrt{\frac{\log n}{b}}\right)$ and \vmn{(b)} $  \frac{5 \Delta }{\theta b }=O\left(\frac{1}{a(1-p)p} \sqrt{\frac{\log n}{b}}\right)$. To prove \vmn{(c)} $\frac{\frac{k}{p^2} \log n}{ab}=O\left(\frac{1}{a(1-p)p} \sqrt{\frac{\log n}{b}}\right)$, we first note that \eqref{ineq:Alg1-non-trivial-case} and \begin{align} \vmn{\bar\epsilon} \leq 1. \label{condition:ALG_11:k'<=1}
\end{align}
implies $\log n \leq a^2 p^2 b$, and thus $\log n = \sqrt{\log n} \sqrt{\log n} \leq ap\sqrt{b\log n}$.
As a result,
\begin{align*}
\frac{\frac{k}{p^2} \log n }{a b} \leq &
\frac{\frac{k}{p} \sqrt{b\log n} }{b} &(\log n\leq ap\sqrt{b\log n} ) \\
\leq & \frac{k}{a(1-p)p}\sqrt{\frac {\log n}{b}} &(0<p<1 \text{ and } a<1) \\
 = &   O\left(\frac{1}{a(1-p)p}\sqrt{\frac {\log n}{b}}\right).
\end{align*}

Similarly, to prove \vmn{(d)} $\frac{\frac{k}{p^2} \log n + 4 \Delta}{\theta b} =O\left(\frac{1}{a(1-p)p} \sqrt{\frac{\log n}{b}}\right)$, we first note that \eqref{ineq:Alg1-non-trivial-case} and \eqref{condition:ALG_11:k'<=1}
implies $\log n \leq bp^2$, and thus $\log n = \sqrt{\log n} \sqrt{\log n} \leq p\sqrt{b\log n}$.
As a result,
\begin{align*}
\frac{\frac{k}{p^2} \log n + 4 \Delta}{\theta b} \leq &
\frac{\frac{k}{p} \sqrt{b\log n} + 4 \alpha \sqrt{b\log n}}{\theta b} &(\log n\leq p\sqrt{b\log n} \text{ and }\Delta = \alpha \sqrt{b\log n}) \\
\leq & \frac{k+4\alpha}{pa}\sqrt{\frac{ \log n}{b}} &(p<1 \text{ and }\theta > a) \\
\leq & (k+4\alpha)\left(\frac{1}{a(1-p)p} \
\sqrt{\frac {\log n}{b}}\right) \\ = & O\left(\frac{1}{a(1-p)p}\sqrt{\frac {\log n}{b}}\right).
\end{align*}
\end{proof}

\section{Missing proofs of Section~\ref{sec:alg2}}
%Proof of Proposition~\ref{prop:values-math-program}}
\label{sec:proof-value-MP1}
%We begin by reiterating the proposition:
%\begin{repproposition}{prop:values-math-program}
%For any $b\leq n$, we have: $c^* > p+\frac{1-p}{2-a}$. Further, if $b=n$, then $c^* =1$.
%\end{repproposition}
%Recall that $c^*$ is the optimal objective in (\ref{MP1}).

\vmn{Before proceeding with the proofs, we state and prove an auxiliary lemma that establishes an upper bound on $n_1$ and $n_1 + n_2$ using the deterministic approximation functions $\tilde o_j(\cdot)$.}

\vmn{
\begin{lemma}
\label{prop:Ubounds}
For $\lambda \in \{1/n, 2/n, \ldots, 1\}$, we have:
\if flase
If for some $\lambda$ and $\vec{v}$, the estimation~\eqref{eq:estimate} holds with equality for both $o_1(\lambda)$ and $o_2(\lambda)$, then
\fi
\begin{subequations}
\begin{align} & n_1 \leq  \min \left \{ \frac{\tilde o_1(\lambda)}{\lambda p}, \frac{ \tilde o_1(\lambda) + (1-\lambda) (1-p) n}{1-p+\lambda p} \right\} \text{, and } \label{eq:u_1}
\\ & n_1 + n_2 \leq  \min \left \{ \frac{\tilde o_1(\lambda)+\tilde o_2(\lambda)}{\lambda p}, \frac{ \tilde o_1(\lambda)+\tilde o_2(\lambda) + (1-\lambda) (1-p) n}{1-p+\lambda p} \right\}. \label{eq:u_2}
\end{align}
\end{subequations}
\end{lemma}
\begin{proof}
The proof essentially follows from the definition of the deterministic approximation functions. Here we just prove \eqref{eq:u_1}.
First note that $\eta_1(\lambda) \geq 0$, thus:
\begin{align*}
\tilde o_1(\lambda) = (1-p)\eta_1(\lambda)+p \lambda n_1 \geq p \lambda n_1  ~~\Rightarrow~~  n_1 \leq  \frac{\tilde o_1(\lambda)}{\lambda p}.
\end{align*}
Second note that $n_1 - \eta_1(\lambda) \leq (1-\lambda) n$, therefore,
\begin{align*}
\tilde o_1(\lambda) = (1-p)\eta_1(\lambda)+p \lambda n_1 \geq (1-p)(n_1 -(1-\lambda) n) +  p \lambda n_1  ~~\Rightarrow~~  n_1 \leq \frac{ \tilde o_1(\lambda) + (1-\lambda) (1-p) n}{1-p+\lambda p}.
\end{align*}
Which completes the proof of \eqref{eq:u_1}. Proof of \eqref{eq:u_2} follows similar steps.
\if false
The simple inequality $\eta_1(\lambda) \geq 0$ and \eqref{eq:estimate} imply that $n_1 \leq \frac{o_1(\lambda)}{\lambda p}$.
Furthermore, in the adversarial customer arrival sequence $\vec{v'}\vmn{ZZ\vec{v}_I}$, the number of \st{class}\vmn{type}-$1$ customers arriving after time $\lambda$ cannot exceed the number of the remaining time slots, i.e., $n_1 - \eta_1(\lambda) \leq (1-\lambda) n$.
This means $\eta_1(\lambda) \geq n_1 - (1-\lambda) n$. Substituting this in \eqref{eq:estimate}, and rearranging term we get: $n_1 \leq \frac{ o_1(\lambda) + (1-\lambda) (1-p) n}{1-p+\lambda p}$.
\fi
\end{proof}
}

\begin{proof}{\textbf{Proof of Lemma~\ref{prop:full-Ubounds}:}}
Since Inequalities~\eqref{eq:full-u_1} and~\eqref{eq:full-u_2} are similar, we only present the proof of  Inequality~\eqref{eq:full-u_1}.
When $\lambda < \delta$, $u_1(\lambda) \triangleq b$, and thus~\eqref{eq:full-u_1} trivially holds.
The more interesting case is when $\lambda \geq \delta$.
Without loss of generality, we assume $n_1 \leq b+\frac{2 \Delta}{\delta p}$.
\vmn{Otherwise, similar to the proof of Lemma~\ref{lem:HYB:full-case1}, we construct a modified adversarial instance with only $b+\frac{2 \Delta}{\delta p}$ type-$1$ customers and argue that, for the same realization of the {\RG} group and random permutation, $u_1(\lambda)$ is lower bounded by the one corresponding to the modified instance.}
\if flase
Otherwise, under the same realization of the {\RG} group and the random permutation, the number of observed \st{class}\vmn{type}-$1$ customers is at least the number in the case where we consider only an arbitrary subset of $b+\frac{2 \Delta}{\delta p}$ \st{class}\vmn{type}-$1$ customers in the original adversarial sequence ${\vec{v}}'$.
Therefore, we still get a valid lower bound on $o_1(\lambda)$ (and hence $u_1(\lambda)$) when we assume $n_1 \leq b+\frac{2 \Delta}{\delta p} $.
\fi
Note that we can apply Inequality~\eqref{inequality:good-approximation-o_1} \vmn{to this modified instance, because}\st{to those} $ b+\frac{2 \Delta}{\delta p}   \geq \frac{k}{p^2} \log n$ \vmn{under the condition imposed on $b$}.

\noindent Note that $\delta =\frac{\phi b}{n}=\frac{(1-c)b}{(1-a)n} \geq \frac{(1-c)b}{n}$, Condition \vmn{imposed on $b$, and} %\eqref{ineq:ADP-non-trivial-case}
\begin{align} \vmn{\bar\epsilon}\leq \frac{3}{2\alpha}\label{ineq:ADP-k'}
\end{align}
\vmn{which holds by the definition of constants $\alpha$ and $\bar\epsilon$ given in Lemma~\ref{lemma:needed-centrality-result-for-m=2}, }give \vmn{us}
\begin{align}
\frac{\Delta}{\delta p}\leq \frac{3}{2}b.\label{epsilon:ADP-small-Delta}
\end{align}
Therefore, $n_1 \leq b+\frac{2 \Delta}{\delta p} \leq 4b$, and thus Inequality~\eqref{inequality:good-approximation-o_1} implies $o_1(\lambda) \geq \tilde o_1(\lambda) - \alpha \sqrt{n_1 \log n} \geq \tilde o_1(\lambda) - \alpha \sqrt{4 b \log n} = \tilde o_1(\lambda) - 2\Delta$.
As a result,

\begin{align*}
u_1(\lambda) \triangleq & \min \left \{ \frac{o_1(\lambda)}{\lambda p}, \frac{ o_1(\lambda) + (1-\lambda) (1-p) n}{1-p+\lambda p} \right\} \\
\geq & \min \left \{ \frac{\tilde o_1(\lambda) - 2\Delta}{\lambda p}, \frac{ \tilde o_1(\lambda) - 2\Delta + (1-\lambda) (1-p) n}{1-p+\lambda p} \right\} &(o_1(\lambda) \geq \tilde o_1(\lambda) - 2\Delta ) \\
\geq & n_1 - \max \left\{ \frac{2\Delta}{\lambda p} , \frac{2\Delta}{1-p+\lambda p}\right\} = n_1-\frac{2\Delta}{\lambda p} &(\text{Lemma~\ref{prop:Ubounds}}) \\
\geq & n_1 - \frac{2\Delta}{\delta p}. &(\lambda \geq \delta)
\end{align*}

\end{proof}

\if flase
\vm{TEMP}
\begin{lemma}
\label{prop:Ubounds}
If for some $\lambda$ and $\vec{v}$, the estimation~\eqref{eq:estimate} holds with equality for both $o_1(\lambda)$ and $o_2(\lambda)$, then
\begin{subequations}
\begin{align} & n_1 \leq  \min \left \{ \frac{o_1(\lambda)}{\lambda p}, \frac{ o_1(\lambda) + (1-\lambda) (1-p) n}{1-p+\lambda p} \right\} \text{, and } \label{eq:u_1}
\\ & n_1 + n_2 \leq  \min \left \{ \frac{o_1(\lambda)+o_2(\lambda)}{\lambda p}, \frac{ o_1(\lambda)+o_2(\lambda) + (1-\lambda) (1-p) n}{1-p+\lambda p} \right\}. \label{eq:u_2}
\end{align}
\end{subequations}
\end{lemma}
\begin{proof}
Note that \eqref{eq:u_2} and \eqref{eq:u_1} are symmetric and hence it is sufficient to derive \eqref{eq:u_1}.
The simple inequality $\eta_1(\lambda) \geq 0$ and \eqref{eq:estimate} imply that $n_1 \leq \frac{o_1(\lambda)}{\lambda p}$.
Furthermore, in the adversarial customer arrival sequence $\vec{v'}\vmn{ZZ\vec{v}_I}$, the number of \st{class}\vmn{type}-$1$ customers arriving after time $\lambda$ cannot exceed the number of the remaining time slots, i.e., $n_1 - \eta_1(\lambda) \leq (1-\lambda) n$.
This means $\eta_1(\lambda) \geq n_1 - (1-\lambda) n$. Substituting this in \eqref{eq:estimate}, and rearranging term we get: $n_1 \leq \frac{ o_1(\lambda) + (1-\lambda) (1-p) n}{1-p+\lambda p}$.
\end{proof}
\vm{TEMP}
\fi

\begin{proof}{\textbf{Proof of Proposition~\ref{prop:values-math-program}:}}

\textbf{Case $b < n$:}
For proving this case, we can relax some of the constraints and only keep Constraints~(\ref{constraint:not-c-competitive}), (\ref{constraint:x<=1}), (\ref{constraint:n1'<=n1}), and (\ref{constraint:n2'<=n2}) and show that for every point in this superset of the feasibility region, $c\geq p+ \frac{1-p}{2-a}$.

We first notice that, for fixed $n_1$ and $n_2$, the right hand side of Constraint~(\ref{constraint:not-c-competitive}) is non-increasing in each of $\tilde u_1$ and $\tilde o_2$, and hence non-increasing in each of $\eta_1$ and $\eta_2$.
By using Constraints~(\ref{constraint:n1'<=n1}) and~(\ref{constraint:n2'<=n2}), we can obtain upper bounds $\tilde o_1 \leq (1-p+\vmn{\l} p)n_1$ and $\tilde o_2 \leq (1-p+\vmn{\l} p) n_2$ .
With these upper bounds and the fact $\tilde u_1 \leq \frac{\tilde o_1}{\vmn{\l} p } $, Constraint~(\ref{constraint:not-c-competitive}) gives:
\begin{align*}
c \geq \frac{a(n_2- (1-p+p\vmn{\l})n_2 +\frac{b}{1-a})+n_1}{a\min\{n_1+n_2, b\}+(1-a)n_1+\frac{a^2b}{1-a}+a \min \left\{ \frac{(1-p+p\vmn{\l})n_1}{\vmn{\l} p } , b \right\}},
\end{align*}
which, after rearranging terms, leads to
\begin{align}
c \geq \frac{a n_2 p (1-\vmn{\l}) +\frac{ab}{1-a}+n_1}{a\min\{n_2, b-n_1\}+n_1+\frac{a^2b}{1-a}+a \min \left\{ \left( \frac{1-p}{\vmn{\l} p }+1 \right) n_1, b \right\}}. \label{inequality:raw-form}
\end{align}
We focus on lower bounding the right hand side of~(\ref{inequality:raw-form}).
When $n_2\geq b-n_1$, the right hand side of~(\ref{inequality:raw-form}) is non-decreasing in $n_2$ because the denominator remains the same while the numerator is non-decreasing when $n_2$ increases (due to Constraint~(\ref{constraint:x<=1})).
Therefore, for the sake of obtaining a lower bound, we can assume, without loss of generality,
\begin{align}
n_2 \leq b-n_1. \label{assumption:n_1+n_2<=b}
\end{align}
With~(\ref{assumption:n_1+n_2<=b}), the right hand side of~(\ref{inequality:raw-form}) can be written as
\begin{subequations}
\begin{align}
&\vmn{f_1(\vmn{\l}) \triangleq }\frac{a n_2 p (1-\vmn{\l}) +\frac{ab}{1-a}+n_1}{an_2+n_1+\frac{a^2b}{1-a}+a b} & \text{ if }\vmn{\l} \leq \frac{(1-p)n_1}{p(b-n_1)},\label{to-prove:small-lambda}
\\ &\vmn{f_2(\vmn{\l}) \triangleq } \frac{a n_2 p (1-\vmn{\l}) +\frac{ab}{1-a}+n_1}{an_2+n_1+\frac{a^2b}{1-a}+a\left( \frac{1-p}{\vmn{\l} p }+1 \right) n_1} &\text{ if }\vmn{\l} > \frac{(1-p)n_1}{p(b-n_1)}.\label{to-prove:big-lambda}
\end{align}
\end{subequations}
We prove \vmn{
\begin{align}
& f_1(\vmn{\l}) \geq p+\frac{1-p}{2-a}, ~~~\text{ if } \vmn{\l} \leq \frac{(1-p)n_1}{p(b-n_1)}\label{eq:inter1}\\
& f_2(\vmn{\l}) \geq p+\frac{1-p}{2-a}, ~~~\text{ if } \vmn{\l} > \frac{(1-p)n_1}{p(b-n_1)}\label{eq:inter2}
\end{align}}
\if false
$f_1(\vmn{\l}) \geq p+\frac{1-p}{2-a}$, for $\vmn{\l} \leq \frac{(1-p)n_1}{p(b-n_1)}$, and
$f_2(\vmn{\l}) \geq p+\frac{1-p}{2-a}$, for $\vmn{\l} > \frac{(1-p)n_1}{p(b-n_1)}$
\fi
\vmn{separately. We start by the former:}
\if false
expressions~(\ref{to-prove:small-lambda}) and~(\ref{to-prove:big-lambda}) being at least $p+\frac{1-p}{2-a}$ separately.
We first prove that expression (\ref{to-prove:small-lambda}) is at least $p+\frac{1-p}{2-a}$.
\fi
Because
%(\ref{to-prove:small-lambda})
\vmn{$f_1(\vmn{\l})$} is non-increasing in $\vmn{\l}$, we only need to prove for the case $\vmn{\l} = \frac{(1-p)n_1}{p(b-n_1)}$; \vmn{
$f_1(\vmn{\l})$ at $\vmn{\l} = \frac{(1-p)n_1}{p(b-n_1)}$ can be rearranged as}
%With this, quantity (\ref{to-prove:small-lambda}) can be rearranged as
\begin{align*}
\vmn{f_1\left(\frac{(1-p)n_1}{p(b-n_1)}\right) =} 1- \frac{an_2 \left(1-p+\frac{(1-p)n_1}{p(b-n_1)} p\right)}{an_2+n_1+\frac{ab}{1-a}},
\end{align*}
which, for fixed $n_1$ and $n_2$, is non-decreasing in $b$.
Therefore, according to~(\ref{assumption:n_1+n_2<=b}), we only need to consider the case $b=n_1+n_2$ (in the degenerated case $n_1=n_2=0$, \st{the above quantity}\vmn{$f_1\left(\frac{(1-p)n_1}{p(b-n_1)}\right)$} is $1$, which is greater than $p+\frac{1-p}{2-a}$, so we can assume, without loss of generality, $n_1+n_2>0$), in which case, \st{the above quantity equals to}\vmn{we have:}
\begin{align*}
\vmn{f_1\left(\frac{(1-p)n_1}{p(b-n_1)}\right) =} 1- \frac{a (1-p) (n_1+n_2)}{an_2+n_1+\frac{a(n_1+n_2)}{1-a}}
\geq 1- \frac{a (1-p) (n_1+n_2)}{an_2+an_1+\frac{a(n_1+n_2)}{1-a}}=1-\frac{1-p}{1+\frac{1}{1-a}}=p+\frac{1-p}{2-a},
\end{align*}
\vmn{which completes the proof of \eqref{eq:inter1}.}
\if false
which is at least
\begin{align*}
1- \frac{a (1-p) (n_1+n_2)}{an_2+an_1+\frac{a(n_1+n_2)}{1-a}}=1-\frac{1-p}{1+\frac{1}{1-a}}=p+\frac{1-p}{2-a}.
\end{align*}
Thus, we are done proving that~(\ref{to-prove:small-lambda}) is at least $p+\frac{1-p}{2-a}$.
\fi
\vmn{Next we prove \eqref{eq:inter2}.}
Due to Constraint~(\ref{constraint:x<=1}), \vmn{i.e., $\l \leq 1$, } case~(\ref{to-prove:big-lambda}) only happens when $\frac{(1-p)n_1}{p(b-n_1)}<1$, or equivalently,
\begin{align}
b > \frac{n_1}{p}. \label{lower-bound-b}
\end{align}
Proving \eqref{eq:inter2}
\if false
~(\ref{to-prove:big-lambda}) being at least $p+\frac{1-p}{2-a}$
\fi
is trickier because both the numerator and denominator decrease in $\vmn{\l}$.
To address this issue,
\vmn{we first remark that by the definition of $f_2(\l)$, inequality $f_2(\l) \geq p+\frac{1-p}{2-a}$ is equivalent to}
%we express what we need to prove as follow,
\begin{align*}
a n_2 p (1-\vmn{\l}) +\frac{ab}{1-a}+n_1 \geq \left( p+\frac{1-p}{2-a} \right) \left(an_2+n_1+\frac{a^2b}{1-a}+a\left( \frac{1-p}{\vmn{\l} p }+1 \right) n_1\right),
\end{align*}
which is \vmn{in turn} equivalent to
\begin{align}
a n_2 p +\frac{ab}{1-a}+n_1 \geq \left( p+\frac{1-p}{2-a} \right) \left(an_2+n_1+\frac{a^2b}{1-a}+a\left( \frac{1-p}{\vmn{\l} p }+1 \right) n_1\right) + an_2p\vmn{\l}. \label{to-prove:function-of-lambda}
\end{align}
\vmn{Now} the left hand side of \st{statement}\vmn{Inequality}~(\ref{to-prove:function-of-lambda}) is not a function of $\vmn{\l}$ while the right hand side of~(\ref{to-prove:function-of-lambda}) is a function of $\vmn{\l}$ of the form
\begin{align} x\vmn{\l}+\frac{y}{\vmn{\l}}+z \label{function-of-lambda}
\end{align} where $x,y,z$ are non-negative\st{constants}.
Clearly, the second derivative of~(\ref{function-of-lambda}) (\st{over}\vmn{with respect to} $\vmn{\l}$), $\frac{2y}{\vmn{\l}^3}$, is non-negative for $\vmn{\l} \in \left[\frac{(1-p)n_1}{p(b-n_1)}, 1\right]$.
As a result, (\ref{function-of-lambda}) is convex and is maximized at extreme values of $\vmn{\l}$, which in our case is at either $\vmn{\l} = \frac{(1-p)n_1}{p(b-n_1)}$ or $\vmn{\l}=1$.
Therefore, we only need to prove \st{statement}\vmn{Inequality}~(\ref{to-prove:function-of-lambda}) at these extreme two values of $\vmn{\l}$.
The former case, $\vmn{\l} = \frac{(1-p)n_1}{p(b-n_1)}$, is covered in
\vmn{\eqref{eq:inter1}}.
%~(\ref{to-prove:small-lambda}).
Thus, we only need to prove~(\ref{to-prove:function-of-lambda}) for $\vmn{\l} = 1$.
When $\vmn{\l} =1$, (\ref{to-prove:function-of-lambda}) can \st{easily} be rearranged as
\begin{align}
\label{eq:chain1}
\frac{ab}{1-a}+n_1 - \left( p+\frac{1-p}{2-a} \right) \left(an_2+n_1+\frac{a^2b}{1-a}+\frac{an_1}{p }\right) \geq 0.
\end{align}
Because the left hand side of Inequality \vmn{\eqref{eq:chain1}} is a decreasing function of $n_2$, using~(\ref{assumption:n_1+n_2<=b}), we only need to prove that the inequality holds when $n_2=b-n_1$, which is equivalent to
\begin{align}
\label{eq:chain2}
\frac{ab}{1-a}+n_1 - \left( p+\frac{1-p}{2-a} \right) \left(a(b-n_1)+n_1+\frac{a^2b}{1-a}+\frac{an_1}{p }\right) \geq 0.
\end{align}
By separating terms associated with $n_1$ and $b$, and using $$1-\left( p+\frac{1-p}{2-a} \right)=\frac{(1-p)(1-a)}{2-a},$$
\st{above}Inequality \vmn{\eqref{eq:chain2}} is equivalent to
\begin{align}
\label{eq:chain3}
\frac{a(1-p)}{2-a}b + \frac{(1-p)(1-a)}{2-a}n_1 - \left( p+\frac{1-p}{2-a} \right) \left(\frac{a(1-p)}{p }\right)n_1 \geq 0.
\end{align}
Using the lower bound on $b$ given by~(\ref{lower-bound-b}), \vmn{i.e., $b > \frac{n_1}{p}$} and then dividing $(1-p)n_1$ on both sides (in the degenerated case where $(1-p)n_1=0$, both sides are $0$ so we are done), Inequality \vmn{\eqref{eq:chain3}} is implied by
\begin{align*}
\frac{a}{(2-a)p} + \frac{1-a}{2-a} - \left( p+\frac{1-p}{2-a} \right) \left(\frac{a}{p }\right)\geq 0,
\end{align*}
 \vmn{which holds because after canceling the two terms involving $1/p$, it} is equivalent to
\begin{align*}
\frac{(1-a)^2}{2-a} \geq 0.
\end{align*}
\vmn{Therefore, we proved Inequality \vmn{\eqref{eq:chain1}}, and thus \eqref{eq:inter2}. This completes the proof of proposition for the case $b < n$.}
\st{so we are done.}

\textbf{Case $b=n$:}
\st{Now let us prove the $b=n$ case.}
For proving this case, we  relax some of the constraints and only keep Constraints~(\ref{constraint:not-c-competitive}),~(\ref{constraint:x<=1}),~(\ref{constraint:eta_1+eta_2<=lambdan}) and~(\ref{constraint:n1+n2<=n}) and show that for every point in this superset of the feasibility region, $c\geq 1$.
According to Constraint~(\ref{constraint:not-c-competitive}) and using $a\min\{n_1+n_2, b\}+(1-a)n_1 = a\min\{n_2, b-n_1\}+n_1 $, it suffices to prove
\begin{align}
\label{ineq:inter:1}
\frac{a(n_2-\tilde o_2+\frac{b}{1-a})+n_1}{a\min\{n_2, b-n_1\}+n_1+\frac{a^2b}{1-a}+a \tilde u_1} \geq 1.
\end{align}
\vmn{or equivalently,}

\begin{align}
\label{ineq:inter:2}
{a(n_2-\tilde o_2+\frac{b}{1-a})+n_1} \geq {a\min\{n_2, b-n_1\}+n_1+\frac{a^2b}{1-a}+a \tilde u_1}.
\end{align}

\st{Multiplying both sides by the denominator of the left hand side,} Using $\min\{n_2, b-n_1\}\leq n_2$, subtracting $an_2+n_1$ on both sides \vmn{of \eqref{ineq:inter:2}}, and dividing both sides by $a$, \st{the above}Inequality \vmn{\eqref{ineq:inter:2}} is implied by

%
%For proving this case, we relax some of the constraints and only keeps Constraints~(\ref{constraint:not-c-competitive}),~(\ref{constraint:x<=1}),~(\ref{constraint:eta_1+eta_2<=lambdan}) and~(\ref{constraint:n1+n2<=n}) and show that for every point in this superset of the feasibility region, $c\geq 1$.
%
%
%According to Constraint~(\ref{constraint:not-c-competitive}) and using $a\min\{n_1+n_2, b\}+(1-a)n_1 = a\min\{n_2, b-n_1\}+n_1 $, it suffices to prove
%\begin{align*}
%\frac{a(n_2-\tilde o_2+\frac{b}{1-a})+n_1}{a\min\{n_2, b-n_1\}+n_1+\frac{a^2b}{1-a}+a \tilde u_1} \geq 1.
%\end{align*}
%Multiplying both sides by the denominator of the left hand side, using $\min\{n_2, b-n_1\}\leq n_2$, subtracting $an_2+n_1$ on both sides, and dividing both sides by $a$, the above inequality is implied by
\begin{align*}
-\tilde o_2+\frac{b}{1-a} \geq \frac{ab}{1-a}+ \tilde u_1.
\end{align*}
Subtracting $\frac{ab}{1-a}$ on both sides and using $\tilde u_1 \leq \frac{\tilde o_1 + (1-\vmn{\l} ) (1-p)n}{(1-p+\vmn{\l} p)}$, the above inequality is implied by
\begin{align*}
b -\tilde o_2 \geq \frac{\tilde o_1 + (1-\vmn{\l} ) (1-p)n}{(1-p+\vmn{\l} p)}.
\end{align*}
Multiplying $1-p+\vmn{\l} p$ on both sides and using $b=n$ (which is the assumption of this case), the above inequality is equivalent to
\begin{align*}
\vmn{\l} n- (1-p+\vmn{\l} p) \tilde o_2 \geq \tilde o_1.
\end{align*}
Due to Constraint~(\ref{constraint:x<=1}), $1-p+\vmn{\l} p \leq 1$, and thus the inequality above is implied by
\begin{align*}
\vmn{\l} n \geq \tilde o_2 + \tilde o_1, \end{align*}
or equivalently,
\begin{align*}
\vmn{\l} n \geq (1-p)\eta_2+p n_2 \vmn{\l} + (1-p)\eta_1+p n_1 \vmn{\l}.
\end{align*}
\vmn{The above inequality} follows straightforwardly from Constraints~(\ref{constraint:eta_1+eta_2<=lambdan}) and~(\ref{constraint:n1+n2<=n}). \vmn{This completes our proof of {\eqref{ineq:inter:2}}, and consequently that of the proposition in the case $b = n$.}
\end{proof}

\if false

\begin{proof}{\textbf{Proof of Proposition~\ref{prop:values-math-program}:}}

\textbf{Case $b < n$:}
For proving this case, we\st{can} relax some of the constraints and only keep Constraints~(\ref{constraint:not-c-competitive}), (\ref{constraint:x<=1}), (\ref{constraint:n1'<=n1}), and (\ref{constraint:n2'<=n2}) and show that for every point in this superset of the feasibility region, $c\geq p+ \frac{1-p}{2-a}$.

We first notice that, for fixed $n_1$ and $n_2$, the right hand side of Constraint~(\ref{constraint:not-c-competitive}) is non-increasing in each of $\tilde u_1$ and $\tilde o_2$, and hence non-increasing in each of $\eta_1$ and $\eta_2$.
By using Constraints~(\ref{constraint:n1'<=n1}) and~(\ref{constraint:n2'<=n2}), we can obtain upper bounds $\tilde o_1 \leq (1-p+\lambda p)n_1$ and $\tilde o_2 \leq (1-p+\lambda p) n_2$ .
With these upper bounds and the fact $\tilde u_1 \leq \frac{\tilde o_1}{\lambda p } $, Constraint~(\ref{constraint:not-c-competitive}) gives:
\begin{align*}
c \geq \frac{a(n_2- (1-p+p\lambda)n_2 +\frac{b}{1-a})+n_1}{a\min\{n_1+n_2, b\}+(1-a)n_1+\frac{a^2b}{1-a}+a \min \left\{ \frac{(1-p+p\lambda)n_1}{\lambda p } , b \right\}},
\end{align*}
which, after rearranging terms, leads to
\begin{align}
c \geq \frac{a n_2 p (1-\lambda) +\frac{ab}{1-a}+n_1}{a\min\{n_2, b-n_1\}+n_1+\frac{a^2b}{1-a}+a \min \left\{ \left( \frac{1-p}{\lambda p }+1 \right) n_1, b \right\}}. \label{inequality:raw-form}
\end{align}
We focus on lower bounding the right hand side of~(\ref{inequality:raw-form}).
When $n_2\geq b-n_1$, the right hand side of~(\ref{inequality:raw-form}) is non-decreasing in $n_2$ because the denominator remains the same while the numerator is non-decreasing when $n_2$ increases (due to Constraint~(\ref{constraint:x<=1})).
Therefore, for the sake of obtaining a lower bound, we can assume, without loss of generality,
\begin{align}
n_2 \leq b-n_1. \label{assumption:n_1+n_2<=b}
\end{align}
With~(\ref{assumption:n_1+n_2<=b}), the right hand side of~(\ref{inequality:raw-form}) can be written as
\begin{subequations}
\begin{align}
&\vmn{f_1(\lambda) \triangleq }\frac{a n_2 p (1-\lambda) +\frac{ab}{1-a}+n_1}{an_2+n_1+\frac{a^2b}{1-a}+a b} & \text{ if }\lambda \leq \frac{(1-p)n_1}{p(b-n_1)},\label{to-prove:small-lambda}
\\ &\vmn{f_2(\lambda) \triangleq } \frac{a n_2 p (1-\lambda) +\frac{ab}{1-a}+n_1}{an_2+n_1+\frac{a^2b}{1-a}+a\left( \frac{1-p}{\lambda p }+1 \right) n_1} &\text{ if }\lambda > \frac{(1-p)n_1}{p(b-n_1)}.\label{to-prove:big-lambda}
\end{align}
\end{subequations}
We prove \vmn{
\begin{align}
& f_1(\lambda) \geq p+\frac{1-p}{2-a}, ~~~\lambda \leq \frac{(1-p)n_1}{p(b-n_1)}\label{eq:inter1}\\
& f_2(\lambda) \geq p+\frac{1-p}{2-a}, ~~~\lambda > \frac{(1-p)n_1}{p(b-n_1)}\label{eq:inter2}
\end{align}}
\if false
$f_1(\lambda) \geq p+\frac{1-p}{2-a}$, for $\lambda \leq \frac{(1-p)n_1}{p(b-n_1)}$, and
$f_2(\lambda) \geq p+\frac{1-p}{2-a}$, for $\lambda > \frac{(1-p)n_1}{p(b-n_1)}$
\fi
\vmn{separately. We start by the former:}
\if false
expressions~(\ref{to-prove:small-lambda}) and~(\ref{to-prove:big-lambda}) being at least $p+\frac{1-p}{2-a}$ separately.
We first prove that expression (\ref{to-prove:small-lambda}) is at least $p+\frac{1-p}{2-a}$.
\fi
Because
%(\ref{to-prove:small-lambda})
\vmn{$f_1(\lambda)$} is non-increasing in $\lambda$, we only need to prove for the case $\lambda = \frac{(1-p)n_1}{p(b-n_1)}$; \vmn{
$f_1(\lambda)$ at $\lambda = \frac{(1-p)n_1}{p(b-n_1)}$ can be rearranged as}
%With this, quantity (\ref{to-prove:small-lambda}) can be rearranged as
\begin{align*}
\vmn{f_1\left(\frac{(1-p)n_1}{p(b-n_1)}\right) =} 1- \frac{an_2 \left(1-p+\frac{(1-p)n_1}{p(b-n_1)} p\right)}{an_2+n_1+\frac{ab}{1-a}},
\end{align*}
which, for fixed $n_1$ and $n_2$, is non-decreasing in $b$.
Therefore, according to~(\ref{assumption:n_1+n_2<=b}), we only need to consider the case $b=n_1+n_2$ (in the degenerated case $n_1=n_2=0$, \st{the above quantity}\vmn{$f_1\left(\frac{(1-p)n_1}{p(b-n_1)}\right)$} is $1$, which is greater than $p+\frac{1-p}{2-a}$, so we can assume, without loss of generality, $n_1+n_2>0$), in which case, \st{the above quantity equals to}\vmn{we have:}
\begin{align*}
\vmn{f_1\left(\frac{(1-p)n_1}{p(b-n_1)}\right) =} 1- \frac{a (1-p) (n_1+n_2)}{an_2+n_1+\frac{a(n_1+n_2)}{1-a}}
\geq 1- \frac{a (1-p) (n_1+n_2)}{an_2+an_1+\frac{a(n_1+n_2)}{1-a}}=1-\frac{1-p}{1+\frac{1}{1-a}}=p+\frac{1-p}{2-a},
\end{align*}
\vmn{which completes the proof of \eqref{eq:inter1}.}
\if false
which is at least
\begin{align*}
1- \frac{a (1-p) (n_1+n_2)}{an_2+an_1+\frac{a(n_1+n_2)}{1-a}}=1-\frac{1-p}{1+\frac{1}{1-a}}=p+\frac{1-p}{2-a}.
\end{align*}
Thus, we are done proving that~(\ref{to-prove:small-lambda}) is at least $p+\frac{1-p}{2-a}$.
\fi
\vmn{Next we prove \eqref{eq:inter2}.}
Due to Constraint~(\ref{constraint:x<=1}), case~(\ref{to-prove:big-lambda}) only happens when $\frac{(1-p)n_1}{p(b-n_1)}<1$, or equivalently,
\begin{align}
b > \frac{n_1}{p}. \label{lower-bound-b}
\end{align}
Proving~(\ref{to-prove:big-lambda}) being at least $p+\frac{1-p}{2-a}$ is trickier because both the numerator and denominator decrease in $\lambda$.
To address this issue, we express what we need to prove as follow,
\begin{align*}
a n_2 p (1-\lambda) +\frac{ab}{1-a}+n_1 \geq \left( p+\frac{1-p}{2-a} \right) \left(an_2+n_1+\frac{a^2b}{1-a}+a\left( \frac{1-p}{\lambda p }+1 \right) n_1\right),
\end{align*}
which is equivalent to
\begin{align}
a n_2 p +\frac{ab}{1-a}+n_1 \geq \left( p+\frac{1-p}{2-a} \right) \left(an_2+n_1+\frac{a^2b}{1-a}+a\left( \frac{1-p}{\lambda p }+1 \right) n_1\right) + an_2p\lambda. \label{to-prove:function-of-lambda}
\end{align}
The left hand side of statement~(\ref{to-prove:function-of-lambda}) is not a function of $\lambda$ while the right hand side of~(\ref{to-prove:function-of-lambda}) is a function of $\lambda$ of the form
\begin{align} x\lambda+\frac{y}{\lambda}+z \label{function-of-lambda}
\end{align} where $x,y,z$ are non-negative constants.
Clearly, the second derivative of~(\ref{function-of-lambda}) (over $\lambda$), $\frac{2y}{\lambda^3}$, is non-negative for $\lambda \in \left[\frac{(1-p)n_1}{p(b-n_1)}, 1\right]$.
As a result, (\ref{function-of-lambda}) is convex and is maximized at extreme values of $\lambda$, which in our case is at either $\lambda = \frac{(1-p)n_1}{p(b-n_1)}$ or $\lambda=1$.
Therefore, we only need to prove statement~(\ref{to-prove:function-of-lambda}) at these extreme two values of $\lambda$.
The former case, $\lambda = \frac{(1-p)n_1}{p(b-n_1)}$, is covered in~(\ref{to-prove:small-lambda}).
Thus, we only need to prove~(\ref{to-prove:function-of-lambda}) for $\lambda = 1$.
When $\lambda =1$, (\ref{to-prove:function-of-lambda}) can easily be rearranged as
\begin{align*}
\frac{ab}{1-a}+n_1 - \left( p+\frac{1-p}{2-a} \right) \left(an_2+n_1+\frac{a^2b}{1-a}+\frac{an_1}{p }\right) \geq 0.
\end{align*}
Because the left hand side of the inequality above is a decreasing function of $n_2$, using~(\ref{assumption:n_1+n_2<=b}), we only need to prove that the inequality holds when $n_2=b-n_1$, which is equivalent to
\begin{align*}
\frac{ab}{1-a}+n_1 - \left( p+\frac{1-p}{2-a} \right) \left(a(b-n_1)+n_1+\frac{a^2b}{1-a}+\frac{an_1}{p }\right) \geq 0.
\end{align*}
By separating terms associated with $n_1$ and $b$, and using $$1-\left( p+\frac{1-p}{2-a} \right)=\frac{(1-p)(1-a)}{2-a},$$
the above inequality is equivalent to
\begin{align*}
\frac{a(1-p)}{2-a}b + \frac{(1-p)(1-a)}{2-a}n_1 - \left( p+\frac{1-p}{2-a} \right) \left(\frac{a(1-p)}{p }\right)n_1 \geq 0.
\end{align*}
Using the lower bound of $b$ given by~(\ref{lower-bound-b}), and then dividing $(1-p)n_1$ on both sides (in the degenerated case where $(1-p)n_1=0$, both sides are $0$ so we are done), the above inequality is implied by
\begin{align*}
\frac{a}{(2-a)p} + \frac{1-a}{2-a} - \left( p+\frac{1-p}{2-a} \right) \left(\frac{a}{p }\right)\geq 0.
\end{align*}
Canceling the two terms involving $1/p$, the above inequality is equivalent to
\begin{align*}
\frac{(1-a)^2}{2-a} \geq 0,
\end{align*}
so we are done.

\textbf{case $b=n$:}

Now let us prove the $b=n$ case.
For proving this case, we can relax some of the constraints and only keep Constraints~(\ref{constraint:not-c-competitive}),~(\ref{constraint:x<=1}),~(\ref{constraint:eta_1+eta_2<=lambdan}) and~(\ref{constraint:n1+n2<=n}) and show that for every point in this superset of the feasibility region, $c\geq 1$.

According to Constraint~(\ref{constraint:not-c-competitive}) and using $a\min\{n_1+n_2, b\}+(1-a)n_1 = a\min\{n_2, b-n_1\}+n_1 $, it suffices to prove
\begin{align}
\label{ineq:inter:1}
\frac{a(n_2-\tilde o_2+\frac{b}{1-a})+n_1}{a\min\{n_2, b-n_1\}+n_1+\frac{a^2b}{1-a}+a \tilde u_1} \geq 1.
\end{align}
\vmn{or equivalently,}

\begin{align}
\label{ineq:inter:2}
{a(n_2-\tilde o_2+\frac{b}{1-a})+n_1} \geq {a\min\{n_2, b-n_1\}+n_1+\frac{a^2b}{1-a}+a \tilde u_1}.
\end{align}

\st{Multiplying both sides by the denominator of the left hand side,} Using $\min\{n_2, b-n_1\}\leq n_2$, subtracting $an_2+n_1$ on both sides \vmn{of \eqref{ineq:inter:2}}, and dividing both sides by $a$, \st{the above}  Inequality \vmn{\eqref{ineq:inter:1}} is implied by
\begin{align*}
-\tilde o_2+\frac{b}{1-a} \geq \frac{ab}{1-a}+ \tilde u_1.
\end{align*}
Subtracting $\frac{ab}{1-a}$ on both sides and using $\tilde u_1 \leq \frac{\tilde o_1 + (1-\lambda ) (1-p)n}{(1-p+\lambda p)}$, the above inequality is implied by
\begin{align*}
b -\tilde o_2 \geq \frac{\tilde o_1 + (1-\lambda ) (1-p)n}{(1-p+\lambda p)}.
\end{align*}
Multiplying $1-p+\lambda p$ on both sides and using $b=n$, the above inequality is equivalent to
\begin{align*}
\lambda n- (1-p+\lambda p) \tilde o_2 \geq \tilde o_1.
\end{align*}
Due to Constraint~(\ref{constraint:x<=1}), $1-p+\lambda p \leq 1$, and thus the inequality above is implied by
\begin{align*}
\lambda n \geq \tilde o_2 + \tilde o_1, \end{align*}
or equivalently,
\begin{align*}
\lambda n \geq (1-p)\eta_2+p n_2 \lambda + (1-p)\eta_1+p n_1 \lambda.
\end{align*}
This follows straightforwardly from Constraints~(\ref{constraint:eta_1+eta_2<=lambdan}) and~(\ref{constraint:n1+n2<=n}).
\end{proof}

\fi

\begin{proof}{\textbf{Proof of Lemma~\ref{lemma:upper-bound-q2-ADP-case-1}:}}
The only interesting case is \vmn{case (b), i.e., }when $n_1+n_2 > b+ \frac{2\Delta}{\delta}$.
If $q_2(1)=0$, then we are done.
Otherwise, let $\bar\lambda$ be the last time we accept a \st{class}\vmn{type}-$2$ customer.
By Lemma~\ref{prop:full-Ubounds}, $u_{1,2}(\bar\lambda) \geq \min \left\{ b, n_1 + n_2 -\frac{2\Delta}{\delta p}\right\} = b$.
Therefore, according to the definition of $\bar\lambda$, Condition~\eqref{eq:Ccomp2} must be satisfied.
Thus,
\begin{align*}
q_2(1) = & q_2 ( \bar \lambda) \\ \leq & \frac{1-c}{1-a} b + c \left(b - u_1( \bar \lambda) \right)^+ {+1} &(\text{Condition~\eqref{eq:Ccomp2}}) \\
\leq & \frac{1-c}{1-a} b + c \left(b - \min \left\{ b, n_1-\frac{2\Delta}{\delta p}\right\} \right)^+ {+1} &(\text{Lemma~\ref{prop:full-Ubounds}}) \\
\leq & \frac{1-c}{1-a} b + c \left(b - n_1 \right)^+ + c \frac{2\Delta}{\delta p} {+1}.
\end{align*}
\end{proof}

\begin{proof}{\textbf{Proof of Lemma~\ref{claim:q1+q2=b-ADP-ratio}:}}
We consider the two cases of Lemma~\ref{lemma:upper-bound-q2-ADP-case-1} separately.
For case \vmn{(a)}, $n_1+n_2 \leq b + \frac{2\Delta}{\delta p}$, we note that
\begin{align*}
OPT(\vec{v}) \leq & n_1+ n_2 a \\
\leq & \left( b+ \frac{2\Delta}{\delta p} -n_2 \right) + n_2 a &(n_1+n_2 \leq b+ \frac{2\Delta}{\delta p}) \\
\leq & ALG_{2,c} ({\vec{v}})+ \frac{2\Delta}{\delta p} .&(ALG_{2,c} ({\vec{v}})\geq (b-n_2)+an_2)
\end{align*}
Therefore,
\begin{align*}
\frac{ALG_{2,c}({\vec{v}})}{OPT(\vec{v})}\geq & \frac{ALG_{2,c}({\vec{v}})}{ ALG_{2,c}({\vec{v}}) + \frac{2\Delta}{\delta p}} \geq\frac{ALG_{2,c}({\vec{v}}) - \frac{2\Delta}{\delta p}}{ ALG_{2,c} ({\vec{v}})} &((ALG_{2,c} ({\vec{v}}))^2 \geq (ALG_{2,c}({\vec{v}}) )^2 - \left(\frac{2\Delta}{\delta p}\right)^2) \\
= & 1- \frac{\frac{2\Delta}{\delta p}}{ALG_{2,c}({\vec{v}})} \geq 1 -\frac{2\Delta}{ab\delta p} \geq c- \frac{2\Delta}{ab\delta p}. &(ALG_{2,c} ({\vec{v}})\geq ab)
\end{align*}
For case \vmn{(b)}, $n_1+n_2 > b + \frac{2\Delta}{\delta p}$ and $q_2(1) \leq \frac{1-c}{1-a} b + c \left(b - n_1 \right)^+ + c \frac{2\Delta}{\delta p} {+1}$, we have
\begin{align*}
\frac{ALG_{2,c}({\vec{v}})}{OPT(\vec{v})} = & \frac{b - (1-a)q_2(1)}{OPT(\vec{v})} &(q_1(1)+q_2(1)=b)\\
\geq & \frac{b - (1-a)\left( \frac{1-c}{1-a} b + c \left(b - n_1 \right)^+ + c \frac{2\Delta}{\delta p} {+1}\right) }{OPT(\vec{v})} &(\vmn{q_2(1) \leq \frac{1-c}{1-a} b + c \left(b - n_1 \right)^+ + c \frac{2\Delta}{\delta p} {+1}}) \\
= & \frac{c(b-(1-a)(b-n_1)^+)}{OPT(\vec{v})}  - \frac{(1-a)c\frac{2\Delta}{\delta p}}{OPT(\vec{v}) } {- \frac{1-a}{OPT(\vec{v})} }\\
\geq & c - \frac{(1-a)c\frac{2\Delta}{\delta p}}{OPT(\vec{v}) } {- \frac{1-a}{OPT(\vec{v})} } &(OPT(\vec{v}) \leq b-(1-a)(b-n_1)^+ ) \\
\geq & c-\frac{2(1-a)c\Delta}{ab\delta p} {- \frac{1-a}{ab} }  &(OPT(\vec{v}) \geq ab) \\
\geq & c - \frac{{3}\Delta}{ab\delta p}. &(1-a <1, c \leq 1, \delta \leq 1, p<1)
\end{align*}
\end{proof}

\begin{proof}{\textbf{Proof of Lemma~\ref{lemma:feasible-sol}:}}
Let us define $\tilde o_1', \tilde o_2', \tilde u_1'$ and $\tilde u_{1,2}'$ to be the \vmn{corresponding functions defined for} \st{values of $\tilde o_1, \tilde o_2, \tilde u_1$ and $\tilde u_{1,2}$ corresponding to}the modified tuple $(\vmn{l}', n_1', n_2', \eta_1', \eta_2', c')$.
It is easy to check that $(\vmn{\l}', n_1', n_2', \eta_1', \eta_2', c')$ satisfies Constraints~\eqref{constraint:x<=1}-\eqref{constraint:n1'+n2'big}.
The interesting part is to show that it satisfies Constraint~\eqref{constraint:u_2>=b}.
When $n_1+n_2\geq b$, we can prove it directly from Lemma~\ref{prop:Ubounds} (since $\tilde u_{1,2}'\geq n_1'+n_2' = n_1+n_2\geq b$).

\vmn{Next we focus on }\st{Below we consider}the case $n_1+n_2<b$\st{For the case $n_1+n_2<b$}, \vmn{and} we  prove $\tilde u_{1,2}' \geq b$ by showing $\tilde u_{1,2}' \geq u_{1,2}(\vmn{\l})$; \vmn{note that we have $u_{1,2}(\vmn{\l}) \geq b$ by Inequality \eqref{ineq:u2>=b}}.  %(because we have \eqref{ineq:u2>=b}).

Recall that we reject a customer at time $\vmn{\l}$ and that the threshold of rejecting a customer is at least $\phi b$; \vmn{thus} we have $\vmn{\l} n  \geq o_2(\vmn{\l}) \geq \phi b $.
This gives \begin{align}
\vmn{\l} \geq \frac{\phi b}{n}= \delta \label{ineq:big-bar-lambda}
\end{align}
%(which we will show at the end of the proof).
\vmn{Note that by definition for $\vmn{\l} \geq \delta$, we have $u_{1,2}(\vmn{\l})  = \min \left\{ \frac{o_1(\vmn{\l})+o_2(\vmn{\l})}{\vmn{\l} p}, \frac{o_1(\vmn{\l}) +o_2(\vmn{\l}) + (1-\vmn{\l}) (1-p) n}{1-p+\vmn{\l} p} \right\}$, which is a non-decreasing function of $o_1(\vmn{\l})+o_2(\vmn{\l})$.}
Thus, $\tilde u_{1,2}' \geq u_{1,2}(\vmn{\l}) $ is implied by $\tilde o_1' + \tilde o_2 ' \geq o_1(\vmn{\l})+o_2(\vmn{\l})$.
We prove this by breaking down into two cases based on the value of $\xi$: \st{For the first case,}\vmn{Case (1)} $\xi = n-(n_1+n_2)$: we have $n_1'+n_2'=n$\vmn{, i.e., there is no time period without a customer.} \vmn{Thus} $\eta_1'+\eta_2' = \vmn{\l} n$, and  $\tilde o_1' + \tilde o_2 ' = \vmn{\l} n \geq o_1(\vmn{\l})+o_2(\vmn{\l})$.
\st{For the second case,}\vmn{Case (2)} $\xi = \frac{\Delta n }{\phi b p}$, we have
\begin{align*}
\tilde o_1' + \tilde o_2' = & \vmn{\l}' p n_1' + (1-p ) \eta_1' +\vmn{\l}' p n_2' + (1-p ) \eta_2' \\
\geq & \vmn{\l} p n_1 + (1-p ) \eta_1 +\vmn{\l} p (n_2 + \frac{\Delta n }{\phi b p}) + (1-p ) \eta_2 &( \xi = \frac{\Delta n }{\phi b p}, \bar \xi \geq 0)\\
= & \tilde o_1(\vmn{\l}) + \tilde o_2(\vmn{\l})+\frac{\Delta n }{\phi b }\vmn{\l}
\geq \tilde o_1(\vmn{\l}) + \tilde o_2(\vmn{\l}) +\Delta &(\eqref{ineq:big-bar-lambda}) \\
\geq & o_1(\vmn{\l})+o_2(\vmn{\l}) &(\eqref{inequality:good-approximation-o_1+o_2})
\end{align*}

\end{proof}

\begin{proof}{\textbf{Proof of Lemma~\ref{claim:ADP-non-exahust}:}}
We first show that, for all $c\leq c^*$, Constraint~\eqref{constraint:not-c-competitive} (same as \eqref{eq:Ccomp6}) is either violated or holds with equality.
First, we note that, for all real number $x$, the tuple
$(\vmn{\l}', n_1', n_2', \eta_1', \eta_2', x)$ satisfies Constraints~\eqref{constraint:u_2>=b}-\eqref{constraint:n1'+n2'big} because those constraints are not related to the last element in the tuple.
For all $x<c^*$, $(\vmn{\l}', n_1', n_2', \eta_1', \eta_2', x)$ is not in the feasible set of (\ref{MP1}), and hence Constraint~\eqref{constraint:not-c-competitive} is violated.
Taking the limit $x\to c^*$, Constraint~\eqref{constraint:not-c-competitive} is either violated or hold with equality.
This means, for $ALG_{2,c}$ (with any $ c \leq c^*$),
\begin{align}
c = c' \leq \frac{a(n_2' - \tilde o_2'+\frac{b}{1-a})+n_1'}{a\min\{n_1'+n_2' , b\}+(1-a)n_1'+\frac{a^2b}{1-a}+a \min\{ \tilde u_1', b\}}.\label{ineq:good-ratio}
\end{align}
After rearranging terms
%(similar to that Inequality~\eqref{eq:Ccomp6} is equivalent to the right hand side of~\eqref{ineq:lower-bound-ratio-ADP} being smaller than or equal to $c$),
~\eqref{ineq:good-ratio} is equivalent to
\begin{align}
\label{ineq:key-ADP-full}\frac{n_1' + a\left(\frac{1-c}{1-a} b + c \left(b - \tilde u_1' \right)^+ + \left[n_2' - \tilde o_2'\right]\right) }{n_1' + a \min \{b-n_1',n_2 '\}} \geq c. ~~~~~~~~(n_1 \geq \frac{k}{p^2} \log n)
\end{align}

\noindent \vmn{Repeating \eqref{ineq:lower-bound-ratio-ADP}, recall that we have:}
\vmn{
\begin{align*}
%\label{ineq:lower-bound-ratio-ADP}
\frac{ALG_{2,c}({\vec{v}})}{OPT(\vec{v})}
\geq \frac{n_1 + a\left(\frac{1-c}{1-a} b + c \left(b - u_1(\vmn{\l}) \right)^+ + \left[n_2 - o_2(\vmn{\l})\right]\right) }{n_1 + a \min \{b-n_1,n_2\}}.
\end{align*}
}

We want to compare the right hand side of~\eqref{ineq:lower-bound-ratio-ADP} with the left hand side of~\eqref{ineq:key-ADP-full}.
First, we compare $\left(b - \tilde u_1' \right)^+$ with $\left(b - u_1(\vmn{\l}) \right)^+$ and show
\begin{align}
\left(b - \tilde u_1' \right)^+ \leq \left(b - u_1(\vmn{\l}) \right)^+ + \xi.
\label{ineq:u_1-and-tilde-u_1}
\end{align}
Recall that we do not exhaust the inventory, and thus $n_1 < b$. \vmn{Further $\Delta = \alpha \sqrt{b \log n}$, thus we have: $\Delta \geq \alpha \sqrt{n_1 \log n}$.}
%\vm{Dawsen. where is this last sentence used?!}
\st{Therefore, a}According to \eqref{inequality:good-approximation-o_1}, $\tilde o_1 ' = \tilde o_1(\vmn{\l}) \geq o_1(\vmn{\l})- \Delta$.
Combining this and using an argument similar to the proof of Lemma~\ref{prop:full-Ubounds},
\begin{align*}
\tilde u_1' \triangleq & \min \left \{ \frac{\tilde o_1'}{\vmn{\l}' p}, \frac{ \tilde o_1' + (1-\vmn{\l}') (1-p) n}{1-p+\vmn{\l}' p} \right\}\\
= &
\min \left \{ \frac{\tilde o_1(\vmn{\l})}{\vmn{\l} p}, \frac{ \tilde o_1(\vmn{\l}) + (1-\vmn{\l}) (1-p) n}{1-p+\vmn{\l} p} \right\} \\
\geq & \min \left \{ \frac{ o_1(\vmn{\l}) - \Delta}{\vmn{\l} p}, \frac{ o_1(\vmn{\l}) - \Delta + (1-\vmn{\l}) (1-p) n}{1-p+\vmn{\l} p} \right\} &(\tilde o_1 ' \geq o_1(\vmn{\l})- \Delta) \\
\geq & u_1(\vmn{\l})- \max \left\{ \frac{\Delta}{\vmn{\l} p} , \frac{\Delta}{1-p+\vmn{\l} p}\right\} = u_1(\vmn{\l})-\frac{\Delta}{\vmn{\l} p} \\
\geq & u_1(\vmn{\l}) - \frac{\Delta n}{\phi b p}. &(\vmn{\eqref{ineq:big-bar-lambda}})
\end{align*}

\vmn{Note that by the definition of $\xi$, the above inequality implies Inequality \eqref{ineq:u_1-and-tilde-u_1}.}
\vmn{Next,} we compare $\tilde o_2 ' $ with $o_2(\vmn{\l})$ and \st{obtain}\vmn{ we show that}
\begin{align}
\tilde o_2' \geq o_2(\vmn{\l}) -2\Delta. \label{ineq:o2'>o2(lambda)}
\end{align}
\vmn{In order to prove \eqref{ineq:o2'>o2(lambda)}, }
we first show that we can assume, without loss of generality, $\phi b < n_2 \leq b+ \frac{2\Delta}{\delta p}$.
To see this, we note that when $q_1(1)+q_2(1) < b$, $q_1(1)=n_1$.
Therefore, the only ``mistakes'' that the algorithm may make is to reject too many \st{class}\vmn{type}-$2$ customers.
When $n_2 \leq \phi b$, we never reject a \st{class}\vmn{type}-$2$ customer and so $q_2(1)=n_2$ and $ALG_{2,c}({\vec{v}}) = OPT(\vec{v})$.
For proving the upper bound on $n_2$, \vmn{i.e., $n_2 \leq b+ \frac{2\Delta}{\delta p}$,} we first note that, clearly, if $n_2 > b+ \frac{2\Delta}{\delta p}$, decreasing $n_2$ to $b+ \frac{2\Delta}{\delta p}$ (while fixing $n_1$) does not modify the optimal revenue $OPT(\vec{v})$.
Using Lemma~\ref{prop:full-Ubounds}, we know that, when $n_2 \geq b+ \frac{2\Delta}{\delta p}$, $u_{1,2}(\vmn{\l}) \geq \vmn{\min \left\{ b, n_1 + n_2 -\frac{2\Delta}{\delta p}\right\} = } b$.
Therefore, we accept a \st{class}\vmn{type}-$2$ customer arriving at time $\vmn{\l}$ only if the number of \vmn{type}-$2$ customer accepted so far does not reach the dynamic threshold \vmn{(i.e., the third rule in Algorithm~\ref{algorithm:adaptive-threshold})} that depends only on $o_1(\vmn{\l})$ but not \vmn{on} $o_2(\vmn{\l})$.
\vmn{Given all the above, similar to the proof of Lemma~\ref{lem:HYB:full-case2}, we can construct an alternative adversarial instance where we reduce the number of type-$2$
customers to $b+ \frac{2\Delta}{\delta p}$, and show that, for the same realization of the {\RG} group and random permutation,
the number of accepted type-$2$ customers in the alternative instance serves as a lower bound on its counterpart in the original instance.
}
\if false
Thus, we can apply the argument in the proof of Lemma~\ref{lem:HYB:full-case2} to show that when $n_2$ increases, because there are more \st{class}\vmn{type}-$2$ customers, we will accept more \st{class}\vmn{type}-$2$ customers.
\fi

\vmn{Next we show that condition $n_2 \geq \frac{ k \log n}{p^2}$ is satisfied which implies we can apply concentration result~\eqref{inequality:good-approximation-o_2} from Lemma~\ref{lemma:needed-centrality-result-for-m=2}. Because  $n_2>\phi b$, it suffices to show:}
%Note that we can apply Inequalities~\eqref{inequality:good-approximation-o_2} if
%$n_2 \geq \frac{ k \log n}{p^2},$ which, due to $n_2>\phi b$, is implied by
\begin{align}
\phi b \geq \frac{k}{p^2} \log n. \label{epsilon-phi-b-condition}
\end{align}
Inequality \eqref{ineq:ADP-non-trivial-case} and
\begin{align}\vmn{\bar\epsilon}\leq \frac{1}{\sqrt{k}} \label{ineq:k'-condition-ADP}
\end{align}
\vmn{which holds by the definition of constants $k$ and $\bar\epsilon$ given  in Lemma~\ref{lemma:needed-centrality-result-for-m=2}, }implies
$$\sqrt{b} = \frac{b^{\frac{3}{2}}}{b} \geq \frac{b^{\frac{3}{2}}}{n} > \frac{1}{\vmn{\bar\epsilon}}\frac{\sqrt{\log n}}{(1-c)^2a^2p^{3/2}}\geq \sqrt{k}\frac{\sqrt{\log n}}{p\sqrt{1-c}}.$$
This, together with $\phi =\frac{1-c}{1-a} \geq 1-c$, \st{gives}\vmn{proves} \eqref{epsilon-phi-b-condition}. \vmn{Thus, we can apply~\eqref{inequality:good-approximation-o_2}.}
\vmn{Further note that \eqref{epsilon:ADP-small-Delta} implies that $n_2\leq b+ \frac{2\Delta}{\delta p} \leq 4b$. Finally note that by definition $\xi\geq 0$, $\xi' \geq 0$. Putting all these together, we have:}
%Using \eqref{epsilon:ADP-small-Delta}, $n_2\leq b+ \frac{2\Delta}{\delta p} \leq 4b$.
%Therefore, $\xi\geq 0$, $\xi' \geq 0$ and Inequality~\eqref{inequality:good-approximation-o_2} give,
\begin{align*}
\tilde o_2' \geq \tilde o_2(\vmn{\l}) \geq o_2(\vmn{\l}) - \alpha \sqrt{n_2 \log n} \geq o_2(\vmn{\l}) - \alpha \sqrt{4b \log n} = o_2(\vmn{\l}) -2\Delta. %\label{ineq:o2'>o2(lambda)}
\end{align*}

\vmn{This proves \eqref{ineq:o2'>o2(lambda)}. Having proved \eqref{ineq:u_1-and-tilde-u_1} and \eqref{ineq:o2'>o2(lambda)},} at last, \vmn{we complete the proof as follows:}
\begin{align*}
c \leq & \frac{n_1' + a\left(\frac{1-c}{1-a} b + c \left(b - \tilde u_1' \right)^+ + \left[n_2' - \tilde o_2'\right]\right) }{n_1' + a \min \{b-n_1',n_2 '\}} &(\eqref{ineq:key-ADP-full}) \\
\leq & \frac{n_1 + a\left(\frac{1-c}{1-a} b + c \left[ \left(b - u_1(\vmn{\l}) \right)^+ + \frac{\Delta n}{\phi b p} \right] + \left[n_2 + \xi - o_2(\vmn{\l})+2\Delta\right]\right) }{n_1 + a \min \{b-n_1,n_2 \}} &(\eqref{ineq:u_1-and-tilde-u_1}, \eqref{ineq:o2'>o2(lambda)}, n_2' = n_2+\xi \geq n_2) \\
\leq & \frac{ALG_{2,c}({\vec{v}})}{OPT(\vec{v})} + \frac{a \left( c\frac{\Delta n}{\phi b p} + \xi + 2\Delta \right)}{OPT(\vec{v})} &(\eqref{ineq:lower-bound-ratio-ADP})\\
\leq & \frac{ALG_{2,c}({\vec{v}})}{OPT(\vec{v})} + \frac{4a\Delta n}{a\phi^2 b^2 p} =\frac{ALG_{2,c}({\vec{v}})}{OPT(\vec{v})} + \frac{4\Delta n}{\phi^2 b^2 p} &(n_2 > \phi b, \Delta \leq \frac{\Delta n}{\phi b p} , \xi \leq \frac{\Delta n}{\phi b p}, OPT \geq a \phi b).
\end{align*}

\end{proof}

\begin{proof}{\textbf{Proof of Lemma~\ref{lemma:n1-small-ADP-full}:}}
Note that \vmn{if we are not in case (a), i.e., $q_1(1)+q_2(1)<b$, then $q_1(1)=n_1$.}
\vmn{Now either $q_2(1) = n_2$, i.e., we are in case (b), or $q_2(1) < n_2$.}
%$q_1(1)=n_1$ when $q_1(1)+q_2(1)<b$.
Therefore, what is remaining is to show that if $q_1(1)+q_2(1)<b$ and $q_2(1)<n_2$, then $q_2(1) \geq cb$ \vmn{, i.e., we are in case (c)}.

Let $\bar \lambda$ be the last time \st{when}\vmn{that} a customer is rejected.
Then, similar to earlier discussion, Inequality~\eqref{ineq:big-bar-lambda} is satisfied.
Therefore,
\begin{align}
u_1(\bar\lambda) = &\min \left \{ \frac{o_1(\bar \lambda)}{\bar \lambda p}, \frac{ o_1(\bar \lambda) + (1-\bar \lambda) (1-p) n}{1-p+\bar \lambda p} \right\} &(\bar \lambda \geq \delta)\nonumber \\
\leq & \frac{o_1(\bar \lambda)}{\bar \lambda p} \nonumber \\
\leq & \frac{n_1 n}{\phi b p} & (\eqref{ineq:big-bar-lambda}\text{ and }o_1(\bar\lambda)\leq n_1) \nonumber \\
< & \frac{n \frac{k}{p^2} \log n }{\phi b p}=\frac{kn \log n}{\phi b p^3}. &(n_1< \frac{k}{p^2} \log n) \nonumber % \label{ineq:ub-u1-ADP-small-n1}
\end{align}
As a result, \vmn{we have:}
\vmn{
\begin{align*}
q_2(1) \vmn{\geq} \phi b + c(b-u_1(\bar\lambda))^+ > (\phi + c) b- c\frac{kn \log n}{\phi b p^3},
\end{align*}}
\vmn{In order to complete the proof, it suffices to show that }
\vmn{
\begin{align}
\label{ineq:lemma:14:1}
q_2(1) > (\phi + c) b- c\frac{kn \log n}{\phi b p^3} \geq cb,
\end{align}}
\vmn{The last inequality in \eqref{ineq:lemma:14:1} holds if $ b^2 \geq c \frac{kn \log n}{\phi ^2 p^3}$. Thus in the following, we show $ b^2 \geq c \frac{kn \log n}{\phi ^2 p^3}$: }
%which is greater than $cb$ when
%$ b^2 \geq c \frac{kn \log n}{\phi ^2 p^3}.$
Using $\phi = \frac{1-c}{1-a} \geq 1-c$, Inequality \eqref{ineq:ADP-non-trivial-case}, and
\begin{align}\vmn{\bar\epsilon}\leq \frac{1}{\sqrt[4]{k}} \label{ineq:k'-condition-ADP<=k4},
\end{align}
\vmn{which holds by the definitions of constants $k$ and $\bar\epsilon$ given in Lemma~\ref{lemma:needed-centrality-result-for-m=2}, }we have $$b^2 = \frac{\left( b^{\frac{3}{2}}\right)^2 }{b} \geq \frac{\left( b^{\frac{3}{2}}\right)^2 }{n} > \frac{1}{\vmn{\bar\epsilon}^4}\frac{n\log n}{(1-c)^4a^2p^3}\geq c \frac{kn \log n}{\phi ^2 p^3}.$$
\st{Therefore,}\vmn{This proves} $ b^2 \geq c \frac{kn \log n}{\phi ^2 p^3}$, and thus $q_2(1) \geq cb$. \vmn{This completes the proof of the lemma.}
\end{proof}

\begin{proof}{\textbf{Proof of Lemma~\ref{claim:small-n1}:}}
We consider each case of Lemma~\ref{lemma:n1-small-ADP-full} separately.
For \st{the first}case \vmn{(a)}, $q_1(1)+q_2(1)=b$, since $n_1 < \frac{k}{p^2} \log n$, it is easy to see that
\begin{align*}
\frac{ALG_{2,c}({\vec{v}})}{OPT(\vec{v})} \geq \frac{ab}{ab+\frac{k}{p^2} \log n} \geq \frac{ab - \frac{k}{p^2} \log n}{ab} = 1- \frac{k \log n}{ab p^2},
\end{align*}
which is at least $c$ if $b \geq \frac{k \log n}{a(1-c) p^2}$.
Inequality \eqref{ineq:ADP-non-trivial-case} and \eqref{ineq:k'-condition-ADP} imply
$$\sqrt{b} = \frac{b^{\frac{3}{2}}}{b} \geq \frac{b^{\frac{3}{2}}}{n} > \frac{1}{\vmn{\bar\epsilon}}\frac{\sqrt{\log n}}{(1-c)^2a^2p^{3/2}}\geq \sqrt{k}\frac{\sqrt{\log n}}{p\sqrt{a(1-c)}},$$
and thus $b \geq \frac{k \log n}{a(1-c) p^2}$; \vmn{therefore we have:} $\frac{ALG_{2,c}({\vec{v}})}{OPT(\vec{v})} \geq c$.

\st{For the second and third}\vmn{In cases (b) and (c), }$q_2(1) \geq \min \{n_2, \vmn{c}b\}$; \vmn{thus} we have
$$ ALG_{2,c}({\vec{v}}) \geq n_1+c(\min \{n_2, b\})a \geq c (n_1+\min \{n_2, b\}a) \geq c OPT(\vec{v}).$$
\end{proof}

\if false
\begin{proof}[Proof of Claim~\ref{claim:check-epsilon}]
We have $\epsilon = \frac{1}{(1-c)^2ap^{3/2}}\sqrt{\frac{n^2 \log n}{b^3}} \geq \sqrt{\frac{1}{b}} \geq \frac{1}{n}$.
In addition, $\frac{1}{(1-c)^2ap^{3/2}}\sqrt{\frac{n^2 \log n}{b^3}} \leq \bar\epsilon $, is implied by \eqref{ineq:ADP-non-trivial-case} and \begin{align} \vmn{\vm{\tilde{k}}} \leq \bar\epsilon. \label{condition:ALG2:k'<=bar-epsilon}
\end{align}
Therefore, the claim statement holds.
%Therefore, the claim statement Condition~\eqref{epsilon-ADP:conditions-lemma-1} holds.
\end{proof}
\fi

\begin{proof}{\textbf{Proof of Lemma~\ref{claim:check-error-terms}:}}
Both follow from definition.
\end{proof}

\begin{proof}{\textbf{Proof of Corollary~\ref{cor:cstar1}:}}
Theorem~\ref{thm:adaptive-threshold} with $c=1- \sqrt[3]{\frac{1}{ap^{3/2}}\sqrt{\frac{n^2 \log n}{b^3}}}$ proves the corollary.
\end{proof}

\section{Missing proofs of Section~\ref{sec:model:discussion}}
\label{sec:app:model:discussion}

\begin{proof}{\textbf{Proof of Proposition~\ref{thm:impossibility}:}}
\st{Here w}We prove that the competitive ratio of any online algorithm is at most $p+\frac{1-p}{2-a}+3\left(\frac{pb^2}{n}\right)$.
\vmn{Note that }when $\frac{pb^2}{n}>1/2$, $p+\frac{1-p}{2-a}+3\left(\frac{pb^2}{n}\right)$ is greater than $1$ and hence \vmn{the upper bound trivially holds}\st{we are done}.
\vmn{Thus} in the following, we assume, without loss of generality, $\frac{pb^2}{n}\leq 1/2$.

We consider two adversarial instances ${\vec{v}_{I}}$ and ${\vec{w}_{I}}$ defined \st{by}\vmn{as}

\begin{equation*}
v_{I,j} = \begin{cases}
a, \qquad & 1 \leq j \leq b, \\
0, \qquad & b < j \leq 2b, \\
0, \qquad & j > 2b.
\end{cases} ~~~~~
w_{I,j} = \begin{cases}
a, \qquad & 1 \leq j \leq b, \\
1, \qquad & b < j \leq 2b, \\
0, \qquad & j > 2b.
\end{cases}
\end{equation*}
\if false
\begin{equation*}
v'_i = \left\{
\begin{array}{lll}
a, \qquad & i = 1, 2, \ldots, b, \\
1, \qquad & i = b+1, b+2, \ldots, 2b \text{ for } {\vec{v_2'}} ,\\
0, \qquad & i = b+1, b+2, \ldots, 2b \text{ for } {\vec{v_1'}} , \\
0, \qquad & i > 2b.
\end{array}
\right.
\end{equation*}
\fi
Let \st{us denote $\mathcal E$}\vmn{$\mathcal U$} \vmn{denote} \st{to be}\vmn{the} event in which
\vmn{in the arrival sequence, none of the first $b$ arrivals belongs to positions $[b+1, 2b]$ in the {\em initial} customer sequence, i.e.,
  for all $i \in [b]$, we have: $i \notin \RGS$  or  $\sigma^{-1}_\RGS(i) \notin [b+1,2b]$, where we use the following definition: For $x, y \in \mathbb{N}$ and $x < y$, $[x,y] \triangleq  \{x, x+1, \ldots,y\}$. Further, we define $[y] \triangleq [1,y]$.}
\vmn{Note that under event $\mathcal U$, no online algorithms can distinguish whether the initial sequence is \vmn{${\vec{v}_{I}}$} or \vmn{${\vec{w}_{I}}$} \vmn{up to time $b/n$}.}
\vmn{We first compute the probability of event $\mathcal U$ as follows:}

\if false
the realization of the {\RG} group and the random permutation is in a way such that the online algorithms cannot distinguish whether the adversarial sequence is \vmn{${\vec{v}_{I}}$} or \vmn{${\vec{w}_{I}}$} \vmn{up to time $b/n$}.
Note that event $\mathcal{E}$ occurs when none of the customers intended to appear in positions $b+1, b+2, \ldots, 2b$ shows up before time $b/n$.
It occurs with probability
\fi
\begin{align}
\prob{\st{\mathcal E}\vmn{\mathcal U}} &= \prob{\text{for all } i \in [b]: i \notin \RGS \text{ or } \sigma^{-1}_\RGS(i) \notin [b+1,2b] } \nonumber \\
&\geq 1- \sum_{i \in [b]} \prob{ i \in \RGS \text{ and } \sigma_\RGS^{-1}(i) \in [b+1, 2b] } &(\text{Union bound}) \nonumber \\
&\geq 1-\frac{pb^2}{n},\label{eq:event:U}
\end{align}
where \st{we have used}\vmn{the last inequality holds because of} the following inequality \vmn{(which we prove next)}\st{in the last inequality}: For all $i\neq j$,
\begin{align}\label{observation:small-prob-go-to-other-location}
\prob{i \in \RGS \text{ and } \sigma_{\RGS}^{-1}(i)=j}\leq \frac{p}{n}.
\end{align}
To prove \eqref{observation:small-prob-go-to-other-location}, we first note for any $i$, we have $p = \prob{i \in \RGS} = \sum_{j=1}^n \prob{i \in \RGS \text{ and } \sigma_{\RGS}^{-1}(i)=j}$.
Second, denoting $R$ the random variable corresponding to the size of the {\RG} group, we have $\prob{ \sigma_{\RGS}^{-1}(i)=i | i \in \RGS , R}=\frac{1}{R} \geq \frac{1}{n}$, and thus $\prob{ \sigma_{\RGS}^{-1}(i)=i | i \in \RGS} \geq \frac{1}{n}$.
Therefore,
\begin{align*}
& \sum_{j\neq i} \prob{i \in \RGS \text{ and } \sigma_{\RGS}^{-1}(i)=j} = p - \prob{i \in \RGS \text{ and } \sigma_{\RGS}^{-1}(i)=i} \\
= & p - \prob{ \sigma_{\RGS}^{-1}(i)=i | i \in \RGS}\prob{ i \in \RGS} \leq p-\frac{p}{n}=\frac{(n-1)p}{n}.
\end{align*}
By symmetry, for each $j \neq i$, $\prob{i \in \RGS \text{ and } \sigma_{\RGS}^{-1}(i)=j}\leq \frac{p}{n}$, which proves \eqref{observation:small-prob-go-to-other-location}.
\vmn{This completes our proof of inequality \eqref{eq:event:U}.}

Under the event \st{$\mathcal{E}$}\vmn{$\mathcal{U}$}, \vmn{in both problem instances,} the \st{value}\vmn{revenue} of each customer accepted \st{by}\vmn{up to} time $b/n$ is\st{necessarily} $a$.
\vmn{Conditioned on \vmn{event} \st{$\mathcal{E}$}\vmn{$\mathcal{U}$}, }we denote $q_2$ the expected number of accepted \st{class}\vmn{type}-$2$ customers up to time $b/n$ \vmn{under either problem instances---recall that under event $\mathcal{U}$, up to time $b/n$, the online algorithm cannot distinguish the two}.

\vmn{We now proceed to find an upper bound on the expected revenue of any online algorithm under the two problem instances. We start by ${\vec{w}_{I}}$:}
%We first find an upper bound on the expected revenue for any online algorithm $ALG$ under ${\vec{v_2'}}$,
\begin{align*}
\E{ALG({\vec{W}})}
&\leq \E{ALG(\vec{W}) \,\middle|\, \st{\mathcal{E}}\vmn{\mathcal{U}} } \prob{\st{\mathcal{E}}\vmn{\mathcal{U}}} + OPT({\vec{w}_{I}}) \left(1- \prob{\st{\mathcal{E}}\vmn{\mathcal{U}}} \right) \\
&\leq q_2 a + (b-q_2) + \frac{pb^2}{n} OPT({\vec{w}_{I}}).
\end{align*}

\vmn{Next, we proceed to establish an upper bound on the expected revenue under ${\vec{v}_{I}}$, by
proving an upper bound on the number of type-$2$ customers that arrive after time $b/n$ conditioned on the event $\vmn{\mathcal{U}}$:}

\if false
Before computing the total revenue for ${\vec{v_1'}}$, we first turn our attention to the customers arriving after time $b/n$ in ${\vec{V}}_1$.
Note that, under $\mathcal E$, the expected number of \st{class}\vmn{type}-$2$ customers that arrive after $b/n$ is upper-bounded as follows
\fi
\begin{align*}
\E{\left|\left\{ i \geq b+1 \,\middle|\, V_i = a \right\} \right| \,\middle|\, \st{\mathcal{E}}\vmn{\mathcal{U}}} &=\sum_{i=b+1}^n \prob{ V_i = a \,\middle|\,\st{\mathcal{E}}\vmn{\mathcal{U}}} \\
& \leq \frac{\sum_{i=b+1}^n \prob{ i \in \RGS \text{ and } \sigma_\RGS^{-1} (i) \in [b] }}{\prob{\st{\mathcal{E}}\vmn{\mathcal{U}}}} \\
&\leq \frac{(n-b) b\frac{p}{n}}{1-\frac{pb^2}{n}} &(\text{Inequalities}~\eqref{eq:event:U},~\eqref{observation:small-prob-go-to-other-location}) \\
& \leq \frac{pb}{1-\frac{pb^2}{n}}
\leq pb \left( 1 + 2\left(\frac{pb^2}{n}\right) \right),
\end{align*}
where we use $\frac{pb^2}{n} \leq 1/2$ in the last inequality. {Note that $\E{ ALG({\vec{V}}) \,\middle|\,\st{\mathcal{E}}\vmn{\mathcal{U}}} \leq q_2 a + a \E{\left|\left\{ i \geq b+1 \,\middle|\, V_i = a \right\} \right| \,\middle|\, \st{\mathcal{E}}\vmn{\mathcal{U}}}$.} As a result,
\begin{align*}
\E {ALG({\vec{V}}) }
&\leq \E{ ALG({\vec{V}}) \,\middle|\,\st{\mathcal{E}}\vmn{\mathcal{U}}} \prob{\st{\mathcal{E}}\vmn{\mathcal{U}}} + OPT(\vec{v}_{I}) \left( 1 - \prob{\st{\mathcal{E}}\vmn{\mathcal{U}}} \right) \\
&\leq q_2 a + \left(1 + 2 \left( \frac{pb^2}{n} \right)\right) pba + \left( \frac{pb^2}{n} \right) OPT(\vec{v}_{I}) \\
&\leq q_2 a + pba + 3\left( \frac{pb^2}{n} \right) OPT(\vec{v}_{I}). &(OPT(\vec{v}_{I}) = ba \geq pba)
\end{align*}
Thus, the competitive ratio is at most
\begin{align*}
\min \left\{ \frac{\E{ALG({\vec{V}})}}{OPT({\vec{v}_{I}})}, \frac{\E{ALG({\vec{W}})}}{OPT(\vec{w}_{I})} \right\}
&\leq \min \left\{ \frac{q_2}{b} +p+ 3 \left( \frac{pb^2}{n} \right), \frac{q_2}{b} a + \left( 1-\frac{q_2}{b}+ \frac{pb^2}{n} \right) \right\} \nonumber \\
&\leq \min \left\{ \frac{q_2}{b} +p, \frac{q_2}{b} a + \left( 1-\frac{q_2}{b} \right) \right\} + 3 \left( \frac{pb^2}{n} \right) \nonumber \\
&\leq p + \frac{1-p}{2-a} + 3 \left( \frac{pb^2}{n} \right),
\end{align*}
where the last inequality holds \vmn{because function $g(q_2) \triangleq \min \left\{ \frac{q_2}{b} +p, \frac{q_2}{b} a + \left( 1-\frac{q_2}{b} \right) \right\}$---defined on $q_2 \in [0,b]$---achieves its maximum at $q^*_2 = \frac{1-p}{2-a} b$, and $g(q^*_2) \leq p + \frac{1-p}{2-a}$.} \st{for $q_2 = \frac{1-p}{2-a} b$.}
\if false
Note that our proof works even if we augment the power of online algorithms by allowing them to see whether an observed customer $v_i$ is from the {\RG} group or not and its intended location in the adversarial instance.
\fi
\end{proof}

\if false

\section{Asymptotic results based on approximate analysis}
\label{sec:app:approx}

\subsection{Upper bounds in Asymptotic Case}\label{sec:asymptotic-ub}
In this section, we assume $b,n \rightarrow \infty$ while $b/n = \kappa$ where $\kappa$ is a positive constant.
We have the following conjecture for this asymptotic case:

%For the asymptotic case where $b$ and $n$ go simultaneously while fixing the $b/n$ ratio, we have the following conjecture.
\begin{conjecture}\label{conj:ata-tightness}
Under the partially \st{learnable}\vmn{predictable} model, when  $ b= \kappa n$ where $\kappa$ is a positive constant and $n \rightarrow \infty$, no online algorithm, deterministic and randomized, can achieve a competitive ratio better than $c^*$, the optimal objective value  of (\ref{MP1}).
\end{conjecture}

If true, this conjecture  would imply the asymptotic tightness of the adaptive algorithm.
%This \vm{true, this implies } conjecture, if true, would imply the asymptotic tightness of the adaptive algorithm.
While rigorous proof of this conjecture is beyond the scope of this article, we show how to prove it under two approximations.

The first approximation allows partial customers to arrive in continuous time:
\begin{approximation}[Continuous Instance]\label{assumption:continuous instance}
A valid problem instance ${\vec{v'}\vmn{ZZ\vec{v}_I}}$ is specified by two non-decreasing functions $\eta_1, \eta_2:[0,1]\to [0,n]$ that represent the number of \st{class}\vmn{type}-$1$ (\st{class}\vmn{type}-$2$) customers arriving up to time $\lambda \in [0,1]$. These functions satisfy:
$\eta_1(0)=\eta_2(0)$, and for any $\lambda\in[0,1]$,
%$$\frac{{\rm d}\left(\eta_1(\lambda)+\eta_2(\lambda)\right)}{{\rm d}\lambda}\leq n,$$
$\frac{\partial \left[\eta_1(\lambda)+\eta_2(\lambda)\right]}{\partial \lambda n}\leq 1.$
We again denote $n_1 = \eta_1(1)$ and $n_2 = \eta_2(1)$ to be the total number of \st{class}\vmn{type}-$1$ and \st{class}\vmn{type}-$2$ customers respectively.
\end{approximation}

The second approximation ignores the error terms and assumes that \eqref{eq:estimate} holds with equality:
\begin{approximation}[Deterministic Observation]\label{assumption:determ-observ}
For any instance ${\vec{v}}$ and every $\lambda$, we have $$o_1(\lambda) =\lambda p n_1 + (1-p)\eta_1(\lambda)\text{ and }o_2(\lambda) =\lambda p n_2 + (1-p)\eta_2(\lambda).$$
\end{approximation}

%

%\vm{remove this}
%A key aspect of this discussion lies in the homogeneity of MP1.
%Namely, the feasibility of a competitive ratio $c$ in MP1 only depends on the ratio $b/n$
%\vm{remove this:}, and thus, in particular, if $c$ is feasible in MP1, it will remain feasible for arbitrary large instances, by scaling the number-of-customer variables $b$, $n$, $n_1$, $n_2$, $\eta_1$ and $\eta_2$ by a same scalar.
Focusing on asymptotic instances allows us to drastically decrease the variances of random events, and make the two approximations almost exact. Note that the feasibility of a competitive ratio $c$ in (MP1) only depends on the ratio $b/n$, and hence it does not change by scaling the problem.

%\vm{remove:}
%Intuitively, such large instances allow to drastically decrease the variances of random events, and nearly correspond to the case where we can apply Assumptions~\ref{assumption:continuous instance} and~\ref{assumption:determ-observ}.
While rigorous proof of Conjecture~\ref{conj:ata-tightness} is beyond the scope of this thesis, we manage to prove the following proposition with these two approximations:
\begin{proposition}\label{prop:ata-tight}
Under Approximations~\ref{assumption:continuous instance} and~\ref{assumption:determ-observ}, no online algorithm, deterministic or randomized, can achieve a competitive ratio better than $c^*$.
\end{proposition}

%$\vec{\widehat{v}'}$

\begin{proof}
Based on the optimal solution of (\ref{MP1}), $(\lambda^*, n_1^*, n_2^*, \eta_1^*, \eta_2^*, c^*)$, we construct two adversarial instances ${\vec{v'}\vmn{ZZ\vec{v}_I}}$ and $\vec{\widehat{v}'}$ that are indistinguishable by any online algorithm up to time $\lambda^*$.
Using Approximations~\ref{assumption:continuous instance}, we will define instance ${\vec{v'}\vmn{ZZ\vec{v}_I}}$ by functions $\eta_1$ and $\eta_2$ and $\vec{\widehat{v}'}$ by functions $\widehat{\eta_1}$ and $\widehat{\eta_2}$ .

We define the first instance ${\vec{v'}\vmn{ZZ\vec{v}_I}}$ by setting, for $j=1,2$, $\eta_j(\lambda^*)= \eta_j^*$ and $\eta_j(1)= n_j^*$, and apply linear interpolation for other values of $\lambda$. Formally,
\begin{align*}
\eta_j(\lambda) \triangleq \begin{cases}
\frac{\lambda}{\lambda^*}\eta_j^* &\text{ for } 0 \leq \lambda \leq \lambda^*. \\
\frac{1-\lambda}{1-\lambda^*}\eta_j^* + \frac{\lambda-\lambda^*}{1-\lambda^*} n_j &\text{ for }\lambda^* < \lambda \leq 1. \\
\end{cases}
\end{align*}

We construct the second instance $\vec{\widehat{v}'}$ with the following properties: (i) the number of observed \st{class}\vmn{type}-$j$ customers up to time $\lambda^*$ is the same as in instance ${\vec{v'}\vmn{ZZ\vec{v}_I}}$, i.e., for all $\lambda \leq \lambda^*$, $\widehat{o}_j(\lambda) = {o}_j(\lambda)$.
(ii) instance $\vec{\widehat{v}'}$ has as many \st{class}\vmn{type}-$1$ customers as possible after time $\lambda^*$.

Recall that at time $\lambda^*$, from an online algorithm's perspective, we have the following upper bounds on $\widehat{n_1}$ and $\widehat{n_1}+\widehat{n_2}$:
\begin{align*}
\widehat{n_1} \leq u_1^* \triangleq &\min \left\{ \frac{\lambda^* p n_1^* + (1-p)\eta_1^* }{\lambda^* p}, \frac{\lambda^* p n_1^* + (1-p)\eta_1^* + p (1-p)(1-\lambda^*)n}{1-p+\lambda^* p} \right\}, \\
\widehat{n_1}+\widehat{n_2} \leq u_{1,2}^* \triangleq & \min \left\{ \frac{\lambda^* p (n_1^*+n_2^*) + (1-p)(\eta_1^* + \eta_2^*) }{\lambda^* p}, \right. \\ & \left. \frac{\lambda^* p (n_1^*+n_2^*) + (1-p)(\eta_1^* + \eta_2^*) + p (1-p)(1-\lambda^*)n}{1-p+\lambda^* p} \right\}.
\end{align*}
Note that the second terms of the minima ensure that $u_1^* \leq n$ and $\bar u_{1,2}^* \leq n$.
Thus, we can define $\vec{\widehat{v}'}$ to have $\widehat{n_1} = u_1^*$,
$\widehat{n_2} = u_{1,2}^* - u_1^*$.
To put as many \st{class}\vmn{type}-$1$ as possible after $\lambda^*$, we set:

\begin{align*}
\widehat{\eta_1}(\lambda^*) = \widehat{\eta_1} &\triangleq \max \left\{ 0, \widehat{n_1} - (1-\lambda^*)n \right\}, \\
\widehat{\eta_2}(\lambda^*) = \widehat{\eta_2} &\triangleq \max \left\{ 0, \widehat{n_2} - (1-\lambda^*) n - \widehat{\eta_1} \right\}.
\end{align*}

Now we define instance $\vec{\widehat{v}'}$ by linear interpolation between time $0$ and $\lambda^*$ and between $\lambda^*$ and $1$:
\begin{align*}
\widehat{\eta_j}(\lambda) \triangleq \begin{cases}
\frac{\lambda}{\lambda^*}\widehat{\eta_j} &\text{ for } 0 \leq \lambda \leq \lambda^*. \\
\frac{1-\lambda}{1-\lambda^*}\widehat{\eta_j} + \frac{\lambda-\lambda^*}{1-\lambda^*} \widehat{n_j} &\text{ for }\lambda^* < \lambda^* \leq 1. \\
\end{cases}
\end{align*}

It is easy to check both ${\vec{v'}\vmn{ZZ\vec{v}_I}}$ and $\vec{\widehat{v}'}$ are valid continuous instances under
Approximations~\ref{assumption:continuous instance}.

Now we consider the arriving instances ${\vec{v}}$ and $\widehat{{\vec{v}}}$ generated from adversary instances ${\vec{v'}\vmn{ZZ\vec{v}_I}}$ and $\vec{\widehat{v}'}$ respectively.
We can show that, under Approximation~\ref{assumption:determ-observ}, online algorithms cannot distinguish between ${\vec{v}}$ and $\widehat{{\vec{v}}}$ up to time $\lambda^*$.
Formally, for all $\lambda \in[0, \lambda^*]$ and $j=1,2$, we have $o_j(\lambda) = \widehat{o_j}(\lambda)$.
To see this, we note both $o_j(\lambda)$ and $\widehat{o_j}(\lambda)$ are linear in $\lambda \in [0, \lambda^*]$. Thus it is enough to show that
$o_j(\lambda^*) = \widehat{o_j}(\lambda^*)$, which is easy to verify.

Now, because ${\vec{v}}$ and $\vec{\widehat{{v}}}$ are indistinguishable until time $\lambda^*$, we know that any online algorithm must accept the same number of customers in both cases.
In particular, the expected number of accepted \st{class}\vmn{type}-$2$ customers by time $\lambda^*$, which we denote as $q_2$, must be the same for both instances.

The competitive ratio of any online algorithm is at most
\begin{align*}
\min \left\{ \frac{\E{ALG({\vec{V}})}}{OPT({\vec{v'}\vmn{ZZ\vec{v}_I}})} ,\frac{\E{ALG(\vec{\widehat{{V}}})}}{OPT(\vec{\widehat{v}'})}\right\}.
\end{align*}

If $q_2 \leq (\frac{1-ac^*}{1-a} )b-c^* \min \{ u_1^* , b\}$, then we have not accepted enough \st{class}\vmn{type}-$2$ customers to guarantee $c^*$-competitiveness in ${\vec{v'}\vmn{ZZ\vec{v}_I}}$:
\begin{align*}
\frac{\E{ALG({\vec{V}})}}{OPT({\vec{v'}\vmn{ZZ\vec{v}_I}})} < & \frac{n_1^* + a (n_2^* - o_2(\lambda^*) + (\frac{1-ac^*}{1-a} )b-c^* \min \{ u_1^* , b\})}{(1-a)n_1^* + a \min \left\{ b, n_1^*+ n_2^* \right\}} &(Constraint~\eqref{constraint:n1<=b}) \\
\leq &c^*. &(Constraint~\eqref{constraint:not-c-competitive})
\end{align*}

On the other hand, if $q_2 > (\frac{1-ac^*}{1-a} )b-c^* \min \{ u_1^* , b\}$, then we will not have enough inventory for \st{class}\vmn{type}-$1$ customers arriving after $\lambda^*$:
\begin{align*}
\frac{\E{ALG(\vec{\widehat{{V}}})}}{OPT(\vec{\widehat{v}'})} \leq & \frac{(b -q_2) + a q_2}{ ab + (1-a)\min\{b, u_1^*\}} &(\widehat{n_1}+ \widehat{n_2} = u_{1,2}^*\text{ and Constraint}~\eqref{constraint:u_2>=b}) \\
< & \frac{b-(1-a) \left( (\frac{1-ac^*}{1-a} )b-c^* \min \{ u_1^* , b\}\right)}{ ab + (1-a)\min\{b, u_1^*\}} \\
= &c^*.
\end{align*}
This completes the proof.
\end{proof}

\subsection{Robustness}\label{sec:robust}
In this section, we show that our adaptive algorithm is robust against a slightly off estimation of the  parameter $p$.
For notational convenience, we define $ALG_{2,c,p}$ to be $ALG_{2,c}$ with parameter $p$, and we study the performance of $ALG_{2,c,p}$ when the true parameter is $\hat{p}$.
For the sake of brevity, we only consider robustness under Approximations~\ref{assumption:continuous instance} and~\ref{assumption:determ-observ}.
Under these approximations, the algorithms are modified so that for all $\lambda$,
\begin{align*}
u_1 (\lambda) & \triangleq
\min \left\{ \frac{o_1(\lambda)}{\lambda p}, \frac{o_1(\lambda) + (1-\lambda) (1-p) n}{1-p+\lambda p} \right\}, \\
u_{1,2} (\lambda) & \triangleq
\min \left\{ \frac{o_1(\lambda)+o_2(\lambda)}{\lambda p}, \frac{o_1(\lambda) +o_2(\lambda) + (1-\lambda) (1-p) n}{1-p+\lambda p} \right\}.
\end{align*}

In this section, we use $c^*( p )$ to denote the solution of (\ref{MP1}) under parameter $p$.
In the next proposition (proved later in this section), we show that $c^*( p )$ is continuous and non-decreasing in $p$:
\begin{proposition}\label{claim:small-c(p)-(p-hat)}
For all $ 0 < \hat{p}\leq p <1 $,
$$ 0 \leq c^*(p) - c^*(\hat{p}) \leq \frac{4(1-a)}{a}\log\left(\frac{p(1-\hat{p})}{\hat{p}(1-p)}\right) .$$
\end{proposition}
Note that the proposition above is a property of (\ref{MP1}) but not a property of the algorithm.

%\vm{what does this sentence say?!!}
%Since the algorithm assumes the problem parameter to be $p$, we assume, without loss of generality that $c<c^*(p)$, where we denote $c^*( p )$ the optimal solution of~\ref{MP1} with parameter $p$ for notational convenience.
%To simplify the discussion, we make Assumptions~\ref{assumption:continuous instance}, and~\ref{assumption:determ-observ}, but we our analysis is still valid up to a small error term similar to that of Theorem~\ref{thm:adaptive-threshold} when relaxing the assumptions.
Furthermore, we have the following proposition (proved later in this section) showing that our adaptive algorithm is robust to a small error in estimating the degree of randomness (in the proof we also give the competitive ratio for the case  $p \geq \hat{p} + \delta_{\hat{p}}$):
\begin{proposition}\label{prop:robust}
Under Approximations~\ref{assumption:continuous instance}, and~\ref{assumption:determ-observ},
if the true probability in the partially \st{learnable}\vmn{predictable} model is $\hat{p}$, then,
\begin{enumerate}
\item if $p < \hat{p}$, then $ALG_{2,c,p}$ has a competitive ratio of at least $c$ for all $c\leq c^*(p)$.
\item if $\hat{p}<p<\hat{p} + \delta_{\hat{p}}$, where $\delta_{\hat{p}}$ is a small positive constant, then $ALG_{2,c,p}$ has a competitive ratio of at least
$$ \begin{cases}
c^*(\hat{p})  - \frac{4(1-a)}{1-c} \log\left(\frac{p(1-\hat{p})}{\hat{p}(1-p)}\right) &,\text{ if }c^*(\hat{p})<c \leq c^*(p),\\
c\left[ 1- \frac{(1-a)(p - \hat{p})}{p} \right] &,\text{ if }c \leq c^*(\hat{p}).\end{cases}$$
\end{enumerate}
\end{proposition}

%\section{Proof of Proposition~\ref{claim:small-c(p)-(p-hat)}}\label{sec:proof-claim-small-cp-p-hat}
%In this section, we prove Proposition~\ref{claim:small-c(p)-(p-hat)}, which we reiterate as follows:
%\begin{repproposition}{claim:small-c(p)-(p-hat)}
%For all $ 0 < \hat{p}\leq p <1 $,
%$$ 0 \leq c^*(p) - c^*(\hat{p}) \leq \frac{4(1-a)}{a}\log\left(\frac{p(1-\hat{p})}{\hat{p}(1-p)}\right) .$$
%\end{repproposition}

\begin{proof}[Proof of Proposition~\ref{claim:small-c(p)-(p-hat)}]
The proof relies on the following lemma:
\begin{lemma}\label{lemma:small:derivative-c*}
For any $p\in (0,1)$,
$$ \frac{{\rm d} c^*(p)}{{\rm d} p} \leq \frac{4(1-a)}{ap(1-p)}.$$
\end{lemma}
\begin{proof}
For any fixed $\bar p  \in (0,1)$, we discuss what happens to $c^*(p)$ when we increase the value of $p$ from $\bar p$ to $\bar p + {\rm d} p$ where ${\rm d} p$ is a small positive number approaching $0$.
We start with an optimal solution $(\lambda^*, n_1^*, n_2^*, \eta_1^*, \eta_2^*, c^*)$ of $\ref{MP1}(\bar p)$.
We first prove that the optimal solution satisfies the following conditions:
\begin{subequations}
\begin{align}
c^* &= f(n_2^*, \eta_2^*, \tilde u_1^*, \bar p) \label{define-f}
\end{align}
where $f(n_2, \eta_2, \tilde u_1, p) \triangleq \frac{a\left((1-\lambda^* p) n_2 - (1-p)\eta_2 + \frac{b}{1-a} \right) + n_1^*}{a\min\{n_1^*+n_2, b\}+(1-a)n_1^*+\frac{a^2b}{1-a}+a \min\{\tilde u_1, b\}}$.
\begin{align}
\eta_1^*+\eta_2 ^*&= \min\{n_1^*+n_2^*, \lambda^* n\}. \label{eta1+eta2-big} & \\
n_1^*+n_2^* &\leq b. & \label{n1+n2=b}
\end{align}
\end{subequations}
Condition~\eqref{define-f} is obtained by setting \eqref{constraint:not-c-competitive} to an equality and expressing $\tilde o_2$ as $(1-p)\eta_2 - p\lambda n_2$.
We can derive Condition~\eqref{eta1+eta2-big} because $f( n_2, \eta_2, \tilde u_1, p)$ is decreasing in both $\eta_2$ and $\eta_1$ (this holds because $f$ is non-increasing in $ \tilde u_1$ and $ \tilde u_1$ is increasing in $\eta_1$).
Also note that \eqref{constraint:eta_1+eta_2<=lambdan}-\eqref{constraint:n2'<=n2} gives $\eta_1+\eta_2 \leq \min\{n_1+n_2, \lambda n\}$.
While fixing $\lambda, n_1$, and $n_2$, $\min\{n_1+n_2, \lambda n\}$ is the maximum value of $\eta_1 +\eta_2$.
Thus, in order to minimize $f$, Condition~\eqref{eta1+eta2-big} holds.
Now let us prove Condition~\eqref{n1+n2=b}.
Assume on the contrary $n_1^*+n_2^*>b$, we note that, decreasing $n_2$ does not change the value of the denominator of the function $f$.
On the other hand, the numerator of the function $f$ is decreasing either when $n_2$ and $\eta_2$ decrease by the same amount or when $n_2$ decreases by itself.
Therefore, replacing $(n_2^*, \eta_2^*)$ with $ (b-n_1^*, \max\{\eta_2^*+b-n_1^*-n_2^*, 0\})$ gives a smaller value of $f$.
Further, it is easy to verify that replacing $(n_2^*, \eta_2^*)$ with $ (b-n_1^*, \max\{\eta_2^*+b-n_1^*-n_2^*, 0\})$ still gives a feasible solution.
Thus, Condition~\eqref{n1+n2=b} holds.

In the following, our goal is to find feasible solutions of~\ref{MP1}($\bar p+{\rm d}p$) for all small enough positive numbers ${\rm d}p$. %where $p'=\bar p+{\rm d}p$ and ${\rm d}p$ denotes a small positive number approaching zero.
We first note that the tuple $(\lambda^*, n_1^*, n_2^*, \eta_1^*, \eta_2^*, c^*)$ satisfies Constraints~\eqref{constraint:x<=1} through~\eqref{constraint:n1'+n2'big} in~\ref{MP1}$(\bar p+{\rm d}p)$.
However, Constraint~\eqref{constraint:u_2>=b} is not necessarily satisfied.
Therefore, we may need to define an alternative tuple with increased values of $n_1, n_2, \eta_1, \eta_2$.

In what follows, we construct a feasible tuple by keeping $n_1=n_1^*$ and $\eta_1=\eta_1^*$, and increasing the values of $n_2$, $\eta_2$, $c$.
More precisely, we define
\begin{subequations}
\begin{align}
n_2(p) & \triangleq n_2^* +\frac{2b}{\bar p(1-\bar p)}(p-\bar p)\text{, }
\label{def:n_2p}\\
\eta_2 (p) & \triangleq \eta_2^* + \min\{n_2(p)-n_2^*, \lambda^* n- \eta_1^*-\eta_2 ^*\}\text{, and }\label{def:eta_2p}\\
c(p)& \triangleq c^* + \frac{4(1-a)}{a\bar p(1-\bar p)}(p-\bar p).\label{def:cp}
\end{align}
\end{subequations}
Note that at $p=\bar p$, $n_2(\bar p)= n_2^*$, $\eta_2(\bar p)= \eta_2^*$, and $c(\bar p) =c^*$.

Next, we show that if ${\rm d}p$ is a small enough positive number, then the tuple $(\lambda^*, n_1^*, n_2(\bar p+{\rm d}p), \eta_1^*, \eta_2(\bar p+{\rm d}p), c(\bar p+{\rm d}p))$ is in the feasible set of~\ref{MP1}$(\bar p +{\rm d}p)$.

First, we note that $(\lambda^*, n_1^*, n_2(\bar p+{\rm d}p), \eta_1^*, \eta_2(\bar p+{\rm d}p), c(\bar p+{\rm d}p))$ still satisfies Constraints~\eqref{constraint:x<=1}-\eqref{constraint:n1'+n2'big} in $\ref{MP1}(p')$ due to \eqref{eta1+eta2-big} and \eqref{n1+n2=b}.
In what follows, we show that it satisfies Constraints~\eqref{constraint:u_2>=b} and~\eqref{constraint:not-c-competitive} separately.

\noindent{\textbf{Constraints~\eqref{constraint:not-c-competitive}}}

% To simplify notation, we write $p'=\bar p + dp$, $n_2'=n_2(\bar p + dp)$, $\eta_2'=\eta_2(\bar p + dp)$, and$c'=c(\bar p + dp)$.

First, we show that for small enough positive ${\rm d}p$, $(\lambda^*, n_1^*, n_2(\bar p+{\rm d}p), \eta_1^*, \eta_2(\bar p+{\rm d}p), c(\bar p+{\rm d}p))$ satisfies Constraint~\eqref{constraint:not-c-competitive} in $\ref{MP1}(\bar p+{\rm d}p)$.
%For simplicity, we denote $ \tilde u_1^* $ the value of $ \tilde u_1 $ calculated with $n_1^*, \eta_1^*$, and $\bar p$, and $ \tilde u_1' $ the value of $ \tilde u_1 $ calculated with $n_1^*, \eta_1^*$, and $p'$.
Due to ~\eqref{define-f}, it suffices to show that when evaluated at $p=\bar p$,
\begin{align}
\frac{{\rm d}c(p)}{{\rm d} p} & > \frac{{\rm d}  f(n_2(p), \eta_2(p), \tilde u_1(p), p)}{{\rm d} p}, \label{ineq:c'-c*dp>=df/dp}
\end{align}
where we note that $\tilde u_1$ is a function of $p$ even though $(n_1,\eta_1, \lambda) $ are fixed at $(n_1^*,\eta_1^*, \lambda^*)$.
For simplicity, we denote $ \tilde u_1^* \triangleq \tilde u_1(\bar p)$.
The right hand sice of \eqref{ineq:c'-c*dp>=df/dp} is the total derivative over $p$, which can be expressed as
\begin{align}
& \frac{\partial f(n_2, \eta_2, \tilde u_1, p)}{\partial n_2} \frac{{\rm d}n_1(p)}{{\rm d}p} + \frac{\partial f(n_2, \eta_2, \tilde u_1, p)}{\partial \eta_2} \frac{{\rm d}\eta_2(p)}{{\rm d}p} +\frac{\partial f(n_2, \eta_2, \tilde u_1, p)}{\partial \tilde u_1} \frac{{\rm d} u_1(p)}{{\rm d}p} \nonumber
\\ & + \frac{\partial f(n_2, \eta_2, \tilde u_1, p)}{\partial p}, \label{equa:total-diff}
\end{align}
where the partial derivatives are evaluated at $(n_2, \eta_2, \tilde u_1,p) = (  n_2^* ,\eta_2^* , \tilde u_1^*, \bar p )$ and the total derivatives are evaluated at $p=\bar p$ throughout the proof.

We can further derive an upper bound for each term in \eqref{equa:total-diff}.
First, because \eqref{n1+n2=b} and $a\min\{n_1^*+n_2', b\}+(1-a)n_1^* \geq a n_2 ^*+ n_1 ^*$, we have
\begin{align}
\frac{\partial f(n_2, \eta_2, \tilde u_1, p)}{\partial n_2}  &= \frac{1}{{\rm d}n_2}\left( \frac{a\left((1-\lambda^* \bar p) (n_2^*+{\rm d}n_2) - (1-\bar p)\eta_2^* + \frac{b}{1-a} \right) + n_1^*}{a (n_2^*+{\rm d}n_2) + n_1^*+\frac{a^2b}{1-a}+a \min\{\tilde u_1^*, b\}} \right. \nonumber \\ & - \left. \frac{a\left((1-\lambda^* \bar p) n_2^* - (1-\bar p)\eta_2^* + \frac{b}{1-a} \right) + n_1^*}{a n_2^* + n_1^*+\frac{a^2b}{1-a}+a \min\{\tilde u_1^*, b\}} \right) \nonumber\\
 \leq & \frac{1}{{\rm d}n_2}\left( \frac{a\left((1-\lambda^* \bar p) (n_2^*+{\rm d}n_2) - (1-\bar p)\eta_2^* + \frac{b}{1-a} \right) + n_1^*}{a n_2^* + n_1^*+\frac{a^2b}{1-a}+a \min\{\tilde u_1^*, b\}} \right. \nonumber \\ & - \left. \frac{a\left((1-\lambda^* \bar p) n_2^* - (1-\bar p)\eta_2^* + \frac{b}{1-a} \right) + n_1^*}{a n_2^* + n_1^*+\frac{a^2b}{1-a}+a \min\{\tilde u_1^*, b\}} \right) &(\text{see below}) \label{ineq:n_2_increasing}\\
 = & \frac{a (1-\lambda^* \bar p )}{a n_2 ^*+ n_1^*+\frac{a^2b}{1-a}+a \min\{\tilde u_1^*, b\}}
\leq \frac{a}{\frac{a^2b}{1-a}} = \frac{1-a}{ab}. \nonumber
\end{align}
For deriving \eqref{ineq:n_2_increasing}, we use the fact that when ${\rm d}p>0$, $n_2(\bar p + {\rm d}p)>n_2^*$, and thus ${\rm d}n_2 > 0$.

Using the above inequality and \eqref{def:n_2p},
\begin{align}
\frac{\partial f(n_2, \eta_2, \tilde u_1, p)}{\partial n_2} \frac{{\rm d}n_2(p)}{{\rm d}p} \leq  \frac{1-a}{ab}  \frac{2b}{p(1-p)} = \frac{2(1-a)}{ap(1-p)}.\label{inequ:partial-n2}
\end{align}

For the term regarding the partial derivative of $\eta_2$, we have
\begin{align*}
\frac{\partial f(n_2, \eta_2, \tilde u_1, p)}{\partial \eta_2}
& = \frac{\partial}{\partial\eta_2} \frac{a\left((1- \lambda^* \bar p) n_2^* - (1-\bar p)\eta_2 + \frac{b}{1-a} \right) + n_1^*}{a n_2^* + n_1^*+\frac{a^2b}{1-a}+a \tilde u_1^*} &(\eqref{n1+n2=b}) \nonumber \\
& = - \frac{a(1-\bar p)}{a n_2^* + n_1^*+\frac{a^2b}{1-a}+a \tilde u_1^*} \leq 0 .\nonumber \\
\end{align*}
Further, from \eqref{def:eta_2p}, $\frac{{\rm d}\eta_2(p)}{{\rm d}p} \geq 0$, and thus
\begin{align} \frac{\partial f(n_2, \eta_2, \tilde u_1, p)}{\partial \eta_2} \frac{{\rm d}\eta_2(p)}{{\rm d}p}  \leq 0\label{ineq:partial-eta-2}
\end{align}

Now let us consider the term regarding the partial derivative of $\tilde u_1$.
We consider two cases $\tilde u_1 ^* > b$ and $\tilde u_1 ^* \leq b$ separately.

For the case $\tilde u_1 ^* > b$, since $\tilde u_1$ is continuous in $p$, for small enough ${\rm d} p$ we have $\tilde u_1 (\bar p +{\rm d} p)>b$, and thus
\begin{align*}
\frac{\partial f(n_2, \eta_2, \tilde u_1, p)}{\partial \tilde u_1} = 0.
\end{align*}
Therefore, in this case,
\begin{align}
\frac{\partial f(n_2, \eta_2, \tilde u_1, p)}{\partial \tilde u_1} \frac{{\rm d} \tilde u_1(p)}{{\rm d}p} =0 .\label{ineq-partial-u1-trivial-0}
\end{align}
For the other case, $\tilde u_1 ^* \leq b$, we have
\begin{align*}
\frac{\partial f(n_2, \eta_2, \tilde u_1, p)}{\partial \tilde u_1}
& = \frac{\partial}{\partial \tilde u_1} \frac{a\left((1- \lambda^* \bar p) n_2 - (1-\bar p)\eta_2 + \frac{b}{1-a} \right) + n_1^*}{a n_2 + n_1^*+\frac{a^2b}{1-a}+a \tilde u_1} &(\eqref{n1+n2=b}) \nonumber \\
& = - a \frac{a\left((1- \lambda^* \bar p) n_2^* - (1-\bar p)\eta_2^* + \frac{b}{1-a} \right) + n_1^*}{(a n_2^* + n_1^*+\frac{a^2b}{1-a}+a \tilde u_1^*)^2} \nonumber \\
& = - \frac{ac^*}{a n_2^* + n_1^*+\frac{a^2b}{1-a}+a \tilde u_1^*} &(\eqref{define-f})
\\
& \geq -\frac{ac^*}{\frac{a^2b}{1-a}} = - \frac{(1-a)c^*}{ab}. %\label{inequ:parial_u1}
\end{align*}
In order to bound $\frac{{\rm d}\tilde u_1(p)}{{\rm d}p}$ (under the case $\tilde u_1 \leq b$), we distinguish two cases based on the definition of $\tilde u_1 ^*$ in~\ref{MP1}($\bar p$).
Note that $\frac{(1-\bar p)\eta_1^* + \lambda^* \bar p n_1^*}{\lambda^* \bar p} \leq \frac{(1-\bar p)\eta_1 ^*+ \lambda ^* \bar p n_1 ^* +(1-\bar p)(1-\lambda^*)n}{1-\bar p+\lambda^* \bar p}$ is equivalent to $\bar p(\lambda^*(1-\lambda^*)n-\lambda^* n_1^* +\eta_1^*) \geq \eta_1^*$.

The first case is $\tilde u_1^*=\frac{(1-\bar p)\eta_1^* + \lambda^* \bar p n_1^*}{\lambda ^* p} \leq \frac{(1-\bar p)\eta_1^* + \lambda^* \bar p n_1^* +(1-\bar p)(1-\lambda^*)n}{1-\bar p+\lambda^* \bar p}$, which happens when $\bar p(\lambda^*(1-\lambda^*)n-\lambda^* n_1^* +\eta_1^*) \geq \eta_1^*$.
Since, $(\bar p+{\rm d}p)(\lambda^*(1-\lambda^*)n-\lambda^* n_1^* +\eta_1^*)\geq \bar p(\lambda^*(1-\lambda^*)n-\lambda^* n_1^* +\eta_1^*) \geq \eta_1^* $, we have $\tilde u_1(p+{\rm d}p) = \frac{(1-(\bar p+{\rm d}p))\eta_1^* + \lambda^* (\bar p+{\rm d}p) n_1^*}{\lambda ^* (\bar p+{\rm d}p)}$.
Furthermore, $ \frac{(1-\bar p)\eta_1^* + \lambda^* \bar p n_1^*}{\lambda ^* \bar p} = \tilde u_1^* \leq b$ implies $ \eta_1^* \leq \frac{\lambda^* \bar p b}{1-\bar p}$.
Combining the above equality and inequality.
\begin{align}
\frac{{\rm d}\tilde u_1(p)}{{\rm d}p} = \frac{{\rm d}}{{\rm d}p}\frac{(1-p)\eta_1^* + \lambda^* p n_1^*}{\lambda^* p}
= -\frac{ \eta_1^*}{\lambda^* \bar p^2} \geq - \frac{b}{\bar p(1-\bar p)}. \label{ineq:partial-u1-dp-case1}
\end{align}

The other case is  $\tilde u_1^* = \frac{(1-\bar p)\eta_1^* + \lambda^* \bar p n_1^* +(1-\bar p)(1-\lambda^*)n}{1-\bar p+\lambda^* \bar p}< \frac{(1-\bar p)\eta_1^* + \lambda^* \bar p n_1^*}{\lambda ^* p} $, which happens when  $\bar p(\lambda^*(1-\lambda^*)n-\lambda^* n_1^* +\eta_1^*) < \eta_1^*$, we have $\tilde u_1^* = \frac{(1-\bar p)\eta_1^* + \lambda^* \bar p n_1^* +(1-\bar p)(1-\lambda^*)n}{1-\bar p+\lambda^* \bar p}$.
For small enough and positive ${\rm d}p$, $(\bar p +{\rm d}p)(\lambda^*(1-\lambda^*)n-\lambda^* n_1^* +\eta_1^*) < \eta_1^*$, and thus $\tilde u_1(\bar p+{\rm d}p) = \frac{(1-(\bar p+{\rm d}p))\eta_1^* + \lambda^* (\bar p+{\rm d}p) n_1^* +(1-(\bar p+{\rm d}p))(1-\lambda^*)n}{1-\bar p+\lambda^* (\bar p+{\rm d}p)}$.
In addition, rearranging terms in the definition of the case, we have $-(1-\lambda^*) \lambda^* n > -\frac{(1-\bar p)}{\bar p}\eta_1^* - \frac{\lambda^* \bar p}{\bar p} n_1^*$.
As a result,
\begin{align}
\frac{{\rm d}\tilde u_1(p)}{{\rm d}p} & = \frac{{\rm d}}{{\rm d}p}\frac{(1-p)\eta_1^* + \lambda^* p n_1^* +(1-p)(1-\lambda^*)n}{1-p+\lambda^* p} \nonumber \\
& = \frac{{\rm d}}{{\rm d}p} \frac{\lambda^* p (n_1^*-\eta_1^*)-\lambda^* n}{1-p+\lambda^* p} \nonumber\\
& = \frac{-(1-\lambda^*) \lambda^* n + \lambda(n_1^*-\eta_1^*)}{(1-\bar p+\lambda^* \bar p )^2} \label{ineq:partial-u1-dp-case2}\\
& > \frac{-\frac{(1-\bar p)}{\bar p}\eta_1^* - \frac{\lambda^* \bar p}{\bar p} n_1^*+ \lambda^*(n_1^*-\eta_1^*)}{(1-\bar p+\lambda^* \bar p )^2} \nonumber
\\ & = - \frac{\eta_1^*}{\bar p(1-\bar p+\lambda^* \bar p) } > - \frac{b}{\bar p(1-\bar p)} .&
(\eqref{constraint:n1'<=n1} \text{ and }\eqref{n1+n2=b}) \nonumber
\end{align}
Overall, if $\tilde u_1^* \leq b$, in either case,
\begin{align}
\frac{\partial f(n_2, \eta_2, \tilde u_1, p)}{\partial \tilde u_1} \frac{{\rm d}\tilde u_1(p)}{{\rm d}p}  \leq \frac{(1-a)c}{ab}\frac{b}{\bar p(1-\bar p)} = \frac{(1-a)c}{a\bar p(1-\bar p)}. \label{inequa:partial:u1-overall}%\label{ineq-parital-u1-overall}
\end{align}
Since this bound is looser than that of ~\eqref{ineq-partial-u1-trivial-0}, the above inequality is true for general $\tilde u_1^*$.

For the terms regarding the partial derivative with respect to $p$, we have
\begin{align}
\frac{\partial f(n_2, \eta_2, \tilde u_1, p)}{\partial p}
& = \frac{\partial}{\partial p} \frac{a\left((1- \lambda^* p) n_2 - (1-p)\eta_2 + \frac{b}{1-a} \right) + n_1^*}{a n_2 + n_1^*+\frac{a^2b}{1-a}+a \min\{\tilde u_1, b\}} &(\eqref{n1+n2=b}) \nonumber \\
& = \frac{a\left(\eta_2^* - \lambda^* n_2 ^*\right) }{a n_2^* + n_1^* +\frac{a^2b}{1-a}+a \min\{\tilde u_1^*, b\}} \nonumber \\
& \leq \frac{ab}{\frac{a^2b}{1-a}} =\frac{1-a}{a}. &(\eqref{constraint:n2'<=n2} \text{ and }\eqref{n1+n2=b}) \label{inequ:parial_p}
\end{align}
Applying~\eqref{inequ:partial-n2},~\eqref{ineq:partial-eta-2},~\eqref{inequa:partial:u1-overall} and~\eqref{inequ:parial_p}, on ~\eqref{equa:total-diff}, we obtain
\begin{align*}
\eqref{equa:total-diff} & \leq \frac{2(1-a)}{a\bar p(1-\bar p)} + \frac{(1-a)c^*}{a\bar p(1-\bar p)} + \frac{1-a}{a} < \frac{4(1-a)}{a\bar p(1-\bar p)} = \frac{{\rm d}c(p)}{{\rm d}p},\end{align*}
which implies \eqref{ineq:c'-c*dp>=df/dp} and thus completes the proof that $(\lambda^*, n_1^*, n_2(\bar p+{\rm d}p), \eta_1^*, \eta_2(\bar p+{\rm d}p), c(\bar p+{\rm d}p))$ satisfies Constraint~\eqref{constraint:u_2>=b}.

\noindent{\textbf{Constraint~\eqref{constraint:u_2>=b}}}

Now, we show that for small enough positive ${\rm d}p$, $(\lambda^*, n_1^*, n_2(\bar p+{\rm d}p), \eta_1^*, \eta_2(\bar p+{\rm d}p), c(\bar p+{\rm d}p))$ satisfies Constraint~\eqref{constraint:u_2>=b} in $\ref{MP1}(\bar p+{\rm d}p)$.

For notational convenience, we view $\tilde u_{1,2}$ as a function of $p$ and denote $\tilde u_{1,2}^* \triangleq \tilde u_{1,2}(\bar p)$ as the value of $\tilde u_{1,2}$ corresponding to the original solution in $\ref{MP1}(\bar p)$.

If $\tilde u_{1,2}^* > b$, then because $u_{1,2}(p)$ is continuous in $p$, for small enough and positive ${\rm d}p$, the modified tuple still gives $\tilde u_{1,2}(\bar p+{\rm d}p) \geq b$ and thus satisfies Constraint~\eqref{constraint:u_2>=b} in $\ref{MP1}(\bar p+{\rm d}p)$.
Due to \eqref{constraint:u_2>=b} in $\ref{MP1}(\bar p)$, $\tilde u_{1,2}^* \geq b$, and hence the case $\tilde u_{1,2}^* < b$ does not exist.
In the following, we consider the remaining case where $\tilde u_{1,2}^* = b$.

Similar to how we distinguish two cases of $u_1^*$ when proving that Constraints~\eqref{constraint:not-c-competitive} is satisfied, we consider different values of $\tilde u_{1,2}^* $ separately.
However, the proof here is trickier and we need to consider three cases separately, based on whether $\frac{(1-\bar p)(\eta_1^*+\eta_2^*) + \lambda ^* \bar p (n_1^*+n_2^*)}{\lambda^* \bar p}$ is smaller, greater, or equal to $\frac{(1-\bar p)(\eta_1^*+\eta_2^*) + \lambda^* \bar p (n_1^*+n_2^*) + (1-\bar p)(1-\lambda^*) n ^*}{1- \bar p + \lambda ^* \bar p}$.

The first case is $\frac{(1-\bar p)(\eta_1^*+\eta_2^*) + \lambda ^* \bar p (n_1^*+n_2^*)}{\lambda^* \bar p} <\frac{(1-\bar p)(\eta_1^*+\eta_2^*) + \lambda^* \bar p (n_1^*+n_2^*) + (1-\bar p)(1-\lambda^*) n ^*}{1- \bar p + \lambda ^* \bar p}$.
In this case,  we have $\tilde u_{1,2}^* = \frac{(1-\bar p)(\eta_1^*+\eta_2^*) + \lambda ^* \bar p (n_1^*+n_2^*)}{\lambda^* \bar p}$.
Furthermore, when ${\rm d}p$ is a small enough positive number,  $n_2(\bar p+{\rm d}p), \eta_2(\bar p+{\rm d}p),$ and $\bar p+{\rm d}p$ are close enough to the respective values of $n_2^*, \eta_2^*,$ and $\bar p$, and thus the corresponding value of $\tilde u_{1,2}$ still takes the same form.
Therefore,
\begin{align*}
& \frac{\tilde u_{1,2}(p)}{{\rm d}p} \\
= &\frac{\partial}{\partial n_2}\left( \frac{(1-p)(\eta_1^*+\eta_2) + \lambda^* p (n_1^*+n_2)}{\lambda^* p} \right) \frac{{\rm d}n_2(p)}{{\rm d}p} \\ &
+ \frac{\partial}{\partial \eta_2}\left( \frac{(1-p)(\eta_1^*+\eta_2) + \lambda^* p (n_1^*+n_2)}{\lambda^* p} \right) \frac{{\rm d}\eta_2(p)}{{\rm d}p} \\
& + \frac{\partial}{\partial p} \frac{(1-p)(\eta_1^*+\eta_2) + \lambda^* p (n_1^*+n_2)}{\lambda^* p} \\
\geq & \frac{{\rm d}n_2(p)}{{\rm d}p} + \frac{1-\bar p}{\lambda^* \bar p}\frac{{\rm d}\eta_2(p)}{{\rm d}p} - \frac{b}{\bar p(1-\bar p)} &(\eqref{ineq:partial-u1-dp-case1}) \\
\geq & \frac{2b}{\bar p(1-\bar p)} - \frac{b}{\bar p(1-\bar p)} = \frac{b}{\bar p(1-\bar p)}>0,&(\eqref{def:n_2p}\text{ and }\eqref{def:eta_2p})
\end{align*}
where the partial derivatives are evaluated at $(  n_2, \eta_2, p) = ( n_2^*,  \eta_2^* ,\bar p)$ throughout the proof.
As a result, when ${\rm d}p$ is a small enough positive number, $\tilde u_{1,2}(\bar p+{\rm d}p) \geq \tilde u_{1,2}(\bar p) \geq b$, and hence Constraint~\eqref{constraint:u_2>=b} is satisfied in $\ref{MP1}(p')$

The second case is $\frac{(1-\bar p)(\eta_1^*+\eta_2^*) + \lambda ^* \bar p (n_1^*+n_2^*)}{\lambda^* \bar p} > \frac{(1-\bar p)(\eta_1^*+\eta_2^*) + \lambda^* \bar p (n_1^*+n_2^*) + (1-\bar p)(1-\lambda^*) n ^*}{1- \bar p + \lambda ^* \bar p}$.
In this case,  we have $\tilde u_{1,2}^* = \frac{(1-\bar p)(\eta_1^*+\eta_2^*) + \lambda^* \bar p (n_1^*+n_2^*) + (1-\bar p)(1-\lambda^*) n ^*}{1- \bar p + \lambda ^* \bar p}.$
Similar to the first case, when ${\rm d}p$ is a small enough positive number,  $n_2(\bar p+{\rm d}p), \eta_2(\bar p+{\rm d}p),$ and $\bar p+{\rm d}p$ are close enough to the respective values of $n_2^*, \eta_2^*,$ and $\bar p$, and thus the corresponding value of $\tilde u_{1,2}$ still takes the same form.
Therefore,
\begin{align*}
& \frac{{\rm d}\tilde u_{1,2}(p)}{{\rm d}p} \\
= &\frac{\partial}{\partial n_2}\left( \frac{(1-p)(\eta_1^*+\eta_2) + \lambda^* p (n_1^*+n_2) + (1-p)(1-\lambda^*) n }{1- p + \lambda^* p} \right) \frac{{\rm d}n_2(p)}{{\rm d}p} \\
& + \frac{\partial}{\partial \eta_2}\left( \frac{(1-p)(\eta_1^*+\eta_2) + \lambda^* p (n_1^*+n_2) + (1-p)(1-\lambda^*) n }{1- p + \lambda^* p} \right) \frac{{\rm d}\eta_2(p)}{{\rm d}p} \\
& + \frac{\partial}{\partial p} \frac{(1-p)(\eta_1^*+\eta_2) + \lambda^* p (n_1^*+n_2) + (1-p)(1-\lambda^*) n }{1- p + \lambda^* p}\\
= & \frac{ \lambda^* \bar p}{1-\bar p+\lambda ^* \bar p}\frac{{\rm d}n_2(p)}{{\rm d}p} + \frac{ 1-\bar p}{1-\bar p+\lambda ^* \bar p}\frac{{\rm d}\eta_2(p)}{{\rm d}p} \\& + \frac{-(1-\lambda^*) \lambda^* n + \lambda^*(n_1^*+n_2^*-\eta_1^*-\eta_2^*)}{(1-\bar p+\lambda^* \bar p )^2} &(\text{similar to }\eqref{ineq:partial-u1-dp-case2})
\\ \geq & \frac{ \lambda^* \bar p}{1-\bar p+\lambda ^*\bar p}\frac{{\rm d}n_2(p)}{{\rm d}p} + \frac{-(1-\lambda^*) \lambda^* n + \lambda^*(n_1^*+n_2^*-\eta_1^*-\eta_2^*)}{(1-\bar p+\lambda^* \bar p )^2},
\end{align*}
which is greater than $0$ when
\begin{align}
\frac{{\rm d}n_2(p)}{{\rm d}p}>
\frac{(1-\lambda^*) n - (n_1^*+n_2^*-\eta_1^*-\eta_2^*)}{(1-\bar p+\lambda^* \bar p )\bar p}. \label{ineq:condition-to-finish-proof}
\end{align}
Therefore, it is sufficient to prove \eqref{ineq:condition-to-finish-proof}.
By the assumption of this case, we know
\begin{align*}
b = \tilde u_{1,2}^* = \frac{(1-\bar p)(\eta_1^*+\eta_2^*) + \lambda^* \bar p (n_1^*+n_2^*) + (1-\bar p)(1-\lambda^*) n }{1- \bar p + \lambda ^*p} \geq \frac{(1-\bar p)(1-\lambda^*) n }{1- \bar p + \lambda ^* \bar p}.
\end{align*}
By using Constraints~\eqref{constraint:n1'<=n1},~\eqref{constraint:n2'<=n2} and the inequality above, we know
\begin{align*}
\frac{(1-\lambda^*) n - (n_1^*+n_2^*-\eta_1^*-\eta_2^*)}{(1-\bar p+\lambda^* \bar p )\bar p} \leq \frac{(1-\lambda^*) n}{(1-\bar p+\lambda^* \bar p )\bar p} \leq \frac{b}{(1-\bar p)\bar p} < \frac{2b}{(1-\bar p)\bar p}  = \frac{{\rm d}n_2(p)}{{\rm d}p},
\end{align*}
which proves \eqref{ineq:condition-to-finish-proof} and hence concludes the proof of this case.

Finally, the third case is $\frac{(1-\bar p)(\eta_1^*+\eta_2^*) + \lambda ^* \bar p (n_1^*+n_2^*)}{\lambda^* \bar p} = \frac{(1-\bar p)(\eta_1^*+\eta_2^*) + \lambda^* \bar p (n_1^*+n_2^*) + (1-\bar p)(1-\lambda^*) n ^*}{1- \bar p + \lambda ^* \bar p}$.
In this case,  we have $\tilde u_{1,2}^* = \frac{(1-\bar p)(\eta_1^*+\eta_2^*) + \lambda ^* \bar p (n_1^*+n_2^*)}{\lambda^* \bar p} = \frac{(1-\bar p)(\eta_1^*+\eta_2^*) + \lambda^* \bar p (n_1^*+n_2^*) + (1-\bar p)(1-\lambda^*) n ^*}{1- \bar p + \lambda ^* \bar p}.$
In this case, we can show that (we will show this shortly), either for all small enough ${\rm d}p$,
\begin{align}\tilde u_{1,2}(\bar p +{\rm d}p) = \frac{(1-(\bar p +{\rm d}p))(\eta_1^*+\eta_2(\bar p +{\rm d}p)) + \lambda ^* \bar p (n_1^*+n_2(\bar p +{\rm d}p))}{\lambda^* \bar p},\label{u12p+dp-first-case}
\end{align}
or for all small enough ${\rm d}p$,
\begin{align}  &\tilde u_{1,2}(\bar p +{\rm d}p) \nonumber \\=&  \frac{(1-(\bar p +{\rm d}p))(\eta_1^*+\eta_2(\bar p +{\rm d}p)) + \lambda^* (\bar p +{\rm d}p)(n_1^*+n_2(\bar p +{\rm d}p)) + (1-(\bar p +{\rm d}p))(1-\lambda^*) n ^*}{1- (\bar p +{\rm d}p) + \lambda ^* (\bar p +{\rm d}p)}.\label{u12p+dp-second-case}
\end{align}
And hence, we can reduce this case to either the first or the second case that we have just proved and completes the proof.

In what follows, we show either \eqref{u12p+dp-first-case} or \eqref{u12p+dp-second-case} holds.
To see this, we first note that, for all $p\in(0,1)$, $\frac{(1-p)(\eta_1^*+\eta_2(p)) + \lambda ^*  p (n_1^*+n_2(p))}{\lambda^* p} \leq \frac{(1-p)(\eta_1^*+\eta_2(p)) + \lambda^*  p (n_1^*+n_2(p)) + (1- p)(1-\lambda^*) n ^*}{1- p + \lambda ^*  p}$ is equivalent to $p (\lambda^*(1-\lambda^*)n-\lambda^* ( n_1^* + n_2(p))  +\eta_1^*+\eta_2(p)) -(\eta_1^* + \eta_2(p)) \geq 0$.
Since the left hand side of this inequality is a quadratic function of $p$, there must exists some positive number $\delta>0$ such that for all ${\rm d}p< \delta$, the quadratic function evaluated at all $p = \bar p+ {\rm d}p$ is either always non-negative or always non-positive, which concludes the proof that the tuple $(\lambda^*, n_1^*, n_2(\bar p+{\rm d}p), \eta_1^*, \eta_2(\bar p+{\rm d}p), c(\bar p+{\rm d}p))$ is in the feasible set of~\ref{MP1}$(\bar p +{\rm d}p)$.

Since the tuple $(\lambda^*, n_1^*, n_2(\bar p+{\rm d}p), \eta_1^*, \eta_2(\bar p+{\rm d}p), c(\bar p+{\rm d}p))$ is in the feasible set of~\ref{MP1}$(\bar p +{\rm d}p)$, when evaluating the derivative of $c^*(p)$ at $p=\bar p$,
\begin{align*}
\frac{{\rm d} c^*(p)}{{\rm d} p} = &
\frac{ c^*(\bar p+{\rm d}p) -c^*(\bar p) }{{\rm d} p} \leq \frac{ c(\bar p+{\rm d}p) -c(\bar p) }{{\rm d} p} &(c^*(\bar p+{\rm d}p) \leq c(\bar p+{\rm d}p) \text{ and }c^*(\bar p)=c(\bar p) ) \\
= & \frac{ c^*+\frac{4(1-a)}{a\bar p(1-\bar p)}{\rm d}p -c^* }{{\rm d} p} &(\eqref{def:cp}) \\
= & \frac{4(1-a)}{a\bar p(1-\bar p)},
\end{align*}
which completes the proof.
\end{proof}

%With Lemma~\ref{lemma:small:derivative-c*}, we are ready to prove Proposition~\ref{claim:small-c(p)-(p-hat)}, which we reiterate as follows:
%\begin{repproposition}{claim:small-c(p)-(p-hat)}
%For all $ 0 < \hat{p}\leq p <1 $,
%$$ 0 \leq c^*(p) - c^*(\hat{p}) \leq \frac{4(1-a)}{a}\log\left(\frac{p(1-\hat{p})}{\hat{p}(1-p)}\right) .$$
%\end{repproposition}
%\begin{proof}

With Lemma~\ref{lemma:small:derivative-c*}, we are ready to prove the proposition:
First, we prove the lower bound, $c^*(p) - c^*(\hat{p}) \geq 0 $.
We start with an optimal solution $(\lambda^*, n_1^*, n_2^*, \eta_1^*, \eta_2^*, c^*)$ of $\ref{MP1}(p)$.
Consider the follwoing tuple $(\lambda^*, n_1^*, n_2^*, \eta_1', \eta_2', c^*)$ where
$\eta_1' = \frac {n_1^*\lambda^* (p-\hat{p} )+\eta_1^*(1-p)}{1- \hat{p} }$ and $\eta_2' = \frac {n_2^*\lambda^* (p-\hat{p} )+\eta_2^*(1-p)}{1- \hat{p} }.$
It is easy to verify that $(\lambda^*, n_1^*, n_2^*, \eta_1', \eta_2', c^*)$ satisfies Constraints~\eqref{constraint:x<=1}-\eqref{constraint:n1'+n2'big} in $\ref{MP1}(\hat{p})$.
Below we prove that $(\lambda^*, n_1^*, n_2^*, \eta_1', \eta_2', c^*)$ also satisfies Constraints~\eqref{constraint:not-c-competitive}-\eqref{constraint:u_2>=b} in $\ref{MP1}(p')$.
To prove this, we denote $\tilde o_{1}^*$, $\tilde o_{2}^*$, $\tilde u_{1}^*$, $\tilde u_{1,2}^*$ the values of $\tilde o_{1}$, $\tilde o_{2}$, $\tilde u_{1}$, $\tilde u_{1,2}$ corresponds to the original (optimal) solution in $\ref{MP1}(p)$, and $\tilde o_{1}'$, $\tilde o_{2}'$, $\tilde u_{1}'$, $\tilde u_{1,2}'$ the values corresponding to the modified solution (the tuple $(\lambda^*, n_1^*, n_2^*, \eta_1', \eta_2', c^*)$) in $\ref{MP1}(\hat{p})$.

By the definition of $\eta_1'$ and $\eta_2'$, we have $\tilde o_1' = \tilde o_1^*$ and $\tilde o_2' = \tilde o_2^*$.
This, conbined with the observation that the right hand side of Constraint~\eqref{constraint:not-c-competitive} is non-increasing in $\tilde u_1$ and the left hand side of Constraint~\eqref{constraint:u_2>=b} is increasing in $\tilde u_{1,2}$, indicates that if $\tilde u_1' \geq \tilde u_1^*$ and $\tilde u_{1,2}' \geq \tilde u_{1,2}^*$, then $(\lambda^*, n_1^*, n_2^*, \eta_1', \eta_2', c^*)$ satisfies Constraints~\eqref{constraint:not-c-competitive}-\eqref{constraint:u_2>=b} in $\ref{MP1}(p')$.
Thus it suffices to prove $\tilde u_1' \geq \tilde u_1^*$ and $\tilde u_{1,2}' \geq \tilde u_{1,2}^*$.

Using the definition of $\tilde u_{1}$ in~\ref{MP1}, for proving $\tilde u_1' \geq \tilde u_1^*$, it is sufficient to show
$$ \frac{\tilde o_1' }{ \lambda^* \hat{p} } \geq \frac{\tilde o_1^*}{\lambda^* p }, $$
and
$$ \frac{\tilde o_1' +(1- p ) (1-\lambda^*)n}{1- p + \lambda^* p } \geq \frac{\tilde o_1^* +(1- \hat{p}) (1-\lambda^*)n}{1- \hat{p}+ \lambda^*\hat{p} } . $$
The first inequality is due to $\tilde o_1' = \tilde o_1^*$ and $p \geq \hat{p}$, and the second is implied by
\begin{align*}
\frac{\tilde o_1' +(1- p ) (1-\lambda^*)n}{1- p + \lambda^* p } - \frac{\tilde o_1^* +(1- \hat{p}) (1-\lambda^*)n}{1- \hat{p}+ \lambda^*\hat{p} }  = \frac{(\lambda^* n- o_1^*)( \hat{p} - p )(1-\lambda^*)}{(1- \hat{p}+ \lambda^* \hat{p} )(1- p + \lambda^* p )} \geq 0.\end{align*}
Similarly, we have $\tilde u_{1,2}' \geq \tilde u_{1,2}^*$, which completes the proof for the lower bound.

The upper bound, $c^*(p) - c^*(\hat{p}) \leq \frac{4(1-a)}{a}\log\left(\frac{p(1-\hat{p})}{\hat{p}(1-p)}\right)$, relies on Lemma~\ref{lemma:small:derivative-c*}:
\begin{align*}
& c^*(p)-c^*(\hat{p}) = \int_{t=\hat{p}}^{p} {\rm d} c^*(t) \\
\leq & \int_{t=\hat{p}}^{p} \frac{4(1-a)}{at(1-t)}{\rm d} t & (\text{Lemma~\ref{lemma:small:derivative-c*}})\\
= & \left[\frac{4(1-a)}{a}\left(\log p - \log (1-p)\right)\right] - \left[\frac{4(1-a)}{a}\left(\log \hat{p} - \log (1-\hat{p})\right)\right] \\
= & \frac{4(1-a)}{a}\log\left(\frac{p(1-\hat{p})}{\hat{p}(1-p)}\right),
\end{align*}
which completes the proof.
\end{proof}

%\section{Proof of Proposition~\ref{prop:robust}}\label{sec:proof-prop-robust}
%In this section, we prove Proposition~\ref{prop:robust}, which we reiterate as follows:
%\begin{repproposition}{prop:robust}
%Under Approximations ~\ref{assumption:continuous instance}, and~\ref{assumption:determ-observ},
%if the true probability in the partially-\st{learnable}\vmn{predictable} model is $\hat{p}$, then,
%\begin{enumerate}
%\item if $p < \hat{p}$, then $ALG_{2,c,p}$ has a competitive ratio of at least $c$ for all $c\leq c^*(p)$.
%\item if $\hat{p}<p<\hat{p} + \delta_{\hat{p}}$, where $\delta_{\hat{p}}$ is a small positive constant, then $ALG_{2,c,p}$ has a competitive ratio of at least
%$$ \begin{cases}
%c^*(\hat{p})  - \frac{4(1-a)}{1-c} \log\left(\frac{p(1-\hat{p})}{\hat{p}(1-p)}\right) &,\text{ if }c^*(\hat{p})<c \leq c^*(p),\\
%c\left[ 1- \frac{(1-a)(p - \hat{p})}{p} \right] &,\text{ if }c \leq c^*(\hat{p}).\end{cases}$$
%\end{enumerate}
%\end{repproposition}

\begin{proof}[Proof of Proposition~\ref{prop:robust}]
Proposition~\ref{prop:robust} consists of two cases based on whether we underestimate the problem parameter ($p < \hat{p}$) or overestimate it ($p > \hat{p}$).
For proving Proposition~\ref{prop:robust}, we first introduce two lemmas, each of them deals with one of the two cases.

When we underestimate the true parameter ($ p < \hat{p}$), we have the following lemma:
\begin{lemma}\label{thm:adaptive-threshold-p-hat}
Under Approximations~\ref{assumption:continuous instance}, and~\ref{assumption:determ-observ},
if the true probability in the customer arrival model is $\hat{p}$ and $p < \hat{p}$, then $ALG_{2,c,p}$ has a competitive ratio of at least $c$ for all $c\leq c^*(p)$.
\end{lemma}
\begin{proof}
For any adversarial problem instance ${\vec{v}}'$ given by functions $\eta_1$ and $\eta_2$ in Approximation~\ref{assumption:continuous instance}, we define the following auxiliary problem instance $\bar {\vec{v}}'$ given by functions $\bar \eta_1$ and $\bar \eta_2$, where for all $\lambda\in [0,1]$, $$\bar\eta_1(\lambda) \triangleq \frac {n_1\lambda (\hat{p}- p )+\eta_1(\lambda)(1-\hat{p})}{1- p }\text{ and }\bar \eta_2(\lambda) \triangleq \frac {n_2\lambda (\hat{p}- p )+\eta_2(\lambda)(1-\hat{p})}{1- p }.$$
The validity of functions $\bar \eta_1(\lambda)$ and $\bar \eta_2(\lambda)$ are easy to verify because our construction of each of them is a convex combination of two obviously valid instances $(n_1\lambda, n_2\lambda)$ and $(\eta_1(\lambda), \eta_2(\lambda))$ with weights $\frac{\hat{p}- p }{1- p }$ and $ \frac{1-\hat{p}}{1- p }$ and the constraints of valid instances are linear.

Note that for all $\lambda\in[0,1]$, $$n_1\lambda \hat{p}+ \eta_1(\lambda)(1-\hat{p}) = n_1\lambda p + \bar \eta_1(\lambda)(1- p )\text{ and }n_2\lambda \hat{p}+ \eta_2(\lambda)(1-\hat{p}) = n_2\lambda p + \bar \eta_2(\lambda)(1- p ).$$
This means that the values of $ \{ o_1(\lambda) , o_2(\lambda)\}_{\lambda \in [0,1]}$ given by ${\vec{v}}'$ with the true problem parameter $\hat{p}$ is the same as those given by $\bar {\vec{v}}'$ with problem parameter $p$.
As a result, $ALG_{2,c,p}$ makes the same decision for the above two scenarios.
The rest follows from that the ratio between $ALG_{2,c,p}$ and $OPT$ is at least $c$ against $\bar {\vec{v}}'$ with problem parameter $p$ and that $OPT({\vec{v}}')=OPT(\bar {\vec{v}}')$ (due to $\eta_1(1)=\bar\eta_1(1)$ and $\eta_2(1)=\bar\eta_2(1)$).
\end{proof}
When we overestimate the true parameter ($ p >\hat{p}$), we have the following lemma:
\begin{lemma}\label{thm:optimistic-error}
Under Approximations~\ref{assumption:continuous instance}, and~\ref{assumption:determ-observ},
if the true probability in the customer arrival model is $\hat{p}$ and $p > \hat{p}$, then for all $c\leq c^*(p)$, $ALG_{2,c,p}$ has a competitive ratio of at least $$ \begin{cases}
\min\left\{1-\frac{(1-a)(p-\hat{p})}{a \hat{p}+(1-a)(p-\hat{p})}, c\left( 1- \frac{(1-a)(p - \hat{p})}{p} \right), c^*(\hat{p}) - \frac{4(1-a)}{1-c} \log\left(\frac{p(1-\hat{p})}{\hat{p}(1-p)}\right) \right\} &,\text{ if }c>c^*(\hat{p}),\\
\min\left\{1-\frac{(1-a)(p-\hat{p})}{a \hat{p}+(1-a)(p-\hat{p})}, c\left( 1- \frac{(1-a)(p - \hat{p})}{p} \right) \right\} &,\text{ if }c \leq c^*(\hat{p}).\end{cases}$$
\end{lemma}

For proving Lemma~\ref{thm:optimistic-error}, for all $\lambda\in[0,1]$, we denote $u_1(\lambda, p)$ ($u_{1,2}(\lambda, p)$, respectively) to be the upper bound on $n_1$ ($n_1+n_2$, respectively) we compute assuming the probability is $p$ at time $\lambda$ (while we still observe the $o_1(\lambda)$ and $o_2(\lambda)$ given by the model with the true probability parameter $\hat{p}$).

For proving Lemma~\ref{thm:optimistic-error}, we need a series of lemmas. The first lemma indicates that when $p$ is not too far away from $\hat{p}$, the upper bounds we compute is not too far from the ones if we used the real problem parameter $\hat{p}$:

\begin{lemma}\label{claim:not-bad-u1-u2}
If $p \geq \hat{p}$, then
\begin{align*} % \label{ineq:u1_hat_p_p}
u_1(\lambda, p ) \leq u_1(\lambda, \hat{p} ) \leq \frac{ p }{\hat{p}}u_1(\lambda, p )& & \text{ and } & &u_{1,2}(\lambda, p ) \leq u_{1,2}(\lambda, \hat{p} ) \leq \frac{ p }{\hat{p}}u_{1,2}(\lambda, p ).
\end{align*}
\end{lemma}

\begin{proof}
The two statements are essentially the same, so proving the inequalities regarding $u_1$ is sufficient.

We first prove the first part of the inequality, $u_1(\lambda, p ) \leq u_1(\lambda, \hat{p} )$.
Using Definition~\ref{eq:u_1}, it is sufficient to show
$$ \frac{o_1(\lambda) }{ \lambda\hat{p} } \geq \frac{o_1(\lambda)}{\lambda p }, $$
and
$$ \frac{o_1(\lambda) +(1- \hat{p}) (1-\lambda)n}{1- \hat{p}+ \lambda\hat{p} } \geq \frac{o_1(\lambda) +(1- p ) (1-\lambda)n}{1- p + \lambda p }. $$
The first inequality is trivial, and the second is implied by
\begin{align*}
\frac{o_1(\lambda) +(1- \hat{p}) (1-\lambda)n}{1- \hat{p}+ \lambda\hat{p} } - \frac{o_1(\lambda) +(1- p ) (1-\lambda)n}{1- p + \lambda p } = \frac{(\lambda n- o_1(\lambda))( p -\hat{p})(1-\lambda)}{(1- \hat{p}+ \lambda\hat{p} )(1- p + \lambda p )} \geq 0.\end{align*}

Let us prove $u_1(\lambda, \hat{p} ) \leq \frac{ p }{\hat{p}}u_1(\lambda, p )$ now.
To do this, we consider the following two cases separately:
$u_1(\lambda, p )= \frac{o_1(\lambda)}{\lambda p }$, and $u_1(\lambda, p ) = \frac{o_1(\lambda) +(1- p ) (1-\lambda)n}{1- p + \lambda p }.$
The first case is easy, we have
$$ u_1(\lambda, \hat{p}) \leq \frac{o_1(\lambda)}{\lambda \hat{p} } = \frac{p}{\hat{p}} \frac{o_1(\lambda)}{\lambda p } = \frac{p}{ \hat{p} } u_1(\lambda, p ) .$$
This case implies $\frac{o_1(\lambda)}{\lambda p } \geq\frac{o_1(\lambda) +(1- p ) (1-\lambda) n}{1- p + \lambda p },$
which is equivalent to
\begin{align}
o_1(\lambda) \geq (1-\lambda)\lambda p n. \label{ineq:o1-big-typer-2-ub}
\end{align}
Replacing $o_1(\lambda)$ in $\frac{o_1(\lambda) +(1- p ) (1-\lambda) n }{1- p + \lambda p }$ using \eqref{ineq:o1-big-typer-2-ub} gives
\begin{align}
u_1(\lambda, p ) = \frac{o_1(\lambda) +(1- p ) (1-\lambda) n }{1- p + \lambda p }\geq (1-\lambda) n. \label{ineq:lower-bound-u-1-case-2}
\end{align}
Using \eqref{ineq:o1-big-typer-2-ub} and \eqref{ineq:lower-bound-u-1-case-2},
\begin{align*}
u_1(\lambda, \hat{p})-u_1(\lambda , p ) & \leq \frac{o_1(\lambda) +(1- \hat{p}) (1-\lambda)n}{1- \hat{p}+ \lambda\hat{p} } - \frac{o_1(\lambda) +(1- p ) (1-\lambda)n}{1- p + \lambda p } \\
& = \frac{(\lambda n- o_1(\lambda))( p -\hat{p})(1-\lambda)}{(1- \hat{p}+ \lambda\hat{p} )(1- p + \lambda p )}\\
& \leq \frac{\lambda n (1- p + \lambda p )( p -\hat{p})(1-\lambda)}{(1- \hat{p}+ \lambda\hat{p} )(1- p + \lambda p )} &(\text{using }~\eqref{ineq:o1-big-typer-2-ub})\\
& = \frac{\lambda n ( p -\hat{p})(1-\lambda)}{(1- \hat{p}+ \lambda\hat{p} ))} \\
& \leq \frac{\lambda ( p -\hat{p})}{(1- \hat{p}+ \lambda\hat{p} ))} u_1(\lambda, p) &(\text{using }~\eqref{ineq:lower-bound-u-1-case-2})
\\
& \leq ( p -\hat{p}) u_1(\lambda, p). &( 1- \hat{p}+ \lambda\hat{p} \geq (1- \hat{p}) \lambda+ \lambda\hat{p} = \lambda)
\end{align*}
Rearranging terms, we get
\begin{align*}
u_1(\lambda, \hat{p}) & \leq ( 1+ p -\hat{p}) u_1(\lambda, p) \leq \left( 1+\frac{p-\hat{p}}{\hat{p}}\right) u_1(\lambda , p)=\frac{p}{\hat{p}} u_1(\lambda, p),\end{align*}
and hence we are done.
\end{proof}

We break down the rest of proof of  Lemma~\ref{thm:optimistic-error} into two cases based on whether $ALG_{2,c,p}$ exhausts the capacity or not, and each of the following lemmas is useful for one of the two cases.

We first consider the case where $ALG_{2,c,p}$ exhausts the capacity, i.e., $q_1(1)+q_2(1)=b$.

We consider the last \st{class}\vmn{type}-$2$ customer accepted by $ALG_{2,c, p }$.
(If $ALG_{2,c, p }$ does not accept any \st{class}\vmn{type}-$2$ customer, then it has a total value of $b$, which is the optimal so the case is trivial.)
Say the customer arrives at time $ \lambda $.
Since $ALG_{2,c,p}$ accepts a customer at time $\lambda$, one of the following condition must hold:
$u_{1,2}(\lambda, p ) \leq b$, or
\begin{align} \frac{ b- {q}_2\left(\lambda \right) + {q}_2\left(\lambda \right) a} {\min \{ u_1(\lambda, p ) , b\}(1-a)+ab} \geq c. \label{ineq:condition-at-least-c}
\end{align}
For each of the two cases, we have one of the following two lemmas.
\begin{lemma}\label{lemma:p>phat-u12<=b} Under the setting of Lemma~\ref{thm:optimistic-error}, if $q_1(1)+q_2(1)=b$ and $u_{1,2}(\lambda, p ) \leq b$, then
$\frac{ALG_{2,c, p }({\vec{v}})}{OPT} \geq 1-\frac{(1-a)(p-\hat{p})}{a \hat{p}+(1-a)(p-\hat{p})}. %\label{ineq:q1+q2=b_u2>b}
$
\end{lemma}
\begin{proof}
From Lemma~\ref{claim:not-bad-u1-u2} and using the fact that $u_{1,2}(\lambda, \hat{p} )$ is an upper bound of $n_1+n_2$ (given by Lemma~\ref{prop:Ubounds}), we have
\begin{align}
n_1+n_2 \leq u_{1,2}(\lambda, \hat{p} ) \leq \frac{ p }{\hat{p}} u_{1,2}(\lambda, p ) \leq \frac{ p }{\hat{p}}b.\label{inequality:n1+n2<=p/phatb}
\end{align}
As a result,
\begin{align}
\frac{ALG_{2,c, p }({\vec{v}})}{OPT} \geq & \frac{a \min\{b, n_2\}+b-\min\{b, n_2\}}{\min\{b, n_1\}+(b-\min\{b, n_1\}) a} = \frac{b-(1-a)\min\{b, n_2\}}{ab + (1-a) \min\{b, n_1\}}. \label{ineq:a_ratio}
\end{align}
For general $p/\hat{p}$ ratio, the above ratio is at least $a$, which happens when $\min\{b, n_1\} = \min\{b, n_2\} = b$.
Therefore, when $p/\hat{p} \geq 2$,
\begin{align*}
\frac{ALG_{2,c, p }({\vec{v}})}{OPT} \geq a \geq \frac{a \hat{p}}{ \hat{p}+(1-a)(p-2\hat{p})} = 1-\frac{(1-a)(p-\hat{p})}{a \hat{p}+(1-a)(p-\hat{p})}.%\label{ineq:general-at-least-a}
\end{align*}
When $p/\hat{p}<2$, we can derive a better bound using~\eqref{inequality:n1+n2<=p/phatb}.
We have
\begin{align*}
\frac{ALG_{2,c, p }({\vec{v}})}{OPT} \geq &  \frac{b-(1-a)\min\{b, n_2\}}{ab + (1-a) \min\{b, n_1\}}    &(\eqref{ineq:a_ratio}) \\
\geq & \frac{b-(1-a) \min\{b, n_2\} }{ab + (1-a) n_1}. &(n_1 \geq \min\{b, n_1\})
\end{align*}
Combine the above in equality and the trick that when $0<x<y$, for any $x > \delta \geq 0$, $(x-\delta) / (y-\delta) \leq x/y$, we have, either
\begin{align}
\frac{ALG_{2,c, p }({\vec{v}})}{OPT} \geq
\frac{b-(1-a) \min\{b, n_2\} }{ab + (1-a) n_1} \geq 1,\label{ineq:adp/opt>=1}
\end{align}
or
\begin{align}
\frac{ALG_{2,c, p }({\vec{v}})}{OPT} \geq &
\frac{b-(1-a) \min\{b, n_2\} }{ab + (1-a) n_1} \geq \frac{b-(1-a) \min\{b, n_2\}- (1-a)(b-\min\{b, n_2\})}{ab + (1-a) n_1-(1-a)(b-\min\{b, n_2\})} \nonumber
\\ = & \frac{b-(1-a)b}{ab + (1-a) ( n_1 + \min\{b, n_2\} -b)}   .\label{ineq:adp/opt<1}
\end{align}
When \eqref{ineq:adp/opt>=1},
$$\frac{ALG_{2,c, p }({\vec{v}})}{OPT} \geq 1
 \geq 1-\frac{(1-a)(p-\hat{p})}{a \hat{p}+(1-a)(p-\hat{p})} ,$$
 and we are done.
 When \eqref{ineq:adp/opt<1},
 \begin{align*}
\frac{ALG_{2,c, p }({\vec{v}})}{OPT} \geq & \frac{b-(1-a)b}{ab + (1-a) ( n_1 + \min\{b, n_2\} -b)} & \\
\geq & \frac{b-(1-a)b}{ab + (1-a) ( n_1 + n_2 -b)} &(\min\{b, n_2\}\leq n_2)\\
\geq & \frac{b-(1-a)b}{ab + (1-a) \left(\frac{p}{\hat{p}}-1\right)b} &(\eqref{inequality:n1+n2<=p/phatb}) \\
 = & \frac{a \hat{p}}{a \hat{p}+(1-a)(p-\hat{p})} = 1-\frac{(1-a)(p-\hat{p})}{a \hat{p}+(1-a)(p-\hat{p})},
\end{align*}
which completes the proof.
% Because when $p> 2\hat{p}$, the above bound is looser than that of~\eqref{ineq:general-at-least-a},~\eqref{ineq:q1+q2=b_u2>b} holds for general pairs of $(p , \hat{p})$.
\end{proof}

\begin{lemma} \label{lemma:inequality:q1+q2=b_ratio>=c}
Under the setting of Lemma~\ref{thm:optimistic-error}, if $q_1(1)+q_2(1)=b$ and~\eqref{ineq:condition-at-least-c}, then
$
\frac{ALG_{2,c, p }}{OPT} \geq c\left( 1- \frac{(1-a)(p - \hat{p})}{p} \right). %\label{inequality:q1+q2=b_ratio>=c}
$
\end{lemma}

\begin{proof}From Lemma~\ref{claim:not-bad-u1-u2}, we have $n_1 \leq u_1(\lambda , \hat{p}) \leq \frac{ p }{\hat{p}} u_1 (\lambda , p)$, and thus
$OPT \leq \min \{ n_1 , b\}(1-a)+ab \leq \min \{ \frac{ p }{\hat{p}} u_1(\lambda, p ) , b\}(1-a)+ab $.
Combing this with $ALG_{2,c, p}({\vec{v}})= b- {q}_2\left(\lambda \right) + {q}_2\left(\lambda \right) a$, we obtain
\begin{align}
\frac{ALG_{2,c, p }({\vec{v}})}{OPT} \geq & \frac{ b- {q}_2\left(\lambda \right) + {q}_2\left(\lambda \right) a}{\min\{ \frac{ p }{\hat{p}} u_1(\lambda, p ) , b\}(1-a)+ab } \nonumber \\
\geq & c \frac{\min\{ u_1(\lambda, p ) , b\}(1-a)+ab }{ \min\{ \frac{ p }{\hat{p}} u_1(\lambda, p ) , b\}(1-a)+ab } &(\eqref{ineq:condition-at-least-c}) \nonumber \\ = & cf(u_1(\lambda, p )),\label{ineq:adp/opt>=a_function}
\end{align}
where $f(x)\triangleq \frac{\min\{x , b\}(1-a)+ab }{ \min\{ \frac{ p }{\hat{p}} x, b\}(1-a)+ab }. $
We next show that $f(x)$ is minimized when $x= \frac{\hat{p}}{p}b$.
When $x \geq \frac{\hat{p}}{p}b$, $f(x) = \frac{\min\{x , b\}(1-a)+ab }{ b(1-a)+ab },$ and hence $f(x)$ is non-decreasing in $x$.
On the other hand, when $x \leq \frac{\hat{p}}{p}b$, $f(x) = \frac{x(1-a)+ab }{  \frac{ p }{\hat{p}} x(1-a)+ab },$ and hence $f(x)$ is non-increasing in $x$. Combining the above two cases, $f(x)$ is minimized at $x=\frac{\hat{p}}{p}b$.
As a result, \eqref{ineq:adp/opt>=a_function} gives
\begin{align*}
\frac{ALG_{2,c, p }({\vec{v}})}{OPT} \geq cf(u_1(\lambda, p )) \geq c f(\frac{\hat{p}}{p}b)
\geq & c \frac{\frac{\hat{p}}{p}b(1-a)+ab}{b(1-a)+ab} = c\left( 1- \frac{(1-a)(p - \hat{p})}{p} \right).
\end{align*}
\end{proof}

Now we consider the case where $ALG_{2,c,p}$ does not exhaust the capacity, i.e., $q_1(1)+q_2(1)<b$.
We discuss the cases $c\leq c^*(p)$ and $c^*(\hat{p}) < c \leq c^*(p)$ separately in the following two lemmas.\begin{lemma}\label{claim:optimistic-p-not-exhausting-small-c}
Under the setting of Lemma~\ref{thm:optimistic-error}, if $q_1(1)+q_2(1)<b$ and $c\leq c^*(p)$, then
$\frac{ALG_{2,c, p }}{OPT}\geq c.$
\end{lemma}
\begin{proof}
According to Lemma~\ref{claim:not-bad-u1-u2}, $u_1(\lambda, p ) \leq u_1(\lambda, \hat{p} )$ and $u_{1,2}(\lambda, p ) \leq u_{1,2}(\lambda, \hat{p} )$.
As a result, when a \st{class}\vmn{type}-$2$ customer arrives, if $ALG_{2,c, \hat{p} }$ accepts the customer, then $ALG_{2,c, p }$ either accepts the customer or has already accepted at least the same amount of \st{class}\vmn{type}-$2$ customers as $ALG_{2,c, \hat{p} }$ has.
As a result, $ALG_{2,c, p }$ accepts at least as many \st{class}\vmn{type}-$2$ customers as $ALG_{2,c, \hat{p} }$ does.
Furthermore,since $q_1(1)+q_2(1)<b$, $ALG_{2,c, p }$ accepts all \st{class}\vmn{type}-$1$ customers.
Therefore, $ALG_{2,c, p } \geq ALG_{2,c, \hat{p} }$.
At last, $\frac{ALG_{2,c, p }}{OPT} \geq \frac{ALG_{2,c, \hat{p} }}{OPT} \geq c$.
\end{proof}

\begin{lemma}\label{lemma:optimistic-p-not-exhausting-large-c}
Under the setting of Lemma~\ref{thm:optimistic-error}, if $q_1(1)+q_2(1)<b$ and $c^*(\hat{p}) < c \leq c^*(p)$, then
$\frac{ALG_{2,c, p }({\vec{v}})}{OPT}\geq c^*(\hat{p}) - \frac{4(1-a)}{1-c} \log\left(\frac{p(1-\hat{p})}{\hat{p}(1-p)}\right).$
\end{lemma}
\begin{proof}
Let $\lambda$ be the last time we reject a \st{class}\vmn{type}-$2$ customer.
According to
\eqref{ineq:lower-bound-ratio-ADP},
\begin{align}
\frac{ALG_{2,c, p }({\vec{v}})}{OPT}
\geq & \frac{n_1 + a\left(\frac{1-c}{1-a} b + c \left(b - u_1(\lambda, p) \right)^+ + \left[n_2 - o_2(\lambda)\right]\right) }{n_1 + a \min \{b-n_1,n_2\}}.\nonumber \\
\geq & \frac{a\left(\frac{b}{1-a}+n_2-o_2(\lambda)\right) - c\left( \frac{a^2}{1-a}b+ a \min\{u_1(\lambda, \hat{p}), b\}\right) + n_1 }{a\min\{n_1+n_2,b\}+(1-a)n_1} .&(\text{Lemma~\ref{claim:not-bad-u1-u2}}) \label{lower-bound-ratio}
\end{align}
Let us consider the tuple $(\lambda , n_1, n_2, \eta_1( \lambda), \eta_2(\lambda), c)$ in~\ref{MP1}$(\hat{p})$.
Recall that under Approximation~\ref{assumption:determ-observ}, $o_1(\lambda)=\tilde o_1$ , $o_2(\lambda)=\tilde o_1$, and $u_1(\lambda, \hat{p})=\tilde u_1$.
Due to Lemma~\ref{claim:not-bad-u1-u2},
$$u_{1,2}(\lambda, \hat{p}) \geq u_{1,2}(\lambda, p) \geq b.$$
Thus, $(\lambda , n_1, n_2, \eta_1( \lambda), \eta_2(\lambda), c)$ satisfies Constraints~\eqref{constraint:u_2>=b}-~\eqref{constraint:n1'+n2'big} in~\ref{MP1}$(\hat{p})$.
Following the argument deriving \eqref{ineq:good-ratio} and \eqref{ineq:key-ADP-full} regarding~\ref{MP1}$(\hat{p})$, the optimal objective value $c^*(\hat{p})$ satisfies
\begin{align} c^*(\hat{p}) \leq \frac{a\left(\frac{b}{1-a}+n_2-\tilde o_2\right) - c^*(\hat{p})\left( \frac{a^2}{1-a}b+ a\min\{\tilde u_1, b\}\right) + n_1 }{a\min\{n_1+n_2,b\}+(1-a)n_1}.\label{inequality:mp1-p-hat}
\end{align}
Therefore,
\begin{align}
\frac{ALG_{2,c, p }({\vec{v}})}{OPT} \geq & \frac{a\left(\frac{b}{1-a}+n_2-\tilde o_2\right) - c\left( \frac{a^2}{1-a}b+ a \min\{\tilde u_1, b\}\right) + n_1 }{a\min\{n_1+n_2,b\}+(1-a)n_1} &(\eqref{lower-bound-ratio} ) \nonumber \\
= & \frac{a\left(\frac{b}{1-a}+n_2-\tilde o_2\right) - c^*(\hat{p})\left( \frac{a^2}{1-a}b+ a \min\{\tilde u_1, b\}\right) + n_1 }{a\min\{n_1+n_2,b\}} \nonumber \\
- & \frac{a (c - c^*(\hat{p}))\left( \frac{a^2}{1-a}b+ a \min\{\tilde u_1, b\}\right) }{a\min\{n_1+n_2,b\}+(1-a)n_1}\nonumber  \\
\geq & c^*(\hat{p}) - \frac{a (c^*(p) - c^*(\hat{p}))\left( \frac{a^2}{1-a}b+ a \min\{\tilde u_1, b\}\right) }{a\min\{n_1+n_2,b\}+(1-a)n_1}. &(\eqref{inequality:mp1-p-hat}\text{ and }c \leq c^*(p)) \label{ineq:74.1}
\end{align}
Recall that we always accept at least $\phi b = \frac{1-c}{1-a}b $ \st{class}\vmn{type}-$2$ customers.
Thus, in this case ($q_1(1)+q_2(1)<b$), if $n_2 < \frac{1-c}{1-a}b$, the we have a ratio of $1$.
Thus, without loss of generality, we assume $n_2 \geq \frac{1-c}{1-a}b$.
As a result,
\begin{align*}
\frac{ALG_{2,c, p }({\vec{v}})}{OPT} \geq & c^*(\hat{p}) - \frac{a (c^*(p) - c^*(\hat{p}))\left( \frac{a}{1-a}b\right) }{a\frac{1-c}{1-a}b} \nonumber &(\eqref{ineq:74.1})
\\ = & c^*(\hat{p}) - (c^*(p) - c^*(\hat{p})) \frac{a}{1-c} \nonumber \\
\geq & c^*(\hat{p}) - \frac{4(1-a)}{1-c} \log\left(\frac{p(1-\hat{p})}{\hat{p}(1-p)}\right), &(Proposition~\ref{claim:small-c(p)-(p-hat)})
\end{align*}
which completes the proof.
\end{proof}

With Lemmas~\ref{lemma:inequality:q1+q2=b_ratio>=c},~\ref{claim:optimistic-p-not-exhausting-small-c},~\ref{claim:optimistic-p-not-exhausting-small-c} and~\ref{lemma:optimistic-p-not-exhausting-large-c}, we are ready complete the proof of Lemma ~\ref{thm:optimistic-error}:
%to prove Lemma~\ref{thm:optimistic-error}, which we reiterate as follows:
%\begin{replemma}{thm:optimistic-error}
%Under Approximations~\ref{assumption:continuous instance}, and~\ref{assumption:determ-observ},
%if the true probability in the customer arrival model is $\hat{p}$ and $p > \hat{p}$, then for all $c\leq c^*(p)$, $ALG_{2,c,p}$ has a competitive ratio of at least $$ \begin{cases}
%\min\left\{1-\frac{(1-a)(p-\hat{p})}{a \hat{p}+(1-a)(p-\hat{p})}, c\left( 1- \frac{(1-a)(p - \hat{p})}{p} \right), c^*(\hat{p}) - \frac{4(1-a)}{1-c} \log\left(\frac{p(1-\hat{p})}{\hat{p}(1-p)}\right) \right\} &,\text{ if }c>c^*(\hat{p}),\\
%\min\left\{1-\frac{(1-a)(p-\hat{p})}{a \hat{p}+(1-a)(p-\hat{p})}, c\left( 1- \frac{(1-a)(p - \hat{p})}{p} \right) \right\} &,\text{ if }c \leq c^*(\hat{p}).\end{cases}$$
%\end{replemma}
%\begin{proof}
The lemma
%Lemma~\ref{thm:optimistic-error}
follows directly from Lemmas~\ref{lemma:inequality:q1+q2=b_ratio>=c},~\ref{claim:optimistic-p-not-exhausting-small-c},~\ref{claim:optimistic-p-not-exhausting-small-c} and~\ref{lemma:optimistic-p-not-exhausting-large-c} and $c\left( 1- \frac{(1-a)(p - \hat{p})}{p} \right)< c $.
%\end{proof}

With Lemmas~\ref{thm:adaptive-threshold-p-hat} and~\ref{thm:optimistic-error}, we are ready to complete the proof of Proposition~\ref{prop:robust}:
%prove Proposition~\ref{prop:robust}, which we reiterate as follows:
%\begin{repproposition}{prop:robust}
%Under Approximations ~\ref{assumption:continuous instance}, and~\ref{assumption:determ-observ},
%if the true probability in the partially-\st{learnable}\vmn{predictable} model is $\hat{p}$, then,
%\begin{enumerate}
%\item if $p < \hat{p}$, then $ALG_{2,c,p}$ has a competitive ratio of at least $c$ for all $c\leq c^*(p)$.
%\item if $\hat{p}<p<\hat{p} + \delta_{\hat{p}}$, where $\delta_{\hat{p}}$ is a small positive constant, then $ALG_{2,c,p}$ has a competitive ratio of at least
%$$ \begin{cases}
%c^*(\hat{p})  - \frac{4(1-a)}{1-c} \log\left(\frac{p(1-\hat{p})}{\hat{p}(1-p)}\right) &,\text{ if }c^*(\hat{p})<c \leq c^*(p),\\
%c\left[ 1- \frac{(1-a)(p - \hat{p})}{p} \right] &,\text{ if }c \leq c^*(\hat{p}).\end{cases}$$
%\end{enumerate}
%\end{repproposition}
%\begin{proof}
Lemmas~\ref{thm:adaptive-threshold-p-hat} proves the case where $p<\hat{p}$.
For the case where $\hat{p}<p<\hat{p} + \delta_{\hat{p}}$, we compare the terms in each of the minimizations in Lemma~\ref{thm:optimistic-error}.
Clearly, there is a small positive number $\delta_{\hat{p}}>0$ such that when $0<p-\hat{p}<\delta_{\hat{p}}$, the competitive ratio is at least
$$ \begin{cases}
c^*(\hat{p}) - \frac{4(1-a)}{1-c} \log\left(\frac{p(1-\hat{p})}{\hat{p}(1-p)}\right) &,\text{ if }c>c^*(\hat{p}),\\
c\left( 1- \frac{(1-a)(p - \hat{p})}{p} \right) &,\text{ if }c \leq c^*(\hat{p}).\end{cases}$$
Hence the proof of Proposition~\ref{prop:robust} is completed.
\end{proof}
%\end{proof}

%\end{document}

\fi 
%\section{Proof of Theorem~\ref{thm:b=1}}\label{sec:proof-b=1}
\section{Missing proofs of Section~\ref{sec:secretary}}\label{sec:proof-b=1}

\if false
\begin{proof}[Proof of Lemma~\ref{lemma:stopping_time}]
Let $X$ denote the number of customers from the {\RG} group among the first $\frac{2\zeta n}{p}$ of all customers.
Then, $X$ is drawn from the binomial distribution $Bin\left( \frac{2\zeta n}{p}, p\right)$.
Note that $T\geq \frac{2\zeta n}{p}$ implies $X\leq \zeta n$.
Thus, $\prob{T\geq \frac{2\zeta n}{p}} \leq \prob{X\leq \zeta n}$.
Chernoff bound gives, for any $\delta \in (0,1)$ and $\mu = \mathbb{E}(X)$, $ \mathbb{P}(X < (1-\delta)\mu) < e^{-\delta^2\mu/2}.$
Therefore, setting $\delta = 1/2$ and $\mu=\zeta n$,
\begin{align*}
\prob{T\geq \frac{2\zeta n}{p}} \leq \prob{X\leq \zeta n} \leq e^{-\frac{1}{4}\zeta n} \leq \epsilon,
\end{align*}
where the last inequality follows from our assumption $n\geq \frac{4}{\zeta}\log \frac{1}{\epsilon}$.
\end{proof}

\begin{proof}[Proof of Theorem~\ref{theorem:one-time-learning}]
For proving this theorem,
%\ref{theorem:one-time-learning}
we first review the Hoeffding-Berstein's Inequality for sampling without replacement, which first appeared in~\cite{VanderVaart1996}.
We quote from~\cite{Agrawal2009b}:
\begin{theorem}\label{thm:HB-Inequality}
Let $u_1, u_2, \dots, u_r$ be random samples without replacement from the real numbers $
{c_1, c_2, \dots, c_R}$. Then for every $t>0$,
$$ \prob{ \left\vert \sum_{i=1}^r u_i - r \bar{c} \right\vert \geq t } \leq 2 \exp \left( - \frac{t^2}{2r\sigma_R^2 + t \Delta_R}\right)$$
where $\Delta_R = \max_i c_i-\min_i c_i$, $\bar{c}=\frac{1}{R}\sum_{i=1}^R c_i$, and $\sigma_R^2 = \frac{1}{R}\sum_{i=1}^R (c_i-\bar{c})^2$.
\end{theorem}
When applying Theorem~\ref{thm:HB-Inequality} to our proof, the variables take values in $\{0,1\}$, and we find it more convenient to use the following corollary:
\begin{corollary}\label{cor:HB-ineq}
In the setting of Theorem~\ref{thm:HB-Inequality}, if $c_1,c_2,\dots, c_R$ take values in $\{0,1\}$, then for every $t>0$,
$$ \prob{ \left\vert \sum_{i=1}^r u_i - r \bar{c} \right\vert \geq t } \leq 2 \exp \left( - \frac{t^2}{2r \bar c + t }\right).$$
\end{corollary}
\begin{proof}
First we note that $\Delta_R \leq 1$. Therefore, it is sufficient to prove $\sigma_R^2 \leq \bar c$, which is given by
\begin{align*}
\sigma_R^2 = & \frac{1}{R}\sum_{i=1}^R (c_i-\bar{c})^2 = \frac{1}{R}\sum_{i=1}^R (c_i^2-2c_i\bar{c}+\bar{c}^2) \\
= & \frac{1}{R}\left[\sum_{i=1}^R (c_i^2)-2\sum_{i=1}^R(c_i)\bar{c}+R\bar{c}^2\right] \\
= & \frac{1}{R} \left[\sum_{i=1}^R (c_i)-2R\bar{c}^2+R\bar{c}^2 \right] &(c_i^2 = c_i, \sum_{i=1}^R(c_i)= R\bar c ) \\
\geq & \frac{1}{R} \sum_{i=1}^R (c_i) = \bar c.
\end{align*}
\end{proof}

For the subsequent proof, we introduce the following notation:
$$x_i(s^*) \triangleq \begin{cases} 0 \text{ if }v_i' \leq s^*, \\ 1\text{ if }v_i' > s^*. \end{cases} $$

According to the definition of $x_i(s^*)$, the algorithm accepts customer $i$ (where $i$ is the order of the customer before the random permutation) when either $\sigma_{\RGS}(i) \leq T$ or $ x_i(s^*)=1$ .
%In addition, the algorithm also selects all of the first $T$ customers.
As a consequence, $T + \sum_{i=1}^n x_i(s^*)$ is an upper bound on the overall number of customers that would have been selected if we did not have the inventory constraint.
In addition, if $T + \sum_{i=1}^n x_i(s^*)\leq b$, then $\sum_{i=1}^n v'_ix_i(s^*)$ is a lower bound on the total revenue of the algorithm (recall that we accept all the first $T$ arrived customers).
The above two observations motivate us to prove the following two lemmas.

The first lemma shows that with high probability, the one-time-learning algorithm does not attempt to sell products to too many customers:
\begin{lemma}\label{lemma:inventory_constraint_hold_high_probabiliy}
If $\epsilon \leq 1$, $ \frac{2\zeta n}{p} \leq \frac{b}{2} $ and
$ \frac{b}{6} \geq \frac{1}{\epsilon^2 \zeta}\log \left( \frac{n+1}{\epsilon}\right)$, then $$\prob{\sum_{i=1}^n x_i(s^*) > \left(b-\frac{2\zeta n}{p}\right)|T\leq n}\leq \epsilon.$$
\end{lemma}
\begin{proof}
%\dawsen{remove this paragraph and condition $T< 2 \zeta n.$}
%We note that when $T\neq \infty$, $S_\RGS$ is a random set of size $\zeta n$ from $\{ 1,2,\dots , n\}$ so we can apply the Hoffding-Bernstein's Inequality for sampling without replacement to find the probability.
%To formalize this, we define the set $S_\RGS'$ as follows.
%If, $T < 2 \zeta n$, we define the set $S_\RGS'$ to be $S_\RGS$.
%Otherwise, we define the set $S_\RGS'$ to be a set of size $\zeta n$ uniformly at random.
%In either case, $S_\RGS'$ is a set of $\zeta n$ customers selected uniformly at random.
%The probability of interest is at most that of $\sum_{i=1}^n x_i(s)\leq \left( 1-\frac{2\zeta n}{pb} \right)b$ when $s$ is calculated with respect to $S_\RGS'$.

We first note that all values of $s$ collectively give a total of $n+1$ distinct sequences of $(x_i(s))_{i=1}^n$.
This is because the number of zeros in the sequence $(x_i(s))_{i=1}^n$ is between $0$ and $n$, and when we fix the number of zeros in a sequence to be $m$, those zeros must correspond to the $m$ $i$'s with the lowest values of $v_i'$.
Therefore, it is sufficient to show that for each particular realization $(\hat{x}_i)_{i=1}^n$ satisfying \begin{align}\sum_{i=1}^n \hat{x}_i > b-\frac{2\zeta n}{p}, \label{ineq:exaust-inventory} \end{align} $
\prob{ (x_i(s^*))_{i=1}^n = (\hat{x}_i)_{i=1}^n } \leq \frac{\epsilon}{n+1}$.
We first make a connection between $x_i(s^*)$ and the optimal solution of~\ref{Primal} and~\ref{Dual}, denoted by $(\{x^* _i, y^* _i\}_{i\in S_\RGS}, s^*)$.
For all $i$ in the {\RG} group ($i \in \RGS$), if $x_{\sigma_\RGS^{-1}(i)}(s^*)=1$, then $v_i = v_{\sigma_\RGS^{-1}(i)}'> s^*$ and thus $v_i + y_i^* > s^*$.
Due to complementary slackness, $x_i^*=1$.
Therefore, for all $i \in \RGS$,
\begin{align*}x_{\sigma_\RGS^{-1}(i)}(s^*) \leq x_i^*. %\label{ineq:xi(s)<=x_i}
\end{align*}
Thus, when $(x_i(s^*))_{i=1}^n = (\hat{x}_i)_{i=1}^n$,
$$ \sum_{i\in S_\RGS} \hat{x}_{\sigma_\RGS^{-1}(i)} = \sum_{i\in S_\RGS} x_{\sigma_\RGS^{-1}(i)}(s^*) \leq \sum_{i\in S_\RGS} x_i^* \leq (1-\epsilon) \zeta \left(b-\frac{2\zeta n}{p}\right) . $$
As a result, \begin{align*}
& \prob{ (x_i(s^*))_{i=1}^n = (\hat{x}_i)_{i=1}^n |T \leq n} \\
\leq & \prob{ \sum_{i\in S_\RGS} \hat{x}_{\sigma_\RGS^{-1}(i)} \leq (1-\epsilon) \zeta \left(b-\frac{2\zeta n}{p}\right) |T \leq n} \\
\leq & \prob{ \sum_{i\in S_\RGS} \hat{x}_{\sigma_\RGS^{-1}(i)} \leq (1-\epsilon) \zeta \sum_{i=1}^n \hat{x}_i |T \leq n} &(\eqref{ineq:exaust-inventory})\\
\leq & \prob{ \left| \sum_{i\in S_\RGS} \hat{x}_{\sigma_\RGS^{-1}(i)} - \zeta \sum_{i=1}^n \hat{x}_i \right| \geq \epsilon \zeta \sum_{i=1}^n \hat{x}_i |T \leq n}
\end{align*}

Conditioned on $T \leq n$, at the end of the algorithm, $S_{\RGS}$ is a random set of size $\zeta n$ (whose indices are $\{\sigma_\RGS^{-1}(i)| i \in S_{\RGS}\}$ before the random permutation) from $\{ 1,2,\dots , n\}$.
Thus, we can use Corollary~\ref{cor:HB-ineq} and give the following upper bound on the probability presented int the above expression.
\begin{align*}
& 2\exp\left( - \frac{\left[\epsilon\zeta \sum_{i=1}^n \hat{x}_i \right]^2}{2\zeta n \frac{\sum_{i=1}^n \hat{x}_i }{n}+ \epsilon\zeta \sum_{i=1}^n \hat{x}_i }\right)
=2 \exp\left(\frac{ - \epsilon^2 \zeta \sum_{i=1}^n \hat{x}_i}{2+\epsilon}\right)
\\
\leq & 2 \exp\left(\frac{ - \epsilon^2 \zeta \left(b-\frac{2\zeta n}{p}\right)}{2+\epsilon}\right) \leq \exp\left(\frac{ - \epsilon^2 \zeta \left(b-\frac{2\zeta n}{p}\right)}{3}\right) \leq \frac{\epsilon}{n+1},
\end{align*}
where the first inequality holds due to \eqref{ineq:exaust-inventory} and last inequality holds because $ \frac{2\zeta n}{p} \leq \frac{b}{2} $ and
$ \frac{b}{6} \geq \frac{1}{\epsilon^2 \zeta}\log \left( \frac{n+1}{\epsilon}\right).$
%Since any infeasible sequence $(x_i(s))_{i=1}^n$ occurs with probability no more than $\epsilon / (n+1)$, and since there are at most $n+1$ of them, the probability that $S_\RGS$ defines an infeasible sequence $(x_i(s))_{i=1}^n$ is at most $\epsilon$.
\end{proof}
The following lemma shows that $\sum_{i=1}^n x_i(s^*)v_i'$ is close to $OPT$:
\begin{lemma}
\label{lemma:high_objective_value}
If $ \frac{2\zeta n}{p} \leq \frac{b}{2} $ and $b \geq \frac{12 \log\left(\frac{n}{\epsilon}\right)}{\epsilon^3}$,
then $$\prob{\sum_{i=1}^n x_i(s^*)v_i' < (1-3\epsilon)\left( 1-\frac{2\zeta n}{pb} \right) OPT |T \leq n } \leq \epsilon.$$
\end{lemma}
\begin{proof}
Conditioned on $T \leq n$, and thus at the end of the algorithm, $s^*$ is well defined and at the end of the algorithm, $S_{\RGS}$ is a random set of size $\zeta n$ (whose indices are $\{\sigma_\RGS^{-1}(i)| i \in S_{\RGS}\}$ before the random permutation) from $\{ 1,2,\dots , n\}$.

Recall that $OPT$ is the optimal offline value to the original problem.
Denote $OPT'$ the optimal offline value when the inventory is $b' \triangleq \left(b-\frac{2\zeta n}{p}\right)$ rather than $b$ and the binary decision variables are relaxed as continuous $[0,1]$ variables.
Then, by multiplying the offline optimal solution by the same factor of $\left( 1-\frac{2\zeta n}{pb} \right)$, we have, $OPT' \geq (1-\frac{2\zeta n}{pb}) OPT$.

Consider Lemma $3$ of \cite{Agrawal2009b} with $m=1$ and $b'=\left(b-\frac{2\zeta n}{p}\right) $.
The condition for the lemma holds because $b'=\left(b-\frac{2\zeta n}{p}\right) \geq \frac{b}{2} \geq \frac{6m \log\left(\frac{n}{\epsilon}\right)}{\epsilon^3}.$
Therefore, with proability at least $1-\epsilon$,
$$\sum_{i=1}^n x_i(s^*)v_i' \geq (1-3\epsilon) OPT' \geq (1-3\epsilon)\left( 1-\frac{2\zeta n}{pb} \right) OPT.$$
\end{proof}
With Lemmas~\ref{lemma:stopping_time},~\ref{lemma:inventory_constraint_hold_high_probabiliy}, and~\ref{lemma:high_objective_value}, we can complete the proof of the theorem as follows:
%%
%%
%%are ready to prove Theorem~\ref{theorem:one-time-learning}, which we reiterate as follows:
%%\begin{reptheorem}{theorem:one-time-learning}
%%If $\epsilon < \frac{p}{2}$ and $b^2 \geq \frac{12n}{p \epsilon ^3} \log \left( \frac{n+1}{\epsilon}\right)$, then the one-time learning algorithm is $(1-O(\epsilon))-$competitive in the partially \st{\st{learnable}\vmn{predictable}}\vmn{predictable} model.
%%\end{reptheorem}
%\begin{proof}
First we show that can apply Lemmas~\ref{lemma:stopping_time},~\ref{lemma:inventory_constraint_hold_high_probabiliy}, and~\ref{lemma:high_objective_value}.
To apply Lemma~\ref{lemma:stopping_time}, we need
$n\geq \frac{4}{\zeta}\log \frac{1}{\epsilon}$, which is equivalent to $b \geq \frac{8}{p \epsilon}\log \left( \frac{1}{\epsilon}\right)$, which is implied by our assumption $b^2 \geq \frac{12n}{p \epsilon ^3} \log \left( \frac{n+1}{\epsilon}\right)$.
To apply Lemma~\ref{lemma:inventory_constraint_hold_high_probabiliy}, we need
$ \frac{2\zeta n}{p} \leq \frac{b}{2} $ and $ \frac{2\zeta n}{p} \leq \frac{b}{2} $.
The first condition is implied by $ \frac{2\zeta n}{p} =\epsilon b < \epsilon \frac{p}{2} <\frac{b}{2} $.
The second condition is equivalent to our assumption $b^2 \geq \frac{12n}{p \epsilon ^3} \log \left( \frac{n+1}{\epsilon}\right)$.
To apply Lemma~\ref{lemma:high_objective_value}, we need $b \geq \frac{12 \log\left(\frac{n}{\epsilon}\right)}{\epsilon^3}$, which is implied by our assumption $b^2 \geq \frac{12n}{p \epsilon ^3} \log \left( \frac{n+1}{\epsilon}\right)$.

Applying the union bound to Lemmas~\ref{lemma:stopping_time},~\ref{lemma:inventory_constraint_hold_high_probabiliy}, and~\ref{lemma:high_objective_value}, we have with probability at least $1-3\epsilon$, $T<\frac{2\epsilon n}{p}(<n)$ (and thus $s^*$ is computed), $\sum_{i=1}^n x_i(s^*)\leq b-\frac{2\zeta n}{p} < b-T$ (and thus we do not run out of inventory), and $\sum_{i=1}^n x_i(s^*)v_i' \geq (1-3\epsilon)\left( 1-\frac{2\zeta n}{pb} \right) OPT$.
This implies that with probability at least $1-3\epsilon$, the total revenue achieved by the online algorithm is at least $(1-3\epsilon)\left( 1-\frac{2\zeta n}{pb} \right) OPT$.
Therefore, the competitive ratio is at least $(1-3\epsilon)^2\left( 1-\frac{2\zeta n}{pb} \right) = (1-3\epsilon)^2\left( 1-\epsilon \right) \geq 1- 7 \epsilon$, which completes the proof.
\end{proof}

\fi

%Now we prove Theorem~\ref{thm:b=1}, which we reiterate as follows:
%\begin{reptheorem}{thm:b=1}
%Under the partially \st{\st{learnable}\vmn{predictable}}\vmn{predictable} model where the probability a customer being at the {\RG} group is $p$, as $n\to \infty$, the success probability of $\text{OSA}_\gamma$ is $\gamma p \log \frac{1}{\gamma p + 1-p}.$
%\end{reptheorem}
\begin{proof}{\textbf{Proof of Theorem~\ref{thm:b=1}:}}
For each integer $k \geq 2$, we denote $\vmn{\mathcal{F}_k}$ the event satisfying all the following conditions:
\begin{enumerate}
\item All of the $k$ \vmn{highest-revenue} customers are in the {\RG} group.
\item The $k^{\text{th}}$-\vmn{highest-revenue} customer arrives in the observation period and the $k-1$ customers with the highest \st{values}\vmn{revenue} do not.
\item The \vmn{highest-revenue} customer arrives first among the $k-1$ customers with the highest \st{values}\vmn{revenue}.
\end{enumerate}
Clearly, for any $k\geq 2$, $\vmn{\mathcal{F}_k}$ is a success event and those events are mutually exclusive for different values of $k$.
Furthermore, we note that, conditioned on \st{that a customer is}\vmn{being} in the {\RG} group, the probability that \st{it}\vmn{a customer} arrives before time $\gamma$ approaches $\gamma$ as $n\to \infty$.
This can be done by using \vmn{a concentration result for random permutations} \st{a slightly modified version of}\vmn{similar to the one used in the proof of} \eqref{inequality:good-approximation-app--o^R_2} \st{that takes into account all customers instead of only \st{class}\vmn{type}-$2$ customers}.
\vmn{Further, for two customers $l$ and $\tilde{l}$, conditioned on being in  the {\RG} group, the events that customer $l$ arrives before time $\gamma$ and customer $\tilde{l}$ arrives after  time $\gamma$ are asymptotically independent.}
Thus, we can write the probability of event $\vmn{\mathcal{F}_k}$ as:
\begin{align*}
\prob{\vmn{\mathcal{F}_k}}= p^k \gamma (1-\gamma)^{k-1}\frac{1}{k-1} + o(1),
\end{align*}
where $o(1)$ represents a small (positive or negative) real number approaching zero as $n\to \infty$.
As a result,\st{$|\prob{\vmn{\mathcal{F}_k}}- p^k \gamma (1-\gamma)^{k-1}\frac{1}{k-1}| \leq p^k \frac{k}{k-1}|o(1)| $, which approaches $0$ as $n\to \infty$, so we can denote $\prob{\vmn{\mathcal{F}_k}} \geq p^k \gamma (1-\gamma)^{k-1}\frac{1}{k-1} - |o(1)|$.
Thus,} for any fixed $m$, we have
\begin{align*}
\prob{\text{success}}\geq &\sum_{k=2}^n\prob{\vmn{\mathcal{F}_k}} \geq \sum_{k=2}^m\prob{\vmn{\mathcal{F}_k}} \\
\geq & \sum_{k=2}^m \left( p^k \gamma (1-\gamma)^{k-1}\frac{1}{k-1} - |o(1)| \right)\\
\geq & \sum_{k=2}^m \left(p^k \gamma (1-\gamma)^{k-1}\frac{1}{k-1}\right) - m|o(1)|,
\end{align*}
which approaches $\sum_{k=2}^m \left(p^k \gamma (1-\gamma)^{k-1}\frac{1}{k-1}\right)$, for fixed $m$, as $n \to \infty$.
Since the above inequality holds for all $m$, we have
\begin{align}
\label{eq:LB}
 \prob{\text{success}} \geq \lim_{n\to \infty }\sum_{k=2}^n \prob{\vmn{\mathcal{F}_k}} \geq \lim_{m \to \infty}\sum_{k=2}^m p^k \gamma (1-\gamma)^{k-1}\frac{1}{k-1}=\gamma p \log \frac{1}{\gamma p + 1-p}.
\end{align}
\if false
For simplicity, in the remaining proofs related to this algorithm, we skip the above argument, and simply write the probabity that a customer (given that it is in the {\RG} group) arrives between time $\lambda_1$ and $\lambda_2(>\lambda_1)$ to be $\lambda_2-\lambda_1$.
\fi

%${\vec{v'}\vmn{ZZ\vec{v}_I}}$

\vmn{Next, }we show that the \st{analysis}\vmn{lower bound given in \eqref{eq:LB}} is tight \vmn{by presenting an instance for which $\text{OSA}_\gamma$ achieves a success probability of at most $\gamma p \log \frac{1}{\gamma p + 1-p}$}.
Consider the following adversarial instance.
The \vmn{highest-revenue} customer is the first customer in $\st{\vec{v'}}\vmn{\vec{v}_I}$ \vmn{(As a reminder, the subscript $I$ indicates that this is the initial sequence determined by the adversary.)}\st{(before the random permutation.)}
For each $k=2, 3, \dots, (1-\gamma)n+1$, the $k^{\text{th}}$-\vmn{highest-revenue} customer is the $(\gamma n + k-1)^{\text{th}}$ customer in $\st{\vec{v'}}\vmn{\vec{v}_I}$.
Other customers arbitrarily fill other positions in $\st{\vec{v'}}\vmn{\vec{v}_I}$.

For each positive integer $1\leq l\leq (1-\gamma)n$, we denote $\vmn{\mathcal{H}_l}$ the event where all of the $l$ \vmn{highest-revenue} customers are in the {\RG} group but the $(l+1)^{\text{th}}$-highest is not.
Similarly, we denote $\vmn{\mathcal{H}_{(1-\gamma)n+1}}$ the event where all of the $((1-\gamma)n+1)$ \vmn{highest-revenue} customers are in the {\RG} group; \vmn{finally, we denote} $\vmn{\mathcal{H}_0}$ the event where the \vmn{highest-revenue} customer is in the {\UPG} group.
Clearly, $\{ \vmn{\mathcal{H}_l} \}_{l=0}^{(1-\gamma)n+1}$ is a partition of the sample space.
In addition, for all $l$, $\prob{\vmn{\mathcal{H}_l}} \leq p^l$.
Conditioned on $\vmn{\mathcal{H}_0}$, the \vmn{highest-revenue} customer arrives during the observation period, and thus the algorithm has a success probability of $0$.
Conditioned on $\vmn{\mathcal{H}_1}$, the algorithm either accepts the customer with the second-highest value or does not accept any customer (if the \vmn{highest-revenue} customer arrives before $\gamma$), and hence \vmn{again it }has a success probability of $0$.
Further,\st{it is easy to check that} for any $2\leq l \leq (1-\gamma)n$, conditioned on $\vmn{\mathcal{H}_l}$, to have a success, either one of \vmn{the events} $\vmn{\mathcal{F}}_2, \vmn{\mathcal{F}}_3 , \dots \vmn{\mathcal{F}}_l$ occurs or the \vmn{highest-revenue} customer \st{must be one of the}\vmn{arrives between time} $(\gamma n+1)$ \st{through}\vmn{and} $(\gamma n+l-1)$; \vmn{note that conditioned on $\vmn{\mathcal{H}_l}$, the latter has a probability of $\frac{l-1}{n}$.}\st{arrived customers, which has a probability of $\frac{l-1}{n}$ (conditioned on $\vmn{\mathcal{H}_l}$).} As a result, the total success probability is at most
\begin{align*}
& \sum_{l=2}^{(1-\gamma ) n} \prob{\vmn{\mathcal{H}_l}}\left[ \prob{ \bigcup_{m=2}^l\vmn{\mathcal{F}}_m | \vmn{\mathcal{H}_l}} +\frac{l-1}{n} \right] + \prob{\vmn{\mathcal{H}}_{(1-\gamma ) n+1}} \\
\leq & \sum_{l=2}^{(1-\gamma ) n} \sum_{m=2}^{l} \prob{\vmn{\mathcal{F}}_m \cap \vmn{\mathcal{H}_l}} + \sum_{l=2}^{(1-\gamma ) n} p^l \frac{l-1}{n} + p^{(1-\gamma ) n+1} &( \prob{\vmn{\mathcal{H}}_l} \leq p^l) \\
\leq & \sum_{m=2}^{\infty} \prob{\vmn{\mathcal{F}}_m} + \sum_{l=2}^{\infty}p^l \frac{l-1}{n} + p^{(1-\gamma ) n+1}\\
= &\sum_{l=2}^{\infty} \prob{\vmn{\mathcal{F}}_l} + \frac{p^2}{(1-p)^2 n}+ p^{(1-\gamma ) n+1},
\end{align*}
which converges to $\sum_{l=2}^{\infty} \prob{\vmn{\mathcal{F}}_l}$ as $n$ approaches infinity.
\end{proof}

\begin{proof}{\textbf{Proof of Proposition~\ref{thm:randomized-b=1}:}}
The key in the proof of the proposition
%Proposition~\ref{thm:randomized-b=1}
is that when the position of the second-\vmn{highest-revenue} customer in $\vmn{\vec{v}_I}$ is before $\gamma_2 n$, $OSA_{\gamma_2}$ has a success probability greater than $s_2$; otherwise, $OSA_{\gamma_1}$ has a success probability greater than $s_1$.
To formalize this idea, we introduce  two lemmas.
\begin{lemma}\label{lemma:random-OSA-good-gamma-2}
If the second-\vmn{highest-revenue} customer is among the first $\gamma_2 n$ customers in $\vmn{\vec{v}_I}$, then $OSA_{\gamma_2}$ has a success probability of at least $s_2+p(1-p)(1-\gamma_2)$\vmn{, when $n\to \infty$}.
\end{lemma}
\begin{proof}{\textbf{Proof:}}
\vmn{Note that}
the events $\{\vmn{\mathcal{F}_k}\}_{k=2}^\infty$ \st{introduced}\vmn{defined} in the proof of Theorem~\ref{thm:b=1} collectively give an \vmn{asymptotic} success probability of $s_2$.
\vmn{We identify}\st{Now we consider} \vmn{another disjoint event which also results in a success.}
\vmn{In particular, we define event $\vmn{\mathcal{\bar F}}$ that satisfies the following conditions:}
\st{the event, denoted by $\mathcal{\bar E}$, which satisfies the following conditions:}
\begin{enumerate}
\item The \vmn{highest-revenue} customer is in the {\RG} group and arrives after time $\gamma_2$.
\item The second-\vmn{highest-revenue} customer is in the {\UPG} group.
\end{enumerate}
Note that $\vmn{\mathcal{\bar F}}$ is a success event that \st{does not overlap with events}\vmn{is disjoint from} $\{\vmn{\mathcal{F}_k}\}_{k=2}^\infty$.
Therefore, $\vmn{\mathcal{\bar F}}$ gives an additional success probability of $p(1-p)(1-\gamma_2) + o(1)$.
\end{proof}

\begin{lemma}\label{lemma:random-OSA-good-gamma-1}
If the second-\vmn{highest-revenue} customer is not among the first $\gamma_2 n$ customers in $\vmn{\vec{v}_I}$, then $OSA_{\gamma_1}$ has a success probability of at least $s_1+(1-p)\frac{\gamma_2-\gamma_1}{1-\gamma_1}s_1$\vmn{, when $n\to \infty$}.
\end{lemma}
\begin{proof}{\textbf{Proof:}}
\vmn{Similar to the previous lemma, we first note that} the events $\{\vmn{\mathcal{F}_k}\}_{k=2}^\infty$ introduced in the proof of Theorem~\ref{thm:b=1} collectively give an \vmn{asymptotic} success probability of $s_1$. \vmn{We identify another set of disjoint events that also results in success.}
\vmn{In particular,} for each positive integer $k\geq 2$, we \st{consider the event, denoted by}\vmn{define the event} $\vmn{\mathcal{\widehat{F}}_k}$ that satisfies all of the following conditions:
\begin{enumerate}
\item Among the $k$ \vmn{highest-revenue} customers, all except the second \vmn{highest-revenue} customer are in the {\RG} group and arrive after time $\gamma_1$.
\item The second-\vmn{highest-revenue} customer is in the {\UPG} group.
\item The $(k+1)^{\text{th}}$-\vmn{highest-revenue} customer is in the {\RG} group and arrives before time $\gamma_1$.
\item The \vmn{highest-revenue} customer arrives no later than $\gamma_2 $ and arrives first among the $k$ \vmn{highest-revenue} customers (except for the second-\vmn{highest-revenue} customer).
\end{enumerate}
Clearly, for any $k\geq 2$, $\vmn{\mathcal{\widehat{F}}_k}$ is a success event.
In addition, those events are mutually exclusive for different values of $k$.
Furthermore, $\{\vmn{\mathcal{\widehat{F}}_k}\}_{k=2}^\infty$ does not overlap $\{ \vmn{\mathcal{F}_k} \}_{k \geq 2}$.
Note that \vmn{the probability that the \vmn{highest-revenue} customer arrives between time $\gamma_1$ and $\gamma_2$,} and  it arrives first among the $k$ \vmn{highest-revenue} customers (except for the second-\vmn{highest-revenue} customer) is at least \vmn{$\frac{\gamma_2 - \gamma_1}{k-1}+o(1)$}.\st{conditioning on the \vmn{highest-revenue} customer arriving before time $\gamma_2 $, the probability that it arrives first among the $k$ \vmn{highest-revenue} customers (except for the second-\vmn{highest-revenue} customer) is at least $\frac{1}{k-1}$.}
Therefore, $\prob{\vmn{\mathcal{\widehat{F}}_k}}\geq p^k (1-p) \gamma_1 (1-\gamma_1)^{k-2}(\gamma_2 - \gamma_1)\frac{1}{k-1} \vmn{+ o(1)}$.
As a result, the \vmn{asymptotic} success probability is at least
\begin{align*}
s_1 + \sum_{k=2}^\infty \prob{\vmn{\mathcal{\widehat{F}}_k}}
\geq & s_1+\sum_{k=2}^\infty p^k (1-p) \gamma_1 (1-\gamma_1)^{k-2}(\gamma_2 - \gamma_1)\frac{1}{k-1}
\\ = & s_1 + (1-p)\frac{\gamma_2-\gamma_1}{1-\gamma_1}s_1 .\end{align*}
\end{proof}

With Lemmas~\ref{lemma:random-OSA-good-gamma-2} and~\ref{lemma:random-OSA-good-gamma-1}, we complete the proof of the proposition as follows:
%
%are ready to prove Proposition~\ref{thm:randomized-b=1}, which we reiterate as follows:
%\begin{repproposition}{thm:randomized-b=1}
%For any $0\leq \gamma_1 < \gamma_2 \leq 1$ and $0<q<1$, the randomized algorithm that runs $\text{OSA}_{\gamma_1}$ with probability $q$ and $\text{OSA}_{\gamma_2}$ with probability $1-q$ has a success probability of at least $$ q s_1 +(1-q) s_2 +\min \left\{ (1-q)p(1-p)(1-\gamma_2), q(1-p)\frac{\gamma_2-\gamma_1}{1-\gamma_1}s_1 \right\} $$
%where for $i=1,2$, $p_i$ denotes the success probability of ${OSA}_{\gamma_i}$.
%\end{repproposition}
%\begin{proof}
Recall that for any position of the second-\vmn{highest-revenue} customer in $\vmn{\vec{v}_I}$, $OSA_{\gamma_1}$ has a success probability of at least $s_1$ and $OSA_{\gamma_2}$ has a success probability of at least $s_2$.
As a result, using Lemma~\ref{lemma:random-OSA-good-gamma-2}, if the second-\vmn{highest-revenue} customer is among the first $\gamma_2 n$ customers in $\vmn{\vec{v}_I}$, then we have a success probability of at least $qs_1 + (1-q)(s_2+p(1-p)(1-\gamma_2))$.
Similarly, using Lemma~\ref{lemma:random-OSA-good-gamma-1}, if the second-\vmn{highest-revenue} customer is not among the first $\gamma_2 n$ customers in $\vmn{\vec{v}_I}$, then we have a success probability of at least $q(s_1+(1-p)\frac{\gamma_2-\gamma_1}{1-\gamma_1}s_1)+(1-q)s_2$.
As a result, for any adversarial problem instance $\vmn{\vec{v}_I}$, \st{our} \vmn{the asymptotic} success probability is at least
\begin{align*} & \min \left \{ qs_1 + (1-q)(s_2+p(1-p)(1-\gamma_2)) , q(s_1+(1-p)\frac{\gamma_2-\gamma_1}{1-\gamma_1}s_1)+(1-q)s_2 \right \} \\
= & q s_1 +(1-q) s_2 +\min \left\{ (1-q)p(1-p)(1-\gamma_2), q(1-p)\frac{\gamma_2-\gamma_1}{1-\gamma_1}s_1 \right\} ,
\end{align*}
which completes the proof.
%\end{proof}
\end{proof}

%\end{document}

%%%%%%%%%%%%%%%%%
\end{document}